\documentclass{report}
\usepackage[T1]{fontenc}
\usepackage[utf8]{inputenc}
\usepackage{amsthm}
\usepackage{amssymb}
\usepackage{amsfonts}
\usepackage{bm}
\usepackage{graphicx}
\usepackage{float}
\usepackage{mathtools}
\usepackage{pifont}
\usepackage{xcolor}
\usepackage[disable]{todonotes}
\usepackage{proof}
\usepackage{alltt}
\usepackage{enumerate}
\usepackage{multicol}
\usepackage{lineno}
\usepackage{quiver}
\usepackage{adjustbox}
\usepackage{stmaryrd}
\usepackage[multiple]{footmisc}
\usepackage{wasysym}
\usepackage{hyperref}

\newcommand{\In}{\mathtt{input}}

\newcommand{\No}{\textrm{No}}
\newcommand{\State}{\mathtt{state}}
\newcommand{\Trees}{\textrm{Trees}}
\newcommand{\Yes}{\textrm{Yes}}
\newcommand{\abstr}[1]{\langle #1 \rangle}

\newcommand{\assoc}{\mathtt{assoc}}
\newcommand{\atoms}{\mathbb{A}}

\newcommand{\bigpicc}[1]{\begin{figure}[H]\centering\includegraphics[width=\textwidth]{#1.pdf} \end{figure}}
\newcommand{\bigpicenum}[1]{\begin{minipage}[t]{\linewidth}
  \raggedright
  \adjustbox{valign=t}{%
     \includegraphics[width=\textwidth]{#1}
  }
  \end{minipage}
  }

\newcommand{\coassoc}{\mathtt{coassoc}}
\newcommand{\comma}{\textrm{,}}
\newcommand{\coproj}{\mathtt{coproj}}
\newcommand{\core}{\textrm{core}}
\newcommand{\cosym}{\mathtt{cosym}}
\newcommand{\ctx}{\textrm{ctx}}
\newcommand{\distr}{\mathtt{distr}}
\newcommand{\eqto}{\to_{\textrm{eq}}}

\newcommand{\fsto}{\to_{\textrm{fs}}}
\newcommand{\fs}{\textrm{fs}}

\newcommand{\orf}{\mathtt{?}}

\newcommand{\longsuto}{\relbar\joinrel\relbar\joinrel\suto}

\newcommand{\longtransform}[1]{\ \stackrel{\underrightarrow{#1}}{}\ }

\renewcommand{\merge}{\mathtt{merge}}
\newcommand{\nat}{\mathbb{N}}
\newcommand{\no}{\mathtt{No}}

\newcommand{\orbitf}{\mathtt{orbit}}
\newcommand{\orb}{\textrm{orb}}

\newcommand{\pfs}{P_{\textrm{fs}}}
\newcommand{\picc}[1]{\begin{figure}[H]\centering\includegraphics[width=0.8\textwidth]{#1.pdf} \end{figure}}

\newcommand{\proj}{\mathtt{proj}}

\newcommand{\reg}{\mathtt{r}}

\newcommand{\rightDistr}{\mathtt{rightDistr}}
\newcommand{\rightI}{\mathtt{rightI}}

\newcommand{\shuffle}{\mathtt{shuffle}}
\newcommand{\smallpicc}[1]{\begin{figure}[H]\centering\includegraphics[width=0.5\textwidth]{#1.pdf} \end{figure}}
\newcommand{\custompicc}[2]{\begin{figure}[H]\centering\includegraphics[width=#2\textwidth]{#1.pdf} \end{figure}}
\newcommand{\vsmallpicc}[1]{\begin{figure}[H]\centering\includegraphics[width=0.3\textwidth]{#1.pdf} \end{figure}}
\newcommand{\vvsmallpicc}[1]{\begin{figure}[H]\centering\includegraphics[width=0.2\textwidth]{#1.pdf} \end{figure}}

\newcommand{\supp}{\textrm{supp}}
\newcommand{\suto}{\multimap}
\newcommand{\sutransform}[1]{\stackrel{#1}{\longsuto}}

\newcommand{\swap}[2]{\sigma_{(#1 \  #2)}}
\newcommand{\sym}{\mathtt{sym}}
\newcommand{\tdot}{\textrm{.}}
\newcommand{\tofs}{\to_{\textrm{fs}}}
\newcommand{\transform}[1]{\stackrel{#1}{\longrightarrow}}
\newcommand{\treeRepr}{\mathtt{treeRepr}}
\newcommand{\yes}{\mathtt{Yes}}
\newcommand{\copyf}{\mathtt{copy}}
\newcommand{\leftI}{\mathtt{leftI}}

\newcommand{\eqf}{\mathtt{eq}}
\newcommand{\idf}{\mathtt{id}}
\newcommand{\const}{\mathtt{const}}
\newcommand{\constI}{\mathtt{constI}}
\newcommand{\msupp}{\textrm{msup}}
\newcommand{\suppt}{\overline{\supp}}
\newcommand{\poforb}{\overline{\textrm{orb}}}

\newcommand{\map}{\mathtt{map}}
\newcommand{\cons}{\mathtt{cons}}
\newcommand{\destruct}{\mathtt{destruct}}
\newcommand{\concatf}{\mathtt{concat}}
\newcommand{\reverse}{\mathtt{reverse}}
\newcommand{\blocks}{\mathtt{blocks}}
\newcommand{\group}{\mathtt{group}}

\newcommand{\concat}{\mathtt{concat}}
\newcommand{\singleton}{\mathtt{singleton}}

\newcommand{\toprightarrow}{\reflectbox{\rotatebox[origin=c]{180}{$\Rsh$}}}

\DeclareUnicodeCharacter{1EC5}{\eth}

\makeatletter
\newtheorem*{rep@theorem}{\rep@title}
\newcommand{\newreptheorem}[2]{%
\newenvironment{rep#1}[1]{%
 \def\rep@title{#2 \ref{##1}}%
 \begin{rep@theorem}}%
 {\end{rep@theorem}}}
\makeatother

\newtheorem{theorem}{Theorem}
\newtheorem{lemma}{Lemma}
\newreptheorem{lemma}{Lemma}

\newtheorem{claim}{Claim}
\newtheorem*{theorem*}{Theorem}

\theoremstyle{definition}
\newtheorem{examplex}{Example}
\newenvironment{example}
  {\pushQED{\qed}\examplex}
  {\popQED\endexamplex}
\newtheorem{definitionx}{Definition}
\newenvironment{definition}
  {\pushQED{\qed}\definitionx}
  {\popQED\enddefinitionx}

\title{The single-use restriction for register automata and transducers over infinite alphabets}

\author{Rafał Stefański}
\date{June 2021}

\begin{document}

\maketitle

\tableofcontents

\chapter*{Acknowledgements}

I would like to start by directing words of gratitude to the entire community of the Faculty of Mathematics, Informatics and Mechanics of the University of Warsaw, among whom I have had the opportunity to develop as a computer scientist over the past 10 years. Thanks to them, 
almost every day spent physically at work is a real pleasure. I especially want to thank Mikołaj Bojańczyk -- I could not have imagined a better supervisor. Further thanks go to Edon Kelmendi for introducing me to the world of theoretical computer science and helping me write my first review; to Nathan Lhote for our invaluable cooperation during the lockdown; and to Janusz Schmude for our joint PhD adventure.\\

For the past year, I had the opportunity to work in the PPLV group at University College London. It was a very valuable time for me. I sincerely thank everyone with whom I had the opportunity to share it. I would like to give 
special thanks to Samson Abramsky for getting me excited to work with category theory, for his words of advice, and for his patience.\\

I would like to express my gratitude to Diego Figueira, Emmanuel Filiot, and Andrzej Murawski for their time and effort dedicated to reviewing this thesis and providing valuable feedback.
I owe special thanks to Emmanuel Filiot for his extensive, detailed and insightful comments that have significantly improved the current version of this thesis.\\ 

I would also like to thank the community associated with the 14th High School in Warsaw, especially to my counselling teachers -- the late Stanislaw Lipiński and the late Jack Banasik; my math teachers, Jerzy Konarski and Filip Smętek; and my computer science teachers, Hanna Stachera and Joanna Śmigielska. I would also like to thank Maciej Matraszek for the work he put into running the computer science club.\\

Recalling my period of high school, it is impossible not to appreciate the people involved in the organization of the Polish Olympiad in Informatics and the editors of the monthly magazine Delta. They opened up new perspectives for me and steered me towards the academic world. At this point, I would also like to thank Damian Niwiński for his article ``Impossible Shortcut''\footnote{Article in Polish: \url{https://deltami.edu.pl/temat/matematyka/logika/2012/07/30/Niemozliwy_skrot/}}.\\

A very important stage in my education was the time spent at the 16th Primary and Middle School in Warsaw. My further thanks go to all my teachers from that period. Especially to Agnieszka Nowińska-Samsel for introducing me to the principles of mathematical rigour, and to Ewa Kietlińska-Zaleska for believing in my abilities.\\

I warmly thank my parents Ewa Stefańska and Grzegorz Stefański for the love, inspiration and effort they put into my education, and my grandmothers Teresa Rybak and Monika Stefańska for their care and devotion.\\

I save my final words of gratitude for my future wife Klaudia Grygoruk: Thank you for being who you are. 
\chapter*{Podziękowania}
Podziękowania chciałbym rozpocząć od skierowania słów wdzięczności do całej społeczności Wydziału Matematyki, Informatyki i Mechaniki Uniwersytetu Warszawskiego, pośród której przez ostatnie 10 lat miałem okazję rozwijać się jako informatyk. Dzięki wszystkim tworzącym ją ludziom prawie każdy dzień spędzony fizycznie w pracy jest prawdziwą przyjemnością. Szczególnie pragnę podziękować Mikołajowi Bojańczykowi – trudno mi sobie wyobrazić lepszego promotora. Dalsze podziękowania kieruję do Edona Kelmendi za wprowadzenie mnie w tajniki świata teoretycznej informatyki i pomoc w napisaniu pierwszej recenzji; do Nathana Lhote'a  za nieocenioną współpracę podczas lockdownu oraz do Janusza Schmude za wspólną doktorancką przygodę.\\

Przez ostatni rok miałem okazję pracować w grupie PPLV na University College London. Był to dla mnie bardzo cenny czas. Serdecznie dziękuję wszystkim, z którymi dane mi było go spędzić. Szczególne podziękowania kieruję pod adresem Samsona Abramsky’ego za inspirację do pracy z teorią kategorii, dobre rady oraz cierpliwość.\\

Jestem bardzo wdzięczny Diego Figueirze, Emmanuelowi Filiot i Andrzejowi Murawskiemu za czas
i wysiłek włożony w zrecenzowanie tej rozprawy oraz za ich wartościowe uwagi. Szczególne wyrazy uznania kieruję do
Emmanuela Filiot za jego obszerne, szczegółowe i wnikliwe komentarze, które wzbogaciły ostateczną wersję tej pracy.\\ 

Słowa podziękowań należą się również środowisku związanemu z XIV L.O. im. Stanisława Staszica w Warszawie, a zwłaszcza moim wychowawcom -- ś.p. Stanisławowi Lipińskiemu oraz ś.p. Jackowi Banasikowi; nauczycielom matematyki – Jerzemu Konarskiemu oraz Filipowi Smętkowi, a także nauczycielom informatyki – Hannie Stacherze i Joannie Śmigielskiej. Ponadto chciałbym podziękować Maciejowi Matraszkowi za pracę włożoną w prowadzenie kółka informatycznego.\\

Wspominając okres liceum, nie sposób nie docenić osób zaangażowanych w organizację Olimpiady Informatycznej oraz redakcję miesięcznika ,,Delta''. To one otworzyły przede mną nowe perspektywy i pokierowały mnie w stronę świata akademickiego. W tym miejscu chciałbym również podziękować Damianowi Niwińskiemu za jego artykuł ,,Niemożliwy Skrót''\footnote{\url{https://deltami.edu.pl/temat/matematyka/logika/2012/07/30/Niemozliwy_skrot/}}\\

Bardzo ważnym etapem mojego rozwoju był czas spędzony w Społecznej Szkole Podstawowej i Społecznym Gimnazjum nr 16 STO. Kolejne podziękowania kieruję do wszystkich nauczycieli obu tych placówek, w tym w sposób szczególny do Agnieszki Nowińskiej-Samsel za wprowadzenie mnie w zasady matematycznego rygoru oraz do Ewy Kietlińskiej-Zaleskiej, która zawsze wierzyła w moje możliwości.\\

Gorąco dziękuję moim rodzicom Ewie i Grzegorzowi Stefańskim za miłość, inspirację i wysiłek włożony w moje wykształcenie oraz moim babciom Teresie Rybak i Monice Stefańskiej za troskę i oddanie.\\

Wreszcie, wieńczące słowa kieruję do mojej przyszłej żony Klaudii Grygoruk: Dziękuję Ci za to kim jesteś. 
\chapter*{Introduction}
\section*{Regular languages}
Regular languages are great. One reason for this is their remarkable
definitional robustness: 
They have many substantially different definitions, all of which turn out
to be equivalent. Examples include one-way deterministic finite automata, two-way deterministic 
finite automata, finite monoids and  \textsc{mso} logic on words. (Other well-known 
definitions include regular expressions and nondeterministic automata, but we do not mention them 
because they do not fit the narrative of this thesis.)\\

Apart from being aesthetically 
pleasing, having so many equivalent definitions can be used to simplify proofs and algorithms. For example, 
the simplest way to test a regular language,
for nonemptiness, is to look at the language's automaton representation 
and check if some accepting state is reachable.
Thanks to the equivalence of the definitions, this gives us an
algorithm for every other representation of a regular language.
On the other hand, the declarative syntax of \textsc{mso} formulas
is usually more convenient for defining properties than the 
operational, low-level syntax of automata.
Also monoids and two-way automata have their advantages --
e.g. monoids play an important role in the renowned 
Krohn-Rhodes theorem\footnote{See \cite[Equation~2.2]{krohn1965prime} or Theorem~\ref{thm:classical-kr}.},
and two-way automata are very useful
when defining regular transductions\footnote{See \cite{engelfriet2001mso} or 
Item~3 in the introduction to Chapter~3.}.\\

These desirable properties 
of regular languages have inspired an effort to extend them to object 
such as trees, graphs or words over infinite alphabets.  
This line of research has already seen many important results: for example,
B\"uchi showed that automata and \textsc{mso} coincide for $\omega$-words,
and Rabin showed the same for infinite trees.
More recent examples include work on regular languages of graphs\footnote{See \cite{courcelle2012graph}.}
or regular languages over abstract monads\footnote{See \cite{bojanczyk2020languages}.}.\\

The direction relevant to this thesis is the study of infinite alphabets. 
This line of research, started by Kaminski and Francez in the paper
\cite{kaminski_finite_memory_paper}, has proved itself to be challenging.
The study of regularity over infinite alphabets can be divided into three phases:

\begin{enumerate}
\item \textbf{Register automata.} This phase was initiated by \cite{kaminski_finite_memory_paper}.
The idea is to equip finite automata with registers so that they can store 
the input letters and compare them with each other\footnote{See Section~\ref{sec:dra} for details.}.
(In this model, letters from the infinite alphabet can only be compared for equality; this restriction will be true for all models discussed in this thesis.) However, it was quickly\footnote{Some of the inequivalence results are already present in \cite{kaminski_finite_memory_paper}, 
see \cite{neven2004finite} for a complete overview. A few of them are also presented in Section~\ref{sec:dofa-variants} of this thesis.}
discovered that most of the models of language recognizers over infinite alphabets are pairwise inequivalent, with only trivial inclusions being valid. In particular 
one-way register automata, two-way register automata, orbit-finite monoids\footnote{
    A natural extension of finite monoids 
    for infinite alphabets. See \cite[Defnition 3.1]{bojanczyk2013nominal}. 
}, and \textsc{mso}$^\sim$
logic\footnote{
    A natural extension of the \textsc{mso} logic 
    for infinite alphabets. See \cite[Section~2.4]{neven2004finite}.}
on words are pairwise nonequivalent.
Moreover, 
already two of those models -- two-way register automata and \textsc{mso}$^\sim$ -- 
have undecidable emptiness\footnote{
    For proof for two-way register automata see Section~\ref{subsec:2dra} or \cite[Theorem~5.3]{neven2004finite}. 
    For \textsc{mso}$^\sim$, see \cite[Theorem~3.2~and~Theorem~5.1]{neven2004finite} combined with the fact that \textsc{mso}$^\sim$ is closed under compositions.
}.
Those results oriented
the first phase towards finding possibly strongest 
models that still have decidable emptiness. Examples of such models include data automata\footnote{
    See \cite{bojanczyk2011two} or \cite[Section~2.1]{bojanczyk2019slightly}.
} and alternating automata with one register\footnote{
    See \cite{demri2009ltl}, \cite{jurdzinski2011alternating},
    and \cite[Section~1.3]{bojanczyk2019slightly}.
}.
For more details about the first phase, see \cite[Part~1]{bojanczyk2019slightly}.\\

\item \textbf{Orbit-finite automata}. This second phase was initiated by \cite{bojanczyk2013nominal} and \cite{lasota2014automata}, inspired by work on nominal sets \cite{pitts2013nominal}.
The general goal of this phase is to build a definitional framework, in which the various existing automata models for infinite alphabets become simply the same as the classical models for finite alphabets, except with new notions of finite sets and functions between them.
The appropriate 
framework turned out to be the already existing\footnote{
    See \cite{pitts2013nominal}, or Section~\ref{sec:sets-with-atoms} for details and more bibliographical notes.
} category of nominal sets with finitely supported functions (denoted as $X \fsto Y$), but enhanced 
with a novel notion of \emph{orbit-finiteness}.
Notable successes of this phase include 
Myhill-Nerode theorem and Angluin-style learning for 
deterministic orbit-finite automata\footnote{
    See \cite[Theorem~5.2]{lasota2014automata} and
    \cite{moerman2017learning}.
}. 
More recent results include the development 
of a theory of orbit-finite-dimensional vector spaces 
and decidability for equivalence 
of weighted register automata\footnote{See \cite{bojanczyk2021orbit} for details on both of those results.}.
For more details about the second phase, see \cite[Sections~3~and~5]{bojanczyk2019slightly}. \\

\item \textbf{Single-use automata}. Finally (what I hope can be described as) the third phase of studies of infinite alphabets was started 
by \cite{stefanski2018automaton} and \cite{single-use-paper} and
is further developed in this thesis
(with some of the ideas traceable back to \cite{bojanczyk2013nominal} and \cite{ley2015logics}).
The key observation is that we can recover the definitional robustness of regular languages 
by introducing the \emph{single-use restriction}. 
In the setting of register automata, this restriction
amounts to the requirement that every read access to a register should
have the side effect of destroying the register's value\footnote{See Section~\ref{sec:su-register-automata} for details. }.
In a more abstract setting, 
the single-use restriction means replacing finitely supported functions ($X \fsto Y$)
with \emph{single-use functions}\footnote{The notion of single-use functions is a novel contribution of this thesis. See Definition~\ref{def:single-use-functions}. } (denoted as $X \suto Y$). Under the single-use restriction,
   one recovers many of the equivalences that were true for finite alphabets but failed for infinite alphabets,
   including the equivalence of the following models: one-way and two-way automata, orbit-finite monoids, and
   a version of \textsc{mso}$^\sim$\footnote{
    Orbit-finite monoids were introduced in \cite{bojanczyk2013nominal}. 
    Rigidly-guarded \textsc{mso}$^\sim$ was introduced and shown to be equivalent 
    to orbit-finite monoids in \cite{ley2015logics}. One-way and two-way 
    single-use register automata were introduced in \cite{stefanski2018automaton}, 
    and were shown to be equivalent to each other and to orbit-finite monoids 
    in \cite{single-use-paper}. The definitions and the equivalence 
    proof for orbit-finite monoids, single-use one-way automata
    and single-use two-way automata can also be found in this thesis. 
    See Section~\ref{subsec:orbit-fintie-monoids}, Definition~\ref{ref:single-use-automaton}, 
    Definition~\ref{def:single-use-two-way-automataon}, and Theorem~\ref{thm:dsua-2dsua-ofm}.
    For a novel topological perspective see \cite{urbat2023nominal}.},
    which is the first such four-way equivalence result for infinite alphabets.
    Presenting and extending the research about the third phase is the main 
    contribution of this thesis. 
\end{enumerate}

\section*{Transductions}
Another important extension of the theory of regular languages is the study of 
transductions, i.e. functions of type $\Sigma^* \to \Gamma^*$. The theory 
of transducers is mainly interested in transduction classes of low complexity -- 
a good litmus test is the decidability of the equivalence problem
(i.e. given two transducers $f$ and $g$ from a given class, is it possible to check 
if $f(w) = g(w)$ for every $w$). 
The classification of transductions is finer than the one of languages. Relevant classes include:
\[
    \begin{tabular}{ccccc}
        Mealy machines & $\subseteq$ &\small \begin{tabular}{c}Letter-to-letter\\ rational functions \end{tabular} & $\subseteq$ & Regular functions
    \end{tabular}
\]
Below, we present a brief description of each of these classes.
For more details, see the introduction to Chapter~3.\\

\emph{Mealy machines}\footnote{See \cite{mealy1955method} or Item~1 in the introduction to Chapter~3.} are a version of DFAs where all states are accepting, and 
every edge is additionally labelled by an output letter. Each such machine implements a 
length-preserving function, where every input position is replaced by the output label of
the corresponding transition.
Here is an example of a Mealy machine that implements the following 
function:
\[\textrm{``Change every other a to b''} \in  \{\textrm{a, b}\}^* \to \{\textrm{a, b}\}^*\]
\vsmallpicc{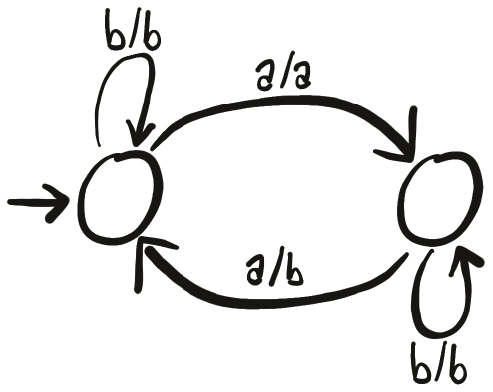}

\emph{Rational functions}\footnote{
    See \cite{eilenberg1974automata} or Item~2 in the introduction to Chapter~3.}
are all functions 
that can be expressed as a composition of a left-to-right Mealy machine 
and a right-to-left Mealy machine. Equivalent definitions include unambiguous nondeterministic Mealy machines
and Mealy machines with a regular look-ahead.
Here is an example of an unambiguous nondeterministic 
Mealy machine that implements the following function:
\[\textrm{``Swap the first and the last letter''} \in  \{\textrm{a, b}\}^* \to \{\textrm{a, b}\}^*\]
\smallpicc{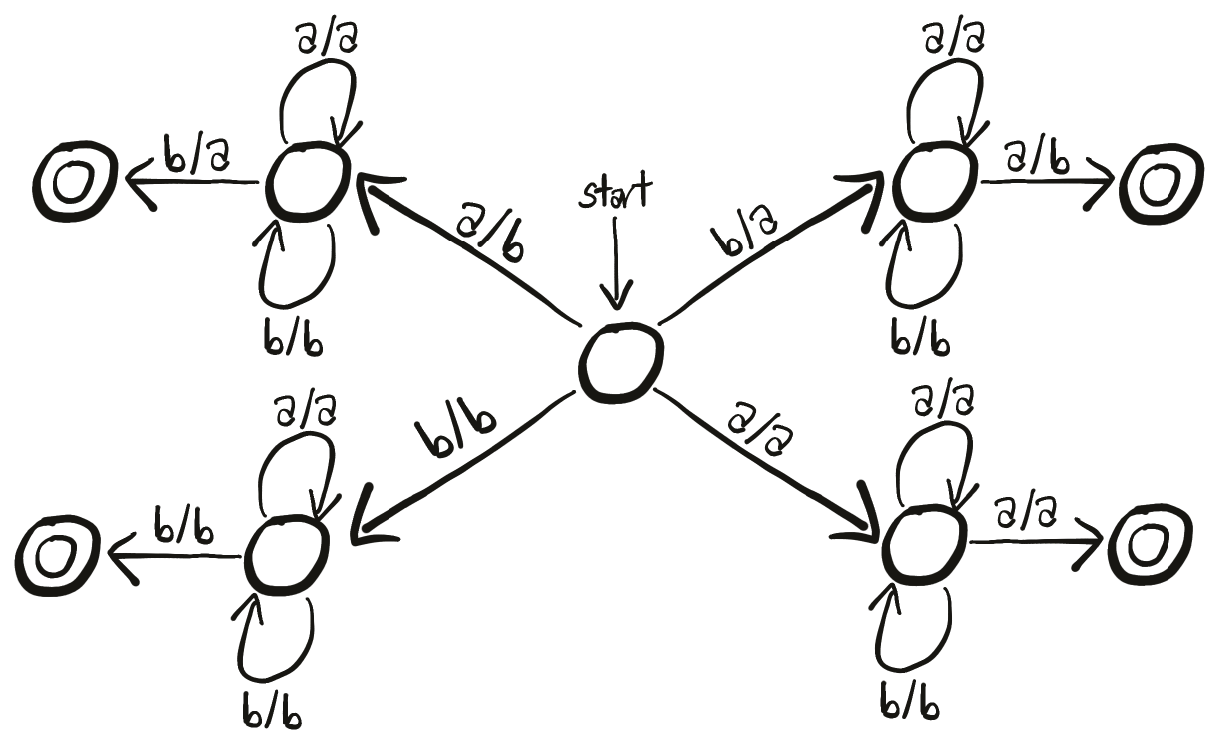}

\emph{Regular functions}\footnote{See \cite{engelfriet2001mso} 
Item~3 in the introduction to Chapter~3.} are all functions 
that can be implemented by two-way transducers, 
i.e. a version of two-way automata, 
where every transition may be additionally labelled with an output letter.
Regular functions exhibit a similar definitional robustness as 
regular languages. Equivalent definitions include
\textsc{mso}-transductions\footnote{See \cite[Section~2]{courcelle1994monadic}.},
string streaming transducers\footnote{See \cite[Section~3]{alur2010expressiveness} or Section~\ref{subsec:sst-def}.},
regular expressions with output\footnote{See \cite[Theorems~13~and~15]{alur2014combinators}.},
and many others.
Possibly for this reason, they have been the subject of significant research attention in recent years.\footnote{
    For examples, see \cite{dartois2017reversible}, \cite{dolatian2019learning}, \cite{chen2019decision}, \cite{nelson2020probing},
    or \cite{bojaczyk2023algebraic}.
}
Here is an example of a two-way transducer that implements the following 
regular function:
\[\textrm{``Reverse the input word''} \in  \{\textrm{a, b}\}^* \to \{\textrm{a, b}\}^*\]
\smallpicc{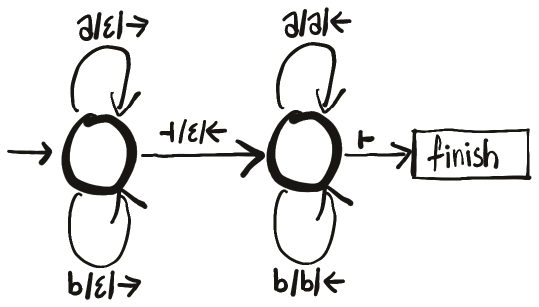}

Without a robust theory of regularity, one cannot hope for a good 
theory of transducers. In particular, deciding the equivalence
of two transducers is at least as hard as the emptiness problem 
for their underlying automaton model. For example,  
the equivalence problem for (multiple-use)
two-way register transducers is undecidable.
In Chapters 3 and 4 of this thesis, we show that under the single-use restriction, 
one recovers a robust theory of transducers for infinite alphabets.
The main results are as follows.\\

In Chapter~3, we define and study \emph{single-use register Mealy machines}.
In particular, we show that they:
\begin{enumerate}
    \item admit a Krohn-Rhodes-like decomposition\footnote{This was initially shown in \cite{single-use-paper}. In this thesis this is Theorem~\ref{thm:kr}.}; and
    \item have an equivalent algebraic definition\footnote{This is a novel contribution of this thesis. See Section~\ref{sec:monoids-and-transductions}.}.
\end{enumerate}
The Krohn-Rhodes decomposition theorem for single-use Mealy machine
is the main technical contribution of this thesis
(it constitutes around one fourth of its total volume).
In Chapter~3, we also define an infinite-alphabet version of the rational functions,
show that it also admits Krohn-Rhodes-like decompositions, and 
we develop a similar algebraic theory\footnote{See Section~\ref{sec:rational-transductions-with-atoms}.} 
for them.\\

Finally, in Chapter~4, we present the theory of single-use two-way transducers.
We prove that the deterministic two-way single-use register\footnote{All of the following results were initially shown in \cite{single-use-paper}.}:
\begin{enumerate}
    \item also admit a Krohn-Rhodes-like decomposition;
    \item are closed under compositions;
    \item have decidable equivalence;
    \item are equivalent to the single-use variant of copyless SSTs over infinite alphabets
          (modification of SSTs from \cite[Section 2.2]{alur2011streaming});
    \item are equivalent to the infinite alphabet variant of regular list functions
          (modification of \cite[Section 6]{bojanczyk2018regular}).
\end{enumerate}

We believe that the results presented in this thesis justify the name \emph{regular languages over infinite alphabets} for 
the class of languages recognized by single-use deterministic register automata,
and the name \emph{regular functions over infinite alphabets} 
for the class of functions definable by single-use two-way register transducers.\\

\subsection*{Contributions}
This thesis is based on a single publication \cite{single-use-paper},
with the objective of delivering its results in a cohesive narrative.
Additionally, this thesis contains two novel (unpublished before) contributions:
The first one is defining \emph{single-use functions}
and using them to simplify the definitions of different 
models of single-use automata and transducers. 
The second one is developing the algebraic theory 
for single-use Mealy machines and single-use rational functions, 
and using this theory to simplify the proofs of the Krohn-Rhodes decomposition 
theorems for single-use Mealy machines and single-use two-way transducers. 
The results presented in this thesis are an outcome
of a close collaboration with my advisor Mikołaj Bojańczyk, 
and it is mostly impossible to partition the contributions individually.
However, for the sake of strengthening my Ph.D. application,
it might be worth noting that the original idea for the single-use restriction,
which initiated this line of research, was mine.

\chapter{Infinite alphabets}
\label{ch:inifnite-alphabets}
This chapter serves two purposes: The first purpose is to put this thesis in context.
The chapter presents some previously studied models of computation over infinite alphabets.
The second purpose is to lay out foundations for the upcoming chapters.
We present the theory of \emph{sets with atoms} and
\emph{orbit finiteness}, which is the modern abstract vocabulary for discussing infinite alphabets.
In the chapter, we follow the narrative of \cite{bojanczyk2019slightly} -- we start with
a rather concrete model of \emph{deterministic register automata}
and we abstract away towards \emph{deterministic orbit-finite automata}. Then, we prove that
the two models are equivalent. Although this proof is quite technical,
we include it in the thesis, because it illustrates some useful techniques for working with sets with atoms.
Finally, in the last section we define some of the well-studied variants of register automata
(including orbit-finite monoids), and we compare their expressive powers.

\section{Deterministic register automata}
\label{sec:dra}

We start describing the model of \emph{deterministic register automata}
with some intuition and examples. 
Let us fix a countably infinite alphabet and call it $\atoms$. Introducing infinite sets to 
automata theory is dangerous: infinite alphabets can only be processed using infinite state spaces, 
and automata with unrestricted infinite state spaces can recognize all possible languages.
To avoid this, we say that elements of $\atoms$ can only be compared for equality. This means that we are only going to consider languages that 
can be defined in terms of equality, for example:
\begin{enumerate}
    \item $\{ w \in \atoms^* \ | \ \textrm{the first letter of } w \textrm{ appears again in } w \}$ 
    \item $\{ w \in \atoms^* \ | \ \textrm{the first and the last letter of } w \textrm{ are equal} \}$
    \item $\{ w \in \atoms^* \ | \ \textrm{there are at most three different letters in } w \}$.
\end{enumerate}
(For now, we only talk about the intuition, and we do not define formally what it means for a language 
to be definable only in terms of equality.) A deterministic register automaton is a model that can
recognize some languages of this type, including the three example languages from the list.
It has a finite set of control states
and a finite set of registers, in which it stores some of the letters it has seen in the input word. 
It can compare the values of the registers with each other and with the input letter. Here are two examples:
\begin{example}
\label{ex:appears-again}
Let us describe a deterministic register automaton, that recognizes the language. 
\[ \textrm{``The first letter appears again''} \subseteq \atoms^* \]
The automaton has one register and $3$ control states:
$\{q_{\textrm{start}}, q_{\textrm{check}}, q_{\textrm{found}}\}$.
Let us go through the automaton's run on the word $1\ 2\ 3\ 2\ 1\ 3 \in \atoms^*$
(we use natural numbers to denote elements of $\atoms$).
Here is the initial configuration:
\smallpicc{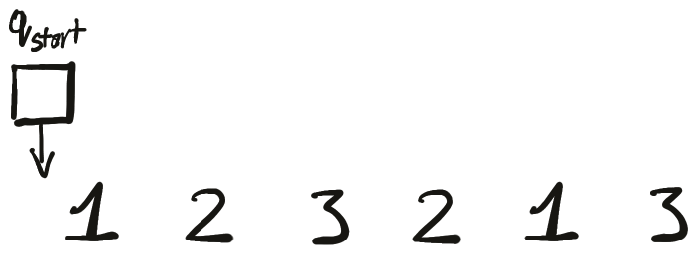}
The automaton starts in $q_{\textrm{start}}$ with an empty register (represented in the picture as the empty box).
In its first step, the automaton stores the first letter in the register, sets its control state to $q_\textrm{check}$, 
and moves forward:
\smallpicc{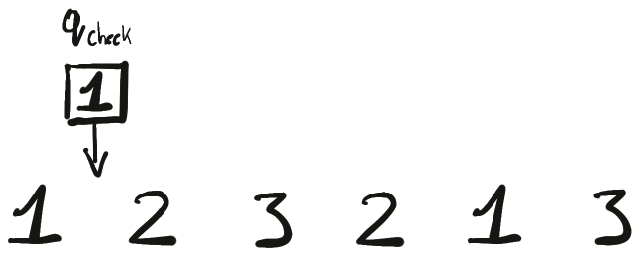}
In the next step it compares the register value with the current letter.
Since they are different, it simply moves forward:
\smallpicc{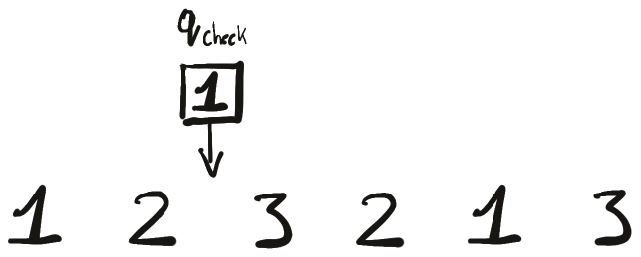}
\noindent
The same transition happens two more times:
\smallpicc{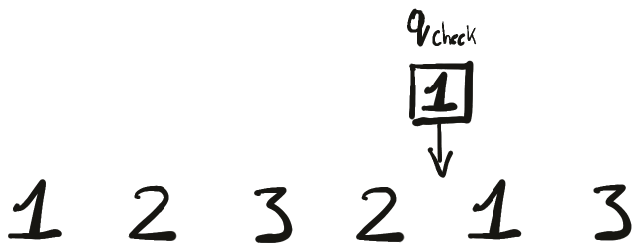}
Now, when the automaton compares the current letter with the register value, it finds out that they are equal.
It has found a reappearance of the initial letter, so it sets its state to $q_{\textrm{found}}$ and proceeds to the next letter.
\smallpicc{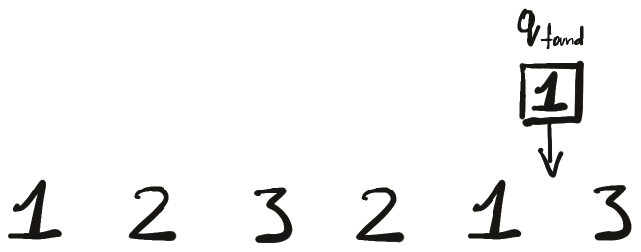}
In the state $q_\textrm{found}$ the automaton ignores the input and keeps moving forward until the end of the word.
\smallpicc{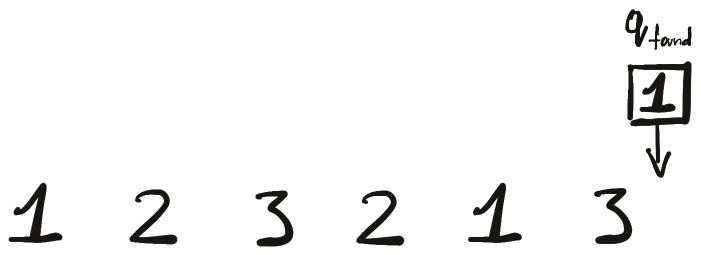}
At the end of the word, the automaton accepts the input word if it is in the state $q_\textrm{found}$.
\end{example}

\begin{example}
\label{ex:at-most-3}
The automaton that recognizes the language 
\[ \textrm{``There at most $3$ letters in the word''} \subseteq \atoms^*  \] 
has three registers and four control states ($q_0$, $q_1$, $q_2$, $q_3$, and $q_{>3}$).
After reading some part of the input, the automaton remembers the letters it has already seen
(without repetitions) unless it has already seen more than $3$ of them.
Here is an example run of the automaton:
\picc{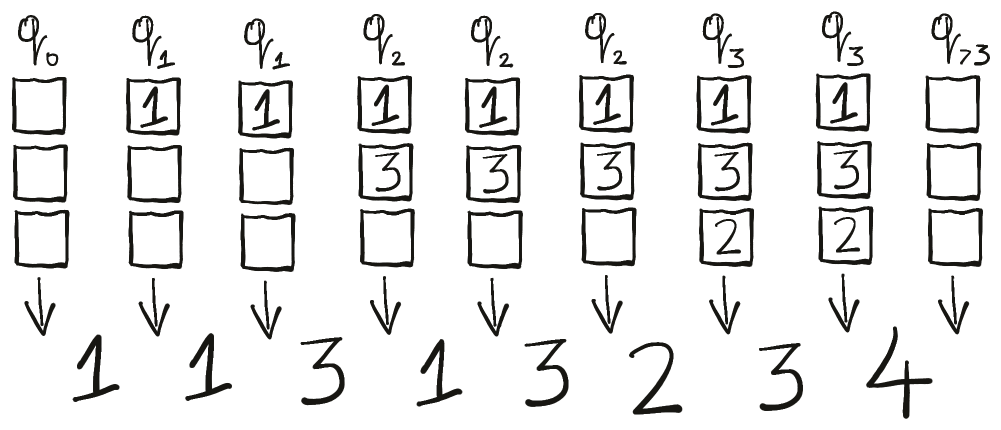}
\end{example}

Now, let us discuss the transition function of a deterministic register automaton.
The contents of the automaton's memory can be described as an element of the following set:
\[\underbrace{Q}_{\substack{
    \textrm{automatons's control state}
}} \times 
\underbrace{(\atoms + \bot)^R}_{\substack{
    \textrm{contents of each register, where}\\
    \textrm{$\bot$ represents empty registers}
}} \]
This means that the transition function should have the following type:
\[\underbrace{Q \times (\atoms + \bot)^R}_{\substack{
    \textrm{current memory state}
}} \times 
\underbrace{\atoms}_{\substack{
    \textrm{current letter}
}} \to
\underbrace{Q \times (\atoms + \bot)^R}_{\substack{
    \textrm{updated memory state}
}}  \]
As we have mentioned before, we cannot allow all transition functions of this type, 
or else a deterministic register automaton would become an unrestricted infinite state 
machine (and such machines can recognize all languages, including ``the length of the word is prime'',
and ``the length of the word encodes a halting Turing machine'').
In order to avoid that, we restrict the power of the transition function.
The intuition behind this restriction is that the only allowed operations on registers should be:
\begin{enumerate}
    \item comparing for equality two register values;
    \item comparing for equality a register value with the input letter; and
    \item saving the input letter in one of the registers.
\end{enumerate}
We describe three ways of formalizing this restriction, and prove that they are equivalent:
\begin{description}
    \item[1. Syntactic equivariance]\label{it:syntactic-equivariance}
    For this definition we provide
    syntax for specifying the transition function -- every function that
    can be specified in this syntax is syntactically equivariant.
    A transition function is specified
    as a list of conditional commands, each of the following form:
    \[ \textrm{list of conditions} \rightarrow \textrm{list of actions} \]
    To apply the transition, the automaton goes through the list of conditional comments (top-down), finds
    the first command in which every condition is satisfied, and performs all the
    actions from that command. (To make this a total function, we say that if there is no command 
    whose all conditions are satisfied, then the automaton stays in the same configuration.)
    Conditions and actions are of the following types (each condition can appear both with $=$
    and with $\neq$):
    \begin{center}
        \begin{tabular}{|l | l|}
            \hline
            Example & Description \\
            \hline
            $\mathtt{r}_1 = \mathtt{r}_2$ & Compare two register values \\
            $\mathtt{r}_2 \neq \mathtt{input}$ & Compare a register value with the current input letter \\
            $\mathtt{r_5} = \bot$ & Check if a register is empty \\
            $\mathtt{state} = \mathtt{q}_5$ & Check that the automaton is in a particular state\\
            \hline
        \end{tabular}
    \end{center}

    \begin{center}
        \begin{tabular}{| l | l |}
            \hline
            Example & Description \\
            \hline
            $\mathtt{r}_2 \, := \mathtt{input}$ & Save the current letter to a register\\ 
            $\mathtt{r}_1 := \mathtt{r}_3$ & Copy a register value into another register. \\
            $\textrm{swap}(\reg_2, \reg_4)$ & Swaps the contents of two registers.\\
            $\mathtt{r}_5 := \bot$ & Clear the contents of a register \\
            $\mathtt{state} := \mathtt{q}_3$ & Update the state of the automaton \\
            \hline
        \end{tabular}
    \end{center}

    \item[2. Semantic equivariance] The intuition behind this definition is that the valid transition
        functions are the ones that do not discriminate between letters i.e. the ones
        that commute with permutations of letters. To express this formally, we have
        to extend letters permutations ($\atoms \to \atoms$) to act on the set
        of all memory states of a register automaton, i.e. on the set:
        \[Q \times (\atoms + \bot)^R \]
        We do it in a natural way, by applying the permutation to every register
        value, leaving the control state and the empty registers unchanged.
        Here is an example of an action of $\pi$ that swaps letters $3$ and $7$ (and does 
        not touch other letters):
        \vsmallpicc{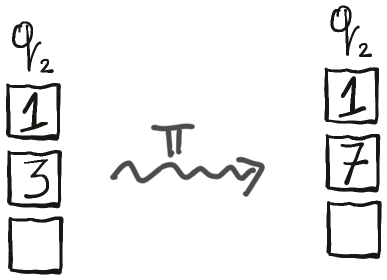}
        We can naturally extend this action to the domain of the transition function:
        \[(Q \times (\atoms + \bot)^R)  \times \atoms\] 
        Finally, we say that a transition function $\delta$ is
        semantically equivariant if for every letters permutation $\pi$ and for every $x$:
        \[ \delta(\pi(x)) = \pi(\delta(x)) \]

    \item[3. Haskell-style equivariance] In this definition, we use Haskell's type system.
    The key idea is to encode $\atoms$ as the polymorphic type:
    \[\mathtt{Eq \, a \ \Rightarrow \  a}\]
    We continue by encoding the finite sets $R = \{r_1, \ldots, r_k\}$ and $Q = \{q_1, \ldots, q_n\}$ as variant types,
    explicitly enumerating all their elements:
    \[
        \mathtt{data\, Q = q_1 \, | \, q_2 \, | \, \ldots \, |  \, q_n}
    \]
    \[
        \mathtt{data\, R = r_1 \, | \, r_2 \, | \, \ldots \, |  \, r_k}
    \]
    Then, we encode a configuration of a register automaton as:
    \[ \mathtt{Eq \, a \Rightarrow (Q \ , \  R \to Maybe\, a)} \]
    Finally, we define the valid transition functions to be the Haskell-definable total functions of the following type:
    \[
        \mathtt{Eq\, a\ \Rightarrow\  (Q,\, R \to Maybe \, a) \to a \to (Q,\, R \to Maybe \, a)} 
    \]
\end{description}

We prove that those definitions are all equivalent, starting with equivalence between semantic and syntactic equivariances.
They are versions of definitions from \cite[Section 1.1]{bojanczyk2019slightly} adapted for transition functions (original definitions work with transition relations), so this proof is very similar to the proof of Lemma~1.3
from \cite{bojanczyk2019slightly}:\\

\textbf{Syntactic equivariance $\Rightarrow$ Semantic equivariance:}
Notice that applying letter permutations to an element of the domain (i.e. $(Q \times (\atoms + \bot)^R)  \times \atoms$) preserves all
conditions from the syntactic definition. This means that $x$ and $\pi(x)$ will be dispatched to the
same conditional command. Also, all actions only move
letters around, so they commute with letter permutations.\\

\textbf{Semantic equivariance $\Rightarrow$ Syntactic equivariance:}
In this proof we look into the structure of the domains of transition functions. We say that
two elements of $(Q \times (\atoms + \bot)^R)  \times \atoms$ are in the same \emph{orbit} if they differ only by a letter permutation. Formally, we
define an orbit of $x$, to be the following set:
\[ \{ \pi(x) \, | \, \pi \textrm{ is an $\atoms$-permutation}\} \]
Notice, that being in the same orbit is an equivalence relation -- two orbits are either equal or
disjoint. It is not hard to see that two elements of the domain belong to the same orbit, if and only if
every condition from the definition of syntactic equivariance is satisfied either in both of them or
in none of them. There are only finitely many of those conditions,
so the following two claims are consequences of this observation:
\begin{claim}
    There are finitely many orbits in $(Q \times (\atoms + \bot)^R)  \times \atoms$.
\end{claim}
\begin{claim}
    For every orbit $O \subseteq (Q \times (\atoms + \bot)^R)$, there is a list of conditions,
    such that an element belongs to $O$, if and only if it satisfies all the conditions from the list. 
\end{claim}
We finish the proof by showing that we can use actions from the syntactic definition to implement
the transition function for each orbit:
\begin{claim}
    \label{claim:universal-sequence-of-actions}
        For every orbit, there is a universal sequence of actions, that transforms every $x$ in that orbit into $\delta(x)$.
    \end{claim}
    \begin{proof}
        We start by showing that there is a sequence of actions that transforms some $x$ from the orbit into 
        $\delta(x)$. Take some $x$.  The actions can modify the state, clear the registers, and move the letters around,
        so the only difficulty is to show that every letter that appears in $\delta(x)$ also appears in $x$.
        Assume towards a contradiction that there exists $a \in \atoms$ that appears in $\delta(x)$, but
        not in $x$. Let $\pi$ be a letter permutation that swaps $a$ with a fresh letter (i.e. a letter that
        appears neither in $x$ nor in $\delta(x)$) and does not touch other letters. This means that:
        \[
            \begin{tabular}{ccc}
            $\pi(x) = x $,&but  $\pi(\delta(x)) \neq \delta(x)$,& so $\delta(\pi(x)) \neq \pi(\delta(x))$
            \end{tabular}
        \]
        This contradicts the assumption that $\delta$ is semantically equivariant.
        It follows 
        that there exists a sequence of actions that transforms $x$ into $\delta(x)$. Call it $a_\delta$,
        and let us prove that it is universal, i.e. that for every $x'$ from the orbit of $x$: 
        \[\delta(x') = a_\delta(x')\]
        If $x$ and $x'$ are in the same orbit, then $x' = \pi(x)$, for some $\pi$. The function
        $\delta$ is semantically equivariant, so:
        \[ \delta(x') = \pi(\delta(x)) \]
        By definition of $a_\delta$:
        \[\pi(\delta(x)) = \pi(a_\delta(x)) \]
        Notice that every individual action commutes with letter permutations. This means
        that $a_\delta$, which a composition of individual actions, also commutes with letter permutations:
        \[\pi(a_\delta(x)) = a_\delta(\pi(x))\]
        Finally, since $x =\pi(x') $:
        \[a_\delta(\pi(x)) = a_\delta(x')\]
\end{proof}

\textbf{Syntactic equivariance $\Rightarrow$ Haskell-style equivariance:}
This implication is immediate -- the syntax of syntactically equivariant functions can be
directly translated into Haskell.\\

\textbf{Haskell-style equivariance $\Rightarrow$ Semantic equivariance:}
This implication follows (for free) from \cite[Section 3.4]{wadler1989theorems}.

\section{Sets with atoms}
\label{sec:sets-with-atoms}
In this section, we introduce the abstract notion of \emph{sets with atoms}.
Intuitively, 
\emph{atoms} is a fixed countably infinite set whose elements can only be compared for equality
(e.g. the infinite alphabet of register automata)
and 
\emph{sets with atoms} are sets whose elements can store a finite number of atoms.
One example of a set with atoms, that we have already seen is:
\[Q \times (\atoms + \bot)^R \]

Before we start, we include a short bibliographical note: Sets with atoms were first studied
by Fraenkel in 1920, and then by Mostowski in the 1930s. Both of those authors studied
them as potential alternative models of set theory, which do not admit
the axiom of choice. In computer science, they were
rediscovered (under the name of nominal sets) by Gabbay and Pitts in \cite{gabbay2002new}
who consider their applications to semantics. In the context of formal language theory,
they were first studied by Bojańczyk in \cite{bojanczyk2013nominal}
and by Bojańczyk, Klin and Lasota in \cite{lasota2014automata}. This section is
mainly based on \cite[Chapter~2]{bojanczyk2019slightly} and on
\cite[Section~1,~2,~and~5]{pitts2013nominal}.

\subsection{Action of atoms permutations and its supports}
Semantic equivariance was defined only in terms of the action of atom\footnote{From now on, we will refer to the elements of $\atoms$ as atoms.} permutations on the set $Q \times (\atoms + \bot)^R$.
This means that we could define semantic equivariance for a function $X \to Y$, as long as we know that
both $X$ and $Y$ are equipped with an action of the group of atom permutations. Another important property 
of $Q \times (\atoms + \bot)^R$ is that each of its elements contains only finitely many atoms. This can 
be abstractly defined in terms of \emph{supports}: 
\begin{definition}
    Let $X$ be a set equipped with an action of the group of atom permutations.
    We say that a subset of atoms $\alpha \subseteq \atoms$ supports an element $x \in X$,
    if for every permutation $\pi$:
    \[\begin{tabular}{ccc}
       $\substack{\textrm{for every } a \in \alpha \comma\\ \pi(a) = a}$ &
       $\Rightarrow$ &
       $\pi(x) = x$
    \end{tabular}\]
We say that $x$ is \emph{equivariant}, if it is supported by the empty set.
\end{definition}
For example, the element $(1, 2, 1) \in \atoms^3$ (equipped with the natural action)
is supported by sets $\{1, 2\}$ and $\{1, 2, 3\}$, but not by $\{1\}$.
More generally, every element $ (x, y, z) \in \atoms^3$ is supported 
by the finite set $\{x, y, z\}$.
\begin{definition}
A \emph{set with atoms} is a set equipped with an action of atom permutations
such that all its elements have finite supports.
\end{definition}
\noindent
For example, the set 
\[ Q \times (\atoms + \bot)^R\]
is a set with atoms, as long as $R$ is finite. This is because every element of the set
is supported by the set of  at most $|R|$ atoms. Another example of set with atoms is 
the set of all finite words over atoms i.e. $\atoms^*$. This is because every word
has a support no bigger than the word's length. On the other hand, 
the set of all subsets of atoms $P(\atoms)$ is not finitely supported. This is because 
elements of $P(\atoms)$ that are neither finite nor cofinite have infinite supports:
\begin{lemma}
\label{lem:pfsa-fin-cofin}
A subset $X \subseteq \atoms$ is finitely supported (i.e. has some finite support), if and only if either
$X$ or $(\atoms - X)$
is finite. 	
\end{lemma}
\begin{proof}
	Take a subset of atoms $X$, and its potential finite support $\alpha$.
    An $\alpha$-permutation $\pi$ modifies $X$ (meaning that $\pi(X) \neq X$),
    if and only if $\pi$ transforms some atom from $X$
    into an atom from outside of $X$. This is only prevented if $\alpha$
	contains all elements of $X$ or all elements from $\atoms - X$. This is possible
    if and only if one of those sets is finite.
\end{proof}
It follows that sets with atoms are not closed under the powersets (i.e. $P(X)$ might not be a set with atoms even if $X$ is).
Instead,
we define $P_{\textrm{fs}}(X)$, which is the set of all finitely supported
subsets of $X$. 
For example, $P_\fs(\atoms)$ is the set of all finite and cofinite subsets of $\atoms$.
It is worth pointing out that $P_{\textrm{fs}}(X)$ is usually 
different from the set of all finite and cofinite subsets of $X$.
For example, the set $\{ (x, 4) \ | \ x \in \atoms \}$ belongs to $P_\fs(\atoms^2)$, and it is neither finite nor cofinite.
It is not hard to see that if $X$ is a set with atoms, then so is $P_\fs(X)$. \\

Notice that equivariant subsets of sets with atoms are
sets with atoms themselves (with permutation action inherited from the superset).
Note that this is only true for equivariant subsets:
consider $\{7, 8\}$ -- a finitely supported subset of $\atoms$. It is very
easy to find a $\pi$, for which $\pi(7) \not \in \{7, 8\}$. This means
that the permutation action inherited from the superset (i.e. $\atoms$)
is not valid action for the set $\{7, 8\}$.\\

Another class of sets with atoms are the \emph{atomless sets},
which are sets equipped with the trivial action:
\[
    \begin{tabular}{cc}
         $\pi(x) = x$ & for every $\pi$
    \end{tabular}
\]
Every element of an atomless set is equivariant.

\subsubsection{The finitely supported relations and functions}
Sets with atoms are closed under many classical operations, including the product:
\begin{lemma}
If $X$ and $Y$ are sets with atoms, then so is $X \times Y$, with the following
permutation action:
\[ \pi((x, y)) = (\pi(x), \pi(y)) \]
\end{lemma}
\begin{proof}
	It is easy to see that, if $x$ is supported by $\alpha$
    and $y$ is supported by $\beta$, then $(x, y)$ is supported by $\alpha \cup \beta$. 
    (Note that the lemma does not extend to infinite products.)
\end{proof}
A \emph{finitely supported relation} between two sets with atoms $X$ and $Y$ is a finitely supported subset of $X \times Y$,
i.e. an element of $P_\fs(X \times Y)$. 
This means that the permutation action on a relation $(\sim) \in P_\fs(X \times Y)$  is defined as follows:
\[
    \begin{tabular}{ccc}
        $x \  (\pi \! \sim) \ y$ & $\stackrel{\textrm{def}}{\iff}$ & $\pi(x) \sim \pi(y)$
    \end{tabular}
\]
Sets with atoms are closed under equivariant quotients:
Let $(\sim) \subseteq X \times X$ be an equivariant equivalence relation,
then $X_{/\sim}$ is a set with atoms. Its permutation action defined as follows:
\[ \pi \left([x]_\sim\right) = \left[\pi (x) \right]_\sim \]
It is easy to see that $[x]_\sim$ is supported by whatever supports $x$\\

A \emph{finitely supported function} is a finitely supported relation that happens to be a function.
(i.e. for all $x$ there is exactly one $y$ such that $x \sim y$). 
This leads to the following definition of permutation action on functions:
\[ \pi(f) = \pi \circ f \circ \pi^{-1} \]
It follows that a function is supported by $\alpha$ if, and only if for every $\pi$ that preserves 
all atoms from $\alpha$ (we call such $\pi$ an $\alpha$-permutation), it holds that:
\[ f(\pi(x)) = \pi(f(x)) \]
We denote the set of all finitely supported functions between two sets with atoms as $X \fsto Y$.
Similarly, we write $X \eqto Y$ for the set of all equivariant functions. The intuition behind the 
next lemma is that finitely supported functions cannot create new atoms:
\begin{lemma}
\label{lem:fs-functions-preserve-supports}
    If $x$ is supported by $\alpha$ and $f$ is supported by $\beta$ then $f(x)$ supported
    by $\alpha \cup \beta$. In particular, if $f$ is equivariant, then $f(x)$ is supported 
    by $\alpha$. 
\end{lemma}
\begin{proof}
    Choose any $(\alpha \cup \beta)$-permutation $\pi$. Then, because $\pi$ is both
    an $\alpha$-permutation and a $\beta$-permutation, we obtain that:
    \[ \pi(f(x)) \stackrel{\textrm{$\beta$ supports $f$}}{=} f(\pi(x)) \stackrel{\textrm{$\alpha$ supports $x$}}{=} f(x) \]
\end{proof}
Equivariant and finitely supported functions are closed under
compositions -- if $f$ is supported by $\alpha$ and $g$ is supported by $\beta$, then $f \circ g$ is supported by
$\alpha \cup \beta$.\footnote{This means that there are two types of categories over sets with atoms.
The first one is the category of all functions supported by some $\alpha$. When $\alpha = \emptyset$,
this is the category of all equivariant functions -- the $\mathbf{Nom}$ studied in \cite{pitts2013nominal}.
The other type of category is the broader category of all finitely-supported functions (where the support depends on the functions).}\\

We say that two sets with atoms ($X$ and $Y$) are isomorphic, if there exists an equivariant
bijection $f: X \eqto Y$, i.e. an equivariant function that is a surjection and an injection. 
The following lemma states that equivariant bijections have equivariant inverses.
\begin{lemma}
\label{lem:fs-rev-bjection}
Every finitely supported bijection $f : X \to Y$, has a finitely supported inverse function $f^{-1} : Y \to X$. 
Moreover, $f$ and $f^{-1}$ have the same supports, i.e. for every finite $\alpha \subset \atoms$:
\[
    \begin{tabular}{ccc}
    $\alpha$ supports $f$ & $\iff$ & $\alpha$ supports $f^{-1}$
    \end{tabular}
\]
\end{lemma}
\begin{proof}
    Functions are defined as relations, so $f \subseteq X \times Y$. This means that we can define:
    \[ f^{-1} = \{ (y, x) \ | \ (x, y) \in f \} \]
    Since $f$ is a bijection, then $f^{-1}$ is a function. It is easy to see that $f^{-1}$ is
    the inverse of $f$ and that it has the same supports as $f$. 
\end{proof}
A similar argument shows that all finitely supported injections admit a partial one-sided inverse. However,
as we will see in Section~\ref{sec:ra-ofa-straight}, there are equivariant surjections that
do not admit equivariant (or even finitely supported) one-sided inverses.

\subsection{Orbit-finite sets}
\label{subsec:orbit-finite-sets}
Finite sets with atoms are not very interesting -- it is not hard to see
that all finite sets with atoms have to be atomless. It turns out, however,
that there exists a suitable analogue of finiteness for sets with atoms called
\emph{orbit finiteness}\footnote{It was first introduced by Bojańczyk in \cite{bojanczyk2013nominal}}.
Intuitively, a set is orbit-finite if it has only finitely
many elements \emph{up to atom permutations}. Formally, we define an orbit\footnote{
	This is the same orbit as in the proof of Semantic equivariance $\Rightarrow$ Syntactic equivariance from Section~\ref{sec:dra}.}
of an element $x \in X$ to be
\[ \{ \pi(x) \ |\ \pi \textrm{ is an atom permutation } \}\tdot \]
Notice that every two orbits are either equal or disjoint,
which means that they divide $X$ into equivalence classes. We say that $X$ is orbit-finite
if it has finitely many orbits.
For example, the set of all atoms $\atoms$ has only one orbit, so it is orbit-finite.
Tuples of atoms are orbit-finite as well:
\begin{lemma}
\label{lem:atom-tuples-of}
	For every $k$, the set $\atoms^k$ is orbit-finite.
\end{lemma}
\begin{proof}
	An orbit of $\atoms^k$ can be defined by an equality pattern such as the one below (for $k=6$):
	\smallpicc{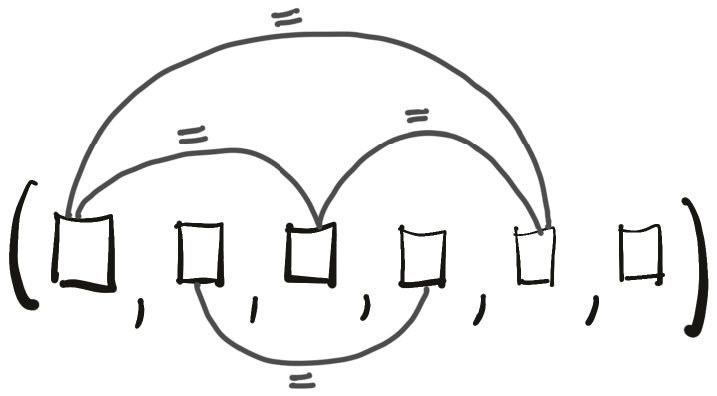}
  It represents a tuple where  atoms on positions $1$, $3$ and $5$ are
   equal to each other (and different from all the other atoms);
   atoms on positions $2$ and $4$ are equal to each other (and different from all the other atoms); 
   and the atom on position $6$ is only equal to itself.
   For every $k$, there are only finitely many such patterns, so $\atoms^k$ is orbit-finite.
\end{proof}
The number of orbits in $\atoms^k$ is finite but very large, so it is sometimes
useful to consider the set $\atoms^{(k)} \subseteq \atoms^k$
of tuples whose atoms are pairwise distinct -- it has only one orbit (regardless of $k$).\\

An example of an orbit-infinite set is $\atoms^*$ -- the length of the word is preserved by 
the permutation action, so there is at least one orbit for each word length.
Perhaps surprisingly, the set $\pfs(\atoms)$ is also orbit-infinite. The argument is the
same as for $\atoms^*$ -- the size of a set is preserved
by the permutation action, so there is at least one orbit for each finite set size.
For the same reason, the following set of finitely supported functions is orbit-infinite as well:
\[\atoms \tofs \{\textrm{yes}, \textrm{no}\}\tdot\]
This means that unlike classical finiteness, orbit finiteness is not 
preserved under powersets and function spaces. (In the next section, we will see that
this causes some of the results from finite automata theory to fail for
infinite alphabets.) Still, orbit finiteness is preserved by many classical combinators 
(and many classical results hold for infinite alphabets).
Here are some of the operations that preserve orbit finiteness:

\begin{lemma}
\label{lem:of-preserved}
If $X$ and $Y$ are orbit-finite, then the following sets are orbit-finite as well:
$X \times Y$, $X + Y$, and $X_{/\sim}$ (where $\sim$ is an equivariant equivalence relation). 
\end{lemma}
\begin{proof}
Cases $X + Y$ and $X_{/\sim}$ are easy: To show that $X + Y$ is orbit-finite, we notice that:
\[\#\textrm{orbits of } (X + Y) = \#\textrm{orbits of } X + \#\textrm{orbits of } Y \]
To show that $X_{\/\sim}$ is orbit-finite, we notice that:
\[ \begin{tabular}{ccc}
    $x \textrm{ is in the same orbit as } y$& $\Rightarrow$ & $[x]_\sim  \textrm{ is in the same orbit as } [y]_\sim$.
\end{tabular}\]
It follows that $X_{\/ \sim }$ has as most as many orbits as $X$.\\

The most interesting case is $X \times Y$. We start the proof by introducing a way of representing orbit-finite sets:
\begin{claim}
\label{claim:of-representation}
Every orbit-finite set is isomorphic\footnote{Remember that we require the isomorphism to be an equivariant function.}
to a set of the form
\[ {\atoms^{(k_{1})}}_{/\sim_{1}} + {\atoms^{(k_{2})}}_{/\sim_{2}} + \ldots + {\atoms^{(k_{n})}}_{/\sim_{n}}  \]
where $\sim_i$ are equivariant equivalence relations.
\end{claim}
\begin{proof}
Every orbit-finite set can be decomposed as a disjoint sum of its orbits, so it suffices
to show that every single-orbit set is isomorphic to 
\[ \begin{tabular}{cc}
    ${\atoms^{(k)}}_{/\sim}$, & for some $k$ and $\sim$.
\end{tabular}\]
Let $X$ be a single-orbit set. Take some $x \in X$, and let $\alpha$ be a finite support of $x$.
Arrange the elements of $\alpha$ in any order to form a tuple $\bar \alpha \in \atoms^{(k)}$,
where $k = |\alpha|$. Define a function $f : \atoms^{(k)} \to X$ as follows (remember that a
function is a special kind of a relation):
\[ f = \{ (\pi(\bar{\alpha}), \pi(x)) \ | \ \textrm{for each atom permutation } \pi \}\]
First, we show that $f$ is well-defined (i.e. that it produces exactly one output for every argument):
The set $\atoms^{(k)}$ only has one orbit, so for every $\bar \beta \in \atoms^{(k)}$ there is
a $\pi$ such that $\pi (\bar{\alpha}) = \bar \beta$. It follows that 
$f$ produces at least one result for every argument.
To show that it produces at most one result for each
argument, we take some $\pi_1$ and $\pi_2$, such that $\pi_1(\bar \alpha) = \pi_2(\bar \alpha)$, 
and  we show that $\pi_1(x) = \pi_2(x)$: Both $\pi_1$ and $\pi_2$ are $\alpha$-permutations,
so $(\pi_2^{-1} \circ \pi_1)$ is an $\alpha$-permutation as well.
Since $x$ is supported by $\alpha$, it follows that:
\[ \pi_2^{-1}(\pi_1(x)) = x\]
It follows that $\pi_1(x) = \pi_2(x)$. Now, let us show that $f$ is
equivariant: take a $\bar \beta \in \atoms^{(k)}$ and an atom
permutation $\pi$. We know that there is some $\rho$,
for which $\bar \beta = \rho(\bar \alpha)$.
It follows that:
\[\pi(f(\bar \beta)) = \pi(f(\rho(\bar \alpha))) = 
\pi(\rho(x)) = f(\pi(\rho(\bar \alpha))) = f(\pi(\bar \beta))\]
It is also not hard to see that $f$ is a surjection (because $X$ is a single orbit). 
It follows that if we divide $\atoms^{(k)}$ by $f$'s kernel,
we get an isomorphism:
\[f' : {\atoms^{(k)}}_{/\textrm{ker } f} \to X \]
To finish the proof of the claim, notice that since $f$
is equivariant then so are $\ker f $ and $f'$.
\end{proof}

Now, take some orbit-finite $X$ and $Y$ and show that $X \times Y$
is also orbit finite. We start by applying Claim~\ref{claim:of-representation},
to both $X$ and $Y$. Then apply the distributivity of products over disjoint sums,
to obtain that $X \times Y$ is isomorphic to a disjoint sum of products
of the following form:
\[ {\atoms^{(k)}}_{/\sim_i} \times {\atoms^{(l)}}_{/\sim_j'} \]
Disjoint sums preserve orbit finiteness, so we are left with showing
that each such product is orbit-finite: First, notice that each
such product is isomorphic to:
\[\left(\atoms^{(k)}\times \atoms^{(l)}\right)_{/(\sim_i \sim_j')}\]
where $(\sim_i \sim_j')$ is a relation that independently
checks that $\sim_i$ holds on the first $k$ coordinates and
$\sim_j'$ holds on the last $l$ coordinates.
This is an equivariant equivalence relation, so dividing by it 
preserves orbit finiteness. This leaves us with showing that
$\atoms^{(k)} \times \atoms^{(l)}$ is orbit-finite. Since
$\atoms^{(k)} \times \atoms^{(l)} \subseteq \atoms^{k+l}$, 
this follows from the following claim (and Lemma~\ref{lem:atom-tuples-of}):
\begin{claim}
\label{claim:of-subsets}
Equivariant subsets of orbit-finite sets are orbit-finite.
\end{claim}
\begin{proof}
Let $Y$ be an equivariant subset of an orbit-finite $X$. Notice, that if two elements are in different orbits in $Y$, 
then they are also in different orbits in $X$. It follows that $X$ contains at least as many orbits as $Y$.
\end{proof}
\end{proof}

We conclude this section, with a table summarizing closure properties of sets with atoms and orbit-finite sets:\\

\begin{tabular}{c|c|c}
  Operation & Preserves sets with atoms? & Preserves orbit-finite sets?  \\
  \hline
  $X \times Y$& Yes & Yes \\
  $X + Y$ & Yes & Yes \\
  $X_{/\sim}$ \tiny{(for equivariant $\sim$)}  & Yes & Yes \\
  Equivariant subsets & Yes & Yes\\
  $X^*$ & Yes & No \\
  $P(X)$ & No & No \\
  $\pfs(X)$ & Yes & No \\
  $X \to Y$ & No & No \\
  $X \tofs Y$ & Yes & No \\
\end{tabular}

\section{Deterministic orbit-finite automaton}
In this section we introduce a model of computation that generalizes the deterministic 
finite automaton. This is a very natural model that deals
with infinite (but orbit-finite) alphabets. Later, we argue that it
is equivalent to register automata. The section is mostly based on
\cite[Section~5.2~and~6.2]{bojanczyk2019slightly}, but it also discusses
techniques from \cite[Section~B.1]{single-use-paper} and
\cite[Section~4]{pitts2013nominal}.\\

\begin{definition}
A \emph{deterministic orbit-finite automaton} consists of:
\begin{enumerate}
\item an orbit-finite alphabet $\Sigma$;
\item an orbit-finite set of states $Q$;
\item an equivariant initial state $q_0 \in Q$;
\item an equivariant subset of accepting states $Q_{\textrm{acc}} \subseteq Q$;
\item and an equivariant transition function
	\[ f : Q \times \Sigma \eqto Q\tdot \]
\end{enumerate}
\end{definition}
Such an automaton defines a language over $\Sigma$. To see this in action, 
let us define a deterministic orbit-finite automaton, recognizing the language
\[ \{ w \in \atoms^* \ |\ \textrm{$w$ has at most $3$ different letters} \}\]
First, we notice that $\Sigma = \atoms$, which is an orbit-finite set.
Then, let us define the automaton's set of states as:
\[Q = \underbrace{\binom{\atoms}{\leq3}}_{
\textrm{subsets of at most $3$ atoms}}+ 
\underbrace{\bot}_{\substack{\textrm{represents sets with}\\
\textrm{more than $3$ atoms}}} \]
This is an orbit-finite set -- it has $5$
orbits: one orbit for each size of the
set (i.e. $0$, $1$, $2$, and $3$) and one for the
element $\bot$. The initial state is the empty set,
and the accepting states are all the states
except of $\bot$. The transition function is defined as follows:
\[f(q, a) = \begin{cases}
 q \cup \{a\} & \textrm{ if } q \neq \bot \textrm{ and } |q \cup \{a\}| \leq 3 \\
 \bot  & \textrm{ otherwise } \\
 \end{cases}
 \]
The transition function is easily seen to be equivariant. Here is an example (rejecting) run of this automaton:
 \picc{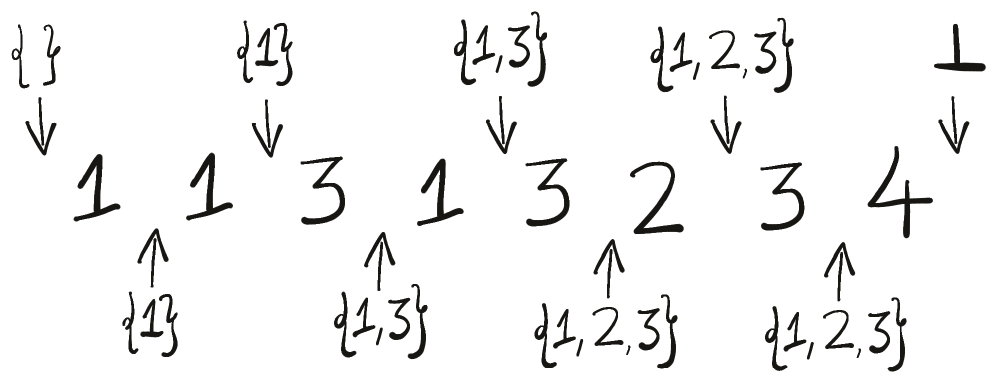}

\subsection{Register automata, orbit-finite automata, and straight sets}
\label{sec:ra-ofa-straight}

We would like to prove that deterministic
orbit-finite automata and deterministic register automata
have equal expressive powers.
However, the two models are slightly incompatible: register
automata can only recognize languages over the alphabet\footnote{Register automata are usually defined to work over the alphabet $\Sigma \times \atoms$
(where $\Sigma$ is some finite set). Theorem~\ref{thm:dra-equiv-dofa} can be adapted to this type of automata as well.} $\atoms$,
whereas orbit-finite automata can recognize languages over
every orbit-finite alphabet. The statement of the equivalence theorem has
to account for this incompatibility:

\begin{theorem}
\label{thm:dra-equiv-dofa}
Deterministic register automata and deterministic orbit-finite automata recognize the same languages over
the alphabet $\atoms$.
\end{theorem}

We start with discussing the more difficult part of the proof, which is translating orbit-finite automata
to register automata. Let us illustrate the problem with this translation, by
discussing a failed attempt: Take an orbit-finite
automaton $\mathcal{A}$ and let $Q$ be its set of states.
Thanks to Claim~\ref{claim:of-representation}, we know that $Q$ is isomorphic to:
\[ {\atoms^{(k_1)}}_{/\sim_1} + {\atoms^{(k_2)}}_{/\sim_2} + \ldots +{\atoms^{(k_n)}_{/\sim_n}} \]
Each element of this set can be represented as an element of the following set:
\[ {\atoms^{(k_1)}}+ {\atoms^{(k_2)}}+ \ldots +{\atoms^{(k_n)}} \]
Which means that each element of $Q$ can be further represented as a memory configuration of the following shape 
(where $k = \max(k_i)$):
\[ Q' := \underbrace{\{1, 2, \ldots, n\}}_{\textrm{the index } i} \times 
\underbrace{(\atoms + \bot)^k}_{
	\substack{
		k_i \textrm{ atoms}\\
		\textrm{followed by $\bot$'s}
	}
}\]
It is important to point out that this
representation is usually partial (i.e. some elements from $Q'$ might not represent any element) and not injective (i.e. some elements from $Q$
might have more than one representation in $Q'$), but it is always
surjective (every element from $Q$ has to have at least one
representation in $Q'$). Now, we would like to lift $\mathcal{A}$'s transition function $\delta : Q \times \atoms \to Q$
to work on this representation, obtaining $\delta' : Q' \times \atoms \to Q'$. Here is an attempt:
\[ \delta'(x, a) =
               \textrm{a representative of } (\delta(\textrm{element represented by } x, a))
\]

Surprisingly, this is not always possible -- we require $\delta'$ to \emph{choose} a representative,
and some sets with atoms do not admit choice:
\begin{claim}
\label{claim:no-choice}
There is no finitely supported function
$f : \binom{\atoms}{2} \fsto \atoms^{(2)}$,
that chooses a tuple representation of a set. 
In other words, there is no such $f$, that for all $a,b \in \atoms$
the value $f(\{a, b\})$ is either equal to
$(a, b)$ or to $(b, a)$.
\end{claim} 
\begin{proof}
    Suppose that there is such $f$, and let $\alpha$ be a finite support of $f$.
	Take some two different atoms $a$ and $b$ from outside of $\alpha$.
	Suppose without the loss of generality that
	$f(\{a, b\}) = (a, b)$. Let $\pi$ be the atom automorphism that swaps $a$ with $b$ and does not touch other atoms.
    Notice that this $\pi$ is an $\alpha$-permutation, and that $\{a, b \} = \pi(\{a, b\})$.
	This leads to a contradiction:
	\[ (a, b) = f(\{a, b\}) = f(\pi(\{a, b\})) = \pi((a, b)) = (b, a)\tdot\]
\end{proof}

The source of those problems with choice are symmetries such as
$\{a, b\} = \{b, a\}$.  In the representation from Claim~\ref{claim:of-representation}
they manifest themselves as the equivalence relations ($\sim_i$).
This observation motivates the definition of \emph{straight sets}, which are
orbit-finite sets that do not exhibit those symmetries:

\begin{definition}
\label{def:straight-set}
An orbit-finite set is \emph{straight} if it is isomorphic to a set of the following form:
\[ \atoms^{(k_1)} + \atoms^{(k_2)} + \ldots + \atoms^{(k_n)} \]
\end{definition}
\noindent
For example, $\atoms^2$ is a straight set:
\[ \atoms^2 \simeq \underbrace{\atoms}_{\textrm{pairs of the form $(x, x)$}} + \underbrace{\atoms^{(2)}}_{\substack{\textrm{pairs of the form $(x, y$)}\\ \textrm{where $x \neq y$}}}\]
In general, every $\atoms^k$ is a straight set -- it is isomorphic
to the disjoint union of one $\atoms^{(k_i)}$ per equality pattern, where $k_i$ is
the number of distinct atoms in that equality pattern (see the proof of Lemma~\ref{lem:atom-tuples-of}).
Using a similar proof as the one of Lemma~\ref{lem:of-preserved},
we can show that straight sets are closed under Cartesian products and disjoint sums.\\

Let us now briefly discuss the structure of straight sets:
One can think of an element
$x \in \atoms^{(k_1)} + \atoms^{(k_2)} + \ldots + \atoms^{(k_n)}$
as a coloured tuple:
\[ \begin{tabular}{cccc}
		$x = i(\bar x)$ & for & $\underbrace{i \in \{1, 2, \ldots, n\}}_{\textrm{the colour}}$
		& $\underbrace{\bar x \in \atoms^{(k_i)}}_{\textrm{the tuple}}$
\end{tabular} \]
We define 
\[  \begin{tabular}{ccc}
     $\dim x = k_i$ & and & $\sigma(x) = \{\textrm{all the atoms from } \bar x\}$
    \end{tabular}
\]
It is easy to see that both $\dim$ and $\sigma$ are equivariant functions,
and that $\sigma(x)$ supports~$x$. Thanks to the following lemma we can extend the functions $\sigma$
and $\dim$ to all straight sets:
\begin{lemma}
\label{lem:straight-unique}
For every straight $X$, for every two isomorphisms:
\[\begin{tabular}{cc}
    $f: X \eqto \atoms^{(k_1)} + \ldots + \atoms^{(k_n)}$ &
    $g: X \eqto \atoms^{(l_1)} + \ldots + \atoms^{(l_m)}$,
\end{tabular}\]
and for every $x$, it holds that $\sigma(f(x)) = \sigma(g(x))$.
(In particular, since $\dim(x) = |\sigma(x)|$, it follows that $\dim(f(x)) = \dim(g(x))$.)
\end{lemma}
\begin{proof}
Let $h = g \circ f^{-1}$. It is an isomorphism of the following type:
\[ h : \atoms^{(k_1)} + \atoms^{(k_2)} + \ldots + \atoms^{(k_n)} \  \eqto \  
       \atoms^{(l_1)} + \atoms^{(l_2)} + \ldots + \atoms^{(l_m)} \]
It is enough to show that for all $x \in \atoms^{(k_1)} + \atoms^{(k_2)} + \ldots + \atoms^{(l_n)}$, 
it holds that $\sigma(x) = \sigma(h(x))$. The set $\sigma(x)$ supports $x$, so 
by Lemma~\ref{lem:fs-functions-preserve-supports} it supports $h(x)$ as well.
It follows that $\sigma(h(x)) \subseteq \sigma(x)$. The other inclusion can be 
proved in the same way, because by Lemma~\ref{lem:fs-rev-bjection}, the function $h^{-1}$ is equivariant.
\end{proof}

\noindent
Finally, let us show that straight sets admit choice:
\begin{lemma}[Straight Uniformisation]
\label{lem:straight-uniformisation}
    Let $X$ and $Y$ be straight sets, and let $R$ be an equivariant relation $R \subseteq X \times Y$.
    If for every $x \in X$, there exists at least one $y \in Y$ for which $x \, R \, y$,
    then there exists a \emph{finitely supported} function $r : X \to Y$ such that 
	$x \, R \, r(x)$, for every $x \in X$.
\end{lemma}

\noindent
Before we prove the lemma, let us notice that $r$ does not have to be equivariant:
\begin{example}
\label{ex:no-eq-unif}
    Consider the following $R \subseteq \atoms \times \atoms^{(2)}$:
    \[ R = \{ (x,\; (x, y)) \ | \ x, y \in \atoms \comma \textrm{ such that } x \neq y\}\]
    It satisfies the conditions of Lemma~\ref{lem:straight-uniformisation}, so it should have 
    a finitely supported uniformization. Here is an example uniformization 
    supported by  $\{4, 5\}$:
    \[
        r(x) = \begin{cases} 
            (x, 5) & \textrm{if } x = 4\\
            (x, 4) & \textrm{otherwise}
        \end{cases}
    \]
    On the other hand, it is not hard to see that $R$ has no equivariant uniformization
    (this follows from Lemma~\ref{lem:fs-functions-preserve-supports}).
\end{example}
As we can see, the reason why there might not be an equivariant $r$, is that 
$x$ might not have enough atoms to construct a matching element of $Y$. The following
lemma formalizes this intuition. We prove it, before we prove Lemma~\ref{lem:straight-uniformisation}:
\begin{lemma}
\label{lem:straight-uniformization-eq}
If $R \subseteq X \times Y$ is as in Lemma~\ref{lem:straight-uniformisation},
but additionally for every $x$ there is a $y$ such that $x\, R\, y$, and 
\[\sigma(x) \textrm{ supports } y\comma\]
then there is an \emph{equivariant} $r$ that uniformizes $R$. 
\end{lemma}
\begin{proof}
    Let us fix straight equivariant isomorphisms for $X$, $Y$:
    \[
        \begin{tabular}{cc}
            $f_X : X \eqto \atoms^{(k_1)} + \ldots + \atoms^{(k_n)}$ &
            $f_Y : Y \eqto \atoms^{(l_1)} + \ldots + \atoms^{(l_m)}$
        \end{tabular}
    \]
    This means that elements of $X$ and $Y$ can be seen as coloured tuples. We take some $x \in X$, 
    and we show how to construct $r(x) \in Y$ in an equivariant way:
	\begin{enumerate}
		\item First, we consider only those $y$'s, such that $x R y$, and 
              $y$ is supported by $\sigma(x)$. (This means that $\sigma(y) \subseteq \sigma(x)$.)
		\item Out of those $y$'s, we prefer the ones labelled with a smaller colour. 
		\item To choose one of the remaining tuples, we annotate each atom in every 
		      remaining
		      tuple $y$ with its position in $x$ (remember that $x$ is a tuple of atoms).
              Then, we chose $r(x)$ to be the $y$ whose annotation is lexicographically smallest. 
    \end{enumerate}
\end{proof}

\noindent We are now ready to proof the Straight Uniformization Lemma:
\begin{proof}
    Let $\dim X$ be the maximal dimension of an element from $X$ (and analogously for $\dim Y$).
    Let $k = \dim(X) + \dim(Y)$, and define $R' \subseteq (X \times \atoms^{(k)}) \times Y$,
    to be a relation that ignores its $\atoms^{(k)}$ part, and otherwise is equal to $R$:
    \[ R' = \{((x, \bar a), y) \ | \ (x, y) \in R,\ \bar a \in \atoms^{(k)}\} \]
    Now, we would like to apply Lemma~\ref{lem:straight-uniformization-eq} to $R'$. 
    The following claim proves that $R$ satisfies the lemma's assumptions:
    \begin{claim}
        \label{claim:good-support}
        For every $\bar a \in \atoms^{(\dim(X) + \dim(Y))}$ and 
        for every $x \in X$, there is a $y \in Y$, such that $(x, \bar a)\, R'\, y$, 
        and $y$ is supported by $\sigma(x, \bar a)$. 
    \end{claim}
    \begin{proof}
        Take some $\bar a \in \atoms^{(\dim(X) + \dim(Y))}$ and some $x \in X$.
        By assumption, we know that there exists $y \in Y$, such that $x\, R\, y$.
        Define $\pi$ to be a permutation that moves all atoms from $\sigma(y) - \sigma(x)$ into $\sigma(\bar{a}) - \sigma(x)$, 
        and does not touch any atoms from $\sigma(x)$. 
        We know that such a $\pi$ always exists, because 
        \[
        \begin{tabular}{ccc}
            $|\sigma(y) - \sigma(x)| \leq \dim Y$ & and & $|\sigma(\bar a) - \sigma(x)| \geq \dim Y$ 
        \end{tabular}
        \]
        It follows that:
        \[ \sigma(\pi(y)) \subseteq \sigma(\bar a) \cup \pi(\sigma(x))\]
        Moreover $\pi$ is a $\sigma(x)$-permutation, $x$ is supported by $\sigma(x)$, and $R$ is equivariant, 
        so $x \, R \, \pi(y)$. It follows that $(x, \bar a) \, R' \, y$.
    \end{proof}
    Let $r'$ be the equivariant uniformization of $R'$ produced by Lemma~\ref{lem:straight-uniformization-eq}. 
    Pick any tuple of atoms $\bar a \in \atoms^{(\dim(X) + \dim(Y))}$, and define $r$ as:
    \[ r(x) = r'(x, \bar a)\]
    It is easy to see that $r$ is an $\bar a$-supported uniformization of $R$. 
\end{proof}

We are now ready to prove Theorem~\ref{thm:dra-equiv-dofa}.
We define a \emph{deterministic straight automaton} to be
a variant of the orbit-finite automaton, where the alphabet and the set of states 
have to be straight. We use this model to break the proof of Theorem~$5$ into smaller steps:

\vspace{0.3cm}
    \adjustbox{width=1\textwidth, center}{
\begin{tikzcd}
	{\substack{\textrm{Deterministic register}\\\textrm{automaton}}} && {\substack{\textrm{Deterministic straight}\\\textrm{automaton}}} && {\substack{\textrm{Deterministic orbit-finite}\\\textrm{automaton}}}
	\arrow["{\textrm{Special case}}", color={rgb,255:red,214;green,153;blue,92}, curve={height=-18pt}, from=1-3, to=1-5]
	\arrow["{\textrm{Section~\ref{subsubsec:ofa-to-sa} }}", curve={height=-18pt}, from=1-5, to=1-3]
	\arrow["{\textrm{Section~\ref{subsubsec:sa-to-ra} }}", curve={height=-18pt}, from=1-3, to=1-1]
	\arrow["{\textrm{Section~\ref{subsubsec:ra-to-sa} }}", curve={height=-18pt}, from=1-1, to=1-3]
\end{tikzcd}
        }
    \vspace{0.3cm}

\subsubsection{Register automaton $\Rightarrow$ Straight automaton}
\label{subsubsec:ra-to-sa}
The set of all possible memory configurations of a register automaton
\[Q \times (\atoms + \bot)^R \]
is straight and as such it can be directly used as the set of states of a straight automaton.
The transition function, initial configuration, and set of accepting functions 
are equivariant by definition.

\subsubsection{Straight automaton $\Rightarrow$ Register automaton}
\label{subsubsec:sa-to-ra}
Let the state space of the straight automaton be equal to :
\[Q \simeq \atoms^{(k_1)} + \atoms^{(k_2)} + \ldots + \atoms^{(k_n)} \]
As mentioned before, one can think of this set as non-repeating tuples of atoms coloured in one of the $n$ colours
and therefore elements of $Q$ can be represented as:
\[ Q' = \underbrace{\{1, 2, \ldots, n\}}_{\textrm{colour of the tuple}} \times 
\underbrace{(\atoms + \bot)^{(\max k_i)}}_{
	\substack{
		\textrm{atoms of the tuple}\\
		\textrm{(followed by $\bot$'s)}
	}
}\]
This representation is partial, but injective and surjective. It follows, by a similar argument as in 
Lemma~\ref{lem:fs-rev-bjection}, that it has a (total) one-sided inverse. This means 
that every element of $Q$ has a canonical representation in $Q'$. We use it to lift the transition function 
$\delta$, to work on representations:
\[\delta'(x, a) = \textrm{\emph{the} representation of } \left( \delta(\textrm{the element represented by } x, a) \right) \]
We are now left with some technical details:
The initial state of the straight automaton is equivariant,
so its representation has to be of the form $c_0(\bot, \bot, \ldots, \bot)$,
for some colour $c_0$. We choose this colour to be the initial state 
of the register automaton.  Moreover, since the register automaton
accepts by control state and not by memory configuration, 
every time it moves forward it has to check if its configuration
is accepting, and remember this information in its control state.

\subsubsection{Orbit-finite automaton $\Rightarrow$ Straight automaton}
\label{subsubsec:ofa-to-sa}
This construction makes use of the Lemma~\ref{lem:straight-uniformisation} to
fix the failed attempt from the beginning of this section.
Take an orbit-finite automaton $\mathcal{A}$. Its set of states ($Q$) is orbit-finite,
so thanks to Claim~\ref{claim:of-representation} it is isomorphic to:
\[ {\atoms^{(k_1)}}_{/\sim_1} + {\atoms^{(k_2)}}_{/\sim_2} + \ldots {\atoms^{(k_n)}_{/\sim_n}} \]
It can be therefore represented as the straight set:
\[ \bar Q = {\atoms^{(k_1)}} + {\atoms^{(k_2)}} + \ldots {\atoms^{(k_n)}} \]
This time the representation function (call it $h : \bar Q \to Q$)
is surjective, but not injective. 
Let $f$ be the transition function of $\mathcal{A}$. Consider the relation:
\[F \subseteq (\bar Q \times \atoms) \times \bar Q\]
\[F = \{ \left( (q,\, a),\;  p\right) \ | \ q,p \in \bar Q,\ a \in \atoms \textrm{, such that } \delta(h(q),\, a) = h(p) \}\]
Since $r$ is surjective, $F$ satisfies the assumption from Lemma~\ref{lem:straight-uniformisation}.
We use it to obtain a \emph{finitely supported} transition function $\bar \delta : \bar Q \times \atoms \to \bar Q$.
The set of accepting states is simply the set of all representations of accepting states. For the initial
state, we repeat the uniformization construction. Define $I = \{ \bullet  \}$ to be the atomless singleton and define:  
\[
\begin{tabular}{cc}
	$B \, \subseteq \, 1 \times \bar Q$ &
	$B = \{ ( \bullet\,,\, q ) \ | \ r(q) \textrm { is the initial state of } \mathcal{A}\}$
\end{tabular}
\]
Let $\bar i$ be $B$'s uniformization, and pick $\bar i ( \bullet )$ as the initial state of the straight automaton.
This almost finishes the construction, but so far
we have constructed a \emph{finitely supported} transition function, and we require an equivariant one.
There are two ways to fix this:
First, we can show that orbit-finite sets can be represented by straight sets in a way that \emph{reflects supports}, which means that
for all $\alpha$ and $x$:
\[
    \textrm{ $\alpha$ supports $r(x)$} \Rightarrow \textrm{$\alpha$ supports $x$}
\]
With this stronger representation, we can use
Lemma~\ref{lem:straight-uniformization-eq} and obtain 
an equivariant transition function and initial state,
finishing the proof of Theorem~\ref{thm:dra-equiv-dofa}. 
Another approach to prove the theorem would be to show how to remove the surplus
atoms from the support of the straight automaton. Both of those approaches illustrate 
techniques that are going to be useful later in this thesis,
so in the next two subsections we present both of them. 
(Although any one of them could finish the proof of Theorem~\ref{thm:dra-equiv-dofa}.)

\subsection{Support-reflecting straight representations}
Most of this section is dedicated to proving the following straight representation lemma. The proof is based
on the proof of \cite[Lemma~6.2]{bojanczyk2019slightly}.
\begin{lemma}
\label{lem:straight-reflect-repr}
For every orbit-finite set $X$ there is a straight set $\bar X$ 
and a surjective representation function:
\[ f: \bar X \to X  \]
that \emph{reflects supports}, i.e. for every $x$ and $\alpha$
\[ \textrm{$\alpha$ supports $f(x)$} \implies \textrm{$\alpha$ supports $x$} \] 
\end{lemma}
\begin{proof}
Straight sets are closed under disjoint unions, so it is enough to prove the
claim for an $X$ that is a single orbit. We start the proof
by applying Claim~\ref{claim:of-representation}
to obtain an isomorphism:
\[ {\atoms^{(n)}}_{/\sim} \to X \]
This defines a natural straight representation:
\[f : \atoms^{(n)} \to X\]
This straight representation may or may not reflect supports.
If it does, then there is nothing more to do.
If it does not, then there is a $\bar x \in \atoms^{(k)}$ and $\alpha$, such that
$\alpha$ supports $f(\bar x)$, but it does not support $\bar x$. This means that,
there is an atom in $\bar x$ that is not present in $\alpha$. Assume without loss of generality
that this is $x_n$ -- the last atom of $\bar x$. Interestingly, 
this means that $f$ does not depend on its last argument:
\begin{claim}
For all pairwise distinct atoms $y_1, y_2, \ldots, y_n, y_n'$:
\[ f(y_1, \ldots, y_{n-1}, y_n) = f(y_1, \ldots, y_{n-1}, y_n') \]
\end{claim}
\begin{proof}
    Since $f$ is equivariant, we can prove the claim, by proving it for some particular
    $y_1, y_2, \ldots, y_n, y_n'$ that we choose 
    (this is because for every $z_1, \ldots, z_n, z'_n$, 
    there is a $\pi$ that maps it to $y_1, y_2, \ldots, y_n, y_n'$).
    We choose to prove it for $x_1, \ldots, x_n, x_n'$, where $x_1, \ldots, x_n$ are elements of 
    $\bar x$ (inherited from the proof of Lemma~\ref{lem:straight-reflect-repr}), and
    $x'_n$ is some atom that does not appear in $\bar x$ or in $\alpha$. 
    Define $\bar x' := x_1, \ldots, x_{n-1}, x_n'$, and let us show that 
    $f(\bar x) = f(\bar x')$. 
	Let $\pi$ be the permutation that swaps $x_n$ with $x_n'$ and does not touch other atoms.
	Such $\pi$ is an $\alpha$-permutation, so:
	\[f(\bar x) = \pi (f(\bar x)) = f( \pi (\bar x) ) = f(\bar x') \]
\end{proof}
\noindent
Using the claim, we define an equivariant $f' : \atoms^{(n-1)} \to X$ in the following way:
\[f'(\bar x) = \textrm{the only element of } \{f(\bar x, x_n)\ |\ x_n \in \atoms\} \]
In this way we obtain a representation of a lower dimension.
We repeat this process until we obtain a representation that reflects the supports. (Note that,
this process will always end, because for $n$ = 0,
we obtain a representation $I \to X$, which always reflects supports, because $I$ is atomless.)
\end{proof}
Having proved the lemma, we go back and finish the translation  \emph{Orbit-Finite Automata $\Rightarrow$ Straight Automata}:
we prove that if $r$ is a straight and support-reflecting representation of $\mathcal{A}$'s set of states,
then, by Lemma~\ref{lem:fs-functions-preserve-supports}, the following transition relation satisfies the condition 
of Lemma~\ref{lem:straight-uniformization-eq}:
\[
\begin{tabular}{cc}
	$F \subseteq (\bar Q \times \atoms) \times \bar Q$ &
	
	$F = \{ ( (\bar p,\, a),\; \bar q ) \ | \ f(r(\bar p),\; a) = r(\bar q) \}$
\end{tabular}
\]
It follows that we can use (Lemma~\ref{lem:straight-uniformization-eq}) and obtain an equivariant transition function,
for the straight automaton.\\

An interesting consequence of Lemma~\ref{lem:straight-reflect-repr} is the following
theorem about least supports. It was first proved as \cite[Proposition 3.4]{gabbay2002new}, but
the proof presented below follows the lines of \cite[Section~6]{bojanczyk2019slightly}. 
\begin{definition}
\label{def:least-support}
    We say that $\alpha \subseteq_\textrm{fin} \atoms$ is the \emph{least support} 
    of $x$ if:
    \begin{enumerate}
        \item $\alpha$ supports $x$; and 
        \item for every other $\beta$, that supports $x$, it holds that $\alpha \subseteq \beta$. 
    \end{enumerate}
    It is not hard to see that if a least support exists, then it is unique. If it exists, 
    we denote it as $\supp(x)$. 
\end{definition}

\begin{theorem}
\label{thm:least-supports}
  Every element of a set with atoms has a least support.
\end{theorem}
\begin{proof}
	Choose a set with atoms $X$ and its element $x \in X$. Let
	$X_x \subseteq X$ be the orbit of $x$ in $X$ i.e.:
	\[X_x = \{ \pi(x) \  | \  \pi \  \textrm{is an atom permutation} \} \]
	Let $h: \atoms^{(n)} \to \bar X_x$ be a straight and support-reflecting
	representation function of $X_x$ (as defined in Lemma~\ref{lem:straight-reflect-repr}).
    By combining the fact $h$ reflects supports, with Lemma~\ref{lem:fs-functions-preserve-supports},
    it is not hard to see that the least support of $h(\bar x)$ is $\sigma(\bar x)$. 
    This finishes the proof, because $h$ is surjective.
\end{proof}

\subsection{Eliminating redundant atoms and name abstraction}
\label{subsec:eliminating-redundant-atoms}
In this section, we show an alternative approach of finishing the translation \emph{Orbit-Finite Automata $\Rightarrow$ Straight automata}. 
It relies on the following lemma, which shows that we can eliminate redundant atoms from the automaton's support:
\begin{lemma}
\label{lem:aut-atom-elim}
    If a language $L$ is equivariant and recognized by a \emph{finitely supported} orbit-finite automaton\footnote{i.e. by an orbit-finite automaton whose transition function, initial state, and subset of accepting states does not
	have to be equivariant (but has to be finitely supported)}, then it is also recognized by an \emph{equivariant} orbit-finite automaton.
    Moreover, if $L$ is recognized by a \emph{finitely supported} straight automaton, then it is also recognized by an \emph{equivariant}
    straight automaton. 
\end{lemma}
To illustrate the intuition behind the proof of the lemma, we start by proving it for register automata:
\subsubsection{Register automata and atom placeholders}
A finitely supported register automaton is a register automaton, whose transition function is 
\emph{syntactically finitely supported}, i.e. it can be defined using the syntax
for an equivariant transition function extended with the following two types queries and 
one type of action:
\[
    \begin{tabular}{ccc}
    $\mathtt{r}_1 = a$,&
    $\mathtt{input} = a$,&
    $\mathtt{r}_1 := a$
    \end{tabular}
\]
It is not hard to see that a transition function is syntactically finitely supported, if and only if
it is a finitely supported function. The proof is almost the same as the one for syntactically equivariant
functions from Section~\ref{sec:dra}, but this time it uses $\alpha$-orbits, defined as:
\[\{ \pi(x) \ |\  \pi \textrm{ is an $\alpha$-permutation of atoms} \}\]

Now, let us take $L$ recognized by a finitely supported register automaton $\mathcal{A}$, and let us construct 
an equivariant $\mathcal{A}'$ that recognizes the same language.
Notice that the group of atom permutations has a natural action on the set of register automata\footnote{The action applies $\pi$ to its set of states, initial state, transition function, and set of accepting states.
The only non-equivariant of those components is the transition function, so this boils down to applying $\pi$ to the transition function.}, which means that we can 
talk about $\pi(\mathcal{A})$. It is not hard to see that the function that maps an automaton to its language is equivariant, i.e. $\pi(\mathcal{A})$ recognizes 
$\pi(L)$. Since $L$ is equivariant, it follows that for every $\pi$:
\[ \pi(\mathcal{A}) \textrm{ recognises } L\]
Every word $w \in \atoms^*$ is finitely supported, so we can always find a permutation $\pi_w$,
such $\pi_w(\mathcal{A})$ does not use any letter from $w$ as a constant. The idea of the proof is to construct one automaton $\mathcal{A}'$,
whose run on every $w$ simulates the run of $\pi_w(\mathcal{A})$. Choose a finite, atomless set $P$ of \emph{atom placeholders},
such that $|P| = |\supp(\mathcal{A})|$. This means that there exists a bijection $p : \supp(\mathcal{A}) \fsto P$
(the bijection is supported by $\supp(\mathcal{\alpha})$). Let us now consider a variant of a register automaton,
which can use its register to store both atoms or placeholders. Now, by replacing queries and actions of $\mathcal{A}$ in the following way
\[
    \begin{tabular}{ccc}
        $\mathtt{r}_1 = a$ & $\rightsquigarrow$ & $\mathtt{r}_1 = p_a$\\
        $\mathtt{input} = a$ & $\rightsquigarrow$ & $\mathtt{input} = p_a$\\
        $\mathtt{r}_1 := a$ & $\rightsquigarrow$ & $\mathtt{r}_1 := p_a$\\
    \end{tabular}
\]
we obtain an equivariant register automaton with placeholders $\mathcal{A}'$. 
Note that the query that compares an atom with a placeholder always returns ``No''. 
Since no atom from $\alpha$ appears in $w$, it follows that the run of $\mathcal{A}'$ on $w$
simulates the run $\pi_w(\mathcal{A})$ on $w$. To finish the proof, notice that a placeholder 
automaton can easily be simulated by a register automaton, that uses its control state to keep 
track of the placeholder values.\\

\subsubsection{Name abstraction}
Before we prove Lemma~\ref{lem:aut-atom-elim} in its generality, let us show how to introduce
placeholders into abstract sets with atoms. This operation is described in 
\cite[Section 4]{pitts2013nominal} under the name of \emph{name abstraction}. The placeholders are 
introduced using operations of the following kind: ``Replace all the occurrences
of an atom $a \in \atoms$ in $x$ by a placeholder'' (where $x$ is an element of a set with atoms $X$).
We denote this operation as $\abstr{a}x$. Note that it is not injective:
For example, if we consider $X = \atoms^{(2)}$, then
\[ \langle 5 \rangle (5, 2) = \langle 4 \rangle (4, 2) \]
because in both cases the first element is replaced with the placeholder. Before
we define $\abstr{a}x$ formally, let us discuss the inverse
operation: ``Replace the placeholder in $\abstr{a}x$ with $b \in \atoms$.'' This
operation is denoted as $(\abstr{a}x)@b$. If $\sigma_{(a\ b)}$ denotes the atom permutation that 
swaps $a$ and $b$, then $@$ is defined in the following way:
\[(\abstr{a}x)@b = \begin{cases}
    x & \textrm{if } a = b\\
    \sigma_{(a\ b)}\, x & \textrm{if $b \not \in \supp(x)$}\\
    \textrm{undefined} & \textrm{otherwise}	
    \end{cases}
\]
The fact that $(\abstr{a}x)@b$ is undefined when $b \neq a$ and $b \in \supp(x)$ 
represents the intuition that the placeholder is different from all the atoms in $\abstr{a}x$.
For example, if we consider $X = \atoms^{(2)}$, then:
\[
    \begin{tabular}{ccc}
        $\abstr{6}(6, 4)@3 = (3,4)$ & but & $\abstr{6}(6, 4)@4$ is undefined 
    \end{tabular} 
\]
Observe that if we had simply replaced $6$ with a $4$ in $(6, 4)$, we would have obtained $(4, 4)$
which does not belong to $\atoms^{(2)}$.\\

Intuitively, in order to check if
$\abstr{a}x$ is equal to $\abstr{b}y$, we take a \emph{fresh} $c \in \atoms$ -- i.e.
such $c$ that $c \not \in (\{a, b\} \, \cup \, \supp(x) \, \cup \, \supp(y)$) --
and check if $\abstr{a}x@c = \abstr{b}y@c$. Notice that this does not 
depend on the choice of $c$ -- if this equality holds for some fresh $c$, then it holds
for all fresh $c$. We can define this as the following relation:
\[ \abstr{a}x \sim  \abstr{b}y \iff \abstr{a}x@c = \abstr{b}y@c \]
As shown in \cite[Lemma 4.1]{pitts2013nominal}, this is an equivalence relation.
We use it to define the set of all possible $\abstr{a}x$'s, denoted as $[\atoms]X$:
\[ [\atoms]X = {(\atoms \times X)_{/\sim}}\comma \] 
We define $\abstr{a}x$ to be the equivalence class of $(a, x)$. For example:
\[ [\atoms]\atoms^2 \simeq
 \underbrace{\atoms^2}_{\substack{
 		\textrm{none of the atoms}\\
 		\textrm{is the placeholder}\\
 		\textrm{e.g.} \abstr{5}(1, 2)}}
 +
 \underbrace{\atoms}_{\substack{
 		\textrm{the first atom }\\
 		\textrm{is  the placeholder}\\
 		\textrm{e.g.} \abstr{5}(5, 2)}}
 +
 \underbrace{\atoms}_{\substack{
 		\textrm{the second atom}\\
 		\textrm{is  the placeholder}\\
 		\textrm{e.g.} \abstr{5}(1, 5)}}
 +\underbrace{1}_{\substack{
 		\textrm{both of the atoms}\\
 		\textrm{are placeholders}\\
 		\textrm{e.g.} \abstr{5}(5, 5)}}\]
Notice that $[\atoms]X$ is a set with atoms, with the following action of the group permutation:
\[ \pi (\abstr{a}x) = \abstr{\pi a} (\pi x)\]
Operation $[\atoms](\cdot)$ commutes with many classical combinators:
\begin{lemma}
\label{lem:abstr-products-sums-words}
For all sets with atoms $X$ and $Y$:
\[
\begin{tabular}{ccc}
$[\atoms](X \times Y) \simeq [\atoms]X \times [\atoms]Y$&
$[\atoms](X + Y) \simeq [\atoms]X + [\atoms]Y$&
$[\atoms](X^*) \simeq ([\atoms]X)^*$ 	
\end{tabular}
\]
\end{lemma}
\begin{proof}
The isomorphism for products and coproducts is proved in \cite[Equations~4.27,~4.28]{pitts2013nominal}.
For the words, we use the following isomorphism:
\[W : [\atoms](X^*) \to ([\atoms]X)^*\]
\[
    \begin{tabular}{cc}
        $W(w) = (\abstr{a}(w@a)_1, \abstr{a}(w@a)_2, \ldots, \abstr{a}(w@a)_n)$ & $\substack{\textrm{where } a \textrm{ is any atom}\\
                                                                            \textrm{that does not appear in } \supp(w)}$
    \end{tabular}
\]
Note that $W(w)$ does not depend on the choice of $a$.
\end{proof}
\noindent
Let us cite two important properties of $[\atoms]$:
\begin{lemma}[{\cite[Proposition~4.5]{pitts2013nominal}}]
\label{lem:abstr-removes-atom}
For every $x \in X$ and every atom $a$:
\[\supp(\abstr{a}x) = \supp(x) - \{a\} \]
\end{lemma} 
\begin{lemma}[{\cite[Proposition~4.14]{pitts2013nominal}}]
\label{lem:abstr-fun}
The following sets are isomorphic:
	\[ [\atoms](X \fsto Y) \, \simeq \, [\atoms]X \fsto [\atoms]Y \]
The isomorphism is given by the following formula:
\[(\abstr{a}f)x = \abstr{b}\left((\abstr{a}f)@b\right)(x@b)\comma \]
where $b \in \atoms$ is some atom that does not appear in $x$ or in $\supp(\abstr{a}f)$.
(As usual, the result does not depend on the choice of $c$.)
\end{lemma}
\noindent
Now let us show that $\abstr{c}(f)$ preserves identities and function compositions\footnote{This means 
that the following mapping is a functor: $X \mapsto [\atoms]X$, $f \mapsto \abstr{c}f$.
To avoid potential confusion, it is worth pointing out that this functor
very similar to the $[\atoms]$ functor defined in \cite[Section~4.4]{pitts2013nominal}, 
but not exactly the same. For this reason the proof of the lemma is very similar 
to the proof of \cite[Lemma~4.10]{pitts2013nominal}}. 
\begin{lemma}
\label{lem:abs-functor}
For every $\abstr{a} id_X = id_{[\atoms]X}$ and $\abstr{a}(f \circ g ) = (\abstr{a} f) \circ (\abstr{a} g)$.
\end{lemma}
\begin{proof}
    We start with two claims that follow immediately from the definitions:
    \begin{claim}
    \label{claim:swap-placeholder}
        If $b$ is fresh for $\abstr{a}x$, then $\abstr{a} x = \abstr{b}( \swap{a}{b} x)$. 
    \end{claim}
    \begin{claim}
    \label{claim:abstr-arg}
    If $a$ is fresh for $x$, then $(\abstr{a} f) x = \abstr{a} (f (x@a))$.
    \end{claim}
    \noindent
    Let us take some $x \in X$ and an atom $b$ that is fresh for $x$, and notice that:
    \[  (\abstr{a}  id_X) \, x \stackrel{\textrm{Claim~\ref{claim:swap-placeholder}}}{=}
        (\abstr{b} id_X) \, x \stackrel{\textrm{Claim~\ref{claim:abstr-arg}}}{=}
        \abstr{b} (x@b) = x\]
    For the second part, we need to show that for all $a$, $f$, $g$, $x$:
    \[(\abstr{a} f) \, \left( (\abstr{a} g) \, x \right) = (\abstr{a} (f \circ g)) \, x\] 
    Take some $a, f, g, x$, and let $b$ be a fresh atom. Define $f' := \sigma_{(a b)} f$,  $g' := \sigma_{(a b)} g$, and observe that:
    \[(\abstr{a} f)\; ((\abstr{a} g)\; x) \stackrel{\textrm{Claim~\ref{claim:swap-placeholder}}}{=} (\abstr{b} f')\; ((\abstr{b} g')\; x) \stackrel{\textrm{Claim~\ref{claim:abstr-arg}}}{=} (\abstr{b} f')\; (\abstr{b} (g'\, x@b)) \stackrel{\textrm{Claim~\ref{claim:abstr-arg}}}{=} \]
    \[=  \abstr{b}\left( f' (\abstr{b} (g'\, x@b)@b) \right) 
    = \abstr{b} (f'\, (g'\, x@b)) =\] \[= \abstr{b} ((f' \circ g') \; x@b)  \stackrel{\textrm{Claim~\ref{claim:abstr-arg}}}{=} (\abstr{b}(f' \circ g')) \; x \stackrel{\textrm{Claim~\ref{claim:swap-placeholder}}}{=} \abstr{a} (f \circ g) \; x\]
\end{proof}
\noindent

\noindent
Finally, we show one more property of $\abstr{c}(\cdot) : X \fsto [\atoms] X$:
\begin{lemma}
\label{lem:abstr-natural}
	For all $a \in \atoms$, $f : X \to Y$ and $x \in X$ it holds that:
	\[\abstr{a}(f x) = (\abstr{a} f)(\abstr{a} x)\]
	This means that the following diagram commutes\footnote{This means that $\abstr{c}(\cdot) : X \fsto [\atoms]X$
    is a natural transformation between the identity functor and the functor $\abstr{c}(\cdot)$ from the previous footnote. The same is true for the functor $[\atoms]$  (see \cite[Equation~4.19]{pitts2013nominal} and the previous footnote).} 
    
    \vspace{0.3cm}
    \adjustbox{width=0.4\textwidth, center}{
    \begin{tikzcd}
        X && Y \\
        {[\mathbb{A}]X} && {[\mathbb{A}]Y}
        \arrow["f", from=1-1, to=1-3]
        \arrow["{\langle c \rangle(\cdot)}"', from=1-1, to=2-1]
        \arrow["{\langle c \rangle f}", from=2-1, to=2-3]
        \arrow["{\langle c \rangle(\cdot)}"', from=1-3, to=2-3]
    \end{tikzcd}
        }
    \vspace{0.3cm}
\end{lemma}
\begin{proof}The lemma follows from Claim~\ref{claim:abstr-arg}:
	\[(\abstr{a} f)(\abstr{a} x) \stackrel{\textrm{Claim~\ref{claim:abstr-arg}}}{=} \abstr{a} (f \; \abstr{a}x@a) = \abstr{a} (f x)\]
\end{proof}

\subsubsection{Redundant atoms in orbit-finite automata}

We are now ready to prove the general version of Lemma~\ref{lem:aut-atom-elim}. Take an orbit-finite automaton $\mathcal{A}$
that recognizes an equivariant language. We pick some $a \in \supp(\mathcal{A})$ and we construct an
automaton $\abstr{a}\mathcal{A}$ such that:
\begin{enumerate}
    \item  $\abstr{a}\mathcal{A}$ recognizes the same language as $\mathcal{A}$;
    \item  $\supp(\abstr{a}\mathcal{A}) = \supp(\mathcal{A}) - \{ a \}$; and
    \item  if $\mathcal{A}$ is straight, then so is $\abstr{a}\mathcal{A}$. 
\end{enumerate}
This is enough to prove the lemma, because if we repeat this construction $|\supp(\mathcal{A})|$ times, we obtain an equivariant automaton, recognizing
the same language as $\mathcal{A}$. Thanks to the third assumption, this construction preserves 
straight automata.

\noindent
If $\mathcal{A} = (Q, \Sigma, q_0, Q_{acc}, f)$, then we define  $\abstr{a}\mathcal{A'}$ as:
\[ 
    \begin{tabular}{ccccc}
        $([\atoms]Q,$ & $[\atoms]\Sigma,$ & $\abstr{a}(q_0),$ & $\{ \abstr{a}q \ | \ a \in Q_{acc}\},$ & $\abstr{a}(\delta))$
    \end{tabular}
\]
This definition results in a slight mismatch of alphabets -- the alphabet of $\abstr{a}\mathcal{A}$ is $[\atoms]\Sigma$, but we want it
to recognize languages over $\Sigma$. The following lemma shows a natural way to inject $\Sigma$ into $[\atoms]\Sigma$:

\begin{claim}
\label{claim:abstr-embedding}
For every $X$, the following $\iota_X$ function is an injection $X\hookrightarrow \atoms[X]$:
\[ \begin{tabular}{cc}
 $\iota_X(x) = \abstr{a}x$ & for any $a \not \in \supp(x)$  	
 \end{tabular}
\]
\end{claim}
\begin{proof}
It is not hard to see that the definition does not depend on the choice of $a$. 
To show that that $\iota_X$ is an injection, let take some $x, y \in X$,
such that $\iota_X(x) = \iota_X(y)$, and show that $x=y$. If we pick $a$ that 
is fresh for $x$ and $y$, then:
\[x = (\abstr{a}x)@a = \iota_X(x)@a = \iota_X(y)@a = (\abstr{a}y)@a = y\]
\end{proof}
\noindent

Formally, this means that, when $\abstr{a}\mathcal{A}$ is used to recognize languages over the alphabet $\Sigma$,
it has the following transition function:
\[ \begin{tabular}{ccc}
 	$\abstr{a}\delta(\iota_\Sigma (\cdot), \cdot): (\Sigma \times [\atoms] Q) \; \to \; [\atoms]Q$ & defined as & $(a, q) \mapsto (\abstr{a} \delta)(\iota_\Sigma \, (a), q)$
  \end{tabular}
\]

Thanks to Lemma~\ref{lem:abstr-removes-atom}, we know that
$\supp(\abstr{a}\mathcal{A}) = \supp(\mathcal{A}) - \{a\}$, 
and it is also not hard to show that if $Q$ is straight then so is $[\atoms]Q$.
This leaves us with showing that $\abstr{a} \mathcal{A}$ recognizes
the same language as $\mathcal{A}$.
This proof is similar to the proof for register automata. 
Take some $w \in \Sigma^*$ and an atom $b$ that is fresh for $w$ and $\mathcal{A}$, and 
define $\mathcal{A}'_w := \swap{a}{b} (\mathcal{A})$.
Since the language of $\mathcal{A}$ is equivariant, we know that
$\mathcal{A}$ and $\mathcal{A}'_w$ recognize the same language.  
Moreover, since $b$ is fresh for $\mathcal{A}$, it follows from Claim~\ref{claim:swap-placeholder}
that $\abstr{a}\mathcal{A} = \abstr{b}\mathcal{A}'_w$. This means that 
it is enough to show that $\abstr{b}\mathcal{A}'_w$
accepts $w$ if and only if $\mathcal{A}$ accepts $w$.
The key observation is that since $b$ is fresh for $w$, then
for every $i$, it holds that $\iota_\Sigma (w_i) = \abstr{b}(w_i)$. It follows that:
\[ (\abstr{b} \delta)(\iota_\Sigma (w_i), \cdot) = (\abstr{b} \delta)(\abstr{b} w_i, \cdot)\]
Notice that $b$ is fresh for every state that appears in the run of $\abstr{b}\mathcal{A}'_w$ on $w$ --
this follows from Lemma~\ref{lem:fs-functions-preserve-supports}, because $b$ does not appear in the initial state, in the transition function, or in any $w_i$.
It follows that we can use Claim~\ref{claim:abstr-arg} to further simplify the transition function:
\[ (\abstr{b} \delta)(\abstr{b} w_i, q) = \abstr{b} (\delta(w_i, q@b))  \]
By the formula from Lemma~\ref{lem:abstr-fun}, it follows that the transition function for $w_i$ is equal to:
\[ \abstr{b}(\delta(w_i, \cdot )) \]
If we treat, the initial state as a function $1 \fsto Q$, and the accepting set as a function $Q \fsto \{\Yes, \No\}$, 
then this leaves us with showing that the following diagram commutes:\\
\bigpicc{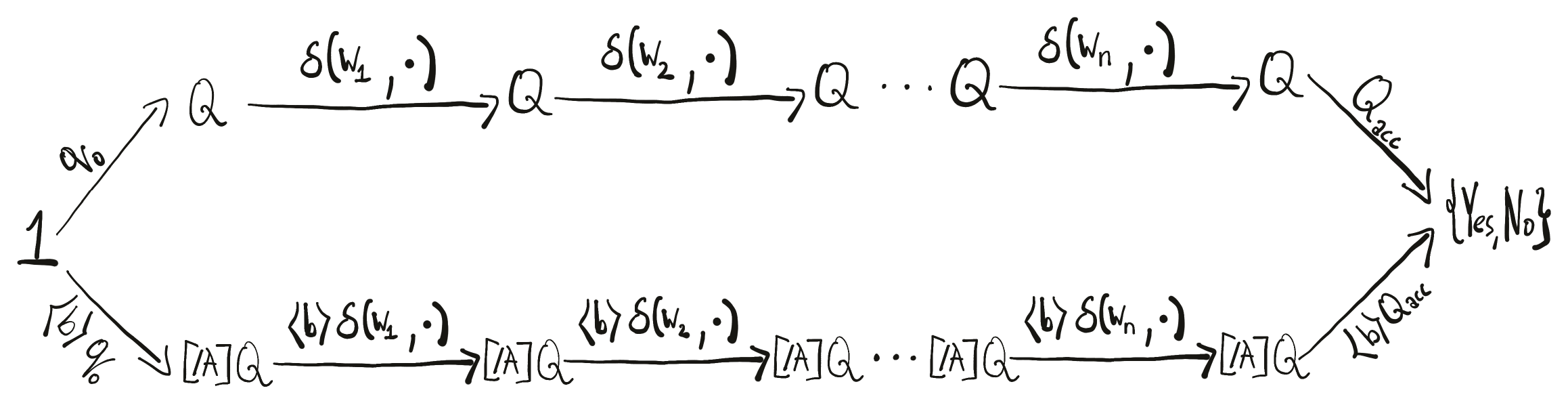}
\noindent
To prove that, we draw some auxiliary arrows so that each face of the diagram becomes an instance of Lemma~\ref{lem:abstr-natural}:
\bigpicc{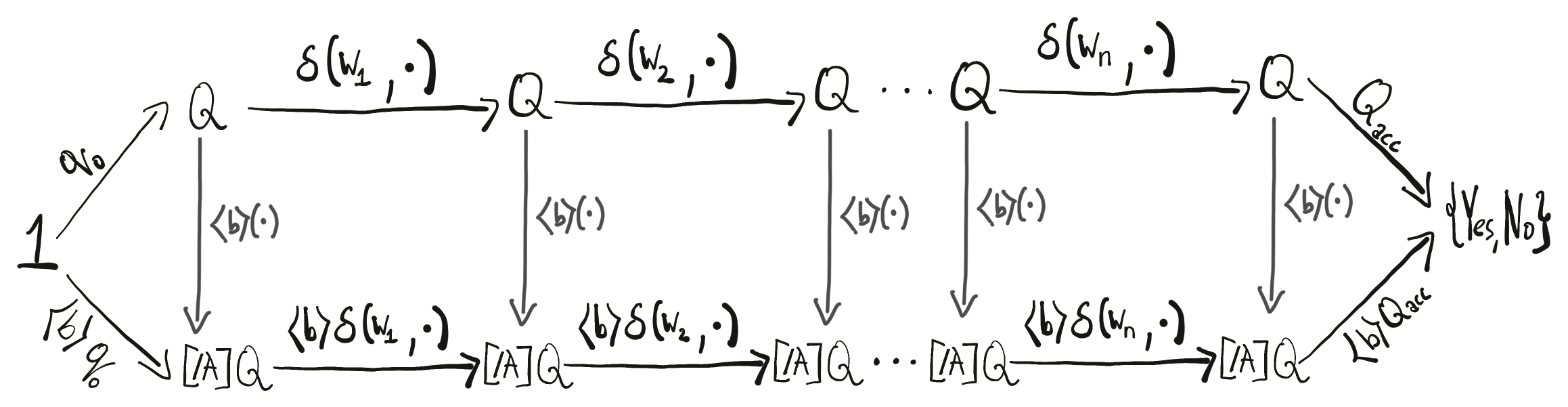}
\noindent
This finishes the proof of Lemma~\ref{lem:aut-atom-elim}, which in turn finishes the (second) proof of Theorem~\ref{thm:dra-equiv-dofa}.

\section{Other models for infinite alphabets}
The theory of infinite alphabets is notorious for the abundance of non-equivalent models that it features --
the expressive powers of different models are well illustrated by
\cite[Figure~1.1]{bojanczyk2019slightly} or \cite[Figure~1]{neven2004finite}.
In this section we present the following part of this landscape:
\begin{enumerate}
\item Right-to-left deterministic orbit-finite automata;
\item Two-way deterministic orbit-finite automata;
\item One-way nondeterministic orbit-finite automata;
\item Orbit-finite monoids
\end{enumerate}
The following picture summarizes the results presented in this section:
\smallpicc{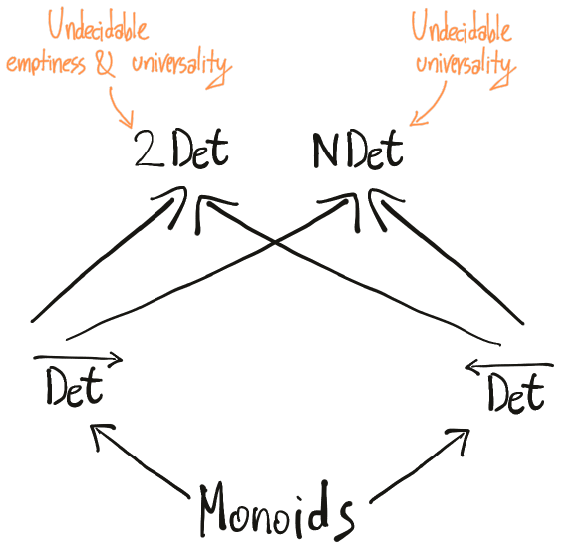}
\subsection{Variants of orbit-finite automata}
\label{sec:dofa-variants}
We begin the discussion with variants of orbit-finite automata.
This section is based on \cite[Section 1.4]{bojanczyk2019slightly}:

\subsubsection{Right-to-left orbit finite automata}
Notice that the class of languages recognized by orbit-finite automata is not closed under the reverse:
\begin{lemma}
\label{lem:dofa-reverse}
	There is a language $L$, such that it is recognized by a deterministic orbit-finite automaton, but its
    reverse is not. 
\end{lemma}
\begin{proof}
	Consider the following language over the alphabet $\atoms^*$:
	\[ L_{\textrm{First}} = \textrm{``The first letter appears again''} \]
	As shown by Example~\ref{ex:appears-again}, it can be recognized by an orbit-finite register automaton. Consider now its reverse:
	\[ L_{\textrm{Last}} =  \textrm{``The last letter appears before''} \]
	We show that this language cannot be recognized by a deterministic orbit-finite automaton.
	First, let us notice that every orbit-finite set has a limit on the size of its supports:
	\begin{claim}
    \label{claim:of-lim-support}
	For every orbit-finite set $A$ there is a number $k$ such 
    that for every $a \in A$, it holds that $|\supp(a)| \leq k$. 
	\end{claim}
	\begin{proof}
	The following function is equivariant:
	\[ a \mapsto |\supp(a)| \]
	It follows that the size of
	least supports is
	fixed in every orbit. There are only 
	finitely many orbits in $A$, so we can set $k$ to be the maximal size of the support among those orbits.
	\end{proof}
	Suppose that $L_{\textrm{Last}}$ is recognized by a deterministic orbit-finite automaton
	$\mathcal{A}$ and let $k$ be the maximal support size for the states of $\mathcal{A}$.
	Let $q$ be the state of $\mathcal{A}$ after it has read a prefix consisting of $k + 1$ different atoms:
	\[ \begin{tabular}{cccc} 
 			$a_1$ & $a_2$ & $\ldots$ & $a_{k + 1}$
       \end{tabular}
    \]
    
    Since $\mathcal{A}$'s states have supports of size at most $k$, one of the input letters
    must be fresh for $q$ (i.e. $a_i \not \in \supp(q)$). Say that it is $a_j$.
    This leads to a contradiction because, if we choose $a_{k+2}$
    different from all $a_i$'s, then $\mathcal{A}$ cannot distinguish between:
    \[ \begin{tabular}{ccc}
 		$a_1\ a_2\ \ldots\ a_{k+1}\ a_j \in L_\textrm{Last}$ & and &$a_1\ a_2\ \ldots\ a_{k+1}\ a_{k+2} \not \in L_\textrm{Last}$
      \end{tabular}
    \]	
\end{proof}
A consequence of Lemma~\ref{lem:dofa-reverse} is that deterministic (left-to-right)
orbit-finite automata recognize a different class of languages than their right-to-left counterparts.\\

We finish the discussion on right-to-left and left-to-right deterministic orbit-finite automata,
by mentioning that both of those models are computationally quite simple: They have decidable\footnote{
    Orbit-finite automata are infinite objects, so before discussing the decidability of their properties,
    we should first briefly mention how to represent them. This involves representing atoms,
    orbit-finite sets, and equivariant functions between them. Here is one way to do this:
    (a) atoms are represented as elements of $\nat$; (b) orbit-finite sets are represented as lists of one representative per orbit; 
    (c) equivariant functions are represented by their behaviour on these representatives.
    For more information on representing orbit-finite sets and their decision problems, 
    see \cite[Chapter~4~and~Part~III]{bojanczyk2019slightly}. 
} emptiness (special case of \cite[Theorem 1.7]{bojanczyk2019slightly})
and, since they are closed under complements, decidable universality. 

\subsubsection{Nondeterministic orbit-finite automata}
The definition of a \emph{nondeterministic orbit-finite automaton} is not surprising --
it looks just like a deterministic orbit-finite automaton, but (a) it may have more than one 
initial state and (b) instead of transition functions, it has a transition relation:
\[ \delta \subseteq_{\textrm{eq}} Q \times \Sigma \times Q \]
Nondeterministic orbit-finite automata can recognize the language
\[ \textrm{``Last letter appears before'':} \]
A nondeterministic automaton can nondeterministically guess a position,
save its letter to a register, and at the end of the word 
verify that it is equal to the last letter. (In general, it is easy to see that the languages
recognized by nondeterministic automata are closed under the reverse.)\\

It follows that
nondeterministic orbit-finite automata are strictly more expressive than their deterministic counterparts.
A reason for this mismatch in expressive powers is that orbit-finite sets are not closed under
finitely supported powersets and therefore they do not admit the powerset construction.\\

Non-deterministic orbit-finite automata are computationally quite complex: They have decidable emptiness \cite[Theorem 1.7]{bojanczyk2019slightly},
but undecidable universality\footnote{It might be worth mentioning that universality is decidable
for the \emph{unambiguous} nondeterministic orbit-finite automata. One way to show this is to show that 
orbit-finite weighted automata have decidable equivalence. See \cite{bojanczyk2021orbit} for details.}
\cite[Theorem 1.8]{bojanczyk2019slightly}.


\subsubsection{Two-way deterministic orbit-finite automaton}
\label{subsec:2dra}
Two-way deterministic orbit-finite automaton is variant of a deterministic orbit-finite automaton that is not forced to read the input left-to-right or right-to-left.
Instead, in each transition, it decides whether it wants to go left or right.
When it leaves the word, it gets notified and might choose to go back.
To finish its run it has to explicitly accept or reject the input\footnote{It may also loop and never finish its run.
We assume that this means that the automaton rejects the input.}. To accommodate those features, 
its transition function has the following type:

\[ (\Sigma + \underbrace{\{\vdash, \dashv\}}_{\substack {\textrm{end of word}\\
\textrm{markers}}} ) \times Q \  \ \longrightarrow_\textrm{eq} \ \  Q \times \{\leftarrow, \rightarrow \} + \{\textrm{accept}, \textrm{reject} \} \]

Two-way orbit-finite automata extend two-way finite automata (defined in \cite[Defniiton~2]{shepherdson1959reduction}).
Over finite alphabets, two-way and one-way automata are equivalent (\cite[Theorem~2]{shepherdson1959reduction}). 
However, this equivalence does not hold for orbit-finite alphabets: It is not hard to see that 
two-way orbit-finite automata can recognize both $L_{\textrm{first}}$ and $L_{\textrm{last}}$. 
It follows they are strictly stronger than one-way deterministic automata.\\

Two-way register automata are computationally very strong.
The following theorem says that they belong more to complexity theory
than to automata theory.
\begin{theorem}[{\cite[Theorem~3.8~(b)]{neven2004finite}}]
\label{thm:2DOFA-logspace}
Take a language over a finite alphabet $L \subseteq \Sigma^*$
and define $L_{\atoms} \subset {(\Sigma \times \atoms)}^*$
to be the language of those words whose $\Sigma$-part belongs to $L$
and whose atoms are pairwise distinct. Then:
\[  L \in \textsc{LogSpace} \iff L_\atoms 
\textrm{ is recognisable by a two-way orbit-finite automaton}
\]
\end{theorem}
\begin{proof}
	
	The key idea is that when the input alphabet is $\Sigma \times \atoms$,
	and all the atoms are distinct, then remembering the atom value is the
	same as remembering the atom's position.\\ 

	We start with the easier right-to-left implication ($\Leftarrow$):
    The alphabet $\Sigma \times \atoms$
	is straight, so using the techniques presented
	in the proof of Theorem~\ref{thm:dra-equiv-dofa}, 
	we can transform a two-way orbit-finite automaton that recognizes $L_\atoms$
	into a two-way register automaton over the alphabet $\Sigma \times \atoms$
    (this type of an automaton is a natural extension of the standard register automaton).
    In order to simulate this two-way register automaton with 
    a $\textsc{LogSpace}$ Turing machine, we notice
    that 
    instead of storing an atom, it is enough to 
    store its position (in binary).\\

   The ($\Rightarrow$) implication is more involved. We
   start with the following claim:
   \begin{claim}
   \label{claim:2dofa-all-distinct}
   There is a two-way orbit-finite automaton, recognizing
   the language ``Each letter appears only once''
   (over $\atoms$).
   \end{claim}
   \begin{proof}
   We show that a two-way orbit-finite automaton can simulate the following program:
  \begin{alltt}
  {\bf for} i {\bf in} 1..n:
      r = atom at i-th position
     {\bf for} j {\bf in} (i + 1)..n:
         {\bf if} r is equal to the atom at j-th position:
             {\bf reject}
  {\bf accept}	
  \end{alltt}
  The execution is mostly straightforward. The only problem appears, when the automaton finishes the inner loop, and has to
  go back to position $i + 1$. At this point the automaton
  knows that the value $r$ appears exactly once in the word, so it can go back to the $i$th position by locating the value $r$.
   \end{proof}
 We are left with showing that once a two-way automaton has checked that all the input atoms are pairwise distinct, it can simulate a
   logspace Turing machine. For this we use an intermediate model of \emph{two-way multi-head deterministic finite automaton}.
   Such an automaton resembles a two-way deterministic finite automaton, but it has many heads that move independently
   from each other. A transition function of a two-way multi-head deterministic finite automaton has the following type:
 \[ Q \times \underbrace{(\Sigma + \{\vdash, \dashv\})^k}_{\textrm{what each head is seeing}} \  \longrightarrow \  Q \times \underbrace{\{{ \leftarrow, \rightarrow\}^k}}_{\substack{
 \textrm{instructions}\\
 \textrm{for each of the heads}}} + \  \{\textrm{accept}, \textrm{reject}\} \] 
 Here is an example configuration of a multi-head automaton:
 \picc{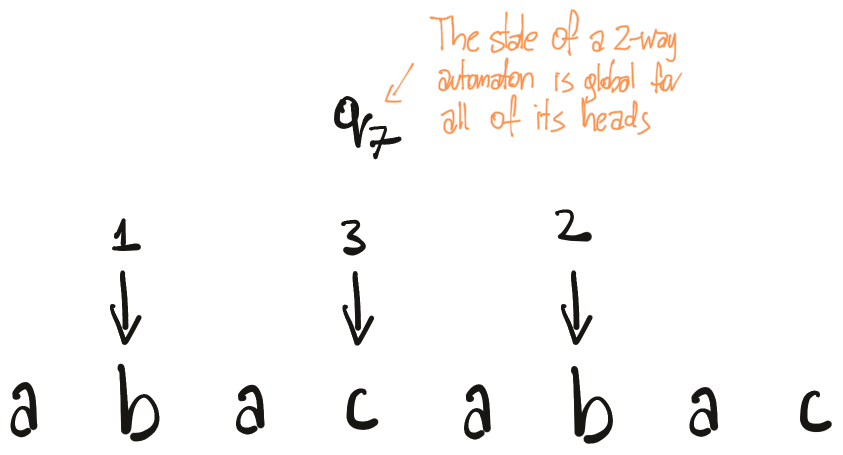}
 The multi-head automata are equivalent to logspace Turing machines in the following sense\footnote{
    To the best of my knowledge, the earliest reference to this result is \cite[Corollary 3.5]{ibarra1971characterizations}. 
    However, in place of a proof, the author states that this is a well-known unpublished result by Alan Cobham and others.
    The earliest reference that contains the proof is \cite[Page~338]{hartmanis1972non}. I would like to thank
    Nguy\~{\^e}n L{\^e} Th{\`a}nh D\~{u}ng for pointing me to those references. 
 }:
 \begin{theorem}
 A language $L \subseteq \Sigma^*$ belongs
 to $\textsc{LogSpace}$ if and only if it
 is recognizable by a two-way deterministic
 multi-head automaton.
 \end{theorem}
 So it suffices to show that, as long as
 its input is equipped with a unique atom on
 every position, a two-way orbit-finite automaton 
 can simulate a deterministic two-way multi-head
 automaton. This is straightforward:
 Instead of remembering the position of each head, 
 an orbit-finite automaton can simply remember the atom
 from that position. For example, the configuration from the example above
 can be represented as follows:
 \picc{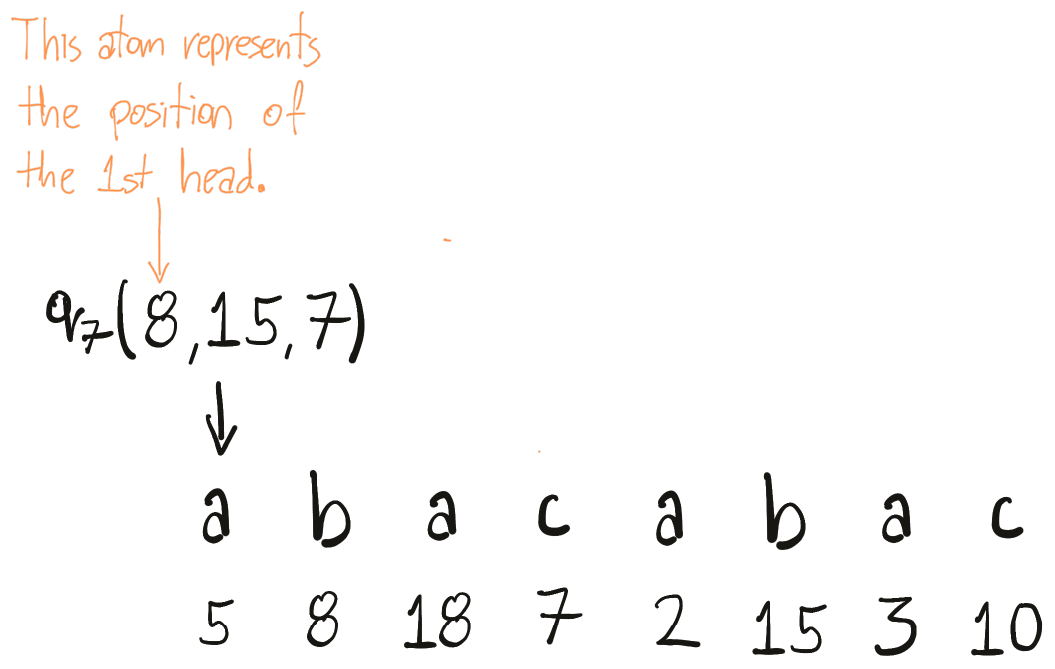}
With this representation, a two-way orbit-finite automaton can simulate a transition of a multi-head automaton, by
doing two sweeps: a left-to-right one to see the letter under each of the heads;
and a right-to-left one to perform necessary updates in the heads' positions.
\end{proof}
The emptiness of $\textsc{LogSpace}$-Turing machines is easily seen to be undecidable, so 
a consequence of Theorem~\ref{thm:2DOFA-logspace} is that two-way deterministic orbit-finite automata, 
have undecidable emptiness. As a deterministic model, they are closed under complement, 
which means that they also have undecidable universality.\\

We finish this section, by showing that one-way nondeterministic orbit-finite automata 
and two-way deterministic orbit-finite automata have incomparable expressive powers:

\begin{lemma}[{\cite[Example~11]{kaminski_finite_memory_paper}}]
\label{lem:2-way-det-not-in-1-way-ndet}
    There is a language $L$, that can be recognized by a two-way deterministic orbit-finite automaton,
    but cannot be recognized by a one-way nondeterministic orbit-finite automaton.
\end{lemma}
\begin{proof}
    One example of such $L$ is the language ``Each letter appears only once''.
    According to Claim~\ref{claim:2dofa-all-distinct} it can be recognized by a two-way deterministic orbit-finite automaton.
    It suffices to show that it cannot be recognized by a one-way nondeterministic orbit-finite automaton:
    Suppose that it is recognized by $\mathcal{A}$ whose set of states is $Q$.
    Let $k$ be maximal support size in $Q$ (see Claim~\ref{claim:of-lim-support}). Consider
    the word consisting of $k + 2$ pairwise distinct atoms:
    \[ \begin{tabular}{ccccc} 
        $a_1$ & $a_2$ & $\ldots$ & $a_{k + 1}$ & $a_{k + 2}$
       \end{tabular}
    \]
    This word belongs to $L$, so $\mathcal{A}$ has to have an accepting run on it. Let $q_k$ be
    $\mathcal{A}$'s state in this run just after it has processed $a_k$. One of the $a_i$ (for $i \leq k + 1$)
    is not present in the support of $q_k$. It follows that $\mathcal{A}$ also
    has to have an accepting run on the following word, which does not belong to $L$:
    \[ \begin{tabular}{ccccc} 
        $a_1$ & $a_2$ & $\ldots$ & $a_{k + 1}$ & $a_i$
       \end{tabular}
    \]
    This contradicts the assumption that $\mathcal{A}$ recognizes $L$.
\end{proof}

\begin{lemma}[{\cite[Exercise~26]{bojanczyk2019slightly}}]
\label{lem:1-way-det-not-in-2-way-ndet}
        There is a language $L$, that can be recognized by a one-way nondeterministic orbit-finite automaton,
        but cannot be recognized by a two-way deterministic orbit-finite automaton. (This lemma is conditional on 
        $\textrm{\textsc{LogSpace}} \neq \textrm{\textsc{NLogSpace}}$.)
\end{lemma}
\begin{proof}
    We start by defining the domino language over the alphabet $\atoms^2$:
    \[L_{\textrm{Domino}} = \{ (a_1, a_2) (a_2, a_3) (a_3, a_4) \ldots (a_n, a_{n+1}) \ | \   a_1, \ldots, a_{n+1} \in \atoms \ \} \]
    Now, we define the language $L_{\textrm{SubDomino}}$ of words that have a valid domino subsequence that contains the first and the last letter.
    For example, the following word belongs to $L_{\textrm{SubDomino}}$ (its valid domino subsequence has been underlined):
    \[ \begin{tabular}{ccccccc}
        $\underline{(7, 5)}$ & $(1, 2)$ & $(5, 9)$ & $\underline{(5, 4)}$ & $(8, 2)$ & $\underline{(4, 3)}$ & $\underline{(3, 2)}$\\
    \end{tabular} \]
    It is easy to see that the language $L_{\textrm{SubDomino}}$ is recognized by a nondeterministic orbit-finite automaton.
    On the other hand, it can be shown that if $L_{\textrm{SubDomino}}$ is
    recognized by a deterministic two-way automaton, then the $\textsc{NLogSpace}$-complete problem of reachability in directed acyclic graphs belongs to $\textsc{LogSpace}$.
    (See \cite[Exercise~26]{bojanczyk2019slightly} for details.)
\end{proof}

\subsection{Orbit-finite monoids}
\label{subsec:orbit-fintie-monoids}
We finish this introduction to languages over infinite alphabets with one more model -- \emph{orbit-finite monoids}.
As noted in the introduction, it plays  a central role in this thesis. This section is based on
$\cite{bojanczyk2019slightly}$, where the model was first introduced.\\ 

First we discuss the well-established model of \emph{finite monoids} (for recognizing languages).
A \emph{monoid} is a set $M$, equipped with an associative binary operation $(\cdot)$ and a neutral element $1$.
This means that, for all $a, b, c \in M$ it holds that
\[ 
    \begin{tabular}{ccc}
        $a \cdot (b \cdot c) = (a \cdot b) \cdot c$ & and & $1 \cdot a = a \cdot 1 = a$.
    \end{tabular}
\]
A monoid is finite, if the set $M$ is finite. 
A finite monoid $M$ together with a function $h : \Sigma \to M$,
and an accepting subset $F \subseteq M$ can be used to recognize 
a language over the finite alphabet $\Sigma$:
to see if a word $w \in \Sigma^*$ belongs to
the language, we check whether:
\[ h(w_1) \cdot h(w_2) \cdot \ldots \cdot h(w_n) \in F \]
\begin{example}
\label{ex:no-rep-monoid}
For example, consider the following language (over a finite alphabet $\Sigma$):
\[\textrm{``No letter appears twice in a row''} \]
It is recognized by the following finite monoid:
\[M = \underbrace{\Sigma^2}_{\substack{
    \textrm{$(a, b$) represents all non-empty words that}\\
    \textrm{do not contain a repetition; }\\
    \textrm{start with $a$; and end with $b$}
}} + \underbrace{\bot}_{\substack{
    \textrm{represents}\\
    \textrm{all the words that}\\
    \textrm{contain a repetition}
}} + \underbrace{1}_{\substack{
    \textrm{represents}\\
    \textrm{the empty word}
}}\]
$M$'s operation is defined as follows:
\[\begin{tabular}{ccc}
	$\bot \cdot x = x \cdot \bot = \bot$,&
	$1 \cdot x = x \cdot 1 = x$,&
	$(x_1, y_1) \cdot (x_2, y_2) = \begin{cases}
		(x_1, y_2) & \textrm{if } y_1 \neq x_2\\
		\bot & \textrm{otherwise}
 \end{cases}$
\end{tabular}
\]
The function $h$ maps a letter $x$ into $(x,x)$, and the accepting subset $F$ is
$M - \{\bot\}$.\\
\end{example}

It is a well-known fact, that the class of languages recognized by finite monoids
is exactly equal to the class of regular languages:
\begin{lemma}
\label{lem:aut-mon-eq}
   A language is recognized by an finite monoid, if and only if it is 
   recognized by a deterministic finite automaton.
\end{lemma}
\begin{proof}
    $(\Rightarrow)$: Let $L \subseteq \Sigma^*$ be a language recognized by a finite monoid $M$
    (together with $h$ and $F$). It is not hard to see that $L$ is also recognized by an automaton
    where the set of states is $M$, the initial state is $1$, the accepting subset of states is $F$, 
    and the transition function is given as:
    \[ \delta(q, a) = q \cdot h(a)\]
    $(\Leftarrow)$: Let $L$ be a language recognized by a finite automaton $\mathcal{A}$ whose 
    set of states is equal to $Q$. It follows that $L$ is recognized by the monoid $Q \to Q$, 
    whose operation is given as $f \cdot g = g \circ f$, together with the following $h$ and $F$:
    \[
        \begin{tabular}{cc}
            $h(a) = (q \mapsto \delta(q, a))$ & $F = \{ f \ | \ f \in Q \to Q,
            \substack{\textrm{ such that } f(q_0)\\ \textrm{ is an accepting state}  }\}$
        \end{tabular}
    \]
    The intuition behind this construction is that a word $w$ can be characterized 
    by its behaviour function $b_w \in Q \to Q$:
    \[ b_w(q) = \substack{\textrm{In which state will $\mathcal{A}$ exit $w$ on the right,}\\\textrm{ if it enters $w$ on the left in the state $q$?}} \]
    (For a more detailed explanation see \cite[Theorem~22]{bojanczyk2020languages}, 
    or Section~\ref{subsec:one-way-sua-to-ofm}).
\end{proof}

Now, let us extend the theory of finite monoids to sets with atoms:
An \emph{orbit-finite monoid}, is a monoid whose underlying set (i.e. $M$) is orbit-finite
and whose operation ($\cdot$) is an equivariant function. 
The language recognized by an orbit-finite monoid $M$ is defined in the same way 
as for a finite $M$, but we require $h$ and $F$ to be equivariant, i.e.:
\[ \begin{tabular}{ccc}
    $h : \Sigma \eqto M$ & and & $ F \subseteq_\textrm{eq} M$
\end{tabular} \]
For example, consider again the language from Example~\ref{ex:no-rep-monoid}, but 
this time over an orbit-finite $\Sigma$:
\[\textrm{``No letter appears twice in a row''} \subseteq \Sigma^* \]
It is easy to see that this language is recognized by an orbit-finite monoid:
In fact, the definitions of $M$, $h$ and $F$ remain the same as in the finite case.
If $\Sigma$ is orbit-finite then so is $1 + \Sigma^2 + \bot$, 
and both $h$ and $F$ are easily seen to be equivariant.\\

\noindent
Let us present one more example:
\begin{example}
\label{ex:monoid-at-most-3}
Consider the following language over the alphabet $\atoms$:
\[ \textrm{``There are at most $3$ different letters in the word ''} \]
It is recognized by the following orbit-finite monoid:
\[ M = \underbrace{\atoms \choose {\leq 3}}_{\textrm{sets with at most $3$ letters}} + \underbrace{\bot}_{\substack{
\textrm{a representation of}\\
\textrm{sets with more than $3$ atoms}}} \]
The monoid operation is defined as:
\[ x \cdot y = \begin{cases}
 x \, \cup \, y & \textrm{if } x \neq \bot \textrm{, } y \neq \bot \textrm{, and } |x \, \cup \, y| < 3   \\\
 \bot & \textrm{otherwise }
 \end{cases}
 \]
Function $h$ is defined as $x \mapsto \{x\}$, and the accepting subset is defined as $F = M - \{ \bot \}$.
\end{example}

The expressive power of orbit-finite monoids is strictly weaker than the one of deterministic one-way orbit-finite automata.
First, let us notice that an orbit-finite monoid can be translated into an orbit-finite automaton, 
using the same construction as in the proof of Lemma~\ref{lem:aut-mon-eq}. To show 
that orbit-finite monoids recognize a different class of languages than deterministic orbit-finite automata, 
it suffices to show that orbit-finite monoids are closed under reverse (we already know 
that the deterministic one-way orbit-finite automata are not):
\begin{lemma}
\label{lem:ofm-reverse}
If a language $L$ is recognized by an orbit-finite monoid, then so is its reverse.
\end{lemma}
\begin{proof}
Let $L$ be recognized by $M, h, F$.
Define $\overleftarrow M$, to be an orbit-finite monoid that has the same underlying set as $M$, but its operation ($\star$) is defined as:
\[a \star b = b \cdot a\]	
where $(\cdot)$ is the operation of the original monoid $M$. This operation is easily seen to be associative.
This finishes the proof, because $\overleftarrow M$ (together with the original $h$ and $F$)
recognizes the reverse of $L$.
\end{proof}
It follows orbit-finite monoids are indeed strictly weaker than deterministic orbit-finite automata.
The following claim gives an explicit witness of this non-containment:
\begin{example}
\label{ex:first-again-not-monoid}
    The language
    \[ \textrm{``First letter appears again''}\]
    is not recognized by any orbit-finite monoid. If it were, 
    then (by Lemma~\ref{lem:ofm-reverse}) so would be its reverse:
    \[L_\textrm{last} = \textrm{``The last letter appears before''}\]
    This would mean that the language $L_\textrm{last}$ is recognized by some deterministic orbit-finite automaton, 
    which was shown in the proof of Lemma~\ref{lem:dofa-reverse} to be false.\\

    Since the proof is a bit circuitous, let us also provide an intuitive argument: Let $w$ be a word in $\atoms^*$,
    and let $a$ be a letter in $\atoms$. To check whether $aw$ belongs to the language, we need to know whether $w$ contains $a$. 
    It follows that the monoid image of $w$ has to contain information about all the letters that appear in $w$. However, such 
    an information is a finite subset of $\atoms$, and the set of all finite subsets of $\atoms$ is orbit-infinite.
\end{example}
\noindent
Finally, let us point out that:
\[ \textrm{orbit-finite monoids} \neq (\substack{\textrm{left-to-right deterministic}\\ \textrm{orbit-finite automata}}) \cap (\substack{\textrm{right-to-left deterministic}\\ \textrm{orbit-finite automata}}) \]
The class on the right is stronger, as witnessed by the following language:
\[\Large \substack{\textrm{``The first letter is equal to the last one}\\
\textrm{and it appears somewhere else in the word''}}\]
\chapter{Single-use restriction}

In this section, we introduce the \emph{single-use restriction}, which weakens register automata in a way that
makes them equivalent to orbit-finite monoids. The restriction was first introduced in my master's thesis \cite{stefanski2018automaton},
which proves that single-use register automata are not stronger than orbit-finite monoids, 
and later studied in the conference paper \cite{single-use-paper}, which proves that the two models are actually equivalent.
The contribution of this thesis is introducing the abstract class of \emph{single-use functions} (Section~\ref{sec:su-functions}), and defining 
the \emph{single-use automaton} in terms of single-use functions (Section~\ref{sec:su-automata}).

\section{Single-use register automaton}
\label{sec:su-register-automata}
We start the chapter with an informal discussion on the model of the \emph{single-use register automaton}, 
which is a variant of the deterministic\footnote{In
    this thesis we assume that all single-use models are deterministic unless we explicitly state otherwise.
    This is because combining the single-use restriction with nondeterminism, which we do not know how to resolve
    (see Section~\ref{sec:non-determinsitic-single-use}).}
register automaton, where 
every register value can be used at most once. This means that:
\begin{enumerate}
    \item the automaton is not allowed to make copies of register values; and
    \item whenever the automaton asks a query about a register value, 
          this has the side effect of destroying that register's contents.
\end{enumerate}
To see this model in practice, consider the following example (see Section~\ref{subsec:single-use-transition-functions} for a more formal definition):
\begin{example}
\label{ex:su-appears-again}
We want single-use automata to be equivalent to orbit-finite monoids.
This means that the following language should not be recognized 
by any single-use register automaton (because, by Example~\ref{ex:first-again-not-monoid},
it is not recognized by any orbit-finite monoid):
\[ \textrm{``The first letter appears again''} \subseteq \atoms^* \]
The formal proof of this fact follows from Lemma~\ref{lem:1sua-incl-ofm} presented later
in this chapter. For now, let us simply show
that the standard register automaton presented in Example~\ref{ex:appears-again}
(in Chapter~\ref{ch:inifnite-alphabets})
fails to recognize the language under the single-use restriction. We follow 
the automaton's run on the word $1\,2\,3\,2\,1\,3$. It starts in the initial configuration:
\smallpicc{First-again-1}
\noindent
The first transition proceeds without difficulties:
the automaton stores the first letter
in its register and moves to the second position:
\smallpicc{First-again-2}
\noindent
Then, the automaton compares its register value with its current input value. It learns that they are different,
but, as a consequence of the single-use restriction, it loses the register value:
\smallpicc{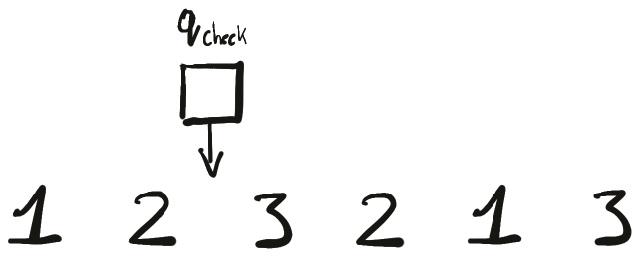}
\noindent
At this point, the automaton does not remember the first atom any more, so it has no chance of
checking if it appears later in the word.
\end{example}
\noindent
Let us now consider a positive example:
\begin{example}
\label{ex:su-no-two-in-a-row}
The following language is recognized by a single-use automaton:
\[ \textrm{``No letter appears twice in a row''} \subseteq \atoms^* \]
The automaton has one register in which it stores a copy of the previous letter.
Whenever the automaton moves forward to a new position, it compares the new letter with the 
previous one. This has the side effect of destroying 
the register's contents. If the letters are (or rather were) equal,
the automaton rejects the input. If they were different,
the automaton saves the current letter in its register and proceeds forward.
\end{example}

It is worth pointing out that a single-use automaton has an unrestricted (i.e. multiple-use)
access to its input letters. This is illustrated by our final example:
\begin{example}
\label{ex:su-at-most-3}
In this example we show that the language:
\[\textrm{``There are at most $3$ different letters in the input word''} \subseteq \atoms \]
is recognized by a single-use register automaton. First, let us point out
that the standard (i.e. multiple-use) automaton presented in Section~\ref{sec:dra} violates 
the single-use restriction. Interestingly, a single-use construction is possible,
but it requires six registers. Suppose that the automaton 
has already seen three different letters $a, b, c \in \atoms$.
Then the automaton should store one of them in three copies, one of them in two copies and one of them in one copy:
\vvsmallpicc{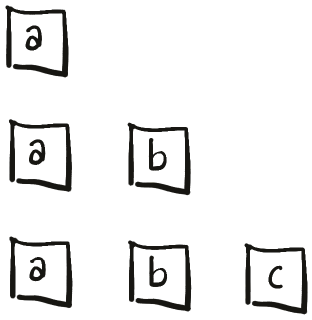}
Suppose that the automaton is about to process the next letter $d \in \atoms$,
which may or may not be equal to $a$, $b$, or $c$.
This transition is explained in the following diagram (remember that the automaton has unrestricted access to the input letter $d$):
\bigpicc{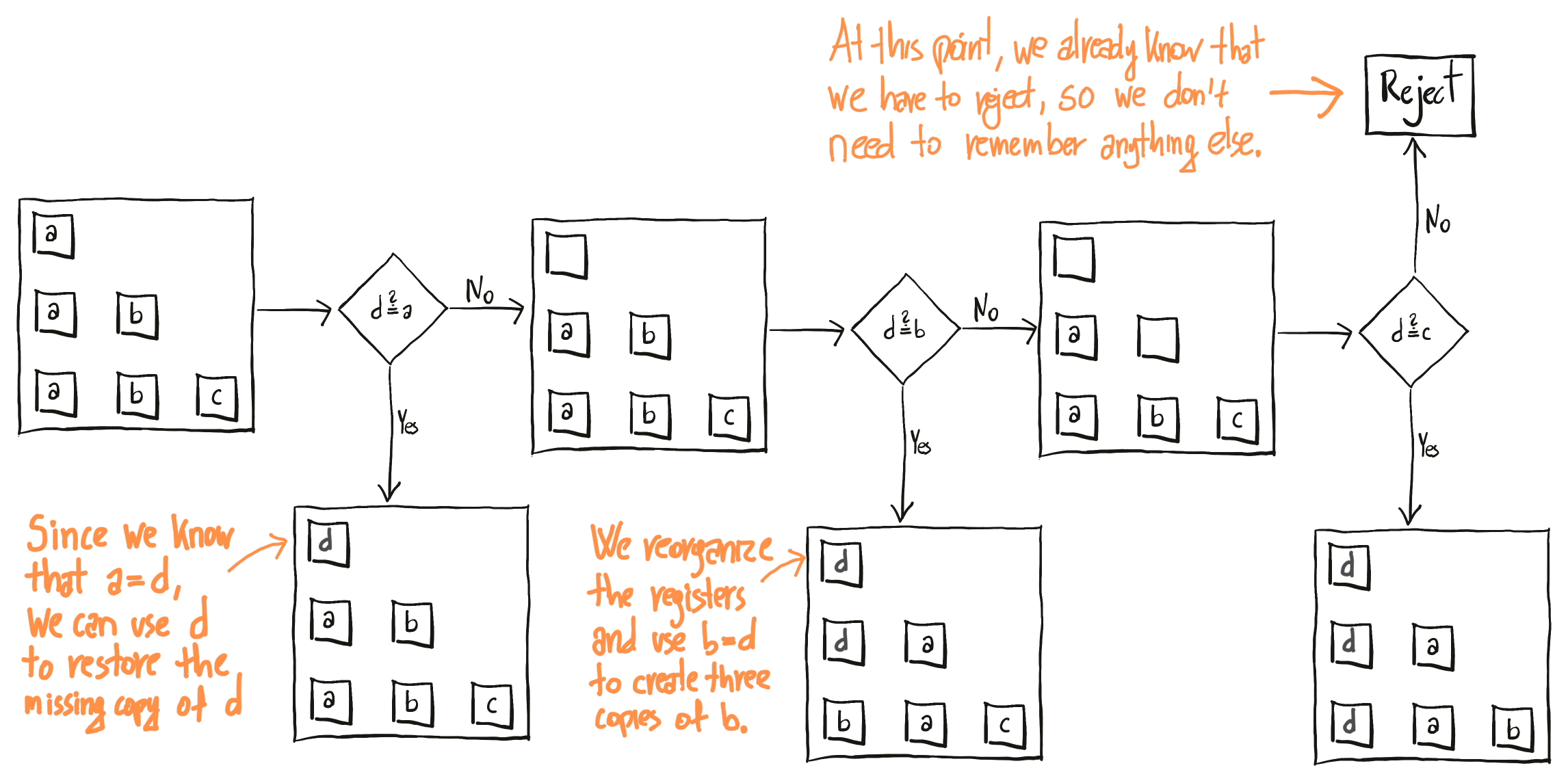}
Using a similar idea, it is easy to extend this construction to cases where the automaton has only seen two or fewer 
different letters so far. 
\end{example}

\subsection{Single-use transition functions}
\label{subsec:single-use-transition-functions}
In this section, we define\footnote{
    The main goal of this definition is to provide an intuitive understanding of the single-use restriction, 
    so it might lack some formal rigor. For a fully formal definition, see Definition~\ref{def:single-use-functions}.
} \emph{single-use transition functions} for register automata.
This definition is specific to the functions of the type:
\[ (Q \times (\atoms + \bot)^R) \times \atoms \eqto (Q \times (\atoms + \bot)^R)\tdot\]

The syntactic definition of single-use transition functions is based on the syntactic definition 
of equivariant transition functions. Recall that in Section~\ref{sec:dra}
we have defined equivariant transition functions using programs of the following shape:
\begin{alltt}
{\bf if} (condition {\bf and} condition {\bf and} \(\ldots\) {\bf and} condition) {\bf then}
    action; action; \(\ldots\); action;
{\bf else} {\bf if} (condition {\bf and} condition {\bf and} \(\ldots\) {\bf and} condition) {\bf then}
    action; action; \(\ldots\); action;
{\bf else} {\bf if} (condition {\bf and} condition {\bf and} \(\ldots\) {\bf and} condition) {\bf then}
    action; action; \(\ldots\); action;
\(\ldots \)
\end{alltt}
Examples of conditions include $\State = q_7$, $\reg_1 = \reg_2$, or $\reg_1 = \In$,
and examples of actions include $\State := q_7$, $\reg_3 := \reg_2$ or $\reg_5 := \In$.
Notice that the if-statements are not allowed to branch.
In the definition of \emph{single-use transition functions,} we use a similar syntax
but with a different semantics. In the single-use semantics, evaluating a condition 
has the side effect of destroying the content of each register that appears in the condition
(by replacing it with $\bot$). For example, in the single-use semantics the following 
two programs are equivalent:
\[
\begin{tabular}{c|c}
\begin{minipage}[t]{5cm}
 \begin{alltt}
    {\bf if} \(\reg\sb{1}\) = \(\reg\sb{3}\) {\bf then}
        \(\State\) := \(q\sb{3}\);
    {\bf else} 
        \(\State\) := \(q\sb{5}\); 
 \end{alltt}
\end{minipage}
&
\begin{minipage}[t]{5cm}
 \begin{alltt}
    {\bf if} \(\reg\sb{1}\) = \(\reg\sb{3}\) {\bf then}
        \(\reg\sb{1}\) := \(\bot\);
        \(\reg\sb{3}\) := \(\bot\);
        \(\State\) := \(q\sb{3}\);
    {\bf else}
        \(\reg\sb{1} := \bot\);
        \(\reg\sb{3} := \bot\);
        \(\State\) := \(q\sb{5}\); 
 \end{alltt}  
\end{minipage} 
\end{tabular}\]

Because of those side effects, the order in which we evaluate the conditions might 
influence the outcome. To make this order clear, we modify the syntax 
of single-use transition functions, by (a) disallowing the use of ${\bf and}$ 
in the if-statements, and (b) allowing nested and branching if expressions.\\

As an example, we provide an implementation of the transition function 
for the register automaton described in Example~\ref{ex:su-at-most-3}.
The automaton has $4$ control states ($q_0$, $q_1$, $q_2$, $q_3$, $q_\textrm{fail}$)
and $6$ registers $(a_1, a_2, a_3, b_1, b_2, c_1)$. For the sake of brevity, we limit 
the implementation to the case where the automaton has already seen $3$ different letters:
\begin{alltt}
{\bf if} state = \(q\sb{3}\) {\bf then}
    {\bf if} \(a\sb{1}\) = \(\In\) {\bf then}
        \(a\sb{1}\) := \(\In\);
    {\bf else} {\bf if} \(b\sb{1}\) = \(\In\) {\bf then}
        \(b\sb{1}\) := \(a\sb{2}\);
        \(b\sb{2}\) := \(a\sb{3}\);
        \(a\sb{1}\) := \(\In\);
        \(a\sb{2}\) := \(\In\);  
        \(a\sb{3}\) := \(\In\);   
    {\bf else} {\bf if} \(c\sb{1}\) = \(\In\) {\bf then}
        \(c\sb{1}\) := \(b\sb{2}\);
        \(b\sb{1}\) := \(a\sb{2}\);
        \(b\sb{2}\) := \(a\sb{3}\);
        \(a\sb{1}\) := \(\In\);
        \(a\sb{2}\) := \(\In\);  
        \(a\sb{3}\) := \(\In\);  
    {\bf else}
        \(\State\) := \(q\sb{\textrm{fail}}\)
{\bf else} ...
\end{alltt}

This particular transition function could have been implemented using ${\bf and}$'s 
instead of the branching if expressions, but this is not true for all single-use transition 
functions. Here is an example of a function that requires branching:

\begin{alltt}
{\bf if} \(r\sb{1}\) = \(\In\) {\bf then}
    {\bf if} \(r\sb{2} = r\sb{3}\)
    {\bf then} \(\State\) := \(q\sb{\textrm{ok}}\);
    {\bf else} \(\State\) := \(\ q\sb{\textrm{fail}}\);
{\bf else}
    {\bf if} \(r\sb{2} = r\sb{4}\)
    {\bf then} \(\State\) := \(q\sb{\textrm{ok}}\);
    {\bf else} \(\State\) := \(q\sb{\textrm{fail}}\);
\end{alltt}

\subsubsection{Single-use acceptance function}
\label{subsec:single-use-acceptance-function}
Finally, let us discuss the way in which the single-use register automaton accepts its input.
There are two possible approaches: In the first one, 
called \emph{acceptance by state}, the automaton decides whether to accept its input
by looking at its final control state -- if it belongs to the accepting subset $F \subseteq Q$, 
it accepts; otherwise it rejects. In the second way, the automaton has an equivariant acceptance
function:
\[ (Q \times (\atoms + \bot)^R)  \eqto \{\Yes, \No\}\tdot \]
Acceptance by control state is used by the standard (multiple-use) register automata, as 
defined in Chapter~\ref{ch:inifnite-alphabets}, and acceptance by configuration is used 
(implicitly) by the orbit-finite automaton. It follows, from Theorem~\ref{thm:dra-equiv-dofa}, 
that the two acceptance models are equivalent for multiple-use automata.
However, for the single-use register automaton, they can change the expressive power of the automaton\footnote{
    We are going to briefly revisit this distinction in Chapter 3,
    while discussing the output function ($\lambda$) of a local monoid transduction
    (see Definition~\ref{def:local-monoid-transduction}).
}:
\begin{example}
\label{ex:su-ac-state-conf}
   The following language
    \[ \textrm{``The first letter is equal to the last one''} \subseteq \atoms^* \] 
    is recognized by a single-use register automaton that accepts by configuration,
    but not by one that accepts by control state.\\

    We start by describing the automaton that accepts by configuration.
    The automaton has two registers: in the first one, it stores 
    a copy of the first letter, and in the second one, it stores a copy of the previous letter.
    In the final configuration, 
    the first register will contain a copy of the first letter, and the 
    second register will contain a copy of the last letter, so
    the acceptance function can check if the two values are equal.\\

    On the other hand, it is not hard to see that a register automaton that accepts 
    by control state and recognizes the language ``The first letter is equal to the last one'',
    has to compare each input letter with the first one (because every letter could be the last one).
    This would violate the single-use restriction.
\end{example}

The language from Example~\ref{ex:su-ac-state-conf} is recognized by an 
orbit-finite monoid. Since we would like single-use register automata to be equivalent to 
orbit-finite monoids, we choose acceptance by configuration as the standard acceptance model 
for the single-use register automaton. One could also argue that acceptance by configuration is more natural than 
acceptance by control state, because it naturally appears in orbit-finite automata.\\

Finally, let us mention that we could also consider a model, which requires 
the acceptance function to be single-use. (With the syntax for a single-use acceptance functions
defined analogously to the one for single-use transition functions.) Fortunately, 
it turns out that this does not influence the expressive power of the automaton.
We prove this later in the chapter as Lemma~\ref{lem:su-mu-acceptance}.

\section{Single-use functions}
\label{sec:su-functions}
In this section, we take a detour from our discussion of automata theory to introduce an abstract concept of \emph{single-use functions}.
Then, in the next section, we will use these functions as a tool for analysing single-use automata.
The step from single-use transition functions to general single-use functions can be compared to the step from equivariant 
transition functions (described in Section~\ref{sec:dra}) to the general equivariant 
functions (described in Section~\ref{sec:sets-with-atoms}).\\ 

Before we define single-use functions, let us discuss a few of their properties.
The defining feature of the class is that it does not contain the following function:
\[ \copyf : \atoms \to \atoms \times \atoms \tdot \]
This motivates the notation $X \suto Y$ for the set of all single-use functions between $X$ and $Y$.
The notation comes from linear logic (\cite{girard1987linear})
and linear type systems\footnote{For more connections with linear type systems see \cite[Claim~1.4.11]{nguyen2021automates}.} (\cite{wadler1990linear}). It is important to point out that
while single-use functions are not allowed to copy atoms,
they are allowed to discard them,  which makes them actually closer to affine logic,
and affine type systems \cite{asperti1998light} (they use the symbol $\suto$ as well).\\

The proof that single-use automata are not stronger than orbit-finite monoids relies on two 
key properties of single-use functions: The first one is that single-use function spaces preserve 
orbit finiteness: if $X$ and $Y$ are orbit-finite, then so is $X \suto Y$. Note that this 
is not true for finitely supported functions, where already $\atoms \fsto \{\Yes, \No\}$ is 
orbit-infinite (see Section~\ref{subsec:orbit-finite-sets} for details). The second important 
property is that single-use functions are closed under compositions:
i.e.  if $f$ belongs $X \suto Y$, and $g$ belongs to $Y \suto Z$, then $(g \circ f)$ belongs to $X \suto Z$.
Thanks to those two properties, we know that if $Q$ is orbit-finite, then $Q \suto Q$ 
is an orbit-finite monoid. We can use it to recognize the language of a single-use automaton whose set of states is $Q$.\\

Unfortunately, as we are going to see, the definition of single-use functions is (at least for now) rather syntactic in its 
nature. It limits their scope to a very specific subclass of orbit-finite sets called
\emph{polynomial orbit-finite sets}. The questions of finding a semantic definition of the class and
extending its scope to all orbit-finite sets (or some other class larger than polynomial orbit-finite sets) remain open. 

\subsection{Polynomial orbit-finite sets}
We start by defining \emph{polynomial orbit-finite sets}\footnote{
The motivation of the name is as follows: The word \emph{orbit-finite} is used, because 
every polynomial orbit-finite set is orbit-finite. The word \emph{polynomial} is used 
because the class of polynomial orbit-finite sets is closed under $\times$ and $+$
(i.e. products and coproducts).}, which are the domains and codomains of single-use functions:
\begin{definition}
\label{def:pofs}
The class of \emph{polynomial orbit-finite sets} is the smallest subclass of sets with atoms that:
\begin{enumerate}
        \item contains the atomless singleton ($1$);
        \item contains the set of all atoms ($\atoms$);
        \item is closed under products ($P_1 \times P_2)$; and
        \item is closed under disjoint sums ($P_1 + P_2$).
\end{enumerate}
\end{definition}

\begin{example}
\label{ex:finite-pof}
    Every finite $S$ can be represented as polynomial orbit-finite, because it is isomorphic to:
    \[\underbrace{1 + 1 + \ldots + 1}_{|S| \textrm{ times}}\tdot\]
\end{example}
\begin{example}
\label{ex:memory-conf-pof}
    The set of all possible memory configurations of a register automaton -- i.e. the
    set $Q \times (\atoms + \bot)^R$, for some finite $R$ and $Q$ -- can be represented 
    as a polynomial orbit finite set, because it is isomorphic to:
   \[ (\underbrace{1 + 1 + \ldots + 1}_{\textrm{$|Q|$ times}}) \times \underbrace{(\atoms + 1) \times (\atoms + 1) \times \ldots \times (\atoms + 1)}_{\textrm{$|R|$ times}} \]
\end{example}

It is easy to see that all polynomial orbit-finite sets are orbit-finite, and that all
polynomial orbit-finite sets are straight (as defined in Definition~\ref{def:straight-set}).
It is worth pointing out that the other inclusion does not hold -- 
there are sets that are straight and orbit-finite but not polynomial orbit finite.
An example of such a set is $\atoms^{(2)}$. In general, thanks to the distributivity 
of $\times$ over $+$, it is not hard to see that every polynomial orbit-finite set 
is isomorphic with a set of the following form (in Lemma~\ref{lemma:pof-normal-form-su}
we are going to show that this isomorphism is a single-use function):
\[ \atoms^{k_1} + \ldots + \atoms^{k_n}\]



\subsection{Single-use functions}
We are now ready to define the single-use functions:
\begin{definition}
\label{def:single-use-functions}
    The class of \emph{single-use functions} is the smallest subclass of functions between polynomial orbit-finite
    that is closed under the following combinators, and contains all the following basic functions:
    \[
    \begin{tabular}{cc}
        \multicolumn{2}{c}{
            \begin{tabular}{|ccc|}
            \hline
            \multicolumn{3}{|c|}{Combinators}\\
            \hline
            $\infer
            { X \transform{g \circ f} Z}
            { X \transform{f} Y &  Y \transform{g} Z}$
            &
            $\infer
            { X_1 \times X_2 \transform{f \times g} Y_1 \times Y_2}
            { X_1 \transform{f} Y_1 &  X_2 \transform{g} Y_2}$
            &
            $\infer
            { X_1 + X_2 \transform{f+g} Y_1 + Y_2}
            { X_1 \transform{f} Y_1 &  X_2 \transform{g} Y_2}$\\
            \hline
            \end{tabular}
        }\\
        &\\
        \multicolumn{2}{c}{
            \begin{tabular}{|ll|}
                \hline
                \multicolumn{2}{|c|}{Functions about $\atoms$}\\
                \hline
                $\eqf:$ & $ \atoms \times \atoms \to 1 + 1$ \\
                $\const_{a \in \atoms}:$ & $ 1 \to \atoms$ \\
                $\idf:$ & $\atoms \to \atoms$\\
                \hline
            \end{tabular}
        }\\
        &\\
        \begin{tabular}{|ll|}
            \hline
            \multicolumn{2}{|c|}{Functions about $\times$}\\
            \hline
            $\proj_1 :$ & $ X \times Y \to X$ \\
            $\proj_2 :$ & $X \times Y \to Y$ \\
            $\sym :$ & $X \times Y \to Y \times X$ \\
            $\assoc : $ & $(X \times Y) \times Z \to X \times (Y \times Z)$ \\
            $\leftI :$ & $X \to 1 \times X$ \\
            \hline
        \end{tabular} &
            \begin{tabular}{|ll|}
                \hline
                \multicolumn{2}{|c|}{Functions about $+$}\\
                \hline
                $\coproj_1$ & $X \to X + Y$ \\
                $\coproj_2$ & $Y \to X + Y$ \\
                $\cosym :$ & $X + Y \to Y + X$ \\
                $\coassoc:$ & $(X + Y) + Z \to X + (Y + Z)$\\
                $\merge :$  & $ X + X  \to X$ \\
                \hline
        \end{tabular}\\
        &\\
        \multicolumn{2}{c}{
            \begin{tabular}{|ll|}
                \hline
                \multicolumn{2}{|c|}{Distributivity}\\
                \hline
                    $\distr :$ & $X \times(Y + Z) \to X \times Y + X \times Z$\\
                \hline
        \end{tabular}}
    \end{tabular}
    \]
We hope that the semantics of the basic functions and combinators follows intuitively from their types and names.
However, to clarify possible confusion, let us define a few of them:
\begin{enumerate}
\item $\coproj_1$ is the natural injection of $X$ into $X + Y$;
\item $f + g : X_1 + X_2 \to Y_1 + Y_2$ is defined as $f$ on $X_1$ and $g$ on $X_2$;
\item $\leftI(x)$ is defined as $\left(\top, x\right)$, where $\top$ is the unique element of $1$. 
\end{enumerate}
\end{definition}
\noindent

The class of single-use functions is a restriction of finitely supported 
functions. To illustrate that, we start with a negative example:

\begin{example}
    The following function is finitely supported, but not single-use:
    \[
        \begin{tabular}{cc}
        $f : \atoms \to {{\underbrace{(1 + 1)}_{\substack {\textrm{represents}\\
                                                         \textrm{true or false}}}}}$ &
        $f(a) = \begin{cases}
                \textrm{true} & \textrm{if } a = 3 \vee a = 5\\
                \textrm{false} & \textrm{otherwise}
        \end{cases}$
        \end{tabular}
    \]
    This is because $f$ needs to compare its input with two different constants, which 
    requires two copies of the input value. (A formal proof
    follows from 
    the decision tree representation, presented later in this section.)
\end{example}

\noindent
Let us now present a few positive examples:
\begin{example}
    \label{ex:constI-su}
        The function $\constI : X \to 1$ is a single-use function, because it can be constructed as the following composition:
        \[X \transform{\leftI} 1 \times X \transform{\proj_1} 1\]
\end{example}

\begin{example}
\label{ex:rightDistr-su}
    The function $\rightDistr : (X + Y) \times Z \to  (X \times Z) + (Y \times Z)$ is a single-use
    function. It can be constructed as the following composition:
    \[
        (X + Y) \times Z \transform{\sym} Z \times (X + Y) \transform{\distr} Z \times X + Z \times Y \transform{\sym + \sym} X \times Z + Y \times Z 
    \]
    Using a similar idea, we can show that all the following functions are single-use:
    \[ 
    \begin{tabular}{cc}
      $\rightI : X \to X \times 1$ & $\assoc^{-1} : X \times (Y \times Z) \to (X \times Y) \times Z$\\
      &\\
      \multicolumn{2}{c}{$\coassoc^{-1} : X + (Y + Z) \to (X + Y) + Z$}
    \end{tabular}
    \]
\end{example}
\begin{example}
\label{ex:maybe-comb-su}
Single-use functions are closed under the following combinator:
\[ \infer{X + Y \transform{[f, g]} Z}
    {X \transform{f} Z & Y \transform{g} Z} \]
Function $[f, g]$ is constructed as follows:
\[X + Y \transform{f + g} Z + Z \transform{\merge}  Z\]
This combinator can easily be generalized to any number of functions:
\[ [f_1, \ldots , f_n] : X_1 + \ldots + X_n \to Y \]
\end{example}
\begin{example}
\label{ex:distr-rev-su}
    The function $\distr^{-1} : X \times Y + X \times Z \to X \times (Y + Z)$ is a single-use function.
    It can be constructed as follows:
    \[  X \times Y + X \times Z \longtransform{[\proj_1 \times \coproj_1, \proj_1 \times \coproj_2]}  X \times (Y + Z) \] 

\end{example}

\noindent
In order to simplify the notation, we 
declare both $\times$ and $+$ to be right associative. This means that:
\[X_1 + X_2 + \ldots + X_n = X_1 + (X_2 + \ldots + X_n)\tdot\]
(and analogously for $\times$). In a similar manner, we define $X^n$ to denote
\[ \underbrace{X \times X \times \ldots \times X}_{\textrm{$n$ times}} = X \times (X \times \ldots \times X)\]

\begin{example}
\label{ex:projs-su}
    The following function is single-use:
    \[ \proj_i : X_1 \times \ldots \times X_n \to X_i \]
    It can be constructed as follows:
    \[ X_1 \times \ldots \times X_n \transform{\proj_2} X_2 \times \ldots \times X_n \transform{\proj_2} \ldots \transform{\proj_2} X_i \times \ldots \times X_n \transform{\proj_1} X_i \]
    Similarly, we show that $\coproj_i : X_i \to X_1 + \ldots + X_n$ is a single-use function.
\end{example}

\begin{example}
\label{ex:assoc-star-su}
    The following function is single-use:
    \[ \assoc^* : (X_1 \times \ldots \times X_n) \times (Y_1 \times \ldots \times Y_m) \to X_1 \times \ldots \times X_n \times Y_1 \ldots \times Y_m\]
    It can be constructed inductively on $n$:
    \[ (X_1 \times(X_2 \ldots \times X_n)) \times (Y_1 \times \ldots \times Y_m) \transform{\assoc} X_1 \times ((X_2 \ldots X_n) \times (Y_1 \times \ldots \times Y_m)) \transform{\idf \times \assoc^*}\]
    \[  \transform{} X_1 \times (X_2 \ldots \times X_n \times Y_1 \ldots \times Y_m) \]
    Using a similar idea, we can extend $\assoc^*$ to map between any two bracketings of $X_1 \times \ldots \times X_n$.
    In an analogous manner, we can show that the following functions are single-use:
    \[ \coassoc^* : (X_1 + \ldots + X_n) + (Y_1 + \ldots + Y_m)  \to X_1 + \ldots + X_n + \ldots Y_1 + \ldots Y_n\]
    \[  \distr^* : (X_1 \times \ldots \times X_n) \times (Y_1 \times \ldots \times Y_m) \to X_1 \times Y_1 \, + \, X_1 \times Y_2 \, + \ldots + \, X_n \times Y_m\]
\end{example}

\begin{example}
\label{ex:perm-su}
    For every permutation $p : \{1, \ldots, k\} \to \{1, \ldots, k\}$, the following function is single-use:
    \[\shuffle_p : X_1 \times \ldots \times X_k \to X_{p(1)} \times \ldots \times X_{p(k)}\]
    In order to see that, we notice that every $\shuffle_p$ can be built from the following to functions:
    \[ \mathtt{swap} : X_1 \times X_2 \ldots \times X_k \to X_2 \times X_1 \times \ldots \times X_k \]
    \[ \mathtt{shift} : X_1 \times X_2 \ldots \times X_k \to X_2 \times X_3 \times \ldots \times X_k \times X_1\]
    This leaves us with constructing $\mathtt{swap}$ and $\mathtt{shift}$. Here is the construction for $\mathtt{swap}$: 
    \[
    \atoms^k \longtransform{\assoc^*} \atoms^2 \times \atoms^{k-2}
    \longtransform{\sym \times \idf} \atoms^2 \times \atoms^{k-2} 
    \longtransform{\assoc^*} \atoms^k\comma
    \]
    and here is the construction for $\mathtt{shift}$:
    \[
    \atoms^k \transform{\sym} \atoms^{k-1} \times \atoms
    \longtransform{\assoc^*} \atoms^k \tdot
    \]

\end{example}
\noindent

Finally, let us show that we can use single-use bijections to present polynomial orbit-finite sets in
a normal form:
\begin{lemma}
\label{lemma:pof-normal-form-su}
For every polynomial orbit-finite $X$, there exists a single-use bijection: 
\[\tau_X : X \suto \atoms^{k_1} + \atoms^{k_2} + \ldots + \atoms^{k_n}\comma\]
such that $\tau^{-1}_X$ is a single-use function as well.
\end{lemma}
\begin{proof}
    We construct the isomorphism inductively on $X$. First, we notice that
    both $\tau_1$ and $\tau_\atoms$ are both equal to $\idf$. Then we construct $\tau_{X_1 + X_2}$: 
    \[ X_1 + X_2 \transform{\tau_1 + \tau_2} (\atoms^{k_1} + \ldots + \atoms^{k_n}) + (\atoms^{l_1} + \ldots + \atoms^{l_m})
    \transform{\coassoc^*} \atoms^{k_1} + \ldots + \atoms^{l_m} \]
    Finally, we construct $\tau_{X_1 \times X_2}$:
    \[X_1 \times X_2 \transform{\tau_{X_1} \times \tau_{X_2}} (\atoms^{k_1} + \ldots + \atoms^{k_n}) \times (\atoms^{l_1} + \ldots + \atoms^{l_m})
    \transform{\distr^*}\]
    \[ \atoms^{k_1} \times \atoms^{l_1} + \atoms^{k_1}\times\atoms^{l_2} + \ldots + \atoms^{k_n}\times\atoms^{l_m} \longtransform{\assoc^* + \ldots + \assoc^*} \atoms^{k_1 + l_1} + \atoms^{k_1 + l_2} + \ldots + \atoms^{k_n + l_m}\]
    This finishes the first part of the proof. To prove that $\tau_X^{-1}$ exists and is a single-use function,
    we use the following claim:
    \begin{claim}
        If a single-use function $f : X \suto Y$ can be constructed using only 
        the following basic functions (and all three combinators): $\sym$, $\assoc$, $\leftI$, $\cosym$,
        $\coassoc$, $\idf$, $\distr$, $\proj_1$ limited to $X \times 1 \to X$, and $\proj_2$ limited 
        to $1 \times X \to X$, then $f^{-1} : Y \to X$ exists and is a single-use function.
    \end{claim}
    \begin{proof}
        The proof goes by induction on the derivation of $f$ as a single-use function. First,
        let us notice that all three combinators preserve reversibility:
        \[ \begin{tabular}{ccc}
            $(f \circ g)^{-1} = g^{-1} \circ f^{-1}$ &
            $(f \times g)^{-1} = f^{-1} \times g^{-1}$ &
            $(f + g)^{-1} = f^{-1} + g^{-1}$. 
        \end{tabular}\]
        This leaves us with showing that all the basic functions listed in the claim 
        have single-use inverses: Each of $\sym$, $\cosym$, and $\idf$ is its own inverse.
        Thanks to Examples~\ref{ex:rightDistr-su}~and~\ref{ex:distr-rev-su}, we know 
        that $\distr^{-1}$, $\assoc^{-1}$ and $\coassoc^{-1}$ are single-use functions. 
        Finally, we notice that $\leftI$ and $\proj_2 : 1 \times X \to X$ are each other's inverses, 
        and the inverse of $\proj_1 : X \times 1 \to X$ is $\rightI$ defined in Example~\ref{ex:rightDistr-su}.
    \end{proof}
\end{proof}

\subsection{Single-use functions + $\copyf$}
\label{subsec:single-use-functions-plus-copy}
In this section we prove the following lemma, which justifies the intuition that:
\[ (\textrm{single-use functions}) + \copyf = (\textrm{multiple-use functions}) \]
\begin{lemma}
\label{lem:su-plus-copy-fs}
    If we extend the class of single-use functions, by including
     \[\copyf~:~\atoms~\to~\atoms~\times~\atoms\]
    as a basic function, we obtain the class of 
    all finitely supported functions between polynomial orbit-finite sets. 
\end{lemma}

We introduce the name $\emph{definable function}$, for a function that belongs 
to the class of single-use functions extended with $\copyf$.\\

\subsubsection{Definable $\Rightarrow$ Finitely supported}

We start the proof of Lemma~\ref{lem:su-plus-copy-fs} with the simpler inclusion:
\begin{lemma}
\label{lem:definable-so-fs}
    Every definable function is finitely supported.
\end{lemma}
\begin{proof}
    Notice that for every $a \in \atoms$, the function $\const_a$ is supported by $\{a\}$, 
    and that all other basic functions equivariant. To finish the proof, 
    we need to show that all the combinators preserve finite supports:
    \[
        \begin{tabular}{cc}
            $\supp(f \circ g) \subseteq \supp(f) \cup \supp(g)$, &
            $\supp(f \times g) \subseteq \supp(f) \cup \supp(g)$, \\
            &\\
            \multicolumn{2}{c}{and $\supp(f + g) \subseteq \supp(f) \cup \supp(g)$}
        \end{tabular}
    \]
    This follows from Lemma~\ref{lem:fs-functions-preserve-supports}, because
    all $\circ$, $\times$, and $+$ are equivariant (higher-order) functions.
\end{proof}

\noindent
Before we proceed with the proof of Lemma~\ref{lem:su-plus-copy-fs}, we state a corollary 
of Lemma~\ref{lem:definable-so-fs}:
\begin{lemma}
\label{lem:su-fs}
    Every single-use function is finitely supported by the set of 
    all constants used in its derivation. In particular, if
    a single-use function can be constructed without the use of any
    $\const_a$, then it is equivariant. 
\end{lemma}

\subsubsection{Finitely-supported $\Rightarrow$ Definable}
This section is dedicated to proving the remaining inclusion:
\begin{lemma}
\label{lem:fs-definable}
    If $X$ and $Y$ are polynomial orbit-finite, then every $f: X \fsto Y$ is definable. 
\end{lemma}

The proof is going to be similar to the proof from Section~\ref{sec:dra},
that every semantically equivariant transition function is also syntactically equivariant.
We start by defining $\alpha$-orbits, which are orbits of $\alpha$-permutations:
\begin{definition}
    Let $\alpha$ be a finite subset of atoms. For every element $x$, of a set with atoms $X$, we 
    define the $\alpha$-orbit of $x$ to be the following set:
    \[ \{\pi(x) \ | \  \pi \textrm{ is an $\alpha$-permutation} \}\]
\end{definition}

Similarly as it was the case for equivariant orbits,
every two $\alpha$-orbits of $X$ are either equal or disjoint. It follows that being in the same $\alpha$-orbit is an equivalence relation on $X$,
which means that every set with atoms $X$ is partitioned into its orbits. By \cite[Theorem~3.16]{bojanczyk2019slightly},
we know that every orbit-finite set is also $\alpha$-orbit-finite:
\begin{lemma}
\label{lem:of-so-aof}
    If a set $X$ is orbit-finite, then it has finitely many $\alpha$-orbits,
    for every $\alpha \subseteq_\textrm{fin} \atoms$. 
\end{lemma}

\noindent
We start the proof of Lemma~\ref{lem:fs-definable}, by proving it in a very special case:
\begin{claim}
\label{claim:o-to-a-def}
Let $\alpha$ be a finite subset of atoms, and let $O \subseteq \atoms^k$ be an  $\alpha$-orbit of $\atoms^k$.
For every $\alpha$-supported $f : O \fsto \atoms$,
there exists a definable function $f' : \atoms^k \to \atoms$, such that $f'$ restricted to $O$ is equal to $f$. 
\end{claim}
\begin{proof}
    Take some $\bar x$ from $O$. It follows from Lemma~\ref{lem:fs-functions-preserve-supports} that $f(\bar x)$ is either
    equal to some $a \in \alpha$ or to some $\bar{x}_i$. In the first we 
    define $f' := \const_a$, and in the second case we define $f' := \proj_i$. 
    Let us show that $f'$ matches with $f$ on $O$: 
    Take some $\bar y \in O$. Since $O$ is a single $\alpha$-orbit, 
    we know that $\bar y = \pi(\bar x)$, for some $\alpha$-permutation $\pi$. If $f(\bar x) = a$, for some $a \in \alpha$, 
    then:
    \[ f(\bar y) = f(\pi(\bar x)) = \pi(f(\bar x)) = \pi(a) = a = \const_a(\bar x)\tdot\]
    If $f(\bar x) = \bar x_i$, then:
    \[ f(\bar y) = f(\pi(\bar x)) = \pi(f(\bar x)) = \pi(\bar{x}_i) = \bar{y}_i = \proj_i(\bar y)\tdot \]
\end{proof}

In the next step, we would like to extend Claim~\ref{claim:o-to-a-def} to functions of type $O \to \atoms^k$.
Before we do that, we need to define a couple of helper functions:

\begin{example}
    \label{ex:copy-general-definable}
    The function $\copyf$ can be extended from $\atoms$ to all polynomial orbit-finite sets.
    We construct $\copyf : X \to X^2$ inductively on $X$:
    If $X = 1$, then $\copyf$ is equal to $\rightI$.  If $X = \atoms$, then $\copyf$ is a basic function.
    If $X = X_1 \times X_2$, then $\copyf$ can be defined as:
    \[ X_1 \times X_2 \longtransform{\copyf \times \copyf} X_1^2 \times X_2^2 \longtransform{\assoc^*\, \circ\, \shuffle\, \circ\, \assoc^*} (X_1 \times X_2)^2 \] 
    Finally, if $X = X_1 + X_2$, then $\copyf$ can be defined as:
    \[ X_1 + X_2 \longtransform{\copyf + \copyf} X_1^2 + X_2^2 \longtransform{[\proj_1 \times \proj_1, \proj_2 \times \proj_2]} (X_1 + X_2)^2 \]
\end{example}

\begin{example}
    \label{ex:pairing-definable}
    Definable functions are closed under the following combinator:
    \[ \infer{X \transform{\langle f, g \rangle} Y_1 \times Y_2}
             {X \transform{f} Y_1 & X \transform{g} Y_2} \]
    The function $\langle f, g \rangle$ can be constructed as:
    \[ X \transform{\copyf} X \times X \longtransform{f \times g} Y_1 \times Y_2\]
    Using a similar construction, we can generalize this combinator to any number of functions:
    \[ \langle f_1, \ldots, f_n \rangle : X \to Y_1 \times \ldots \times Y_n \]
\end{example}

\noindent
Using the $\langle f_1,  \ldots, f_k \rangle$ combinator, we can easily extend Claim~\ref{claim:o-to-a-def} to functions 
of type $O \fsto \atoms^k$. Now, let us notice that every function of type $O \fsto Y$ (where $Y$ can be 
any polynomial orbit-finite set) can be decomposed as:
\[O \transform{f'} \atoms^k \transform{\coproj_i} \atoms^{l_1} + \ldots + \atoms^{l_n} \transform{\tau_Y^{-1}} Y\]
Thanks to this observation, we can extend Claim~\ref{claim:o-to-a-def} to all functions of the type $O \fsto Y$.
In order to further extend it to $\atoms^k \fsto Y$, we need to show that the definable functions can calculate the $\alpha$-orbit of their argument.
(By Lemma~\ref{lem:of-so-aof}, we know that the set of $\alpha$-orbits of $\atoms^k$ is finite,
which means that it can be represented as $1 + \ldots + 1$.) Before that, we define two more helper functions:

\begin{example}
    \label{ex:const-su}
        For every $X$, and for every $x \in X$, the function $\const_x : 1 \to X$ is single-use.
        Let us start by defining $\const$ for $X = \atoms^k$. For a $\bar x \in \atoms^k$,
        we construct $\const_{\bar x}$ as:
        \[ 1 \transform{\rightI} 1 \times 1 \longtransform{\idf \times \rightI} 1 \times 1 \times 1 \transform{} \cdots \transform{} 1^k
             \longtransform{\const_{\bar{x}_1} \times \ldots \times \const_{\bar{x}_k}} \atoms^k \]
        Now, we use Lemma~\ref{lemma:pof-normal-form-su} to extend this construction to an arbitrary $X$:
        We take some $x \in X$, and we define $\bar x := \tau_X(x)$, then we construct $\const_x$ as:
        \[ 1 \transform{\const_{\bar{x}}} \atoms^{k_j} \transform{\coproj_j} \atoms^{k_1} + \ldots + \atoms^{k_n} \transform{\tau_X^{-1}} X\]
\end{example}

\begin{example}
\label{ex:finite-su}
        If $X$ is \emph{finite}, then every function $f: X \to Y$ is
        single-use. In order to see that, we first notice that since $X$ is finite, then $\tau_X$
        (from Lemma~\ref{lemma:pof-normal-form-su}) has to be of the form $\tau_X : X \to 1 + \ldots + 1$.
        It follows that we can construct $f$ as:
        \[ X \transform{\tau_X} 1 + \ldots + 1 \longtransform{\const_{f(x_1)} \times \ldots \times \const_{f(x_k)}} Y\comma \]
        where $x_i = \tau_X^{-1}(\coproj_i(1))$.
\end{example}

\noindent
We are now ready to show how to compute $\alpha$-orbits:
\begin{claim}
\label{claim:orbit-definable}
     The following function, which computes the $\alpha$-orbit of its input is definable for every $\alpha \subseteq_{\textrm{fin}} \atoms$:
     \[\orbitf_\alpha : \atoms^k \to 1 + \ldots 1\tdot\]
\end{claim}
\begin{proof}
    The orbit of every $\bar x \in \atoms^k$  depends only on whether $\bar{x}_i = \bar{x}_j$ for every $i, j \in \{1, \ldots, k\}$,
    and on whether $\bar{x}_i = a$ for every $a \in \alpha$ and $i \in \{1, \ldots, k\}$. 
    The first type of check can be performed by $\mathtt{cmp}_{i,j}$, defined as:
    \[\atoms^k \transform{\copyf} \atoms^k \times \atoms^k \transform{\proj_i \times \proj_j} \atoms \times \atoms \transform{\eqf} (1 + 1)\]
    The second type of check can be performed by $\mathtt{cmp}_{i, a}$, defined as:
    \[\atoms^k \transform{\rightI} \atoms^k \times 1 \longtransform{\proj_i \times \const_a} \atoms \times \atoms \transform{\eqf} (1 + 1)\]
    Let $c : (1 + 1)^{m} \to 1 + \ldots + 1$ be the function that consolidates the results of all those checks, and computes the $\alpha$-orbit.
    Thanks to Example~\ref{ex:finite-su}, we know that $c$ is definable. This means that we can construct 
    $\orbitf_\alpha$ as:
    \[ \atoms^k \longtransform{\langle \mathtt{cmp}_{1, 1}, \ldots, \mathtt{cmp}_{k, a} \rangle} (1 + 1)^{m} \transform{c} 1 + \ldots + 1 \]
\end{proof}

\noindent
In order to extend Claim~\ref{claim:o-to-a-def} from $O \to Y$ to $\atoms^k \to Y$, we combine 
Claim~\ref{claim:orbit-definable} with the following combinator:

\begin{example}
    \label{ex:if-else-su}
        The class of single-use functions is closed under the following if-then-else combinator:
        \[ \infer{(1 + 1) \times X \transform{f \orf g} Y}
                 {X \transform{f} Y & X \transform{g} Y} \]
        The combinator is defined as follows:
        \[ (1 + 1) \times X \longtransform{\rightDistr} X + X \transform{[f, g]} Y \]
        Using a similar technique, we can generalize the combinator to take more than two functions:
        \[ (f_1 \orf \ldots \orf f_n) : (1 + \ldots + 1) \times X \to Y\]
\end{example}

Finally, we use the $[f_1, \ldots, f_n ]$ combinator and $\tau_X$ function,
to extend Claim~\ref{claim:orbit-definable} to all functions of the type $X \to Y$.
This finishes the proof of Lemma~\ref{lem:fs-definable}.

\subsubsection{$k$-fold use functions}
In this section we define and briefly discuss the class of \emph{$k$-fold-use functions}. 
It is situated between single-use functions which can use only one copy of their input, 
and finitely supported functions which can use any number of copies.
\begin{definition}
\label{def:k-fold-use-functions}
    For every $k \in \nat$, we say that a function is $k$-fold-use,
    if it can be constructed as a composition of the following form:
    \[ X \transform{\copyf^k} X^k \stackrel{f'}{\longsuto} Y\comma \]
    where $f'$ is some single-use function. We denote the set of all $k$-fold use functions
    between polynomial orbit-finite $X$ and $Y$ as $X \suto_k Y$. 
\end{definition}

\noindent It is easy to see that: 
\[ (X \suto Y) = (X \suto_1 Y) \subseteq (X \suto_2 Y) \subseteq (X \suto_3 Y) \subseteq \ldots\]
It is also not hard to see that this hierarchy is strict:
\begin{example}
\label{ex:2u-notsu}
Consider the following function $f \in \atoms \fsto \{\Yes, \No\}$, which is supported by $\{4, 7\}$:
\[
    f(x) = \begin{cases}
                \Yes & \textrm{if } x \in \{4, 7\}\\
                \No & \textrm{otherwise}
    \end{cases}
\]
\noindent
This function clearly belongs to $\atoms \suto_2 \{\Yes, \No\}$, but not to
$\atoms \suto \{\Yes, \No\}$.  
(The formal proof that $f \not \in \atoms \suto \{\Yes, \No\}$ follows from the decision-tree
representation of single-use functions, presented later in this chapter.) 
\end{example}

\noindent This example can easily be generalized, to show that:
\[ (X \suto_1 Y) \subsetneq (X \suto_2 Y) \subsetneq (X \suto_3 Y) \subsetneq \ldots \]
In the limit, this sequence reaches $X \fsto Y$:
\begin{lemma}
\label{lem:fs-k-fold}
For every polynomial orbit-finite $X$ and $Y$:
\[ X \fsto Y \ = \ \bigcup_{k \in \nat} X \suto_k Y \]
\end{lemma}
\begin{proof}
    The proof that every $k$-fold use function is finitely supported is almost the same as 
    the proof of Lemma~\ref{lem:su-fs}. Now, let us prove the other inclusion:
    all the basic functions from Definition~\ref{def:single-use-functions} are $1$-fold,
    and $\copyf$ is easily seen to be $2$-fold, so thanks to Lemma~\ref{lem:fs-definable},
    it suffices to show that all three combinators $+$, $\times$, and $\circ$ preserve belonging to $\bigcup_{k \in \nat} X \suto_k Y$:\\

    We start with $\circ$. Let us show that we can present the following composition as an $m$-fold use function, for some $m \in \nat$:
	\[ X \transform{\copyf^k} X^k \sutransform{f} Y \transform{\copyf^l} Y^l \sutransform{g} Z \]
    We do this in the following way, with $m = k \cdot l$:
	\[ X \transform{\copyf^{k \cdot l}} X^{k \cdot l} \sutransform{\assoc^*} \left({(X^k)}^l\right) \sutransform{f \times \ldots \times f} Y^l \sutransform{g} Z \]

    For $\times$ and $+$, it is not hard to see that if $f$ is a $k$-fold-use function, and $l$ is a 
    $l$-fold-use function, then both $f+l$ and $f \times l$ are $\max(k,l)$-fold-use functions.
\end{proof}


\subsection{Single-use decision trees}
\label{subsec:single-use-decision-trees}
In this section we prove the most important property of single-use functions:
\begin{theorem}
\label{thm:su-orbit-finite}
    The set of all single-use functions $X \suto Y$ is orbit-finite, for all polynomial orbit-finite $X$ and $Y$. 
\end{theorem}

We prove Theorem~\ref{thm:su-orbit-finite} by introducing the \emph{single-use decision tree}
representation of single-use functions -- for every polynomial $X$ and $Y$, we define the set 
$\Trees(X, Y)$, such that there is an equivariant surjection\footnote{This tree representation
is not bijective -- many different trees can describe the same function.}:
\[\Trees(X, Y) \eqto (X \suto Y)\]
Then, we show that the set $\Trees(X, Y)$ is polynomial orbit-finite for every $X$ and $Y$.
this is enough to prove Theorem~\ref{thm:su-orbit-finite}, because by \cite[Lemma~3.24]{bojanczyk2019slightly}
images of orbit-finite sets under equivariant functions remain orbit-finite.

\begin{definition}
\label{def:su-tree}
We start defining the \emph{single-use decision trees} with the most
interesting case, i.e. trees for functions of the following type:
\[ \atoms^k \suto \atoms^{l_1} + \ldots + \atoms^{l_n}\tdot\]
A single-use decision tree of this type is a tree whose nodes contain \emph{queries}
and whose leaves contain \emph{constructors}. Each constructor is of the following form:
\[\coproj_i(v_1, \ldots, v_{l_i}) \comma\]
where each $v_j$ is either an input variable $\mathtt{x_v}$ or an atomic constant $a \in \atoms$. Each query is either 
of the form $\mathtt{x_i = x_j}$ or $\mathtt{x_i = a}$, where $i,j \in \{1, \ldots, k\}$.
Moreover, every single-use tree has to satisfy
the \emph{single-use restriction}, which says that on every path from the root to a leaf,
each variable $\mathtt{x_i}$ may appear at most once (in queries or constructors). Here is an example 
of a single-use decision tree of type $\atoms^2 \suto \atoms + \atoms^2$:
\smallpicc{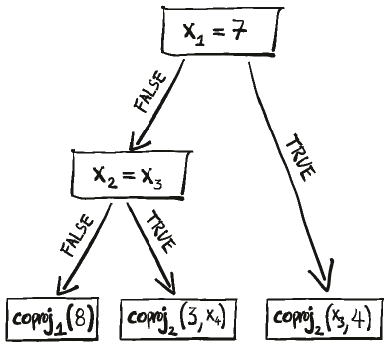}
Each such tree naturally represents a function $f_T : \atoms^k \suto \atoms^{l_1} + \ldots + \atoms^{l_n}$.
It is not hard to see that the single-use restriction on $T$ guarantees that $f_T$ is a single-use function.\\ 

The construction for trees of more general types is standard. The single-use decision 
trees of the type $\atoms^{k_1} + \ldots + \atoms^{k_n} \suto \atoms^{l_1} + \ldots + \atoms^{l_m}$
\noindent
are simply $n$-tuples of trees $T = (T_1, \ldots, T_n)$ such that 
\[T_i \in \Trees(\atoms^{k_1}, \atoms^{l_1} + \ldots + \atoms^{l_m}) \tdot\]
In this case, function $f_T$ is defined using the combinator Example~\ref{ex:maybe-comb-su}:
\[ f_T : \atoms^{k_1} + \ldots + \atoms^{k_n} \sutransform{[f_{T_1}, \ldots, f_{T_n}]}  \atoms^{l_1} + \ldots + \atoms^{l_n}\]
Finally, in order to define single-use decision trees of type $X \suto Y$, 
where $X$ and $Y$ are arbitrary polynomial orbit-finite sets,
we use the isomorphisms from  Lemma~\ref{lemma:pof-normal-form-su}:
\[
    \begin{tabular}{cc}
        $\tau_X : X \to \atoms^{k_1} + \ldots + \atoms^{k_n}$ & 
        $\tau_Y : Y \to \atoms^{l_1} + \ldots + \atoms^{l_m}$. 
    \end{tabular}
\]
A single-use decision tree $T$ representing a function $X \suto Y$ is a tree:
\[ T' \in \Trees(\atoms^{k_1} + \ldots + \atoms^{l_n}, \atoms^{k_1} + \ldots + \atoms^{l_m})\]
It represents the following function:
\[f_T : X \sutransform{\tau_X} \atoms^{k_1} + \ldots + \atoms^{k_n} \sutransform{f_{T'}} \atoms^{l_1} + \ldots + \atoms^{l_k}
\sutransform{\tau_Y^{-1}} Y\]
\end{definition}
It is not hard to see that the mapping $T \mapsto f_T$ is equivariant, i.e. $f_{\pi(T)} = \pi(f_T)$ for all atom permutations.
\todo{Update next sections for the new definition}

\subsubsection{Single-use function $\Rightarrow$ Single-use decision trees}
In this section we show that the single-use decision tree representation is injective:
\begin{lemma}
\label{lem:su-functions-su-trees}
    Every function from $X \suto Y$ is represented by some $T \in \Trees(X, Y)$. 
\end{lemma}
We prove the lemma, by showing that the class of functions recognized by single-use decision trees
is closed under $\circ$, $+$, and $\times$; and contains all basic functions from Definition~\ref{def:single-use-functions}.\\

We start with the most interesting case, which is showing that single-use decision trees are closed under compositions.
We first show it for trees on tuples:
\begin{claim}
\label{claim:su-trees-compose-tup}
        If $g : \atoms^k \suto \atoms^l$ and $f : \atoms^l \suto \atoms^m$ are represented 
        by single-use decision trees, then so is $(f \circ g) : \atoms^k \suto \atoms^l$.
\end{claim}
\begin{proof}
        Let us take two decision trees $F$, $G$ that represent $f$ and $g$, and let us show 
        how to compose them, obtaining $H$ that represent $g \circ f$. Here are 
        some example $F$ and $G$ (for clarity we denote the variables in $F$ as $x_i$,
        and the variables in $G$ as $y_i$):
        \picc{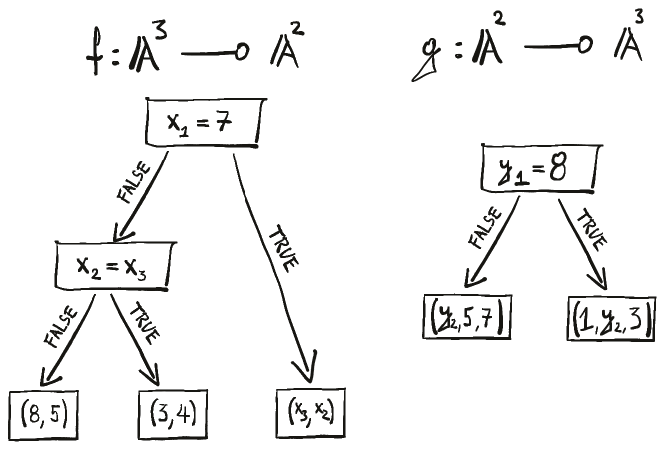} \noindent
        We start the construction by placing one copy of $G$ under each leaf of $F$:
        \picc{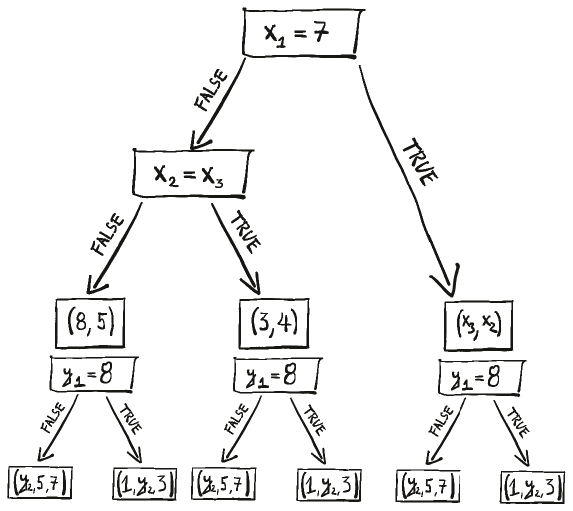} \noindent
        Then we replace every $y_i$ in each copy of $G$, with the $v_i$ from the leaf of $F$
        (each $v_i$ is either some $x_j$ or some $a \in \atoms$):
        \picc{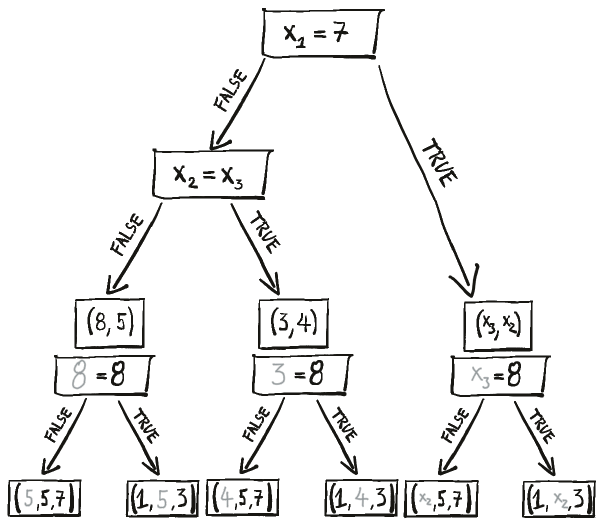}\noindent
        Finally, we resolve all the queries that can be resolved (and forget about $F$'s leaves):
        \picc{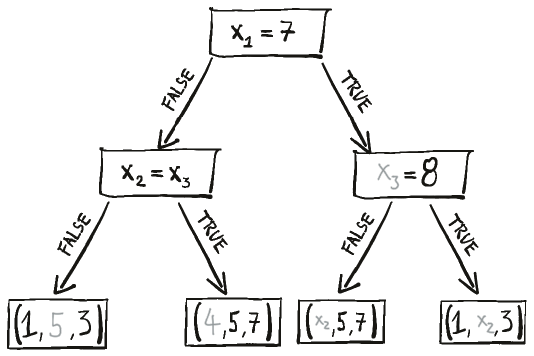}\noindent
        The constructed tree recognizes $g \circ f$ by design. It is also not hard
        to see that the construction preserves the single-use restriction.
\end{proof}
It is not hard to see that Claim~\ref{claim:su-trees-compose-tup} can be first generalized
for trees $\Trees(\atoms^{k_1} + \ldots + \atoms^{k_n}, \atoms^{l_1} + \ldots + \atoms^{l_m})\comma$
and then to all single-use decision trees.\\

We continue the proof of Lemma~\ref{lem:su-functions-su-trees} by showing that functions recognized by single-use decision trees 
are closed under $+$ and $\times$:
\begin{lemma}
\label{ref:su-trees-plus}
    If $f : X_1 \to Y_1$ and $g : X_2 \to Y_2$ are represented by single-use decision trees, then 
    so is $(f + g) : X_1 + X_2 \to Y_1 + Y_2$. 
\end{lemma}
\begin{proof}
    If $f$ is represented by $(F_1, \ldots, F_n)$, and $g$ is represented by $(G_1, \ldots, G_m)$, 
    then  $f + g$ is represented by $(F_1, \ldots, F_n, G'_1, \ldots, G'_m)$,
    where $G'_i$ is a modification of $G_i$ in which each $\coproj_j$ is replaced with $\coproj_{n + j}$.
\end{proof}

\begin{lemma}
\label{ref:su-trees-times}
    If $f : X_1 \to Y_1$ and $g : X_2 \to Y_2$ are represented by single-use decision trees,
    then so is $(f \times g) : X_1 \times X_2 \to Y_1 \times Y_2$. 
\end{lemma}
\begin{proof}
    Let $f$ be represented by $(F_1, \ldots, F_n)$ and $g$ be represented by $(G_1, \ldots, G_m)$, and 
    let us construct $H$ that represents $f \times g$. First of all, notice that $H$ should be a tuple 
    of $n \cdot m$ trees $(H_{1, 1},  \ldots,  H_{i, j},  \ldots, H_{n, m})$. Let us show, how to 
    construct $H_{i, j}$. 
    First, we notice that if $F_i$ belongs to $\Trees(\atoms^{k_i}, \atoms^{l_1} + \ldots + \atoms^{l_m})$
    and $G_j$ belongs to $\Trees(\atoms^{k_j'}, \atoms^{l'_1} + \ldots + \atoms^{l'_{m'}})$, 
    then $H_{i, j}$ should belong to:
    \[ \Trees(\atoms^{k_i + k_j'}, \ \atoms^{l_1 + l'_1 } + \ldots + \atoms^{l_i + l'_j} + \ldots + \atoms^{l_{m} + l'_{m'}})\tdot\]
    In order to construct $H_{i,j}$, we take $F_i$ and replace all its leaves with a modified version of $G_{j}$, where each variable 
    $x_p$ has been replaced by $x_{p + k_i}$, and in which every leaf has been replaced with the following merge of the leaf from $F_i$ with the leaf from $G_j$:
    \[ \coproj_a(v_1, \ldots, v_{l_a}),\ \coproj_b(w_1, \ldots, w_{l'_b}) \mapsto \coproj_c(v_1, \ldots, v_{l_a}, w_1, \ldots, w_{l'_b})\comma\]
    where $c$ is the index of $\atoms^{k'_i + l'_j}$  in $\atoms^{k'_1 + l'_1 } + \ldots + \atoms^{k'_{n'} + l'_{m'}}$.
    For example, consider the following $F_i$, $G_j$:
    \picc{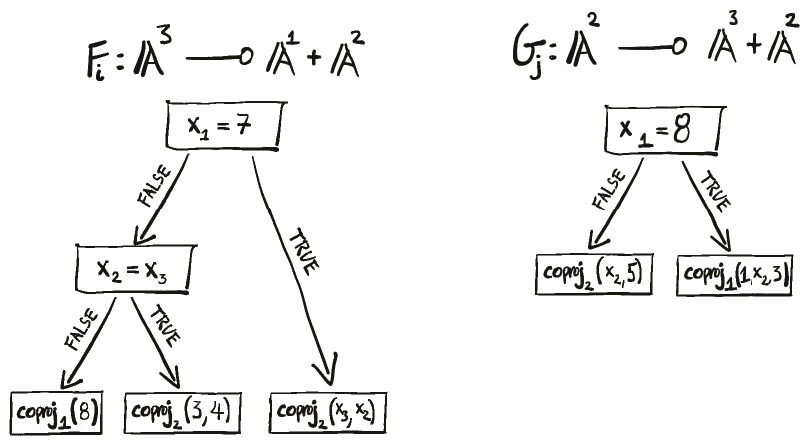}
    \noindent
    For those $F_i$, $G_j$, we construct the following $H_{i,j}$:
    \picc{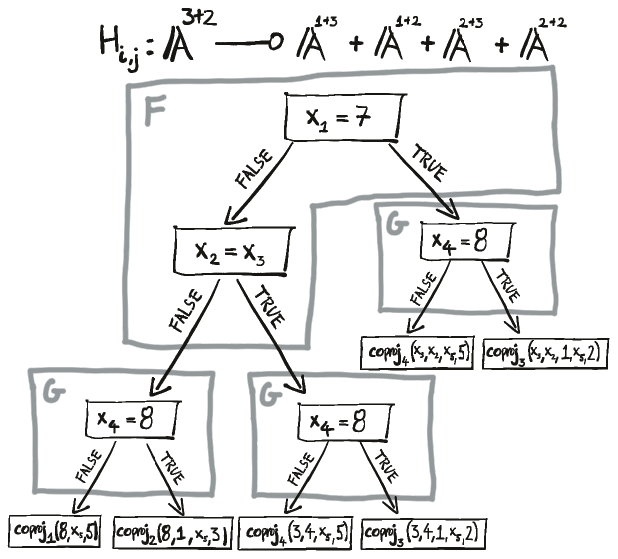}
\end{proof}

Finally, we need to show that all basic functions from Definition~\ref{def:single-use-functions} can be implemented 
as single-use decision trees. For the sake of conciseness, we only show how to implement $\proj_1$ (the implementations
for other basic functions are analogous):
\begin{lemma}
    For every $X$ and $Y$, the function $\proj_1 : X \times Y \to X$ is represented by some $T \in \Trees(X \times Y, X)$. 
\end{lemma}
\begin{proof}
    Let $\tau_X$ and $\tau_Y$ (from Lemma~\ref{lemma:pof-normal-form-su}) have the following types:
    \[
        \begin{tabular}{cc}
            $\tau_X : X \to \atoms^{k_1} + \ldots + \atoms^{k_n}$ & $\tau_Y : Y \to \atoms^{l_1} + \ldots + \atoms^{l_m}$
        \end{tabular}
    \]
    This means that for every $i, j$, we need to define a single-use decision tree
    $T_{i,j} \in \Trees(\atoms^{k_i + l_j}, \atoms^{k_1} + \ldots + \atoms^{k_n})$, such that:
    \[ \proj_1 = \tau_{X}^{-1} \circ (T_{1, 1} + \ldots + T_{i, j} + \ldots + T_{n, m}) \circ \tau_{X \times Y}\]
    It is not hard to see that we can construct $T_{i, j}$ to be a single leaf with the constructor $\coproj_i(x_1, \ldots, x_{k_i})$.
    (Note that the constructor forgets about variables $x_{k_i + 1}$ to $x_{k_i + l_j}$.)
\end{proof}
\noindent
This concludes the proof of Lemma~\ref{lem:su-functions-su-trees}.

\subsubsection{Single-use decision trees are polynomial orbit finite}
This section is dedicated to proving the following lemma:
\begin{lemma}
\label{lem:su-trees-pof}
    For every polynomial orbit-finite $X, Y$, the set $\Trees(X, Y)$ is polynomial orbit-finite as well.
\end{lemma}
\begin{proof}
It suffices to prove the lemma only for the case where $X = \atoms^k$. The general version follows from Lemma~\ref{lem:of-preserved}.
Notice that, thanks to the single-use restriction, the depth of the trees from $\Trees(\atoms^k, Y)$ is bounded by $k + 1$.
Now, the key observation is that there is only finitely many possible \emph{shapes} of binary trees of bounded height
(a \emph{shape} of a tree is a version of the tree where all atoms have been replaced by an atomless blank). 
Let $\{s_1, \ldots, s_k\}$ be the set of these shapes, and denote $|s_i|$ to be the number of blanks (i.e. atoms)
in each tree of shape $s_i$. Then, the set of all from $\Trees(\atoms^k, \atoms^l)$ can be represented 
as the following polynomial orbit-finite set: 
\[ \atoms^{|s_1|} + \ldots + \atoms^{|s_k|} \]
\end{proof}

\subsubsection{Canonical single-use decision trees}
\label{subsec:canonical-trees}
The definition of polynomial orbit-finite sets (Definition~\ref{def:pofs})
does not include the constructor of single-use function spaces ($\suto$). This means 
that if we want to treat $X \suto Y$ as a polynomial orbit-finite set (for example 
to talk about higher-order single-use functions), we need to represent it as 
a polynomial orbit-finite set. For this purpose, 
we are going to use $\Trees(X, Y)$. Observe that, this tree representation
of single-use functions is not injective -- the following two decision trees represent the same function:
\bigpicc{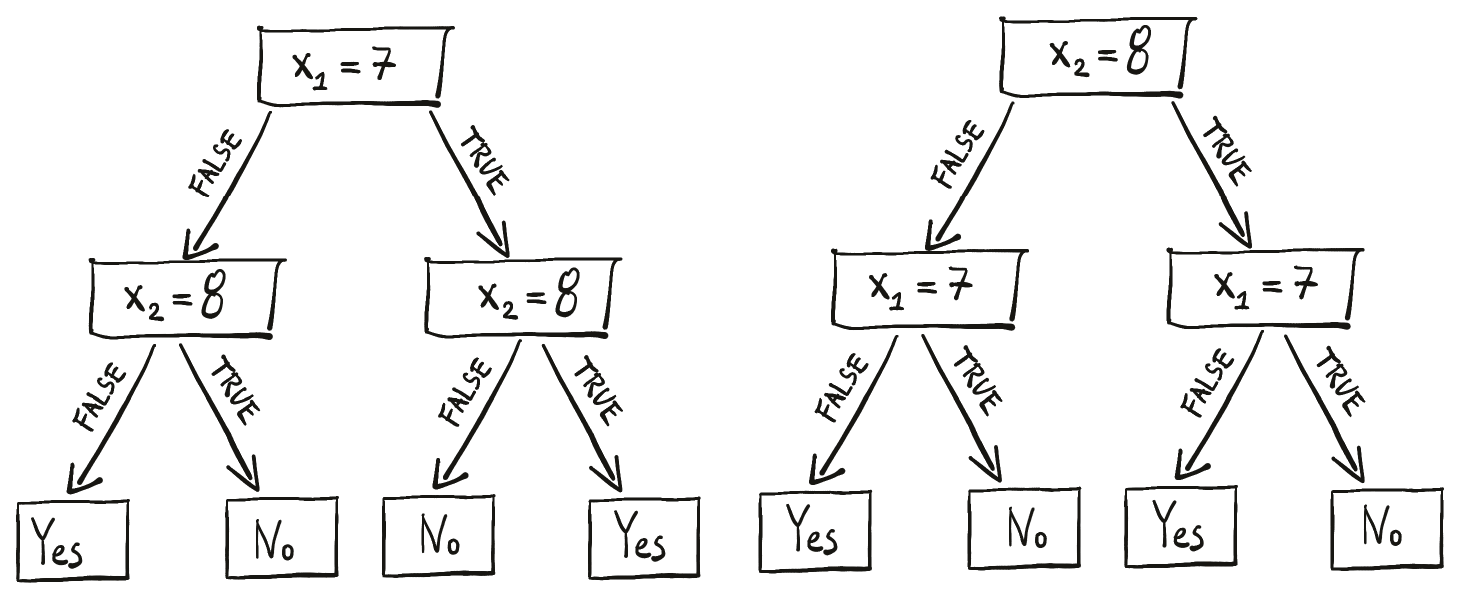}
In this section we show that there is an equivariant way of choosing
a canonical tree $\Trees(X, Y)$ for every function $X \suto Y$
(this is not immediate, because sets with atoms do not always admit choice -- see Claim~\ref{claim:no-choice}):
\begin{lemma}
\label{lem:eq-tree-repr}
For every polynomial orbit-finite $X$ and $Y$ there exists a function:
\[\treeRepr : (X \suto Y) \eqto \Trees(X, Y)\comma \]
such that $\treeRepr(f)$ represents $f$.
\end{lemma}
As usual, we first prove the lemma for the special case where $X = \atoms^k$ 
and $Y = \atoms^{l_1} + \ldots + \atoms^{l_n}$: Define the height of a function
$f \in (X \suto Y)$ to be the smallest height of a tree from $\Trees(X, Y)$ that 
recognizes $f$ and define $\textrm{SU}_{\leq h}(X, Y) \subseteq X \suto Y$ be the set of
all single-use functions recognized by trees that are not taller than $h$. 
We prove by induction on $h$ that for every $h$ there is a function:
\[ \treeRepr_h : \textrm{SU}_{\leq h}(X, Y) \eqto \Trees(X, Y) \]
This is enough to prove the lemma: Thanks to the single-use restriction,
the height of trees from $\Trees(\atoms^k, Y)$ is bounded by $k + 1$, which 
means that $\treeRepr = \treeRepr_{\leq (k+1)}$.\\

We start the inductive proof with $h = 1$. If $f$ has height $1$,
then it is recognized by some leaf. Moreover, it is not hard to see that if
$\coproj_j(v'_1, \ldots, v'_{l_j})$ and $\coproj_{j'}(v'_1, \ldots, v'_{l_{j'}})$ 
recognize the same function, then $j = j'$, and for every $i$, $v_i = v'_i$.
This means that we can define $\treeRepr_1(f)$ to be the only leaf that recognizes $f$.\\

For the induction step we take a $h > 1$. Let $f$ be a function in $\textrm{SU}_h(X, Y)$.
We assume that the height of $h$ is at least $2$ -- if not, we can simply repeat the construction 
from the induction base. We say that a tree is \emph{minimal} if
no tree of smaller height represents the same function. We say that $x_i$ is the leading variable 
of $f$, if $i$ is the smallest index, for which there exists a minimal tree that represents $f$ 
and queries $x_i$ in its root. 
It is not hard to see that the function that maps $f$ to its leading variable is equivariant.
Let $x_i$ be the leading variable of $f$, and let $T$ be a minimal tree that recognizes $f$ and queries 
$x_i$ in its root. Let $\mathtt{x_i = v}$ be the query from the root of $T$, and let $T_\Yes$ and 
$T_\No$ be the subtrees of $T$. The following claim states that $v$, $f_{T_\Yes}$ and $f_{T_\No}$ 
do not depend on the choice of $T$:
\begin{claim}
\label{claim:v-yes-no-decomp-unique}
    Let $T$ and $T'$ be two minimal trees that recognize the same function, and query the same 
    variable $x_i$ in their roots:
    \picc{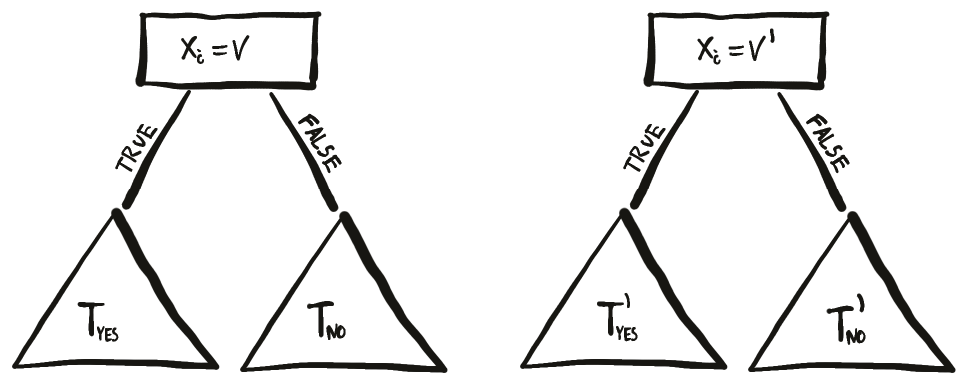}
    \noindent
    It follows that $v = v'$, $f_{T_\Yes} = f_{T'_\Yes}$, and $f_{T_\No} = f_{T'_\No}$. 
\end{claim}
Before we prove the claim, we finish the induction step. Notice that the heights of 
$f_{T_\Yes}$ and $f_{T_\No}$ are strictly lower than $h$. It follows that we can define 
$T'_\Yes = \treeRepr(f_{T_\Yes})$, $T'_\No = \treeRepr(f_{T_\No})$, and define $\treeRepr(f)$ 
to be:
\vsmallpicc{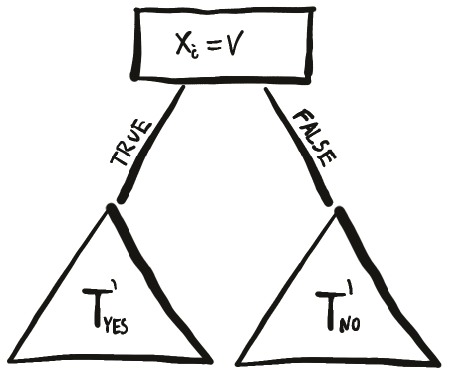}
\noindent
The construction does not depend on the choice of $T$, which makes it equivariant.
This leaves us with proving Claim~\ref{claim:v-yes-no-decomp-unique}:
\begin{proof}[Proof of Claim~\ref{claim:v-yes-no-decomp-unique}:]
    Let us take $T, T' \in \Trees_h(X, Y)$ that recognize the same function $f \in \mathrm{SU}_h(X, Y)$,
    and query the same $x_i$ in their roots. Let $\mathtt{x_i = v}$ and $\mathtt{x_i = v'}$ be the 
    root queries of $T$ and $T'$, and let $T_\Yes$, $T_\No$, $T'_\Yes$, and $T'_\No$ be their subtrees. 
    Notice that $f_{T_\Yes} \neq f_{T_\No}$,
    or otherwise $f$ would have been recognized by $T_\Yes$, and $T$ would not be minimal.
    Similarly, $f_{T'_\Yes} \neq f_{T'_\No}$. Let us now show that $v = v'$.
    Since both $v$ and $v'$ might be equal to an input variable, or to an atomic constant,
    we need to consider four cases. They are all similar to each, so we only present the proof 
    for $v = a \in \atoms$ and $v' = b \in \atoms$:
    Let us assume that $a \neq b$ and show that this leads to a contradiction.
    We start by noticing that 
    thanks to the single-use restriction none of $T_\Yes$, $T_\No$, $T'_\Yes$, 
    $T'_\No$ uses $x_i$. It follows that, for every $\bar x \in \atoms^k$ and for every $c \in \atoms$:
    \[ \begin{tabular}{ccc}
        $T_\Yes(\bar x[x_i := c]) = T_\Yes(\bar x)$ & and & $T_\No(\bar x[x_i := c]) = T_\No(\bar x)$.
        \end{tabular}
    \]
    Similarly for $T'_\Yes$ and $T'_\No$. Moreover, since $f$ is recognized by $T$ and $T'$, 
    it follows that for every $\bar{x} \in \atoms^k$ and for every $c \in \atoms$:
                \[
                    \begin{tabular}{cc}
                        $f(\bar{x}[x_i := c]) = \begin{cases}
                            T_\Yes(\bar x) & \textrm{if } c = a\\
                            T_\No(\bar x)  & \textrm{if } c \neq a
                        \end{cases}$ &
                        $f(\bar{x}[x_i := c]) = \begin{cases}
                            T'_\Yes(\bar x) & \textrm{if } c = b\\
                            T'_\No(\bar x)  & \textrm{if } c \neq b
                        \end{cases}$
                    \end{tabular}
                \]
                Since $f_{T_\Yes} \neq f_{T_\No}$, we know that there is a $\bar{x} \in \atoms^k$, such 
                that $T_\Yes(\bar{x}) \neq T_\No(\bar{x})$. This leads to a contradiction because, 
                if we take $c \in \atoms$, that is different from both $a$ and $b$, we have that:
                \[T_\Yes(\bar x) = f(\bar{x}[x_i := a]) = T'_\No(\bar{x}) = f(\bar{x}[x_i := c]) = T_\No(\bar{x})\]
    It follows that $v = v'$. 
    Let us now show that $f_{T_\Yes} = f_{T'_\Yes}$ and $f_{T_\No} = f_{T'_\No}$.
    Since $v$,  $v'$ might both be equal to some $a \in \atoms$, or to some $x_j$, we have 
    two cases to consider. Again, they are very similar, so we only present the proof for $v = v' = a \in \atoms$. 
    In this case, for all $\bar x \in \atoms^k$:
    \[ T_\Yes(\bar x) = f(\bar{x}[x_i := a]) = T'_\Yes(\bar x) \]
    Similarly, if we take $b \in \atoms$, that is different from $a$, then for all $\bar x$:
    \[ T_\No(\bar x) = f(\bar{x}[x_i := b]) = T'_\No(\bar x) \]
\end{proof}

\noindent
Let us now define $\treeRepr$ for functions $\atoms^{k_1} + \ldots + \atoms^{k_n} \suto \atoms^{l_1} + \ldots + \atoms^{l_m}$:
\[ \treeRepr(f) := (\treeRepr(f \circ \coproj_1), \ldots, \treeRepr(f \circ \coproj_n)) \]
Finally, we define the representation for functions $X \suto Y$:
\[ \treeRepr(f) := \treeRepr(  \tau_Y \circ f \circ \tau^{-1}_X) \]

\section{Single-use automata}
\label{sec:su-automata}
In this section we define the \emph{single-use automaton} -- a model which (slightly) generalizes the single-use register automaton.
A single-use automaton uses a polynomial orbit-finite set of states $Q$ to process words over a polynomial orbit-finite alphabet $\Sigma$.
It has an equivariant initial state $q_0 \in Q$, and a single-use acceptance function $f : Q \suto \{\Yes, \No\}$. 
The type of its transition function is slightly more complicated. We discuss it in the following 
section.
\subsection{Single-use transition function}
\label{sec:su-transitions}
The first idea for the type of the transition function is:
\[ (\Sigma \times Q) \suto Q\tdot \]
Such a transition function would only allow the automaton to 
use one copy of each input letter, which would make single-use
automata weaker than both single-use register automata and orbit-finite monoids:
\begin{example}
	The following language over $\atoms$ is easily seen to be recognized by an orbit-finite monoid,
	and by a single-use register automaton, but not
	by a single-use polynomial orbit-finite automaton of type $\Sigma \times Q \suto Q$:
   \[\textrm{``The length of the word is $3$ and all its letters are pairwise distinct''.}\]
\end{example}

The main focus of this thesis are models equivalent to orbit-finite monoids, so
we are not going to discuss single-use automata of type $\Sigma \times Q \suto Q$
(which might be interesting in another context). Instead, we present two equivalent ways
of typing the transition function that allow the automaton to use multiple copies 
of the input letters:\\

The first one is to let the automaton explicitly specify, how many copies of the input letters 
it requires. We add this information as a number $k \in \nat$ to the automaton's specification,
and then we type its transition function as:
\[ \underbrace{\Sigma^k}_{\substack{\textrm{$k$ copies of}\\\textrm{the same letter}}} \times \underbrace{Q}_{\substack{\textrm{one copy of}\\\textrm{the state}}} \ \  \suto \ \  Q \]
Let us stress that although, it is possible to feed a function of this type with $k$ different
letters from $\Sigma$, we are only going to use it with $k$ identical copies of the input letter.
Another way of typing the transition function is as follows: 
\[\Sigma \eqto (Q \suto Q) \]
The function transforms every letter in a multiple-use, equivariant way
into a single-use state transformation (which usually is not going to be 
equivariant, because it is going to use atoms from the input letter). 
This way, we give the automaton a multiple-use access to $\Sigma$ and a single-use access to $Q$.
A similar approach to typing the transition function can be found in \emph{weighted automata}, where $\Sigma$ is a finite set, 
$Q$ is a finite-dimensional linear space, and the transition function maps 
every letter to a linear transformation:
\[ \Sigma \to (Q \to_{\textrm{lin}} Q)\comma\]

Before we show that the two ways of defining transition functions are equivalent, we need to discuss one 
more property of single-use functions:
\subsubsection{Single-use currying}
The following isomorphism called \emph{currying} holds for many different classes of functions, 
such as all functions, finitely supported functions, linear transformations, \ldots
\[ \begin{tabular}{ccc}
    $A \to (B \to C)$ & $\simeq$ & $A \times B \to C$
   \end{tabular}
\]
\[ \begin{tabular}{cc}
    $F(f)(a, b) = f(a)(b)$ & $F^{-1}(f)(a)(b) = f(a, b)$
   \end{tabular}
\] 

In this section we would like to discuss currying for single-use functions, i.e. 
talk about the relationship between:
\[
    \begin{tabular}{ccc}
        $A \times B \suto C$ & and & $A \suto \Trees(B, C)$
    \end{tabular}
\]
(As mentioned in Section~\ref{subsec:canonical-trees}, since $B \suto C$ is not a polynomial orbit-finite set, 
we need to represent is as $\Trees(B, C)$).\\

First, we notice that every single-use function  $A \suto \Trees(B, C)$,
can be transformed into $A \times B \suto C$:
\begin{lemma}
    \label{lem:su-curring-one-way}
    For every $f : A \suto \Trees(B, C)$, there exists an $f' : A \times B \suto C$, 
    such that for every $a \in A$ and $b \in B$:
    \[ f(a)(b) = f'(a, b) \] 
    Moreover, the mapping $f \mapsto f'$ is an equivariant function.
\end{lemma}
\noindent
In order to construct $f'$, we simply have to unpack the leaves of $f$:
\bigpicc{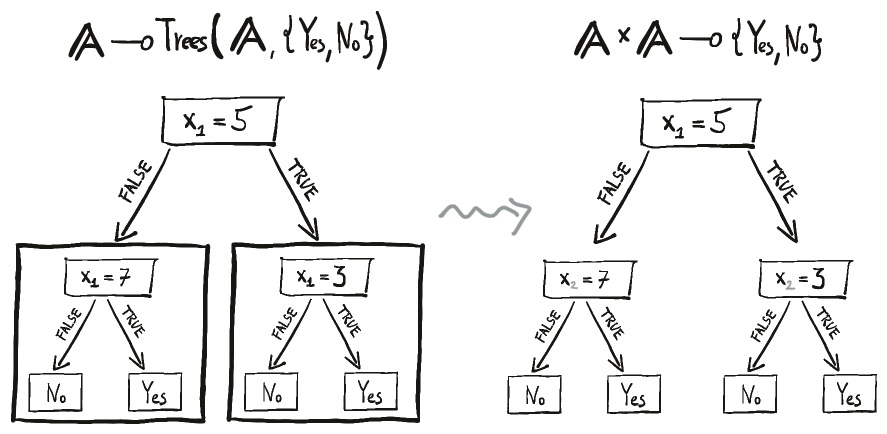}

On the other hand, it might not be possible to translate a function $A \times B \suto C$ into a function $A \suto \Trees(B, C)$.
For example, consider the following function $f : \atoms \times \atoms \suto \atoms$:
\smallpicc{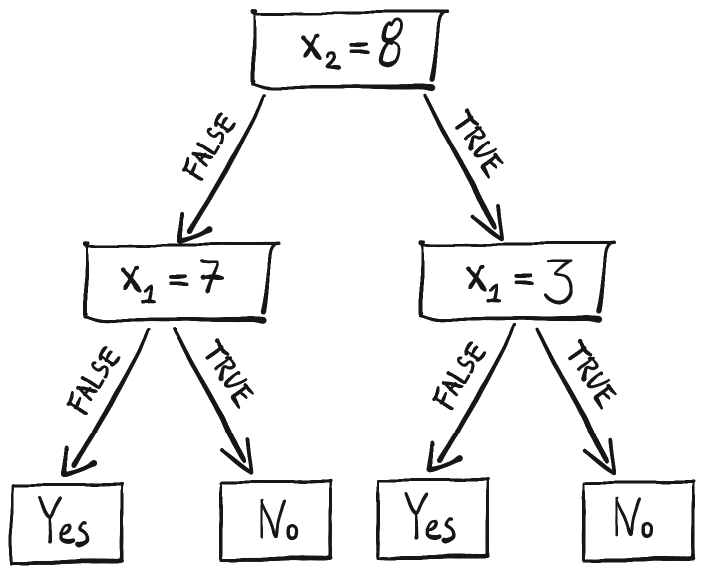}
\noindent
It is not hard to see that this function cannot be translated into a function $\atoms \suto \Trees(\atoms, \{\Yes, \No\})$ -- 
this is because we need to compare the input argument with both $7$ and $3$.
However, the variable $x_1$ appears only twice in the tree, which means that we can express $f$ as
$f' \in \atoms \suto_2 \Trees(\atoms, \{\Yes, \No\})$:
\[ f'(a) = \textrm{resolve all resolvable queries } T[x_1 := a] \] 
For example, let us compute $f'(7)$. We substitute both appearances of $x_1$ with the atomic constant $7$, 
and resolve all resolvable queries:
\picc{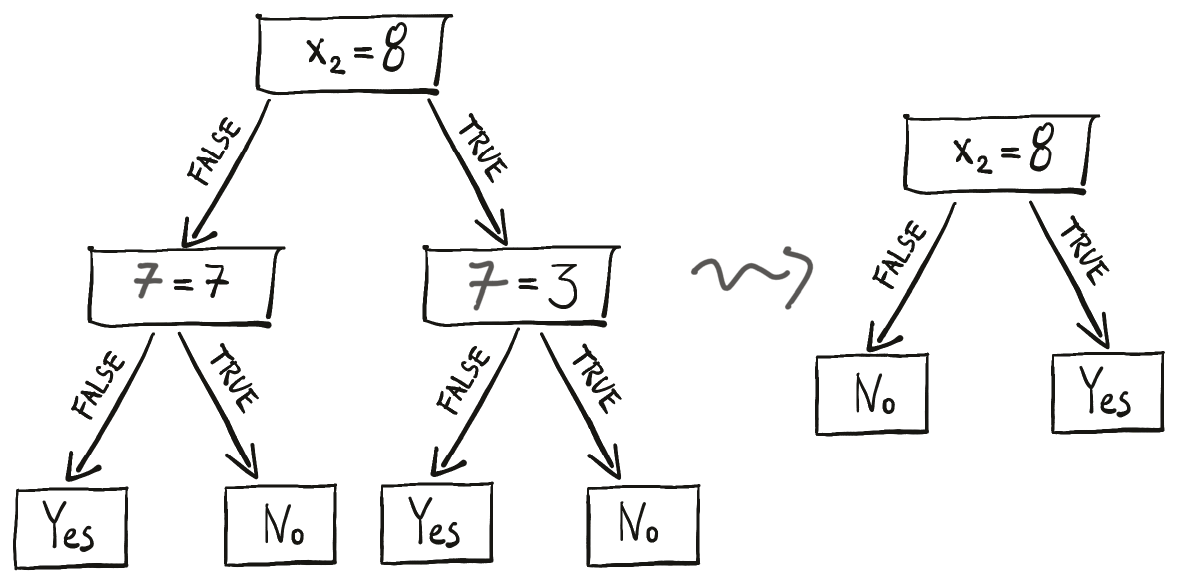}
\noindent
This construction can be generalized:
\begin{lemma}
\label{lem:su-mu-currying}
For every $A$, $B$, and $C$, there is an $l \in \nat$, such that 
for every $f$ of type $A \times B \suto C$, there exists an $f' : A \suto_l \Trees(B, C)$, such that 
for every $a \in A$, $b \in B$ and $c \in C$:
\[ f(a, b) = f'(a)(b) \]
Moreover, the mapping $f \mapsto f'$ is equivariant.
\end{lemma}
\begin{proof}
    Note that thanks to the single-use restriction,
    the size of a tree in $\Trees(A \times B, C)$ is bounded. This means that there is an $l \in \nat$ such that every $x_i$ appears at most
    $l$ times in $\Trees(A \times B, C)$. This means that we can construct $f'$ in the same way as in the previous example.
\end{proof}


\noindent
We are now ready to show that the two types of transition functions are equivalent:
\begin{lemma}
    \label{lem:su-transition-functions-equivalent}
    For every $k \in \nat$ and for every equivariant single-use function $f : \Sigma^k \times Q \suto Q$,
    there is a function $f' : \Sigma \eqto (Q \suto Q)$, such that for every $a \in \Sigma$ and $q \in Q$:
    \[ f'(a)(q) = f(\underbrace{a, a, \ldots, a}_{\textrm{$k$ times}}, q) \]
    And conversely: for every $f \in \Sigma \eqto (Q \suto Q)$, there exists a $k \in \nat$ and an equivariant
    $f' : \Sigma^k \times Q \suto Q$, such that:
    \[ f'(\underbrace{a, a, \ldots, a}_{\textrm{$k$ times}}, q) = f(a)(q) \]
\end{lemma}
\begin{proof}
The first part of the proof follows from Lemma~\ref{lem:su-mu-currying}. 
To prove the second part, we transform $f : \Sigma \eqto (Q \suto Q)$ into an equivalent $f' : \Sigma^k \times (Q \suto Q)$,
in the following way:
\[ \boxed{\Sigma \eqto (Q \suto Q)} \longtransform{\treeRepr\ \circ\ \cdot} \boxed{\Sigma \eqto(\Trees(Q, Q))} \longtransform{\textrm{Lemma~\ref{lem:fs-k-fold} }} \]
\[ \boxed{\Sigma \suto_k \Trees(Q, Q)} \longtransform{\textrm{Definition}~\ref{def:k-fold-use-functions}} \boxed{\Sigma^k \suto \Trees(Q, Q)} \longtransform{\textrm{Lemma~\ref{lem:su-curring-one-way}}} \boxed{\Sigma^k \times Q \suto Q} \]
\end{proof}

\subsection{Single-use automaton}
We are now ready to define the \emph{single-use automaton}:
\begin{definition}
\label{ref:single-use-automaton}
A \emph{single-use automaton} consists of:
\begin{enumerate}
    \item a polynomial orbit-finite alphabet $\Sigma$;
    \item a polynomial orbit-finite set of states $Q$;
    \item an equivariant initial state $q_0 \in Q$;
    \item an equivariant single-use\footnote{See Lemma~\ref{lem:su-mu-acceptance}.} acceptance function $f : Q \suto \{\Yes, \No\}$;
    \item a transition function $\delta : \Sigma \eqto (Q \suto Q)$.
\end{enumerate}
\end{definition}

Observe that for $\Sigma = \atoms$, the single-use automaton model is equivalent to the single-use register automaton:
Thanks to Example~\ref{ex:memory-conf-pof}, we know that the set of memory configurations of every single-use register automaton can be represented as a polynomial orbit-finite set.
Conversely, by Lemma~\ref{lemma:pof-normal-form-su}, every polynomial-orbit-finite set is isomorphic to
\[ \atoms^{k_1} + \atoms^{k_2} + \ldots + \atoms^{k_n}\comma \]
whose elements can be stored in the memory of a single-use register automaton.
The transition and acceptance functions are easily seen to be intertranslatable between the two models
(thanks to the single-use decision tree representation). 

As we have mentioned before, the expressive power of the single-use automaton does not change
if we allow for the acceptance function to be multiple-use:
\begin{lemma}
    \label{lem:su-mu-acceptance}
    For every single-use automaton whose acceptance function has the type $Q \eqto \{\Yes, \No\}$, there exists
    a standard single-use automaton (with a single-use accepting function $Q \suto \{\Yes, \No\}$), which recognizes the same 
    language.
    \end{lemma}
    \begin{proof}
        Let $\mathcal{A}$ be the automaton whose acceptance function $f$ has the type $Q~\eqto~\{\Yes,~\No\}$.  
        By Lemma~\ref{lem:fs-k-fold}, the function $f$ can be transformed into an equivalent
        $f' : Q \suto_k \{\Yes, \No\}$. The acceptance function is used only once,
        at the end of the computation, so we can simulate $\mathcal{A}$, using a single-use
        automaton that maintains $k$ copies of the state of $\mathcal{A}$
        (the automaton's set of states is $Q^k$).
    \end{proof}

\subsubsection{Single-use two-way automaton}
Using single-use functions, we can easily define single-use versions of many classical automata models:
In the next two chapters, we are going to see single-use Mealy machines, single-use two-way transducers, and
single-use SSTs. For now, we define the \emph{single-use two-way automaton}:

\begin{definition}
\label{def:single-use-two-way-automataon}
    A \emph{single-use two-way automaton} consists of:
\begin{enumerate}
    \item a polynomial orbit-finite alphabet $\Sigma$;
    \item a polynomial orbit-finite set of states $Q$;
    \item an equivariant initial state $q_0 \in Q$;
    \item a two-way transition function:
    \[
        \delta \ :\  (\Sigma + \{\vdash, \dashv\}) \ \longrightarrow_{\textrm{eq}}\  \left(Q \ \suto\  (Q \times \{\leftarrow, \rightarrow\}) \; + \;  \{\Yes, \No\}
        \right)
    \]
\end{enumerate}
\end{definition}
The model is analogous to the orbit-finite two-way automaton from Section~\ref{sec:dofa-variants},
but much weaker. An important difference between the 
two models is that the single-use two-way automaton has decidable 
emptiness. This is because, as we are going to show in Theorem~\ref{thm:dsua-2dsua-ofm}, 
every two-way single-use automaton can be effectively translated into 
a one-way single-use automaton. Which is, in turn, a
special case of a one-way orbit-finite automaton whose emptiness 
is decidable by \cite[Theorem~5.12]{bojanczyk2019slightly}. 

\subsection{Single-use automata and orbit-finite monoids}
\label{sec:su-automata-and-of-monoids}
The following theorem states the key property of the single-use automaton:
\begin{theorem}
\label{thm:dsua-2dsua-ofm}
All the following models recognize the same class of languages:
\begin{enumerate}
    \item One-way single-use automaton;
    \item Two-way single-use automaton;
    \item Orbit-finite monoids (limited to polynomial orbit-finite alphabets).
\end{enumerate}
\end{theorem}
\noindent
We prove this theorem using the following strategy:\\

\vspace{0.3cm}
\adjustbox{width=0.8\textwidth, center}{
\begin{tikzcd}
	{\substack{\textrm{Single-use}\\ \textrm{one-way}}} &&&& {\substack{\textrm{Single-use}\\ \textrm{two-way}}} \\
	\\
	&& {\substack{\textrm{Orbit-finite}\\ \textrm{monoids}}}
	\arrow["{\substack{\color{gray}\textrm{Lemma~\ref{lem:ofm-to-sua}}\\ \color{gray}\textrm{in Chapter~\ref{ch:mealey}}}}"{description}, from=3-3, to=1-1]
	\arrow["{\textrm{\color{gray}Obvious}}"{description}, from=1-1, to=1-5]
	\arrow["{\textrm{\color{gray}Lemma~\ref{lem:2sua-incl-ofm}}}"{description}, tail reversed, no head, from=3-3, to=1-5]
\end{tikzcd}}
\vspace{0.3cm}

In this chapter we show how to translate a single-use two-way automaton into an orbit-finite monoid.
In order to make the two-way construction easier to understand, we 
first present the construction for a one-way single-use automaton.

\subsection{One-way single-use automata $\subseteq$ Orbit-finite monoids}
\label{subsec:one-way-sua-to-ofm}
This section is dedicated to proving the following lemma:
\begin{lemma}
\label{lem:1sua-incl-ofm}
If a language $L$ over a polynomial orbit alphabet $\Sigma$ is recognized by a one-way single-use automaton, 
then it is also recognized by an orbit-finite monoid.
\end{lemma}

The proof of the lemma is very similar to the classical translation of a finite automaton into a finite monoid
(see Lemma~\ref{lem:aut-mon-eq}).
Take a single-use one-way automaton $\mathcal{A}$ over a polynomial orbit-finite alphabet $\Sigma$,
and let $Q$ be its polynomial orbit-finite set of states. For every word $w \in \Sigma^*$, define $\mathcal{A}$'s behaviour on $w$ as the following function $Q \fsto Q$:
\[ b_w(q) = \substack{\textrm{In which state will $\mathcal{A}$ exit $w$ on the right,}\\\textrm{ if it enters $w$ on the left in the state $q$?}} \]
Note that behaviours are compositional:  $b_{uv} = b_{v} \circ b_{u}$. Since single-use functions 
are (by definition) closed under compositions, it follows that, all possible behaviours of $\mathcal{A}$ are single-use functions.
This means that the language of $L$ is recognized by the monoid $Q \suto Q$ with the following operation
as $f \cdot g = g \circ f$, and with the following accepting set:
\[  \{f \ |\  f(q_0) \textrm{ is an accepting state}\} \]
This finishes the proof of Lemma~\ref{lem:1sua-incl-ofm}, because by Theorem~\ref{thm:su-orbit-finite},
$Q \suto Q$ is orbit-finite.
\subsection{Two-way single-use automata $\subseteq$ Orbit-finite monoids}
\label{sec:su-two-way-to-monoids}
In this section we show how to translate a two-way single-use automaton into an orbit-finite monoid:
\begin{lemma}
    \label{lem:2sua-incl-ofm}
    If a language $L$ over a polynomial orbit alphabet $\Sigma$ is recognized by a two-way single-use automaton, then it is also recognized by an orbit-finite monoid.
\end{lemma}

\noindent
In the proof we are going to use the theory of \emph{compositional functions} (described, for example, in \cite[page~5]{bojanczyk2020languages}):
\begin{definition}
    Let $\Sigma$ and $R$ be arbitrary sets (possibly infinite).
    We say that a function $h : \Sigma^* \to R$ is compositional if for every $u, w \in \Sigma^*$,
    the value $h(uw)$ is uniquely determined by the values of  $h(u)$ and $h(w)$.     
\end{definition}

For example, the function $f : \Sigma^* \to \nat$ that maps every word to its length is compositional,
and the function $g : \{a, b\}^* \to \{a, b, =\}$, defined as follows, is not compositional:
\[ g(w) = \begin{cases}
    a & \textrm{if there are more $a$'s than $b$'s in $w$}\\
    b & \textrm{if there are more $b$'s than $a$'s in $w$}\\
    = & \textrm{if there is equally many $a$'s and $b$'s in $w$}
\end{cases}\]

A compositional function $h : \Sigma^* \to R$ together with an accepting set $F \subseteq R$
can be used to recognize the language: $\{ w \ | \ f(w) \in F\} \subseteq \Sigma^*$.

\begin{lemma}
\label{lem:compositional-to-monoids}
If a language $L \subseteq \Sigma^*$ is recognized by an equivariant, compositional function
$\Sigma^* \eqto R$, for an orbit-finite $R$, then $L$ is also recognized by an orbit-finite monoid.
\end{lemma}
\begin{proof}
    Let $M$ be the image of $\Sigma$ under $f$, and let us define the following operation on $M$:
    \[ \begin{tabular}{cc}
        $x \cdot y = f(uv)$ & where $u, v \in \Sigma^*$ are such that $f(u) = x$ and $f(v) = y$
    \end{tabular} \]
    Because $M = f(\Sigma)$, we know that such $u$ and $v$ always exist, and thanks to $f$'s 
    compositionality, we know that the value $x \cdot y$ does not depend on the choice of $u$ and $v$.
    This operation is associative: for every $x, y, z$, if we pick some $u, v, w$ 
    such that $f(u) = x$, $f(v) = y$, $f(w) = z$, we have that:
    \[ (x \cdot y) \cdot z = f(uv) \cdot z = f((uv)w) = f(u(vw)) = x \cdot f(vw) = x \cdot (y \cdot z) \]
    In a similar way, we can show that the image of the empty word ($f(\epsilon)$) is the operation's identity element.
    It follows that $M$ is a monoid, and it is easy to see that $M$ recognizes $L$. This leaves us with 
    showing that $M$ is an orbit-finite monoid. To see that $M$ is an orbit-finite set, 
    we notice that $M = h(R)$. This is enough, because orbit-finiteness is preserved by taking images under finitely supported functions (\cite[Lemma~3.24]{bojanczyk2019slightly}).
    Finally, we show that the product operation is equivariant: for every $x, y$ and $\pi$, if we pick some $u, v$ such that 
    $f(u) = x$, and $f(v)= y$, we have that:
    \[\pi(x \cdot y) = \pi(f(uv)) = f(\pi(uv)) = f(\pi(u)\pi(v)) = \pi(x) \cdot \pi(y)\]
\end{proof}

We are now ready to show how to translate two-way single-use automata into orbit-finite monoids. The construction 
is a single-use variant of \cite[Theorem~2]{shepherdson1959reduction}: Take a two-way automaton 
$\mathcal{A}$ over the alphabet $\Sigma$, and let $Q$ be its polynomial orbit-finite set of states.
We define $\mathcal{A}$'s behaviour on a word $w \in \Sigma^*$ to be a function:
\[ \begin{tabular}{cccc}
    $b_w:$ & $ Q\times \{\leftarrow, \rightarrow\}$ & $\longrightarrow$ & $(Q \times \{\leftarrow, \rightarrow\}) + \{\Yes, \No\}$ \\
\end{tabular} 
\]
The function is analogous to the one-way behaviour, but this time $\mathcal{A}$ can enter $w$ from the left or from the right, 
and it might leave $w$ from the left or from the right. It might also never exit $w$, because it accepts,
rejects, or starts to loop (which is considered as rejecting) inside $w$. Let us now show that all behaviours are single-use functions:
\begin{claim}
\label{claim:two-way-behaviour-single-use}
    For every $w \in \Sigma^*$, function $b_w$ is single-use:
    \[ \begin{tabular}{ccccc}
        $b_w$ & $\in$ & $ Q\times \{\leftarrow, \rightarrow\}$ & $\longsuto$ & $(Q \times \{\leftarrow, \rightarrow\}) + \{\Yes, \No\}$ \\
    \end{tabular} 
    \] 
\end{claim}
\begin{proof}
    We prove the claim by induction on the length of $w$. For the empty word, we have that $b_\epsilon = \coproj_1$
    (note that the symbol $\leftarrow$ means entering from the left and exiting from the right).
    Now, for the induction step let us take $w = va$. 
    We know that $b_a$ is equal to $\delta(a)$, which means that it is single-use. Thanks to the induction assumption,
    we also know that $b_v$ is single-use. Now, let us use $b_v$ and $b_a$ to construct $b_{va}$ as a single-use 
    function. Our first (incorrect) approach is to define four mutually recursive functions, each of the type:
    \[ \begin{tabular}{ccc}
        $Q$ & $\longsuto$ & $(Q \times \{\leftarrow, \rightarrow\}) + \{\Yes, \No\}$ \\
    \end{tabular} 
    \]
    They correspond to four possible entry points to $va$:
    $(\shortrightarrow \! va)$, $(va \! \shortleftarrow)$, $(v \! \shortleftarrow \! a)$, and $(v \!\shortrightarrow \! a)$.
    Here are their definitions:
    \[
    \begin{tabular}{cc}
        $(\shortrightarrow va)(q) = \begin{cases}
            (v \shortrightarrow a)(q') & \substack{ \textrm{if }b_v(q, \rightarrow)= (q', \rightarrow)}\\
            b_v(q, \rightarrow)  & \textrm{otherwise} 
        \end{cases}$ &
        $(va\shortleftarrow)(q) = \begin{cases}
            (v \shortleftarrow a)(q') & \substack{\textrm{if } b_a(q, \leftarrow) = (q', \leftarrow)}\\
            b_a(q, \leftarrow)  & \textrm{otherwise} 
        \end{cases}$ \\
        &\\
        $(v \shortleftarrow a)(q) = \begin{cases}
            (v \shortrightarrow a)(q') & \substack{\textrm{if }  b_v(q, \leftarrow) = (q', \rightarrow)}\\
            b_v(q, \leftarrow)  & \textrm{otherwise} 
        \end{cases}$ &
        $(v\shortrightarrow a)(q) = \begin{cases}
            (v \shortleftarrow a)(q') & \substack{\textrm{if }  b_a(q, \rightarrow) = (q', \leftarrow)}\\
            b_a(q, \rightarrow)  & \textrm{otherwise} 
        \end{cases}$
    \end{tabular}
    \]
    There are two problems 
    with this approach: First, if $\mathcal{A}$ starts to loop in $va$, then this recursive definition will never terminate. 
    Second, it is unclear if single-use functions are closed under taking fixpoints. We can deal with both of those problems
    by bounding the number of times that $\mathcal{A}$ can cross between $v$ and $a$:
    Let $\alpha \subset_{\textrm{fin}} \atoms$ be the finite set of all atoms that appear in $va$ (i.e. the support of $va$). 
    Then,  the subset $Q_\alpha \subseteq Q$ of all states supported 
    by $\alpha$ is finite as well -- this can be shown either by induction on the structure of $Q$, or 
    by applying \cite[Lemma~5.2]{bojanczyk2013nominal}. Now, let us notice that all the 
    states that appear in $\mathcal{A}$'s run on $va$ are supported by $\alpha$
    (thanks to Lemma~\ref{lem:fs-functions-preserve-supports}), which means that 
    they belong to $Q_\alpha$. It follows that if $\mathcal{A}$ 
    crosses between $v$ and $a$ more than $2 \cdot |Q_\alpha|$ times,
    it will visit some state for the second time in the same position, which means that it is going to loop\footnote{
    Observe that this is not true for multiple-use two-way automata. In particular, the automaton from Claim~\ref{claim:2dofa-all-distinct} can visit the last position of a word 
    $O(n)$ times, where $n$ is the length of the word. However, as a further research idea, it might be interesting to consider
    a semantically restricted model of \emph{$k$-crossing (multiple-use) orbit-finite two-way automata}.
    I would like to thank Emmanuel Filiot for this suggestion.}
    In order to use this observation, we define $(\shortrightarrow \! va)_k$ to be the limited version of $(\shortrightarrow \! va)$, 
    that only allows $\mathcal{A}$ to cross $k$ times between $v$ and $a$, before it rejects the input.
    And similarly, for $(va \! \shortleftarrow)_k$, $(v \! \shortleftarrow \! a)_k$, and $(v \!\shortrightarrow \! a)_k$. 
    The definition of the functions is inductive on $k$. Here is the definition for$(\shortrightarrow \! va)_k$
    (other definitions are analogous): 
    \[
            (\shortrightarrow va)_k(q) = \begin{cases}
                (v \shortrightarrow a)_{k-1}(q') & \substack{ \textrm{if }b_v(q, \rightarrow)= (q', \rightarrow)}\\
                b_v(q, \rightarrow)  & \textrm{otherwise} 
            \end{cases}
    \]
    For the induction base we pick $k = -1$, in which case all four functions are equal to $\const_\No$.
    Now, thanks to the if-then-else combinator from Example~\ref{ex:if-else-su}, it is not hard 
    to use induction on $k$, to show that all those functions are single-use.
    To finish the proof of the claim, we notice that we can construct $b_{va}$ as:
    \[ Q \times \{\leftarrow, \rightarrow\} \transform{\distr} Q + Q \longtransform{[(va\shortleftarrow)_{k}, (\shortrightarrow va)_{k}]} Q \times \{\leftarrow, \rightarrow\} + \{\Yes, \No\} \comma \]
    for $k$ equal to $2 \cdot |Q_\alpha| + 1$. 
\end{proof}
It is easy to convince oneself that the behaviour $b_{uv}$ depends entirely on the behaviours $b_u$ and $b_v$
(even though the exact formula for combining those behaviours might not be clear). Moreover, it is also 
easy to see that the behaviour $b_w$ uniquely determines whether $w$ is accepted by $\mathcal{A}$.
This means that the language recognized by $\mathcal{A}$ is also recognized by the compositional function $w \mapsto b_w$, 
whose codomain is,
\[ \begin{tabular}{ccc}
    $ Q\times \{\leftarrow, \rightarrow\}$ & $\longsuto$ & $(Q \times \{\leftarrow, \rightarrow\}) + \{\Yes, \No\}$ \\
\end{tabular} 
\]
By Theorem~\ref{thm:su-orbit-finite} it is orbit-finite, which means that we can finish the construction of an orbit-finite monoid, 
by applying Lemma~\ref{lem:compositional-to-monoids}. (Observe that in the non-single-use setting, the proof fails because $b_w$ is not single-use, 
and we cannot apply Theorem~\ref{thm:su-orbit-finite}).

\section{Nondeterministic single-use automata}
\label{sec:non-determinsitic-single-use}
In this section, we explain why combining nondeterminism with the single-use restriction is problematic.
The key concept of nondeterminism is a relation, so in order to define a \emph{nondeterministic single-use automaton}, 
it would be useful to define \emph{single-use relations}. Unfortunately, the question of finding a well-behaved
notion of a single-use relation remains open. Below, we present two possible definitions of single-use relations, 
and explain why they are not well behaved. 
\begin{definition}
\label{def:su-rel-1}
    Let $X$ and $Y$ be polynomial orbit-finite sets.
    We say that a relation $R \subseteq X \times Y$ is \emph{single-use},
    if there is a single-use function $f  : X \times Y \suto \{\Yes, \No\}$ such that:
    \[ R = \{ (x, y) \ |\ x \in X,\ y \in Y,\ \textrm{such that } f(x, y) = \Yes \}\tdot \]
    We denote the set of all single-use relations between $X$ and $Y$ as $X \suto_\textrm{rel} Y$. 
\end{definition}
Now, we define \emph{single-use nondeterministic automaton} to be just  
like the deterministic single-use automaton, except that its transition function, becomes a transition relation of the following type:
\[ \Sigma \eqto (Q \suto_\textrm{rel} Q) \]
Observe that Definition~\ref{def:su-rel-1} allows nondeterministic single-use automata to recognize the language:
\[\textrm{``The first letter appears again''}\tdot\]
This is because a nondeterministic automaton can store the first letter in its state, 
nondeterministically guess the position where the first letter reappears, and verify its 
guess by comparing the two letters. Thanks to nondeterminism, this construction
requires only one copy of the first letter. For example, here is a tree of all possible 
runs of the automaton on the word $1\ 2 \ 3 \ 1 \ 3$:
\bigpicc{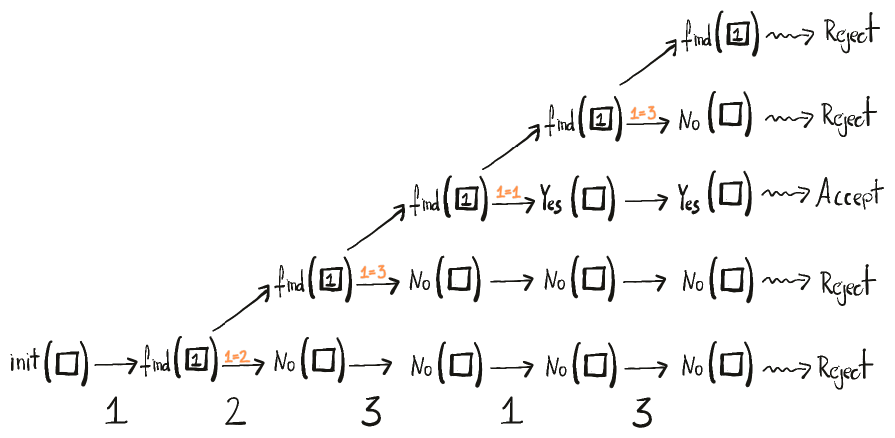}
It follows that nondeterministic single-use automata are more expressive than orbit-finite monoids.
The reason behind this is that  $Q \suto_\textrm{rel} Q$ is not closed under compositions: 
Consider the following relation:
\[ \mathtt{check}_7 \in  (\atoms + \top)\suto_\textrm{rel} (\atoms + \top) \]
\[
    \mathtt{check}_7 = \{(\top, \top)\}\ \cup\ \{ (a, a)\; | \; a \in \atoms \}\ \cup\ \{ (7, \top)\} \tdot
\]
It is not hard to see that both $\mathtt{check}_7$ and an analogous
$\mathtt{check}_8$ are single-use relations, but their composition
$\mathtt{check}_{7} \circ \mathtt{check}_8$ is not.\\

Our second approach to define single-use relations is based on 
the idea that every relation between $X$ and $Y$ can be expressed as a function:
\[ X \to P(Y) \tdot \]
The problem, with lifting this definition to polynomial orbit-finite sets is that 
$P(Y)$ is neither polynomial nor orbit-finite. In order to make it at least
polynomial, we can represent it as $Y^*$. To make it additionally orbit finite, we 
notice that for every function $f : X \fsto Y^*$, there exists a $k_f$,
such that for all $x \in X$:
\[|f(x)| \leq k_f \tdot \]
This is because $X$ has only finitely many $\supp(f)$-orbits, and for each of those orbits
the length of $f(x)$ is constant. This leads to the following (alternative) definition 
of a \emph{single-use relation}:
\begin{definition}
    Let $X$ and $Y$ be polynomial orbit-finite sets:
    We say that a relation $R \subseteq X \times Y$ is \emph{single-use}, if
    there is a $k \in \nat$, and a single-use function $f  : X \suto Y^{\leq k}$, such that:
    \[ R = \{ (x, y) \ |\ x \in X,\ y \in f(x)\} \]
\end{definition}
This time, it is not hard to see that single-use relations are closed under compositions.
However, because $k$ can be arbitrarily large, $Q \suto_\textrm{rel} Q$ 
will usually be orbit-infinite. 
This can be exploited, to construct a nondeterministic single-use automaton that recognizes the language:
\[\textrm{``The first letter appears again''}\tdot\]
The construction is slightly different from the previous one: This time the automaton
saves the first letter in its state, then it nondeterministically picks another position,
saves a second letter in its state, and keeps the two letters until the end of the word.
For example, here is a tree of all possible runs of the automaton on the word $1\ 2 \ 3 \ 1 \ 3$:
\bigpicc{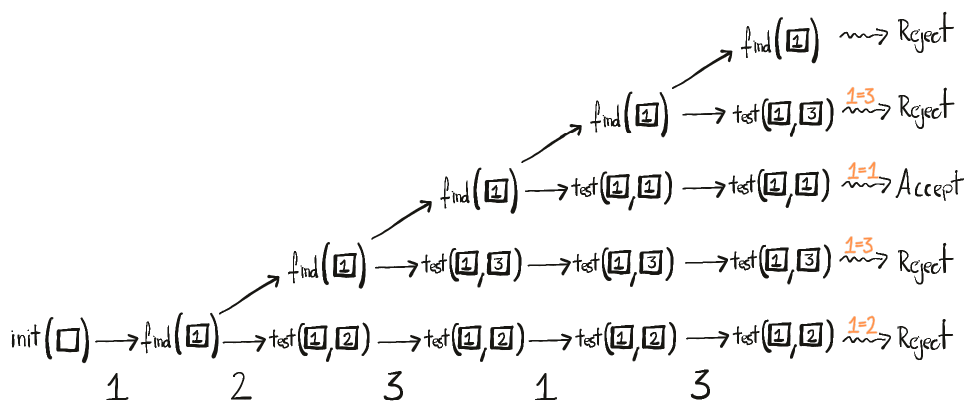}
\chapter{Single-use Mealy machines and their Krohn-Rhodes decompositions}
\label{ch:mealey}

So far we have discussed languages (i.e. subsets of $\Sigma^*$) recognized by different automata models.
In the next two chapters, we are going to discuss word transformations (i.e. functions $\Sigma^* \to \Gamma^*$)
computed by \emph{finite state transducers} (i.e. output-producing variants of automata).
Before we introduce single-use models for infinite alphabets, let us discuss the classical theory of
finite state transducers: Their classification is finer than the one of automata --
two models that define the same class of languages, might define two different
classes of word transformations. To illustrate this, we present three models 
of transducers that define three different classes of word transformations, 
but whose underlying automaton models recognize the same class of languages (i.e. regular languages).
\begin{description}
    \item[1. Mealy machine] This model (introduced in \cite[Section 2.1]{mealy1955method})
          is a version of the deterministic register automaton, where every transition 
          produces exactly one output letter. This is reflected in the type of transition function:
          \[ \underbrace{Q}_{\textrm{current state}} \times \underbrace{\Sigma}_{\textrm{{input letter}}}  \to
          \underbrace{Q}_{\textrm{new state}} \times \underbrace{\Gamma}_{\textrm{output letter}} \]
          A Mealy machine produces output for every input word, which means that it does not 
          have accepting (or rejecting) states. It follows that every
          Mealy machine computes a total, length-preserving function.
          Here is an example of a Mealy machine that recognizes the following transduction:
          \[\textrm{``Change every other a to b''} \in  \{\textrm{a, b}\}^* \to \{\textrm{a, b}\}^*\]
          \vsmallpicc{Mealy-ex}
    \item [2. Unambiguous Mealy machine] This model is a nondeterministic version of a Mealy 
          machine. This means that the transition function becomes a transition relation:
          \[ \underbrace{Q}_{\textrm{current state}} \times \underbrace{\Sigma}_{\textrm{{input letter}}}  \times
          \underbrace{Q}_{\textrm{new state}} \times \underbrace{\Gamma}_{\textrm{output letter}} \]
          The nondeterministic version of a Mealy machine can have a number of runs on a single
          input word. To guarantee that it recognizes a total, function we reintroduce
          the accepting and rejecting states, and we require that for every input word there 
          is exactly one accepting run (this is the unambiguity condition). Here is an 
          example of an unambiguous Mealy machine that computes the following function:
          \[\textrm{``Swap the first and the last letters''} \in  \{\textrm{a, b}\}^* \to \{\textrm{a, b}\}^*\]
          \smallpicc{Rational-ex}
          (Notice that the deterministic Mealy machine is not able to compute this language, 
          as it has no way of guessing what the last letter is going to be).
          This class of transductions is known as \emph{rational letter-to-letter functions}.\footnote{
            The class was introduced in \cite{eilenberg1974automata}. For a more detailed bibliographical
            note, see the footnote in \cite[Section~12.2]{bojanczyk2018automata}. 
          }  
          It is worth 
          pointing out that unambiguous Mealy machines admit the following decomposition into 
          two deterministic Mealy machines\footnote{This is known as the Elgot-Mezei theorem. It was originally 
          shown in \cite[Theorem~7.8]{elgot1963two}. Since I was not able to access the full version 
          of the original paper, I relied on \cite[Theorem~12.1]{bojanczyk2018automata}.
          It is worth pointing out that the two papers prove the theorem for slightly
          different models (functional Mealy machines and unambiguous NFAs with output), 
          but the proof from \cite[Theorem~12.1]{bojanczyk2018automata} can be easily adapted 
          to work with unambiguous Mealy machines.}:
          \[ \left(\substack{\textrm{unambigous}\\ \textrm{Mealy machine}}\right) = 
             \left(\substack{\textrm{left-to-right}\\ \textrm{Mealy machine}} \right) \circ
             \left(\substack{\textrm{right-to-left}\\ \textrm{Mealy machine}} \right)  
          \]

    \item [3. Two-way transudcer] This is a version of the two-way automaton
          whose every transition has an option of producing an output letter
          (it is a classical model discussed, for example, in \cite{engelfriet2001mso}).
          The transition function of a two-way transducer has the following type:
          \[\begin{tabular}{ccccc}
            $\underbrace{Q}_{\substack{\textrm{current}\\ \textrm{state}}} \times
            (\underbrace{\Sigma}_{\substack{\textrm{current}\\ \textrm{letter}}} + \underbrace{\{ \vdash, \dashv \}}_{\substack{\textrm{end of word}\\ \textrm{markers}}})$ 
            & $\longrightarrow$
            & $\underbrace{Q}_{\substack{\textrm{new}\\ \textrm{state}}} \times \underbrace{\{\leftarrow, \rightarrow\}}_{\substack{\textrm{direction of}\\ \textrm{the next step}}} \times \underbrace{(\Gamma + \epsilon)}_{\substack{\textrm{output}\\ \textrm{letter}}}$& $+$ & $\underbrace{\textrm{finish}}_{\substack{\textrm{finish}\\ \textrm{the run}}}$
          \end{tabular}\]
          Here is an example of a two-way automaton that recognizes the function:
          \[\textrm{``Reverse the input''} \subseteq \{a, b\}^* \to \{a, b\}^*\]
          \smallpicc{ex-two-way}
          It can be shown that an unambiguous Mealy machine cannot compute the reverse function.
          Note that it is possible for a two-way automaton to loop. To guarantee that it does 
          not loop we can add a semantic requirement that prohibits two-way automata from looping
          (an alternative approach would be to say that a looping run produces the empty word).
          The class of functions computed by the two-way transducer has many equivalent definitions,
          including streaming string transducers (see \cite[Section~3]{alur2010expressiveness} or Section~\ref{subsec:sst-def} in this thesis)
          and the logical model of MSO-transductions (see \cite[Section~2]{courcelle1994monadic}).
          This class of transductions is known as the \emph{regular functions}.
\end{description}
The three transducer models are well behaved: For example, they are closed under composition and their equivalence
problem (as defined below) is decidable. 
\[
    \begin{tabular}{ll}
        \textbf{Input:} & Two transducers $\mathcal{A}$ and $\mathcal{B}$.\\
        \textbf{Output:} & Do $\mathcal{A}$ and $\mathcal{B}$ compute the same function?
    \end{tabular}
\]
In the next two chapters, we are going to use the single-use restriction to develop a similar theory
for infinite alphabets. Chapter~\ref{ch:mealey} covers single-use (deterministic) Mealy machines
and Chapter~\ref{ch:su-regular-functions} covers single-use two-way transducers. 
The theory of single-use unambiguous Mealy machines is still a work in progress, but it 
is briefly discussed in Chapter~\ref{ch:su-regular-functions}.\\

Finally, let us briefly mention the class of \emph{polyregular functions}: This is a
class over finite alphabets that extends \emph{regular functions} while keeping many of their 
desirable properties (see \cite{bojanczyk2022transducers}). Unfortunately,
the question of finding a notion of \emph{single-use polyregular functions}
that is a well-behaved class of functions remains open, so we do not discuss it 
in this thesis. However, we would like to note that this could be an interesting direction 
for further research. 

\section{Single-use mealy machines}
Single-use Mealy machines are a transducer model that computes length-preserving
functions $\Sigma^* \to \Gamma^*$. Here is its definition:
\begin{definition}
    A \emph{single-use Mealy machine} consists of:
    \begin{enumerate}
        \item a polynomial orbit-finite input alphabet $\Sigma$ and a polynomial orbit-finite output alphabet $\Gamma$;
        \item a polynomial orbit-finite set of states $Q$;
        \item an initial state $q_0 \in Q$;
        \item a single-use transition function:
                \[ \delta : \underbrace{\Sigma}_{
                    \substack{\textrm{current letter} }} \longrightarrow_\textrm{eq} \left(\underbrace{Q}_{\textrm{current state}} \suto
                   \left(\underbrace{Q}_{\textrm{new state}} \times \underbrace{\Gamma}_{
                       \substack{\textrm{output letter} \\
                                }}\right)\right) \]
    \end{enumerate}
\end{definition}

Notice that if $\Sigma$, $\Gamma$ and $Q$ are finite (and not only orbit-finite), 
then this definition matches the classical definition of a Mealy machine -- 
this follows from Example~\ref{ex:finite-su}. Let us now consider
some examples of Mealy machines:

\begin{example}[Length-preserving single-use homomorphism]
\label{ex:ll-homomorphism}
    For any $h: \Sigma \eqto \Gamma$ (where $\Sigma$ and $\Gamma$ are polynomial orbit-finite sets), we define
    $h^* : \Sigma^* \to \Gamma^*$
    to be the length-preserving transduction that applies $h$ to every input letter. 
    This $h^*$ can be computed by a one-state\footnote{in \cite{mealy1955method}
    one-state machines are called \emph{combinatorial circuits}} single-use Mealy machine,
    with the following transition function:
    \[ \delta(a)(1) = (1, h(a)) \]
\end{example}

\begin{example}[Single-use atom propagation]\label{ex:su-propagation}
    \emph{Single-use atom propagation} simulates single-use operations on one register. 
    Its input alphabet is a set of instructions:
    \[\underbrace{\atoms}_{\textrm{store an atom}} +
      \underbrace{\downarrow}_{\substack{\textrm{output the atom from the register}\\
                                        \textrm{and destroy it}}} + \underbrace{\epsilon}_{\textrm{do nothing}} \]
    Its output alphabet is:
    \[\underbrace{\atoms}_{\textrm{outputs of $\downarrow$}} + \underbrace{\epsilon}_{\textrm{empty output}}\]
    Here is an example input and output of the function (the grey arrows are only informative -- they are not part of the input or output):
    \picc{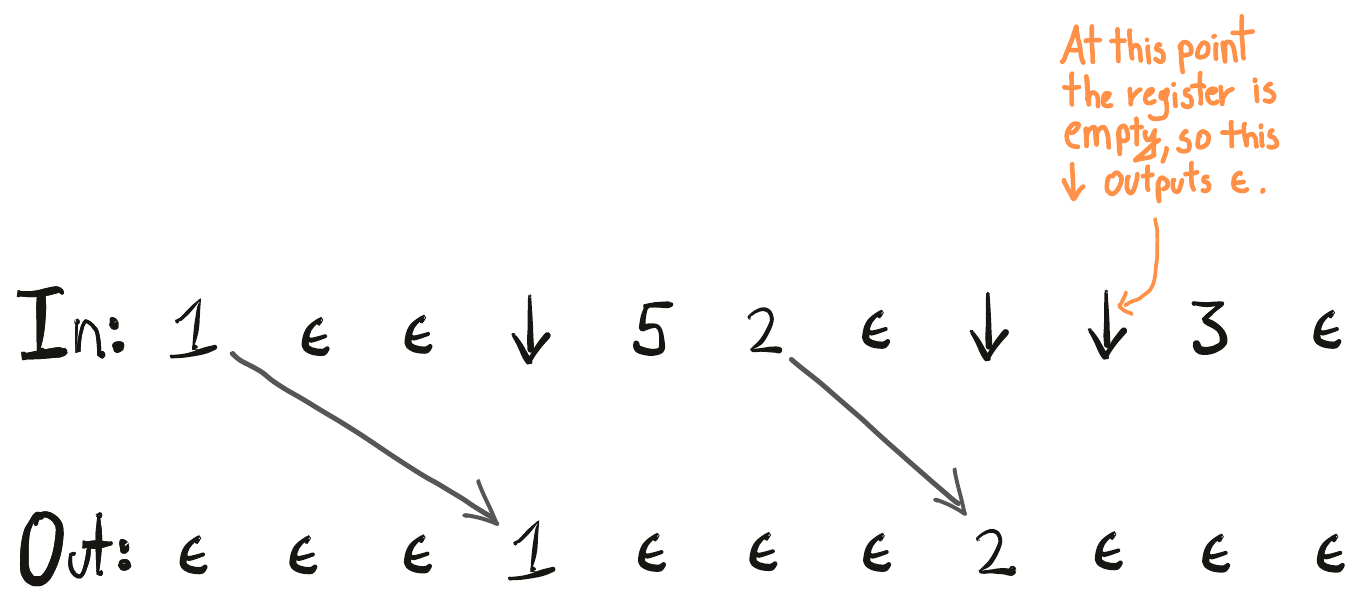}
    \noindent
    The semantic is rather intuitive, but to avoid confusion, we also define it formally:
    the $i$-th output letter is equal to $a \in \atoms$, if (a) $i$-th input letter is equal to $\downarrow$,
    and (b) there is $j < i$ such that the $j$-th input letter is equal to $a$ and 
    every input letter between $i$ and $j$ is equal to $\epsilon$. Otherwise, the $i$-th input
    letter is equal to $\epsilon$. Single-use atom propagation can be computed by a single-use Mealy machine:
    its states are $\atoms + \epsilon$ and its transition function looks as follows:
\[ \begin{tabular}{c c c}
    $\delta(q)(\downarrow) = (\epsilon, q)$ & 
    $\delta(q)(\epsilon) = (q, \epsilon)$ &
    $\delta(q)(a \in \atoms) = (a, \epsilon)$   
\end{tabular}
\]
\end{example}

Notice that since the transition function of a Mealy machine is single-use, it has to forget
every atom that it outputs. For this reason, the multiple-use version of Example~\ref{ex:su-propagation}
cannot be computed by a single-use Mealy machine. In order to prove this, we use a simple 
quantitative reasoning:
\begin{definition}
    For every element $x$ of a polynomial orbit-finite set $X$, we define the \emph{multi-support}
    of $x$ (denoted $\msupp(x)$) to be the \emph{multiset} of all atoms that appear in $x$, 
    with repetitions. This definition can also be extended to work on words over 
    polynomial orbit-finite sets. For example (if we take $X = \atoms^3$):
    \[ \msupp(3, 2, 3) = \{2, 3, 3\} \]
\end{definition}
\begin{lemma}
\label{lem:su-unif-cont}
    If $f$ is a function computed by a single-use Mealy machine, then there exists 
    a $k \in \nat$, such that for every word $w$ and every atom $a$:
    \[ \left(\substack{\textrm{The number of times $a$}\\\textrm{appears in $f(w)$}}\right) \leq  k \cdot \left(\substack{\textrm{The number of times $a$}\\\textrm{appears in $w$}}\right) \]
\end{lemma}
\begin{proof}
    The lemma follows from Lemma~\ref{lem:su-transition-functions-equivalent} combined 
    with the following claim (which can be easily shown using the tree representation):
    \begin{claim}
        Let $f \in X \suto_\textrm{eq} Y$, and 
        for every $x \in X$:
        \[ \msupp(f(x)) \subseteq \msupp(x)\]
    \end{claim}
\end{proof}
It is not hard to see that the multiple-use version of Example~\ref{ex:su-propagation} does not 
satisfy the condition from Lemma~\ref{lem:su-unif-cont}, so it cannot be computed by a single-use 
Mealy machine. The following example shows that the single-use restriction does not apply to the 
finite information:
\begin{example}[Multiple-use bit propagation]
    \label{ex:flip-flop}
        The \emph{multiple-use bit propagation} function can be seen as the classical variant
        of the single-use atom propagation from Example~\ref{ex:su-propagation}.
        It simulates operations on one multiple-use register that stores one bit of information,
        represented as one of two values -- ($\newmoon$ or $\fullmoon$).
        Its input alphabet is the following set of instructions:
        \[ \{\underbrace{\newmoon}_{\substack{
            \textrm{output the register value} \\
            \textrm{and save $\newmoon$ to the register}
        }}, \underbrace{\fullmoon}_{\substack{
            \textrm{output the register value} \\
            \textrm{and save $\fullmoon$ to the register}
        }},
        \underbrace{\epsilon}_{\substack{
            \textrm{output the register value} \\
            \textrm{and keep its contents}
        }}\} \]
        The output alphabet is the same: $\{\fullmoon, \newmoon, \epsilon \}$.  (Value $\epsilon$ denotes
        empty register -- the register cannot be emptied, but being empty is its initial value.)
        Here is an example input and output (again,
        the grey arrows are only informative -- they are not part of the input or output):
        \picc{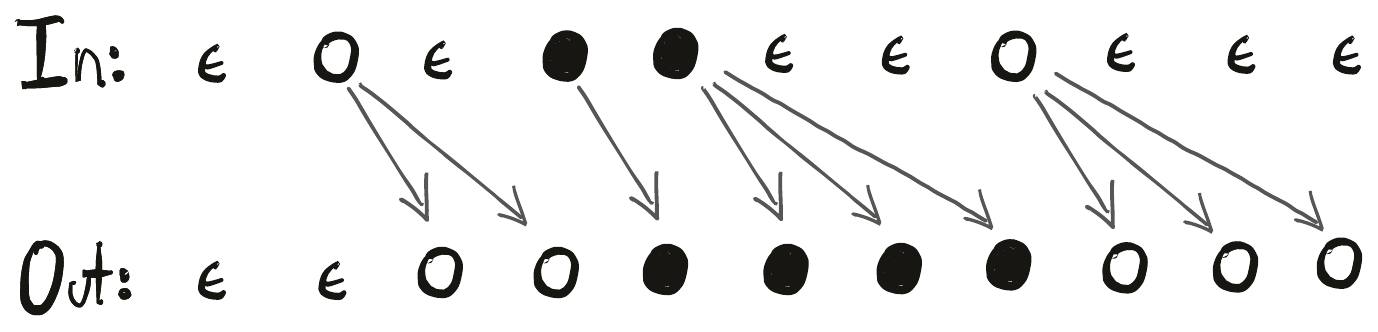}
        \noindent
        The Mealy machine recognizing the bit propagation has three states: $Q = \{\fullmoon, \newmoon, \epsilon\}$.
        Its transition function is as follows:
        \[ \begin{tabular}{c c c} 
           $\delta(q, \fullmoon) = (\fullmoon, q)$ &
           $\delta(q, \newmoon) = (\newmoon, q)$ &
           $\delta(q, \epsilon) = (q, q)$
        \end{tabular}
        \]
(To see that this is a single-use function, we notice that $Q$ is finite and apply Example~\ref{ex:finite-su}.)
\end{example}

Here is an example of a more general class of functions recognized by Mealy machines:
\begin{example}[Monoid prefixes]
\label{ex:monoid-prefix}
    For every finite monoid $M$, define the \emph{$M$-prefix function} $M^* \to M^*$
    to be the function that computes products of the input prefixes:
    \[
    \begin{tabular}{ccccccc}
       $s_1$ & $s_2$ & $s_3$ & $\ldots$ & $s_n$  \\
       & &\rotatebox[origin=c]{270}{$\mapsto$} & &    \\
        $s_1$ & $s_1s_2$ & $s_1s_2s_3$ & $\ldots$ & $s_1 s_2 \ldots s_n$ 
   \end{tabular} \]
   In particular, if we take a monoid $ P = \{\fullmoon, \newmoon, \epsilon\}$ 
   whose multiplication is given by the following Cayley table, then the $P$-prefix function 
   is the multiple-use bit propagation, but with the results shifted one position to the right:
   \[
    \begin{tabular}{l|lll}
        $\cdot_P$ & $\epsilon$ & $\fullmoon$ & $ \newmoon$ \\ \hline
        $\epsilon$  & $\epsilon$ & $\fullmoon$ & $\newmoon$ \\
        $\fullmoon$         & $\fullmoon$        & $\fullmoon$ & $\newmoon$  \\
        $\newmoon$         & $\newmoon$        & $\fullmoon$ & $\newmoon$ 
        \end{tabular}
   \]

   It is not hard to see that for every finite $M$, the $M$-prefix function
   can be computed by a single-use Mealy machine, where 
   $Q = M$, $q_0 = 1_M$ and whose transition function is defined as:
    \[ \delta(p, g) = (p \cdot g, p \cdot g) \]
   (Again, $Q$ is finite, so thanks to Example~\ref{ex:finite-su} we know that $\delta$ is single-use.)
\end{example}

It is worth pointing out that if $M$ polynomial orbit-finite (and not finite), then 
the $M$-prefix function will not necessarily be computable by a single-use Mealy machine:
\begin{example}
\label{ex:monoid-products-not-single-use}
Take $M = \atoms + 1$ such that $1$ is the identity element, and otherwise 
the operation is defined as follows:
\[
    a \cdot b = a
\]
If a sequence $a_1, \ldots, a_n \in \atoms^*$ consists of $n$ different atoms,
then the $M$-prefix function looks as follows:
\[
\begin{tabular}{ccccccc}
    $a_1$ & $a_2$ & $a_3$ & $\ldots$ & $a_n$  \\
    & & \rotatebox[origin=c]{270}{$\mapsto$}  & &    \\
    $a_1$ & $a_1$ & $a_1$ & $\ldots$ & $a_1$  \\ 
\end{tabular}
\]
This function violates the condition from Lemma~\ref{lem:su-unif-cont},
which means that it cannot be computed by a single-use Mealy machine. 
\end{example}

\section{Krohn-Rhodes decomposition}
\label{sec:krohn-rhodes-decomp}

The Krohn-Rhodes theorem \cite[Equation 2.2]{krohn1965prime} states that every function computed
by a classical Mealy machine can be decomposed into certain prime functions.
It uses two types of composition: 
\begin{align*}
    \infer[\text{sequential}]
{ \Sigma^* \stackrel {g \circ f }\longrightarrow \Delta^*}
{ \Sigma^* \stackrel f \longrightarrow \Gamma^* &  \Gamma^* \stackrel g \longrightarrow \Delta^*}
\qquad 
\infer[\text{parallel}]
{ (\Sigma_1 \times \Sigma_2)^* \stackrel {f_1 \times f_2}\longrightarrow (\Gamma_1 \times \Gamma_2)^*}
{ \Sigma_1^* \stackrel {f_1} \longrightarrow \Gamma_1^* &  \Sigma_2^* \stackrel {f_2} \longrightarrow \Gamma_2^*}
\end{align*}
The sequential composition is the usual function composition,
and the parallel composition, which only makes sense for
length-preserving functions, applies one function to the first coordinate of the word, and the other function to the second coordinate of the word.
Here is a schematic depiction of the two compositions:
\smallpicc{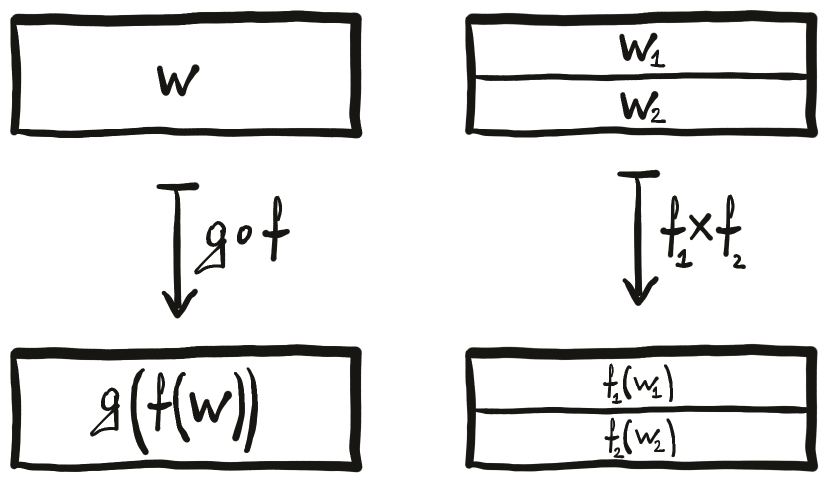}
\begin{theorem}[Krohn-Rhodes, {\cite[Equation 2.2]{krohn1965prime}}]
\label{thm:classical-kr}
        The class of functions computed by Mealy machines over finite alphabets
        is equal to the
        smallest class of function that is closed under sequential and parallel compositions,
        and which contains the following \emph{ (classical) prime functions}:
        \begin{enumerate}
            \item the $h^*$ function from Example~\ref{ex:ll-homomorphism}, for every $h : \Sigma \to \Gamma$
                  (such that $\Sigma$ and $\Gamma$ are finite);
            \item the multiple-use bit propagation function from Example~\ref{ex:flip-flop};
            \item the $G$-prefix function from Example~\ref{ex:monoid-prefix}, for every finite group $G$
                  (note that a group is a special case of a monoid).
        \end{enumerate}
\end{theorem}

The following theorem (proved by Bojańczyk and me in \cite[Theorem~9]{single-use-paper}) shows
that single-use Mealy machines for infinite alphabets admit a similar decomposition:
\begin{theorem}
\label{thm:kr}
    The class of functions computed by single-use Mealy machines (over polynomial orbit-finite alphabets)
    is equal to the smallest class of functions that is closed under sequential and parallel compositions,
    and which contains the following \emph{single-use prime functions}: 
    \begin{enumerate}
        \item all functions recognized by Mealy machines over finite alphabets;
        \item $h^*$ for every equivariant $h : \Sigma \eqto \Gamma$ from example \ref{ex:ll-homomorphism}
        \item single-use propagation from Example~\ref{ex:su-propagation}
    \end{enumerate}
    (Note that thanks to Theorem~\ref{thm:classical-kr} the first item can be further decomposed
    into classical prime functions.)
\end{theorem}
I would like to use this thesis to present a new and (hopefully) improved proof of Theorem~\ref{thm:kr}.
The new proof is also a joint work with Bojańczyk.\\

Since the proof of Theorem~\ref{thm:kr} is very long, let us briefly discuss its structure. 
First, in Section~\ref{subsec:cmp-primes-incl-su-mealy},
we show that compositions of primes are included in the class of single-use Mealy machines. 
This part is relatively straightforward.
The other inclusion, i.e. showing that all single-use Mealy machines can be decomposed into prime functions, is much more involved:
In Section~\ref{sec:monoids-and-transductions} we introduce a new algebraic model called \emph{local semigroup transduction},
which serves as an intermediate step in the decomposition. In Section~\ref{sec:su-mealy-to-local-monoid-transduction}, we 
show that every single-use Mealy machine can be translated into a local semigroup transduction.
This leaves us with showing how to decompose local semigroup transductions. 
We do this in two steps: In Section~\ref{sec:factorisation-forest-theorem}, we prove a variant of the
\emph{factorisation forest theorem} for
orbit finite semigroups. Then, in Section~\ref{sec:local-monoid-transformation-composition-of-primes},
we use the theorem to show that local semigroup transductions can be decomposed into single-use primes.

\subsection{Compositions of primes $\subseteq$ Single-use Mealy machines}
\label{subsec:cmp-primes-incl-su-mealy}
Translating compositions of single-use primes into single-use Mealy machines is the easy part --
we have already seen that all prime functions can be computed by single-use Mealy machines,
so it suffices to show that single-use Mealy machines are closed under both types of composition:
\begin{lemma}
    \label{lem:su-mealy-parallel}
    Single-use Mealy machines are closed under parallel composition.
\end{lemma}
\begin{proof}
    The proof is straightforward -- the machine $\mathcal{A} \times \mathcal{B}$ can be constructed as
    a simple product construction: it keeps one copy of $\mathcal{A}$, one copy of $\mathcal{B}$
    and whenever it receives a new letter $(a, b) \in \Sigma_\mathcal{A} \times \Sigma_\mathcal{B}$
    it feeds $a$ to $\mathcal{A}$ and $b$ to $\mathcal{B}$, updates their states and outputs
    their outputs (as a pair). 
\end{proof}

\begin{lemma}
    \label{lem:su-mealy-compose}
    Single-use Mealy machines are closed under sequential compositions. 
\end{lemma}
\begin{proof}
    We take any two single-use Mealy machines $\mathcal{A}: \Sigma^* \to \Delta^*$ and $\mathcal{B} : \Delta^* \to \Gamma^*$,
    and we construct a single-use Mealy machine $\mathcal{B} \circ \mathcal{A} : \Sigma^* \to \Gamma^*$.
    Let us start by restating the construction for classical (i.e. atomless) Mealy machines:
    the machine $\mathcal{B} \circ \mathcal{A}$ keeps a copy of $\mathcal{A}$ and a copy of $\mathcal{B}$.
    When it reads a new letter $a \in \Sigma$, it:
    \begin{enumerate}
        \item feeds $a \in \Sigma$ to $\mathcal{A}$;
        \item updates $\mathcal{A}$'s state;
        \item feeds $\mathcal{A}'s$ output letter to $\mathcal{B}$;
        \item updates $\mathcal{B}$'s state;
        \item outputs $\mathcal{B}$'s output letter.
    \end{enumerate}
    The problem with using the same construction for single-use Mealy machines is that $\mathcal{A}$
    produces only one copy of the output, whereas $\mathcal{B}$ might require a multiple-use access to its input.
    This is because the first arrow in the type of $\delta_\mathcal{B}$ is equivariant but not necessarily
    single-use:
    \[\delta_\mathcal{B} : \Delta \eqto (Q_\mathcal{B} \suto (Q_\mathcal{B} \times \Gamma))\]
    To deal with this problem, we use a similar reasoning as in the proof of Lemma~\ref{lem:su-transition-functions-equivalent},
    and show that that there is a $k \in \nat$ such that $\delta_\mathcal{B}$ can be represented as:
    \[ \delta_{\mathcal{B}}' : \Delta^k \times Q_\mathcal{B} \suto (Q_\mathcal{B} \times \Gamma)\]
    It follows that $\mathcal{B}$ requires only $k$ copies of its input (for some fixed number $k$). This means that we can repeat the classical construction for
    $\mathcal{B} \circ \mathcal{A}$, but we have to maintain one copy of $\mathcal{B}$ and $k$ identical copies of $\mathcal{A}$.
\end{proof}

\section{An algebraic model for single-use Mealy machines}
\label{sec:monoids-and-transductions}

In the proof of the remaining inclusion of Theorem~\ref{thm:kr}, we would like to use the algebraic theory of single-use machines
developed in Chapter~2. For this purpose, we will define an algebraic transducer model that is equivalent 
to single-use Mealy machines. This model is going to be based on \emph{semigroups}
(a version of monoids without the requirement of having an identity element).
This shift from monoids is justified by Claim~\ref{claim:monoid-transduction-so-atomfree} (stated later in this section), 
which demonstrates that the identity elements would cause some technical problems in the theory of orbit-finite 
transducers. However, it is worth noting that semigroups and monoids are very similar algebraic structures.
Trivially, every monoid is a semigroup, and, as shown by the following claim, every semigroup can be embedded
into a monoid:
\begin{claim}
\label{claim:semigroup-to-monoid}
For every semigroup $S$, there exists a monoid $S^1$ such that $S$ is a subsemigroup of $S^1$.
\end{claim}
\begin{proof}
    If $S$ already happens to contain an identity element, we can set $S^1 = S$.
    Otherwise, we need to adjoin a formal identity element to $S$.
    This means that $S^1 = S + 1$, with the operation defined as follows (for all $x$ and $y$ from the original $S$):
    \[
        \begin{tabular}{cccc}
            $x \cdot y = x \cdot_S y$ & $1 \cdot x = x$ & $x \cdot 1 = x$ & $1 \cdot 1 = 1$
        \end{tabular}
    \]
    (Observe that we could also unconditionally adjoin a formal identity element to $S$. This operation 
    is usually denoted as $S^I$.)
\end{proof}

\subsection{Semigroup transductions over finite alphabets}
Before discussing the algebraic transducer model for orbit-finite alphabets, let us
briefly discuss its classical version for finite alphabets.
We start with a definition\footnote{Although I was not able to find this definition in the literature, 
it is consistent with the commonly understood folklore in the field.}:
\begin{definition}
\label{def:finite-semigroup-transductions}
    A \emph{semigroup transduction} of type $\Sigma^* \to \Gamma^*$ consists
    of\begin{enumerate}
        \item a finite semigroup $S$;
        \item an input function $h : \Sigma \to S$; and
        \item an output function $\lambda : S \to \Sigma$.
    \end{enumerate}
    The semigroup transduction defines the following function $\Sigma^* \to \Gamma^*$:
    \[ \Sigma^* \transform{h^*} S^* \longtransform{\textrm{$S$-prefix function}} S^* \transform{\lambda^*} \Gamma^*\comma\]
    where the $S$-prefix function is defined in Example~\ref{ex:monoid-prefix}, and 
    $h^*$ and $\lambda^*$ are length-preserving homomorphisms defined in Example~\ref{ex:ll-homomorphism}. 
    In other words, this means that a word $w_1 w_2 w_3 \ldots w_n \in \Sigma^*$, is transformed into the following 
    word from $\Gamma^*$:
    \[ \begin{tabular}{cccc}
        $\lambda\left(h(w_1)\right)$ & $\lambda\left(h(w_1) \cdot h(w_2)\right)$ & $\ldots$ & $\lambda\left(h(w_1) \cdot \ldots \cdot h(w_n)\right)$
    \end{tabular} \]
\end{definition}
\noindent
As expected, this class of transductions is equivalent to finite Mealy machines\footnote{Again, this 
result seems to be a part of the field's folklore.  Similar reasoning can be found in the 
literature -- e.g. 
in \cite[Section~4]{krohn1965prime}.}:
\begin{lemma}
\label{lem:mealy-monoids-classical}
    The class of functions computed by Mealy machines over finite alphabets
    is equivalent to the class of functions computed by semigroup transductions.
\end{lemma}
\begin{proof}
$\supseteq$: Thanks to Examples~\ref{ex:monoid-prefix}~and~\ref{ex:ll-homomorphism},
we know that the $S$-prefix function, $h^*$ and $\lambda^*$
can be computed by a classical Mealy machine. This
finishes the proof, because by Lemma~\ref{lem:su-mealy-compose},
Mealy machines are closed under compositions.\\

$\subseteq$: Let us take a classical Mealy machine $\mathcal{A}$, and let us construct an
equivalent semigroup transduction. First, let us define the behaviour of a word $w \in \Sigma^*$
to be the following function $b_{\mathcal{A}}(w): Q \to (Q \times \Gamma)$:
\[ \begin{tabular}{ccc}
    $b_{\mathcal{A}}(w)(q) = (q', c)$ & $\stackrel{\textrm{def}}{\Leftrightarrow}$ &$\substack{\textrm{If $\mathcal{A}$ enters $w$ from the left in the state $q$,}\\
                                             \textrm{it exits $w$ from right in state $q'$, outputting the letter $c \in \Gamma$}.}$
\end{tabular}\]Similarly as it was in the case of finite automata, the set of all possible behaviours forms a semigroup. 
Its operation is defined as follows:
\[(f \cdot g) = g \circ \underbrace{\proj_1}_{
    \substack{
        \textrm{projection}\\
       Q \times \Gamma \to Q 
    }}  
 \circ f\]It is not hard to see that this finite semigroup,
together with the following $h$ and $\lambda$, forms a semigroup transduction
that is equivalent to $\mathcal{A}$:
\[
    \begin{tabular}{cc}
        $h(a) = b_\mathcal{A}(a)$ & $\lambda(f) = \proj_2 (f(q_0))$
    \end{tabular}
\]
\end{proof}

\subsection{Semigroup transductions over orbit-finite alphabets}
For infinite alphabets, things get more complicated. 
Even though orbit-finite monoids are equivalent to single-use automata, 
the orbit-finite semigroup transductions are stronger than single-use Mealy machines. 
We have already seen that in Example~\ref{ex:monoid-products-not-single-use}, 
but the problem persists even if $\Gamma$ is finite:
\begin{example}
\label{ex:monoid-prefix-impossible}
    Consider the function $f_{\textrm{cmp}} : \atoms^* \to \{=, \neq\}^*$ which
    replaces every atom that is equal to the first letter with $=$ and all other atoms with $\neq$.
    Here is an example:
    \[
        \begin{tabular}{ccccccc}
            1 & 2 & 1 & 5 & 6 & 1 & 1 \\
             &   &   & \rotatebox[origin=c]{270}{$\mapsto$}  &   &   &\\
           $=$ & $\neq$ & $=$ & $\neq$ & $\neq$ & $=$ & $=$
       \end{tabular}    
    \]
   This function is a semigroup transduction. We can implement it using $S = \atoms^2$
   with the semigroup operation defined as follows
    \[
        (a,b) \cdot (c, d) = (a, d)\comma
    \]
   and with the following $\lambda$ and $h$:
   \[
   \begin{tabular}{cc}
    $h(a) = (a, a)$ & $\lambda(a, b) = \begin{cases} 
                                            = & \textrm{if } a = b\\
                                            \neq & \textrm{if } a \neq b\\
    \end{cases}$ 
   \end{tabular}
   \]
   Let us now show that this function cannot be computed by a single-use Mealy machine.
   If we consider languages as functions $\Sigma^* \to \{\Yes, \No\}$, we can write that:
   \[ \left(\substack{\textrm{The first letter}\\
                      \textrm{appears again}}\right) =   \left(\substack{\textrm{The letter $=$ appears}\\
                                                                           \textrm{at least twice}} \right) \circ f_{\textrm{first}} \]
  By reasoning analogous to the one presented in the proof of Lemma~\ref{lem:su-mealy-compose}, we obtain that:
   \[
    \left(\substack{\textrm{Single-use register}\\
    \textrm{automata}}\right) \circ
    \left(\substack{\textrm{Single-use}\\
                     \textrm{Mealy machines}}\right)
    =
    \left(\substack{\textrm{Single-use register}\\
                     \textrm{automata}}\right)
    \]
    The language "the letter `$=$' appears at least twice" is a regular (over a finite alphabet).
    Therefore, if $f_\textrm{cmp}$ were recognized by a single-use Mealy machine, then the language
    ``the first letter appears again'' would also be recognized by a single-use Mealy machine.
    However, as shown by Example~\ref{ex:first-again-not-monoid} and Lemma~\ref{lem:1sua-incl-ofm},
    this is not true. Therefore, $f_\textrm{cmp}$ is not recognized by single-use Mealy machines.
\end{example}
It turns out that whether an orbit-finite semigroup transduction is equivalent to a single-use Mealy machine
depends solely
on the output function $\lambda$.
As we will show, it depends on whether $\lambda$ satisfies the following \emph{locality equation} (which, to the best of my knowledge 
is an original contribution of this thesis, based on a joint work with Bojańczyk):
\begin{definition}
\label{def:local-monoid-transduction}
Let $S$ be an orbit-finite semigroup and let $\Gamma$ be an orbit-finite set. 
We say that a function $\lambda : S \eqto \Gamma$ is \emph{local}, if for every $a, e, b \in S$,
such that $e$ is an \emph{idempotent} (i.e. $e e = e$) and 
$b$ is a prefix of $e$ (i.e. $e = bb'$ for some $b'$ in $S$) and for every $\supp(e)$-permutation $\pi$ (i.e.
a permutation $\pi$ such that $\pi(a) = a$, for every $a$ from the least support of $e$), it holds that:
\[ \lambda(aeb) = \lambda(\pi(a)eb)\tdot\]
We further say that a semigroup transduction $(S, h, \lambda)$ is \emph{local}, if its output function $\lambda$ is local.
In this thesis, we focus on local semigroup transductions that are also equivariant and orbit-finite. So, 
to simplify the notation, we additionally require that in a local semigroup transduction the components
$S$, $h$ and $\lambda$ are equivariant and $S$ is orbit-finite. 
\end{definition}
To avoid interrupting the proof of Theorem~\ref{thm:kr} (which is already very long), 
we defer some of the discussion about single-use Mealy machines to Section~\ref{sec:local-semigroup-transductions-revisited}.
For now we only present some intuition and an example in the next couple of paragraphs. 
However, in order to provide context, it is worth mentioning that, later in this thesis we will show that 
local semigroup transductions recognize the same class of transductions as single-use Mealy machines
(see Lemma~\ref{lem:single-use-mealy-monoid-transduction}).
Furthermore, we will show that as long as $S$ does not contain unreachable elements, 
a semigroup transduction $(S, h, \lambda)$ is equivalent to a single-use Mealy machine 
if and only if $\lambda$ is local
(see Lemma~\ref{lem:full-single-use-Mealy-local-iff} in Section~\ref{sec:local-semigroup-transductions-revisited}).\\

Let us now offer some informal intuition behind the locality restriction:
Consider a nonlocal semigroup 
transduction ($S, h, \lambda$). This implies that there exist $a, e, b \in S$ and a $\supp(e)$-permutation $\pi$ 
such that $e$ is an idempotent and $b$ is a prefix of $e$, for which the locality equation does not hold:
\[ \lambda(aeb) \neq \lambda(\pi(a)eb)\]
Since $b$ is a prefix of $e$, there exists $c \in S$ such that $bc=e$. Consider the following sequence:
\[a \underbrace{bc}_e  \underbrace{bc}_e \ldots  \underbrace{bc}_e\]
Observe that such a sequence contains arbitrarily many prefixes that 
evaluate to $aeb$. 
We know that $\lambda(aeb) \neq \lambda(\pi(a)eb)$, which means
that the value of $\lambda(aeb)$ depends on an atom from $(\supp(a) - \supp(e))$.
This means that every Mealy machine computing this semigroup transduction would 
have to use at least one copy of an atom from $a$ while processing each $bc$ part of the input. 
Since $a$ appears only once, and the input sequence can be arbitrarily long,
this would violate the single-use restriction.\\

\begin{example}
\label{ex:semigroup-atom-prop}
Let us construct a local semigroup transduction $(S, \lambda, h)$ that computes the single-use atom propagation from Example~\ref{ex:su-propagation}.
The semigroup $S$ is defined as follows:
\[ \underbrace{\epsilon}_{\textrm{do nothing}} + \underbrace{\atoms}_{\substack{\textrm{save an atom } a\\ \textrm{into the register}}} + \underbrace{\bot}_{\textrm{empty the register}} + \underbrace{\downarrow}_{\substack{\textrm{output and empty}\\\textrm{the register}}} + \underbrace{\atoms\downarrow}_{\substack{\textrm{output $a \in \atoms$ and}\\\textrm{empty the register}\\ \textrm{ (denoted as $a\downarrow$)}}} \]
The operation in $S$ is defined by the following table, for every $a,b \in \atoms$
(the general rule is that $x \cdot \epsilon = x = \epsilon \cdot x$ for every $x \in S$, and that 
$x \cdot y = y$ for every $x, y \in S$, such that $y \neq \epsilon$. All exceptions to this rule are marked in {\color{blue} blue}):
\[
    \begin{tabular}{ccccc}
        $\epsilon \cdot \epsilon = \epsilon$ & $\epsilon \cdot a = a$  & $\epsilon \cdot \bot = \bot$ & $\epsilon \cdot \downarrow = \downarrow$ & $\epsilon \cdot a \downarrow = a \downarrow$\\  
        $a \cdot \epsilon = a$ & $a \cdot b = b$ & $a \cdot \bot = \bot$ & $\color{blue}a \cdot \downarrow = a \downarrow$ & $a \cdot b \downarrow = b \downarrow$\\  
        $\bot \cdot \epsilon = \bot$ & $\bot \cdot a = a$ & $\bot \cdot \bot = \bot$ & $\color{blue}\bot \cdot \downarrow  = \bot$ & $\bot \cdot a \downarrow = a \downarrow$\\ 
        ${\color{blue}\downarrow \cdot \epsilon = \bot}$ & $\downarrow \cdot a = a$ & $\downarrow \cdot \bot = \bot$ & $\color{blue}\downarrow \cdot \downarrow  = \bot$ & $\downarrow \cdot a \downarrow = a \downarrow$\\
        $\color{blue}a \downarrow \cdot \epsilon = \bot$ & $ a \downarrow \cdot b = b$ & $a \downarrow \cdot \bot = \bot$ & $\color{blue}a \downarrow \cdot \downarrow  = \bot$ & $a \downarrow \cdot b \downarrow = b \downarrow$\\    
    \end{tabular} 
\]
Note that $\Sigma = \atoms + \downarrow + \epsilon$, which means that $\Sigma \subseteq S$, so
we can define $h$ to be the natural injection. Finally, we define $\lambda$ as follows:
\[
    \lambda(x) = \begin{cases}
        a & \textrm{if } x = a \downarrow \\ 
        \epsilon & \textrm{otherwise}
    \end{cases} 
\]
This semigroup transduction $(S, h, \lambda)$ defines the single-use atom propagation function 
from Example~\ref{ex:su-propagation}. Let us now show that it satisfies the locality equation. 
We take $x, e, y \in S$, and a $\supp(e)$-permutation $\pi$, such that $e$ is an idempotent 
and $y$ are a prefix of $e$, and we show that $\lambda(xey) = \lambda(\pi(x)ey)$. First, let us notice 
that unless $y = \downarrow$ or $y = a\downarrow$, we know that $\lambda(xey) = \epsilon = \lambda(\pi(x)ey)$. 
Moreover, if $y = a\downarrow$ we know that:
\[xey = a\downarrow = \pi(x)ey\tdot\]
It follows that $\lambda(xey) = a = \lambda(\pi(x)ey)\tdot$
This leaves us with the case where $y = \downarrow$. We need to show that:
\[\lambda(xe\downarrow) = \lambda(\pi(x)e\downarrow)\]
Observe that $e$ cannot be equal to $\downarrow$, because 
$\downarrow$ is not idempotent. Moreover $e$ also cannot be equal to  $\epsilon$, because $\downarrow$ is not a prefix of $\epsilon$.
This means that $e$ is either equal to $a$, $a \downarrow$ or $\bot$. It follows
that $z \cdot e = e$, for every $z \in S$. This means that $x e = e = \pi(x)e$, which in 
turn means that $\lambda(xey) = \lambda(\pi(x)ey)$.  
\end{example}

Finally, let us mention that the locality restriction is very limiting if the underlying semigroup 
happens to be a monoid -- this is the reason why we use semigroup-based models for algebraic transductions rather 
than monoid-based models:
\begin{claim}
\label{claim:monoid-transduction-so-atomfree}
    Let $(S, h, \lambda)$ be a local semigroup transduction. If $S$ contains an identity element, 
    then $\lambda$ can only output equivariant (i.e. atomless) values. 
\end{claim}
\begin{proof}
    Every semigroup $S$ contains at most one identity element
    (see \cite[Section II.1.1]{pin2010mathematical} for details). 
    This means that, if it exists, the identity element of $S$ can be computed from $S$ in 
    an equivariant way. Since the semigroup $S$ is equivariant as a whole, it follows from Lemma~\ref{lem:fs-functions-preserve-supports}
    that the identity element $1 \in S$ has to be equivariant as well. This means that every 
    atom permutation $\pi$ is a $\supp(1)$-permutation. Observe that $1$ is idempotent and
    that it is its own prefix. Since $\lambda$ is equivariant, it follows that for every $x \in S$:
    \[\pi(\lambda(x)) = \lambda(\pi(x)) = \lambda(\pi(x) \cdot 1 \cdot 1) \stackrel{\textrm{locality}}{=} \lambda(x \cdot 1 \cdot 1) = \lambda(x)\]
    This means that for every $x \in S$, the output value $\lambda(x)$ is equivariant (i.e. atomless).
\end{proof}

As mentioned before, we are going to use local semigroup transductions in the proof of Theorem~\ref{thm:kr}.
Here is the plan:\\

\vspace{0.3cm}
\noindent
    \adjustbox{width=1\textwidth, center}{
       \begin{tikzcd}
        & {\substack{\textrm{Single-use}\\ \textrm{Mealy machines}}} \\
        \\
        {\substack{\textrm{Compositions of}\\ \textrm{Krohn-Rhodes primes}}} && {\substack{\textrm{Local semigroup transductions}\\ \textrm{over polynomial orbit-finite alphabets}}}
        \arrow["{\substack{\color{gray}\textrm{Easy, already shown in}\\ \color{gray}\textrm{Section~\ref{sec:krohn-rhodes-decomp}}}}"{description}, from=3-1, to=1-2]
        \arrow["{\color{gray}\textrm{Section~\ref{sec:su-mealy-to-local-monoid-transduction}}}"{description}, from=1-2, to=3-3]
        \arrow["{\color{gray}\textrm{Lemma~\ref{lem:local-monoid-transformation-composition-of-primes}}}"{description}, from=3-3, to=3-1]
    \end{tikzcd}}
\vspace{0.3cm}
\label{fig:kr-proof-plan}

The hardest part of the proof is Lemma~\ref{lem:local-monoid-transformation-composition-of-primes}
(translating local semigroup transductions into compositions of primes).
For this reason, we devote a significant portion of this chapter to explain it.\\ 

Notice, that as a byproduct of the proof strategy, we will obtain the following theorem:
\begin{theorem}
\label{lem:single-use-mealy-monoid-transduction}
Local semigroup transductions over polynomial orbit-finite alphabets\footnote{
    \label{ftn:kr-semigroups-general-of}
    Notice that semigroup transductions work with all orbit-finite $\Gamma$ and $\Sigma$
    (even if they are not polynomial). This general case could be an interesting topic 
    for future work. See also Footnote~\ref{ftn:kr-semigroups-general-rational} on page~\pageref{ftn:kr-semigroups-general-rational}.}
compute the same class of functions as single-use Mealy machines
\end{theorem}

\noindent
Before we proceed with the proof of Theorem~\ref{thm:kr}, let us show how to use it to prove the missing implication from Theorem~\ref{thm:dsua-2dsua-ofm}:
\begin{lemma}
\label{lem:ofm-to-sua}
Every language that can be recognized by an orbit-finite monoid can also be recognized by a one-way single-use automaton.
\end{lemma}
\begin{proof}
	Let $L \subseteq \Sigma^*$ be recognized by an orbit-finite
	monoid $M$. We define a transduction
	$f_L : (\Sigma + \dashv)^* \to \{\epsilon, \yes, \no\}^*$ such that the $i$-th letter of $f_L(w)$ is equal to:
	\begin{itemize}
	\item $\yes$ if $w_i$ is the first $\dashv$ in $w$ and $w_1 w_2 \ldots w_{i-1} \in L$;
	\item $\no$ if $w_i$ is the first $\dashv$ in $w$ and $w_1 w_2 \ldots w_{i-1} \not \in L$;
	\item $\epsilon$ if $w_i$ is  not the first $\dashv$ in $w$.
	\end{itemize}	
	Since $L$ is recognized by $M$, we can use the following semigroup $M^{\dashv}$ to define $f_L$ as a local semigroup monoid transduction.
	\[M^{\dashv} = \underbrace{M}_{\textrm{words without $\dashv$}} + \underbrace{M\dashv}_{\substack{\textrm{words that end with $\dashv$}\\ \textrm{and otherwise do not contain $\dashv$} }} + \underbrace{\bot}_{\substack{\textrm{all other words}\\
	}}  \]
    The operation on $M^{\dashv}$ is defined as follows for every $a, b \in M$ (note that there are only 
    two cases where the result is not $\bot$):
    \[
        \begin{tabular}{lll}
            $a \cdot b = a \cdot_M b$ & $(a \dashv) \cdot b = \bot$ & $\bot \cdot a = \bot$\\
            $a \cdot (b \dashv) = (a \cdot_M b) \dashv $& $(a \dashv) \cdot (b \dashv) = \bot$ & $\bot \cdot (a \dashv) = \bot$\\
            $a \cdot \bot = \bot$ & $ (a \dashv) \cdot \bot = \bot $ & $\bot \cdot \bot = \bot$
        \end{tabular} 
    \]
    The $h$ function is the same as the one used to recognize $L$, and $\lambda$ is defined as follows:
\[ \begin{tabular}{ccc}
	$\lambda(a) = \epsilon$ & 
	$\lambda(a\dashv) = \begin{cases}
 	\yes & \textrm{if } s \textrm{ is accepting}\\
 	\no & \textrm{otheriwise}
 \end{cases}
$ & $\lambda(\bot) = \epsilon$
 \end{tabular}
\]
It is easy to see that $(M^{\dashv}, h, \lambda)$ recognize $f_L$. Let us show that
$\lambda$ satisfies the locality equation: Take $x, e, y$ and a $\supp(e)$-permutation $\pi$
such that $e$ is an idempotent and $b$ is a prefix of $e$. The only interesting 
case is where $y = a\dashv$, or otherwise
$\lambda(xey) = \square = \lambda(\pi(x)ey)$. The
only idempotent that contains $a\dashv$ as a prefix is $\bot$, which means that
$e= \bot$. It follows that $aeb = \bot = \pi(a)eb$, which 
in particular means that $\lambda(aeb) = \lambda(\pi(a)eb)$.\\

It follows by Lemma~\ref{lem:single-use-mealy-monoid-transduction}, that there exists a single-use 
Mealy machine $\mathcal{A}$ that computes $f_L$. If we ignore the output of $\mathcal{A}$ 
and equip it with the following acceptance function, we obtain a single-use automaton for the language $L$.
\[ f_{\textrm{acc}}(q) = \begin{cases}
 	\yes & \textrm{if the output letter of } \delta(q, \dashv) \textrm{ is } \yes\\
 	\no & \textrm{otherwise}
 \end{cases}
\]
\end{proof}

\subsection{Single-use Mealy machines $\subseteq$ Local semigroup transductions}
\label{sec:su-mealy-to-local-monoid-transduction}
In this section, we show how to translate a single-use Mealy machine
into a local semigroup transduction. The construction is a single-use 
version of the classical construction presented in the $\subseteq$-inclusion of Lemma~\ref{lem:mealy-monoids-classical}.
Given a single-use Mealy machine $\mathcal{A}$ of type $\Sigma^* \to \Gamma^*$,
we can translate every non-empty word $w \in \Sigma^*$ into a behaviour $b_{\mathcal{A}}(w)$, which is of the following type:
\[ \underbrace{Q}_{
    \substack{\textrm{The state in}\\ \textrm{which $\mathcal{A}$ enters} \\ \textrm{$w$ from the left.}}
} \suto \underbrace{Q}_{
    \substack{\textrm{The state in}\\ \textrm{which $\mathcal{A}$ exits}\\ \textrm{$w$ from the right.}}
} \times \underbrace{\Gamma}_{
    \substack{\textrm{The letter outputted}\\ \textrm{by $\mathcal{A}$  as it exits}\\ \textrm{$w$ from the right.}}
}\]
To see that every behaviour is a single-use function, we notice that
the behaviours of one-letter words are single-use functions (because the transition functions of Mealy 
machines are single-use), and that the behaviours can be composed according to the following formula:
\[ b_\mathcal{A}(w_1 w_2) = b_\mathcal{A}(w_2) \circ \proj_1 \circ b_\mathcal{A}(w_1) \]
Similarly as in the proof of Lemma~\ref{lem:mealy-monoids-classical}, it is not hard to see
that the following semigroup transduction defined the same function as $\mathcal{A}$: 
\[
    \begin{tabular}{ccc}
        $ \underbrace{S = Q \suto Q \times \Gamma}_{\substack{\textrm{Equipped with function composition,}\\ \textrm{ignoring the  } \Gamma\textrm{-component.}}}$ &
        $ h(w) = b_\mathcal{A}(w)$ & 
        $ \lambda(f) = \proj_2(f(q_0))$, 
    \end{tabular}
\]
where $q_0$ is the initial state of $\mathcal{A}$.
This leaves us with showing that $\lambda$ is local.

\begin{lemma}
    The function $\lambda$ defined above is local.
\end{lemma}
\begin{proof}

We take $x, e, y \in S$ and a $\supp(e)$-permutation
such that $e$ is an idempotent and $y$ is a prefix of $e$,
and we show that they satisfy the locality equation, i.e.:
\[ \lambda(xey) = \lambda(\pi(x)ey) \]
By definition of $\lambda$ this is equivalent to showing that:
\[ \proj_2(q_0\, xey) = \proj_2(q_0\, \pi(x)ey)\comma \]
where $q_0\, xey$ is a notation for $(x e y)(q_0)$ -- remember that $q_0$ is a function.
Denote 
$q := \proj_1(q_0 x e)$, and notice that since $\pi$ is a $\supp(e)$-permutation 
and since $q_0$ is equivariant, we know that $\pi(q) = \proj_1(q_0 \pi(x) e)$. 
This means that we need to show that:
\[ \proj_2(q\, y) = \proj_2(\pi(q)\, y) \]
Let us pick some single-use decision tree $T$, that represents the function $y$, 
and let us consider the trees $\proj_1 \circ T$ and $\proj_2 \circ T$
(obtained using the construction from Claim~\ref{claim:su-trees-compose-tup}).
We are going to show that the leaf of $\proj_1 \circ T$, that is 
reached when computing $(\proj_1 \circ T) (q)$, contains all the input variables $x_i$, such 
that $\pi(x_i) \neq x_i$. This is enough to prove that $\proj_2(q \, y) = \proj_2(\pi(q) \, y )$, 
because thanks to the single-use restriction, we know that if all variables
that are modified by $\pi$ appear in the leaf of $\proj_1 \circ T(q)$,
then they cannot appear in the queries or in the leaf while computing $(\proj_2 \circ T)(q)$.
(This is because the queries for computing $(\proj_2 \circ T)(q)$ and $(\proj_1\circ T)(q)$ 
are equal, and the output variables that appear in the leaf of $T(q)$ are
partitioned between the leaves for $(\proj_1 \circ T)(q)$ and $(\proj_2 \circ T)(q)$.)\\

This leaves us with showing that the leaf of $(\proj_1 \circ T) (q)$
contains all input variables for which $\pi(x_i) \neq x_i$. For this, 
we notice that since $y$ is a prefix of $e$, there exists a $y'$ 
such that $yy' = e$. Notice that $q$ is a fixpoint of $y' \circ \proj_1 \circ y$:
\[ y'( \proj_1 (y(q))) = (y y')(q) = e(q) = e(q_0 \, xe) \stackrel{\textrm{notation}}{=}q_0 \, xee = q_0 \, xe = q\]
Let us pick some tree $T'$ that corresponds to $y'$ and let us consider the tree 
\[T' \circ \proj_1 \circ T\]
By construction (from Claim~\ref{claim:su-trees-compose-tup}), 
we know that the leaf corresponding to $(T' \circ \proj_1 \circ T)(q)$, can only contain 
those input variables that were present in the leaf of $(\proj_1 \circ T)(q)$:
\picc{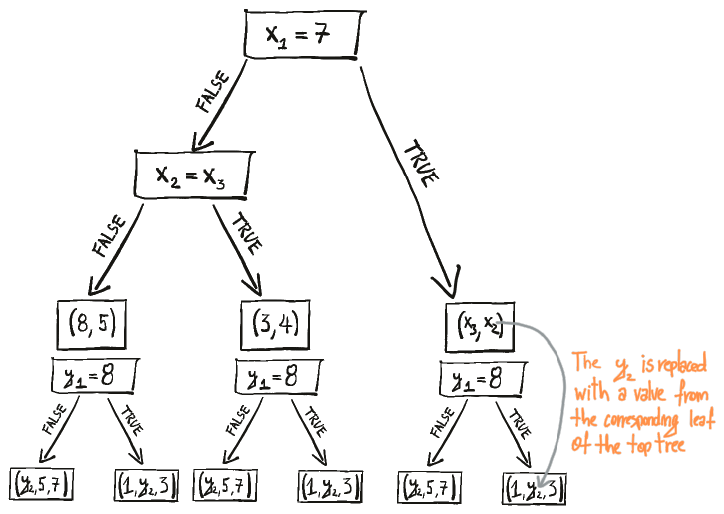}

It follows that all input atoms of $q$ that are not present in $(\proj_1 \circ T)(q)$
as variables have to appear in the leaf of $(T' \circ \proj_1 \circ T)(q)$ as atomic constants.
The tree $T' \circ \proj_1 \circ T$ represents the function $y' \circ \proj_1 \circ y$, 
which (by the definition of the product in $S$) is equal to $y' \cdot y  = e$.
By Lemma~\ref{lem:su-tree-leaves} (stated below), it follows all atoms that appear in the leaves of $(T' \circ \proj_1 \circ T)$
have to appear in $\supp(e)$. It follows that all input variables that do not belong to $\supp(e)$ have to appear 
in the leaf of $(\proj_1 \circ T)(q)$. This finishes the proof, because $\pi$ is a $\supp(e)$-permutation.
\end{proof}

We finish this section by proving Lemma~\ref{lem:su-tree-leaves}: 
\begin{lemma}
\label{lem:su-tree-leaves}
    Let $T$ be a single-use decision tree, and let $f_T$ be the single-use function represented by $T$. 
    If a leaf of $T$ contains an atomic constant $a \in \atoms$, then $a \in \supp(f_T)$.
\end{lemma}
\begin{proof}
    We consider the case where $T$ is of type $\atoms^k \suto Y$ (the proof can be 
    easily extended to the general case). We pick an $a$ that appears in a leaf of $T$ and 
    we show how to use $f_T$ and atomic constants other than $a$ to construct a value 
    $y$ such that $a \in \supp(y)$. As long as this construction is equivariant, 
    it follows by Lemma~\ref{lem:fs-functions-preserve-supports} that $a \in \supp(f_T)$.
    The proof goes by induction on the depth of $T$:\\

    If $T$ is a leaf, we take a tuple of $k$ atoms $(b_1, \ldots, b_k)$
    other than $a$, and define $y$ as $f_T(b_1, \ldots, b_k)$. Since $T$ is a leaf, 
    and $a$ appears in $T$, we know that $a \in \supp(y)$.\\
    
    For the induction step, 
    we assume that the query in the root is of the type $x_i = b$;
    as the case of $x_i = x_j$ is analogous but simpler.
    We consider two subclasses: $b \neq a$ and $b = a$.
    First, let us deal with $b \neq a$. We assume
    that the leaf with $a$ belongs to the $\yes$-subtree
    (the case for the $\no$-subtree is analogous), and 
    we construct  $y \in \atoms^k \suto Y$ as the following function:
    \[ (x_1, \ldots, x_k) \mapsto f_T(x_1, \ldots, x_{i-1}, b, x_{i+1}, \ldots, x_{k}) \]
    Such $y$ is equal to the function defined by the $\yes$-subtree, and 
    since the $\yes$-subtree contains $a$ in its leaf, it follows by the induction assumption that 
    $a \in \supp(y)$.\\

    The case where $b = a$ is harder: We can assume that the $\yes$-subtree 
    and the $\no$-subtree define two different functions (or otherwise we can 
    directly apply the induction assumption to the subtree with $a$) and we define $y \in \atoms$ 
    as \emph{the only atom $c \in \atoms$ such that for every $d \in \atoms$, 
    such that $d \neq c$, the following functions are different}:
    \[ \begin{tabular}{c}
        $(x_1, \ldots, x_k) \mapsto f_T(x_1, \ldots, x_{i-1}, c, x_{i+1}, \ldots, x_{k})$\\
        $(x_1, \ldots, x_k) \mapsto f_T(x_1, \ldots, x_{i-1}, d, x_{i+1}, \ldots, x_{k})$ 
    \end{tabular}\]
    It is not hard to see that the only such $c$ is equal to $a$, which means that $ a \in \supp(y)$. 
\end{proof}

\section{Factorization forest theorem}
\label{sec:factorisation-forest-theorem}
As mentioned before, the translation from local semigroup transductions to compositions of primes
is the hardest part of the proof of Theorem~\ref{thm:kr}. We split it into two sections:
In this section, we define \emph{factorization trees} and show how to construct them using
compositions of primes. In the next section, we use factorization trees to construct the output of local semigroup transductions.\\

A \emph{factorization tree} for a sequence $s_1, s_2, \ldots, s_n$ over a semigroup $S$ is a
tree labelled by elements of $S$. It has $n$ leaves that correspond to the input positions -- 
the $i$-th leaf is labelled with $s_i$. The inner nodes of a factorization tree correspond
to infixes of the input sequence and are labelled by the product of that infix.
Here are three examples of factorization trees for the infinite semigroup $(\nat, +)$, over the following sequence:
\[1\ 2\ 1\ 3\ 2\ 2\ 1\ 3\]
\bigpicc{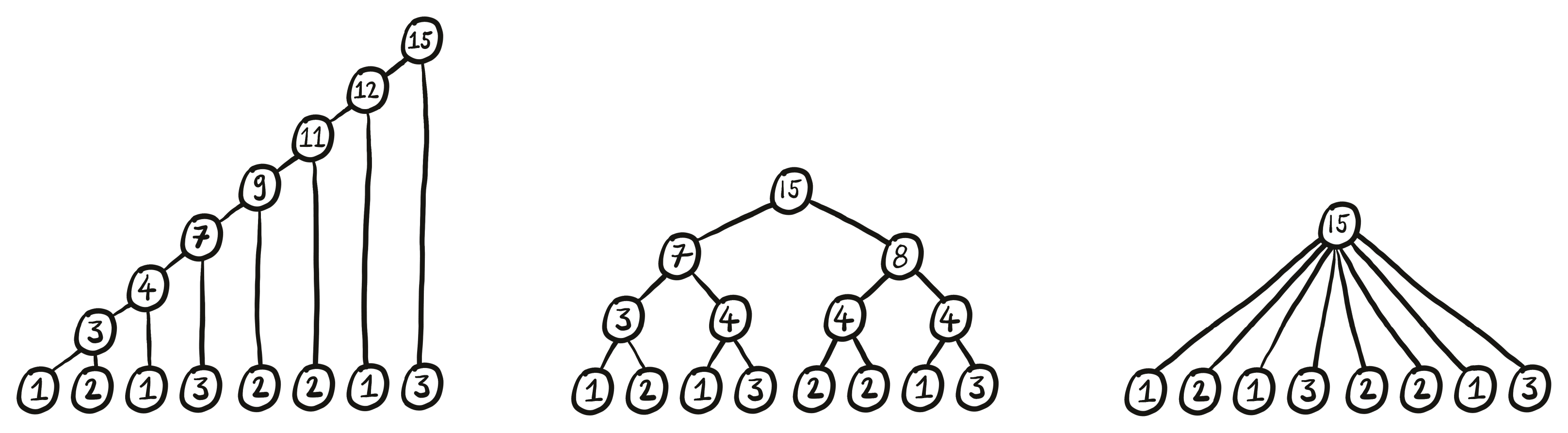}

An important parameter of a factorization tree is its height. Notice that if nodes are allowed to have an unbounded number of children,
then every sequence admits a decomposition tree of height $1$. On the other hand,
if every node can only have at most two children, then the height of a decomposition tree over
$s_1, \ldots, s_n$ cannot be lower than $\log_2(n)$. A compromise between those two
approaches is to require that nodes with more than two children respect the structure
of $S$. An example of this approach is the following \emph{idempotency condition}: We say that
a node is \emph{idempotent} if all of its children are labelled by the same idempotent
from $S$. A factorization tree satisfies the idempotency condition if
all of its nodes are either binary or idempotent. Here is an example of an idempotent
factorization tree for $S= \left(\,P(\{a, b, c\}),\ \cup\,\right)$:
\smallpicc{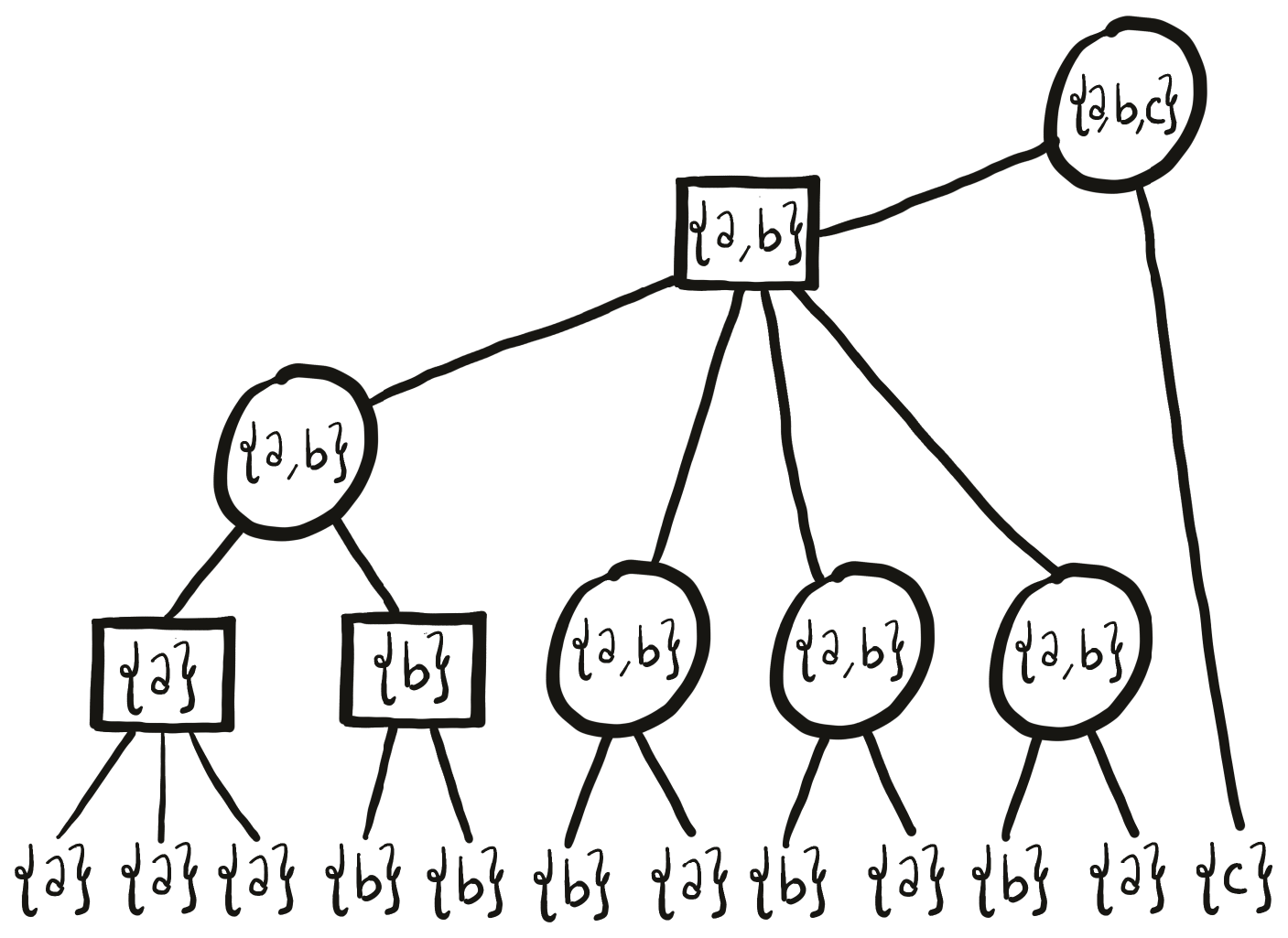}

Idempotent factorization trees are called \emph{Simon's factorization trees} after
Imre Simon, who has shown that finite semigroups admit idempotent factorization trees of
bounded height\footnote{The original theorem shows this for monoids, but the same proof
can be applied to semigroups. It is also worth noting that in \cite{simon1990factorization} Simon's factorization trees 
are referred to as Ramseyan factorization trees.}:

\begin{theorem}[{\cite[Theorem~3.3]{simon1990factorization}}]
\label{thm:simons-factorisation}
    For every finite semigroup $S$ there exists an $h_S \in \nat$, such that every
    sequence $s_1, s_2, \ldots, s_n \in S^*$ admits an idempotent factorization tree
    of height at most $h_S$.
\end{theorem}

\noindent
This theorem does not directly extend to orbit-finite semigroups:
\begin{example}
\label{ex:simons-fails-ofm}
Consider $S = \atoms^2$, with the operation defined as follows:
\[ (x_1, x_2) \cdot (y_1, y_2) = (x_1, y_2)\]
Let us show that $S$ does not admit idempotent factorization trees of 
bounded height. Notice first, that the only idempotents in $S$ are
elements of the form $(a, a)$. Consider the following family of sequences $l_n \in S^*$:
\[  
    l_n = \begin{tabular}{ccccc}
    $(1, 2)$ & $(2, 3)$ & $(3, 4)$ & \ldots & $(n-1, n)$,
    \end{tabular}
\]
and notice that no infix of $l_n$ evaluates to an idempotent. It follows that
all nodes in an idempotent factorization tree over $l_n$ have to be binary.
This means that the height of an idempotent factorization tree over $l_n$ is at least $\log_2(n)$,
which is not bounded by any $h_M$. 
\end{example}

It follows from Example~\ref{ex:simons-fails-ofm} that if we want to have an equivalent
of Theorem~\ref{thm:simons-factorisation} that works for orbit-finite semigroups, 
we need to relax the idempotency condition. We are going to 
define this relaxation in terms of \emph{smooth sequences}, which are in turn defined in 
terms of the \emph{Green's infix relation:}\footnote{
    Green's relations were introduced and studied by James Alexander Green in \cite{green1951structure}.}
\begin{definition}
    Let $S$ be a semigroup. We say that $x \in S$ is \emph{an infix of} $y \in S$ if
    for some $a, b \in S^1$ (from Claim~\ref{claim:semigroup-to-monoid}),
    it holds that $axb = y$. The infix relation is
    a preorder -- i.e. it is reflexive and transitive. If $x$ and $y$ are each other's
    infixes, we say that they are \emph{infix equivalent} (or \emph{$\mathcal{J}$-equivalent}).
    This is an equivalence relation and its equivalence classes are called \emph{infix classes}
    or \emph{$\mathcal{J}$-classes}. 
\end{definition}

\begin{example}
\label{ex:inifx-classes}
   Consider the semigroup $S = \atoms^2 + \bot + 1$, where $1$ is the identity element 
   and otherwise, the operation is defined as follows (this semigroup is equal to the monoid from Example~\ref{ex:no-rep-monoid}):
   \[
       \begin{tabular}{ccc}
        $ x \cdot \bot = \bot$ & $\bot \cdot x = \bot$ & $(x_1, x_2) \cdot (y_1, y_2) = \begin{cases}
            (x_1, y_2) &  \textrm{if } x_2 \neq y_1\\
            \bot & \textrm{otherwise}
        \end{cases}$
           
       \end{tabular} 
   \]
   This semigroup has three infix classes: $\{1\}$, $\{\bot\}$, and $\{(x, y) \ | \ x, y \in \atoms\}$. 
\end{example}

\noindent
We now are ready to define smooth sequences:
\begin{definition}
    Let $S$ be a semigroup. We say that a sequence $s_1, s_2, \ldots, s_n \in S^*$ is \emph{smooth}
    if each $s_i$ is $\mathcal{J}$-equivalent to the product of the sequence (i.e. $s_1 \cdot s_2 \cdot \ldots \cdot s_n$).
    (In particular, this means that all $s_i$'s are pairwise $\mathcal{J}$-equivalent,
    but this is not a sufficient condition for a sequence to be smooth.)
\end{definition}
\begin{example}
\label{ex:smooth-seq}
    Consider the semigroup from Example~\ref{ex:inifx-classes}.
    Here is an example of a smooth sequence:
    \[ (1,2)\ (3, 7)\ (4, 9)\ (7, 19)\]
    Here are three examples of non-smooth sequences:
    \[ \begin{tabular}{ccc}
        $(1,2)\ (3, 7)\ (7, 9)\ (1, 3)$; &
        $(1, 3)\ \bot\ (2, 9)\ (7, 3)$; &
        $(7, 3)\ 1\ (4, 8)$
    \end{tabular}
    \]
\end{example}
\noindent
We are now ready to define \emph{smooth factorization trees}:
we say that a node of a factorization tree is \emph{smooth} if 
the labels of its children form a smooth sequence. 
We say that a factorization tree is \emph{smooth} if all of its nodes are
either binary or smooth. Notice that the smoothness condition is 
a relaxation of the idempotency condition -- every idempotent factorization 
tree is also a smooth factorization tree. Thanks to this relaxation, 
we can extend Theorem~\ref{thm:simons-factorisation} to work with all orbit-finite semigroups
\footnote{Although a proof of this theorem can be found in \cite[Lemma16]{single-use-paper}, we include it in this thesis due to its central role in the proof of Theorem~\ref{thm:kr}.}:
\begin{theorem}
\label{thm:smooth-factorisation-trees}
    For every orbit-finite semigroup $S$, there exists an $h_S \in \nat$, such that every
    sequence $s_1, s_2, \ldots, s_n \in S^*$ admits a smooth factorization tree
    of height at most $h_S$.
\end{theorem}
In fact, to complete the proof of Theorem~\ref{thm:kr},
we require a slightly stronger version of Theorem~\ref{thm:smooth-factorisation-trees}:
in addition to proving the existence of bounded smooth factorization trees, 
we need to show that they can be constructed using compositions of primes.
This means that we need a way of representing factorization trees as words.
For this we use a version of \emph{splits} defined by Colcombet in \cite[Section~2.3]{colcombet2007combinatorial}:
\begin{definition}
    Let $S$ be a semigroup. A split of height $h$ over a sequence $s_1, s_2, \ldots s_n \in S^*$
    is a function $t: \{1, \ldots, n\} \to \{1, \ldots, h\}$, which assigns a height to every position of the input sequence.
    A split defines the following forest structure on the positions of the sequence -- 
    a position $i$ is a \emph{descendant} of $j$ if:
    \begin{enumerate}
        \item the position $i$ is to the left of (or equal to) $j$, i.e. $i \leq j$; and
        \item the position $i$ is visible from the position $j$, which means that 
              the heights of all positions between $i$ and $j$ (including $i$, but excluding $j$)
              are strictly lower than the height of~$j$.
    \end{enumerate}
    Note that the descendants of every position form a contiguous infix of the input sequence. 
    Here is an example split, together with an example set of descendants:
    \picc{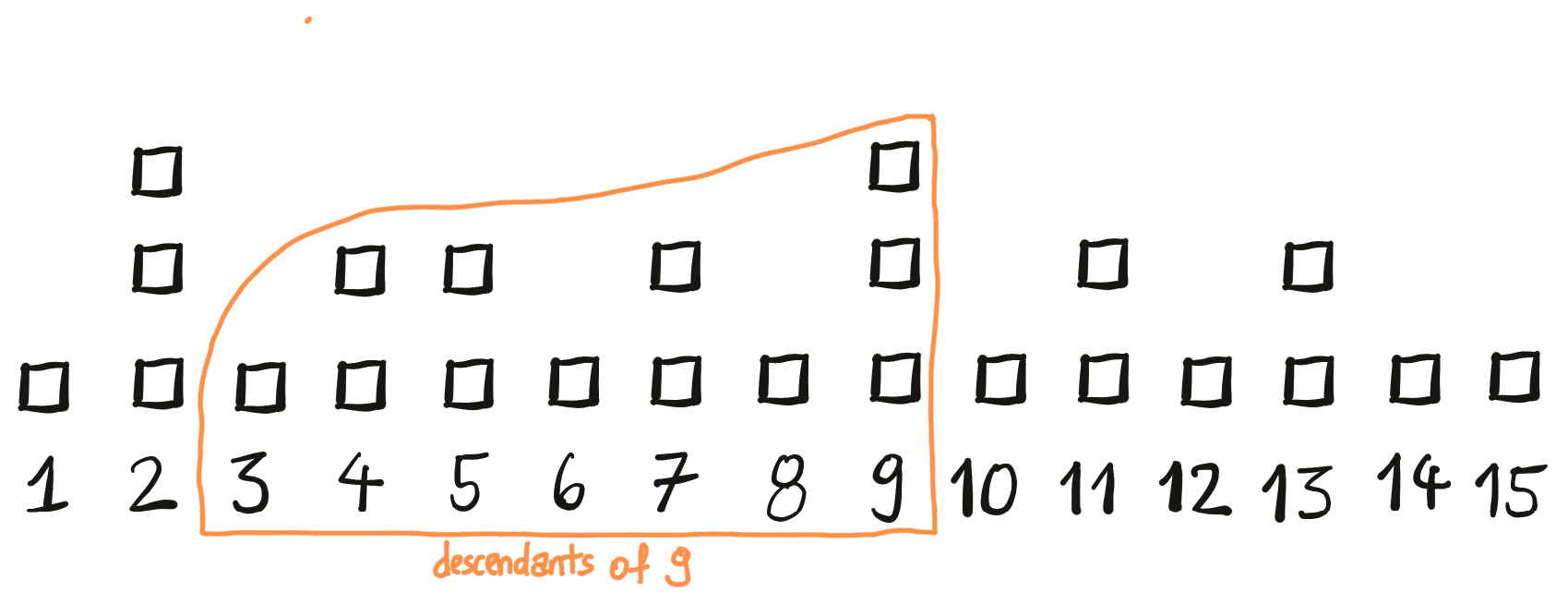}
    \noindent
    We say that two positions $i$ and $j$ are \emph{siblings} if:
    \begin{enumerate}
    \item they have equal heights; and
    \item all positions between $i$ and $j$ (excluding both $i$ and $j$)
          are strictly lower than $i$ and $j$.
    \end{enumerate}
    Being siblings is an equivalence relation. Here is an example split 
    partitioned into sibling equivalence classes:
    \picc{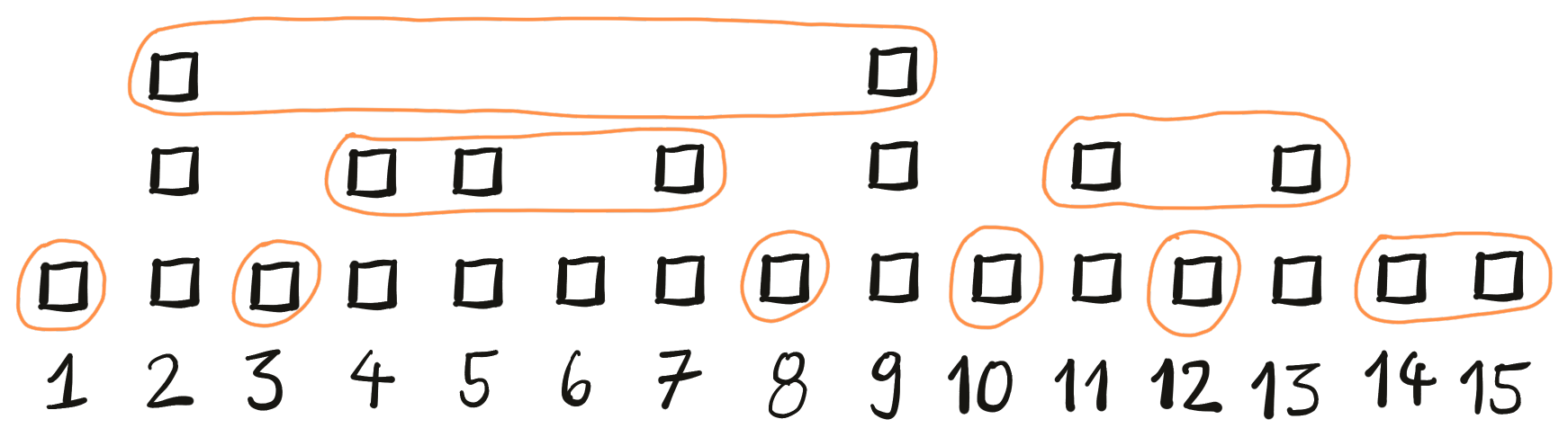}
    To avoid a potential confusion, let us clarify that being siblings is different from being descendants 
    of the same position. For instance, positions $3$ and $4$ are both descendants of position $9$, 
    but they are not siblings. Conversely, positions $2$ and $9$ 
    are siblings, but they are not descendants of any position.\\

    The \emph{split value} of a position is the semigroup product of the $s_i$-values of all of its descendants
    (this includes the position itself). Here is an example split over a sequence of elements of the semigroup 
    defined in Example~\ref{ex:smooth-seq}, where every position is annotated with its split value:
    \bigpicc{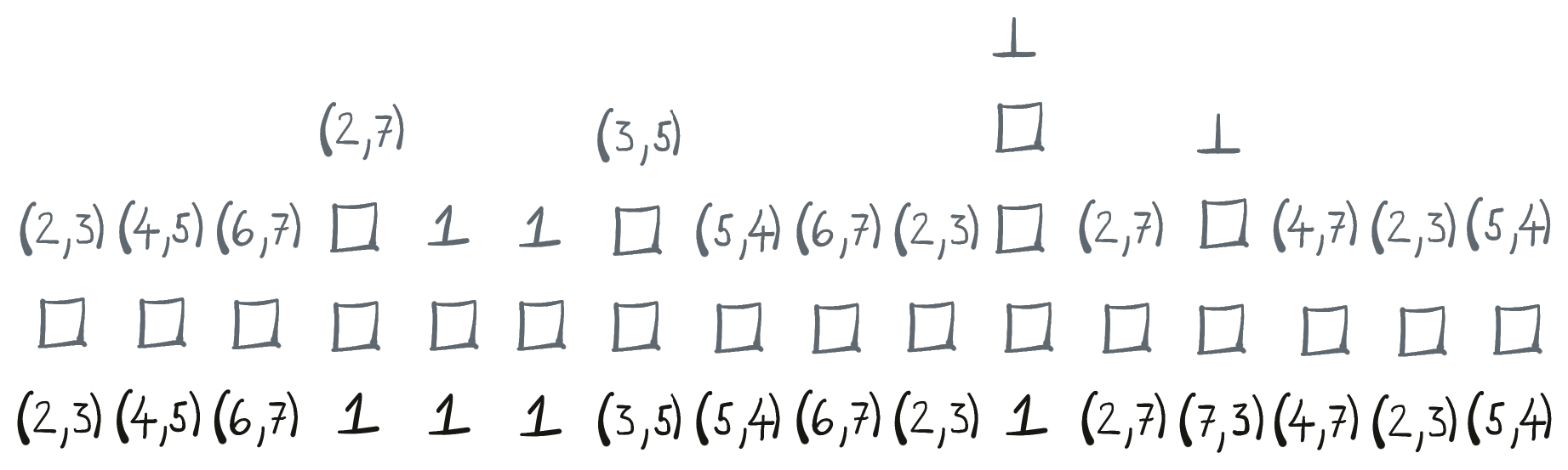}
    A \emph{sibling subsequence} is a sequence containing all positions 
    from a sibling equivalence class, where each position is labelled by its split value.
    We say that a split is \emph{smooth} if all of its sibling subsequences are smooth.
    The split presented in the previous example is smooth. Here is a picture of the same 
    split where all sibling subsequences are marked in orange:
    \bigpicc{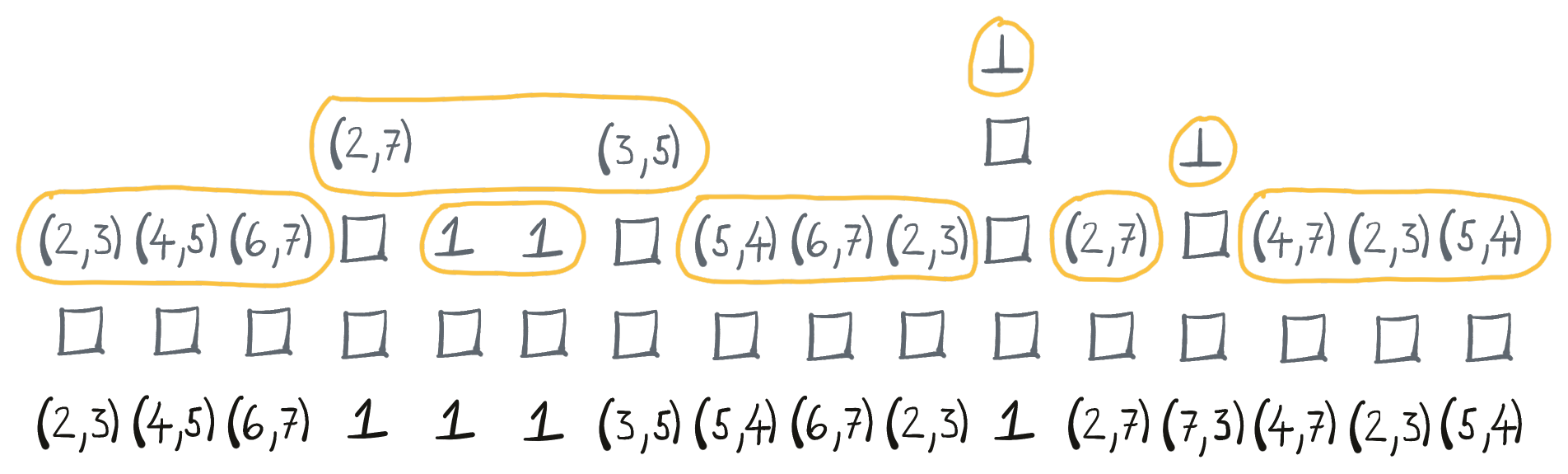}
\end{definition}

Observe that every smooth split of height $h$ can be transformed 
into a smooth factorization tree of height $2h + 1$. For example, consider the following split:
\picc{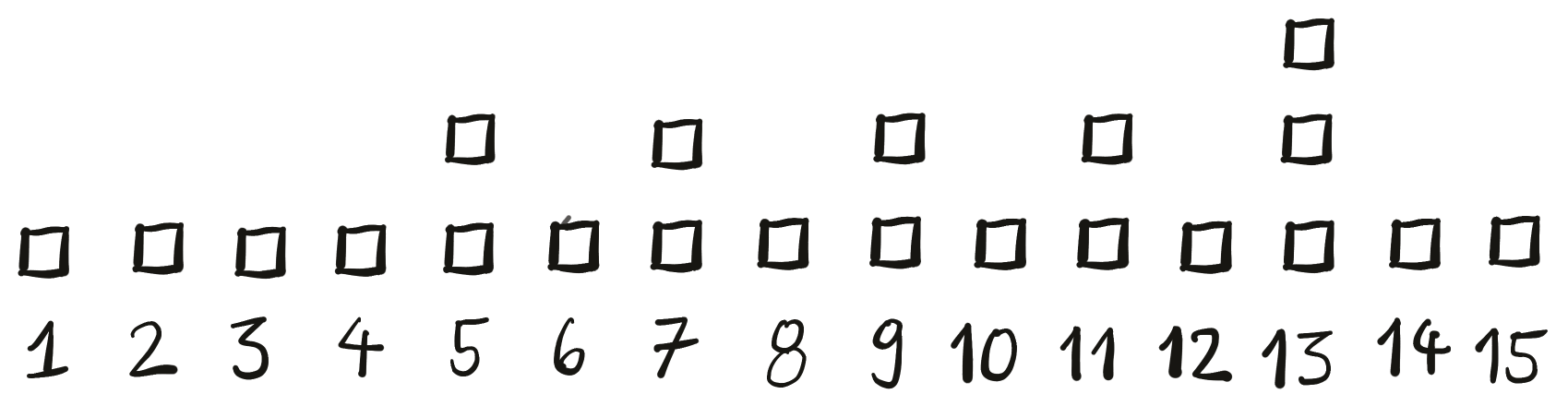}
\noindent
It can be transformed into the following factorization tree:
\picc{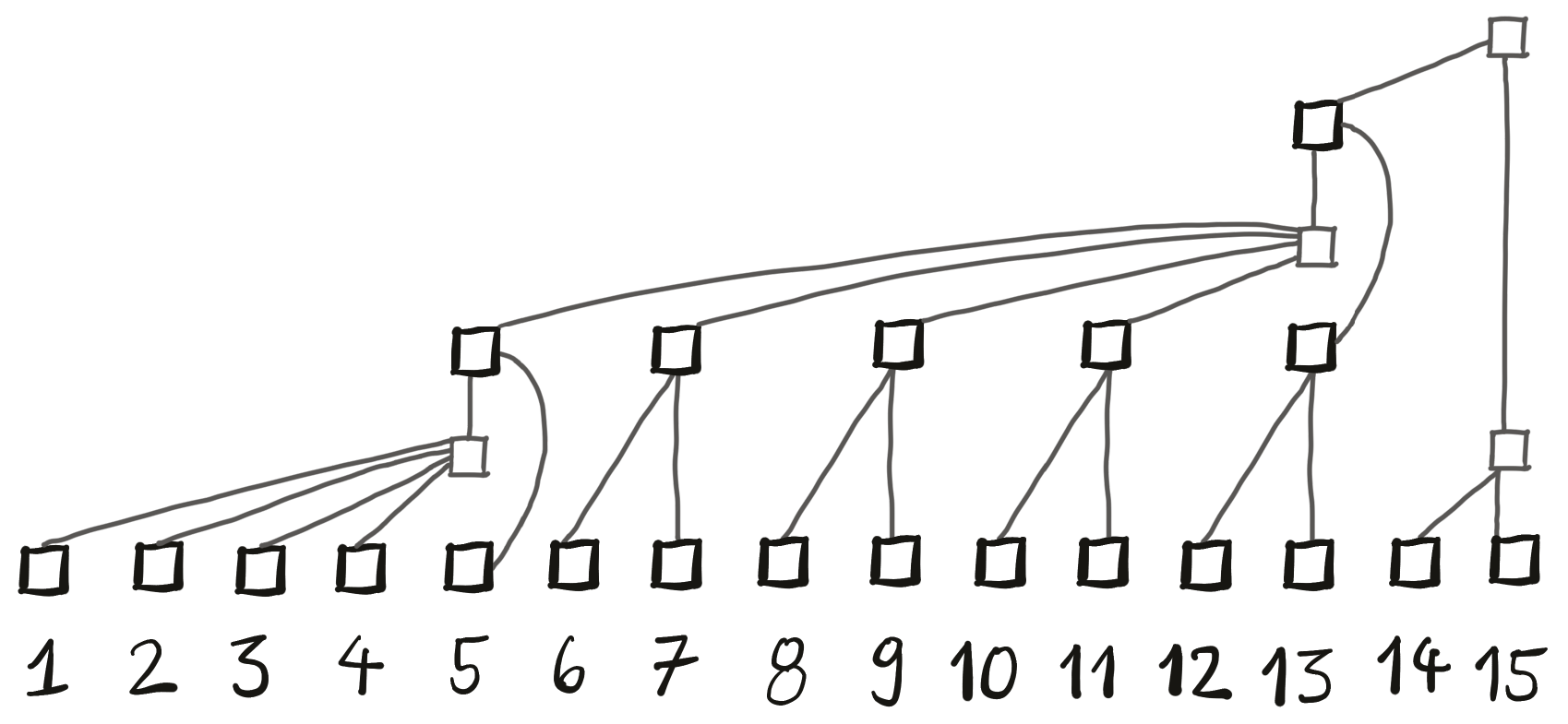}
It is not hard to see that a similar construction works for every smooth split.
This means that we
can prove Theorem~\ref{thm:smooth-factorisation-trees} by showing how to  
construct smooth splits of bounded heights.
Before we do this, we need to briefly discuss polynomial orbit-finite representations
of orbit-finite sets (that are not necessarily polynomial):
\begin{definition}
\label{def:pof-representation}
    Let $X$ be an orbit-finite set. We say that a polynomial orbit-finite set $\Sigma$ together with
    a partial function $r: \Sigma \eqto X + \bot$ is a \emph{polynomial orbit-finite representation}
    of $X$ if:
    \begin{enumerate}
        \item $r$ is surjective -- i.e. every element from $X$ has a representation in $\Sigma$;
        \item $r$ preserves least supports -- i.e. if $r(x)$ is defined, then $\supp(x) = \supp(r(x))$;
    \end{enumerate}
\end{definition}
\begin{lemma}
\label{lem:pof-represenation}
    Every orbit-finite set $X$ has a polynomial orbit-finite representation.
\end{lemma}
\begin{proof}
    Thanks to Lemma~\ref{lem:straight-reflect-repr}, we know that there is a
    surjective total function that preserves least supports:
    \[r' : \atoms^{(k_1)} + \atoms^{(k_2)} + \ldots + \atoms^{(k_n)} \to X\]
    To obtain a polynomial orbit-finite representation, we choose the following $\Sigma$:
    \[\Sigma = \atoms^{k_1} + \atoms^{k_2} + \ldots + \atoms^{k_n}\comma\]
    and we define $r$ as follows:
    \[
        r(x) = \begin{cases}
            r'(x) & \textrm{if } x \in \atoms^{(k_1)} + \atoms^{(k_2)} + \ldots + \atoms^{(k_n)}\\
            \bot & \textrm{otherwise} 
        \end{cases}
    \]
    The function $f$ is surjective and support-preserving supports because $r'$ is surjective and support-preserving.
    It is worth noting that $r(x)$ is undefined for those tuples that contain repeating atoms,
    as those tuples may not contain enough distinct atoms to construct elements of $X$. 
\end{proof}

We are now ready to formulate the lemma about constructing smooth splits with compositions of primes (as defined in Theorem~\ref{thm:kr}). 
This is the main technical lemma of this section.
\begin{lemma}
\label{lem:smooth-split-primes}
    For every orbit-finite semigroup $S$ and its polynomial orbit-finite representation
    $r : \Sigma \eqto S$, there exists a natural number $h$ and a function
    \[f : \Sigma^* \to (\underbrace{\{1,2, \ldots, h\}}_{\textrm{split height}} \times \underbrace{\Sigma}_{\textrm{split value}})^*\comma\]
    such that $f$ can be constructed as a composition of primes and such that $f$ outputs a smooth split over the input sequence,
    annotated with the split values.
    Note that the type of $f$ guarantees that the height of the split is at most~$h$.
\end{lemma}
    
The rest of this section is devoted to proving Lemma~\ref{lem:smooth-split-primes}.
The proof uses induction on the \emph{$\mathcal{J}$-height} of $S$:
\begin{definition}
    Let $S$ be a semigroup. Define a $\mathcal{J}$-chain to be a sequence of elements:
    \[ s_1, s_2, \ldots, s_n\comma \]
    such that every $s_i$ is a \emph{proper infix} of $s_{i+1}$
    (i.e. $s_i$ is an infix of $s_{i+1}$  , but $s_{i+1}$ is not an infix of $s_i$).
    The $\mathcal{J}$-height of $S$ is the length of its longest $\mathcal{J}$-chain (or $\infty$ if $S$ has arbitrarily long
    $\mathcal{J}$ chains). Similarly, the $\mathcal{J}$-height of an element $x \in S$, is the length 
    of the longest $\mathcal{J}$-chain that starts with $x$. 
\end{definition}

\noindent
In order to use $\mathcal{J}$-height as the inductive parameter, we need to know that it is finite:
\begin{lemma}[{\cite[Lemma~9.3]{bojanczyk2013nominal}}]
    If $S$ is orbit-finite then it has a finite $\mathcal{J}$ height. 
\end{lemma}
\begin{proof}
    We are going to show that if two elements $x, y \in S$ belong to the same orbit, then they are
    either $\mathcal{J}$-equivalent or $\mathcal{J}$-incomparable.
    This is enough to prove the lemma because it means that the length of every
    $\mathcal{J}$-chain is limited by the number of orbits in $S$.
    Let us take $x$ and $y$ from the same orbit and show that if $x$ is an infix of $y$, then
    also $y$ is an infix of $x$: Since $x$ and $y$ are in the same orbit, then $y = \pi(x)$ for
    some atom permutation $\pi$. The following claim, which follows from \cite[Lemma~1.14]{pitts2013nominal}, lets us assume that $\pi$ touches only 
    finitely many elements\footnote{In fact, in \cite{pitts2013nominal} Pitts defines
    sets with atoms (\emph{nominal sets}) using only this type of atom permutations
    (called finite permutations).}\footnote{
    \label{ftn:total-order}
    It is worth pointing out that Claim~\ref{claim:finite-shift} is not true for
    some other types 
    of atoms that are sometimes studied in the literature.
    One example of such atoms are rational numbers with comparison $(\mathbb{Q}, \leq)$, 
    called \emph{total-order} atoms (see \cite{bojanczyk2019slightly} or \cite{bojanczyk2013nominal}
    for more details). For this reason, the single-use theory for total order atoms 
    is different from the one for equality atoms. We are currently working 
    on it together with Nathan Lhote.} :
    \begin{claim}
    \label{claim:finite-shift}
            For every $x$ and $y$ that belong to the same orbit, there exists a permutation $\pi$ such that $\pi(x) = y$, with only finitely many atoms $a$ for which $\pi(a) \neq a$.
    \end{claim}
    We have assumed that $x$ is an infix of $\pi(x)$ (i.e. $y$) and we need to show that $\pi(x)$ is an infix of $x$.
    The infix order is an equivariant relation, so $\pi(x)$ is an infix of $\pi^2(x)$,
    and by induction on $k$ we can show that $\pi^{k}(x)$ is an infix of $\pi^{k+1}(x)$.
    Since $\pi$ touches only finitely many atoms, there exists a $k$ such that $\pi^k$ is the identity permutation.
    It follows by transitivity that $\pi(x)$ is an infix of $x$, because $x$ is an infix of $\pi(x)$, 
    which is an infix of $\pi^2(x)$, \ldots, which is an infix of $\pi^{k-1}(x)$,
    which is an infix of $\pi^k(x)$, which is equal to $x$.
\end{proof}

We are now ready to start proving Lemma~\ref{lem:smooth-split-primes}. We slightly strengthen 
its formulation, to make it compatible with the inductive structure of the proof:
\begin{definition}
    The \emph{output sequence} of a split is its sibling subsequence 
    of maximal height (note that it is uniquely defined for every split). 
    We say that a split is \emph{semi-smooth} if all of 
    its sibling subsequences are smooth, with the possible exception of 
    the output subsequence.
\end{definition}
Here is an example of a semi-smooth split for the semigroup from Example~\ref{ex:smooth-seq}.
Its output sequence is marked in orange:
\bigpicc{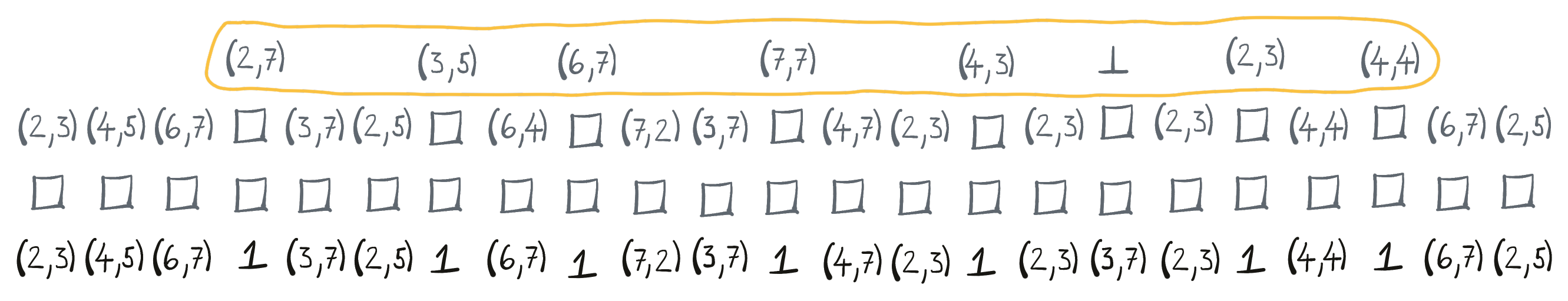} 
\begin{lemma}
\label{lem:partial-splits-primes}
Fix an orbit-finite semigroup $S$, let $H$ be its $\mathcal{J}$-height, and
let  $r : \Sigma \eqto S + \bot$ be its polynomial orbit-finite representation.
For every $h \leq H$, there exists a function:
\[f : \Sigma^* \to (\underbrace{\{1,2, \ldots, h\}}_{\textrm{split height}} \times \underbrace{\Sigma}_{\textrm{split value}})^*\comma\]
that can be constructed as a composition of primes, which constructs \emph{semi-smooth} splits
of height $\leq h$ over the input sequence, and annotates each position with its split value.
Furthermore, the output sequence of every split 
constructed by $f$ consists only of elements whose $\mathcal{J}$-height is at most $H + 1 - h$.
\end{lemma}

Before we prove Lemma~\ref{lem:partial-splits-primes}, let us show that it implies
Lemma~\ref{lem:smooth-split-primes}. We start with the following claim:
\begin{claim}
\label{claim:j-height-one-smooth}
    All elements of $\mathcal{J}$-height one are pairwise $\mathcal{J}$-equivalent.
\end{claim}
\begin{proof}
    Let $x$ and $y$ be elements with $\mathcal{J}$-heights equal to $1$. 
    Assume towards a contradiction that they belong to two different 
    $\mathcal{J}$-classes. It follows that either $x$ or $y$ is a proper infix of $x \cdot y$,
    so either $\{x, xy\}$ or $\{y, xy\}$ forms a $\mathcal{J}$-chain. This contradicts the assumption.
\end{proof}
Thanks to Claim~\ref{claim:j-height-one-smooth}, we know that every sequence of elements with
$\mathcal{J}$-heights equal to $1$ has to be smooth. It follows that for $h = H$
Lemma~\ref{lem:partial-splits-primes} produces a smooth split.\\

Let us now proceed with the proof of Lemma~\ref{lem:partial-splits-primes}. The induction base 
is very simple: For $k = 1$, it suffices to set the height of every position to $1$,
and set the split values to $s_i$'s. This can be expressed as a homomorphism $f^*$, where $f$ is the following function:
\[ f(s) = (1, s) \]

This leaves us with proving the induction step.
We construct the split function for $h+1$ in the following four steps:
\begin{enumerate}
    \item We apply the induction assumption, constructing an almost smooth split of height $h$.
          Here an example for the monoid from Example~\ref{ex:smooth-seq} and $h=2$:\\
          \bigpicenum{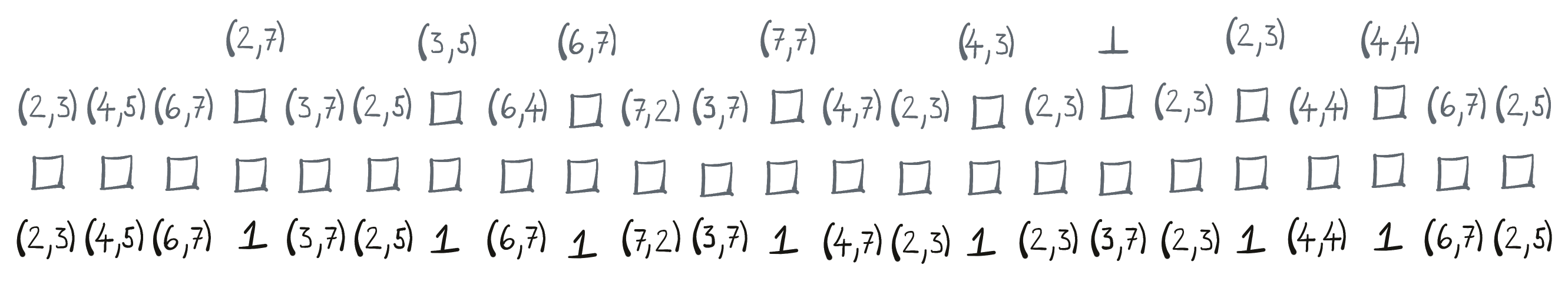} 
    \item Then, we divide the output of the constructed split into \emph{almost smooth blocks}, which 
          are blocks that are smooth sequences except of their last element (this step is explained in detail in Section~\ref{subsec:breaking-up-the-sequence}):
          \[ \overbrace{\underbrace{s_1, s_2, s_3, \ldots, s_n}_{\textrm{smooth}}, s_{n+1}}^{\textrm{non-smooth}} \]
          \bigpicenum{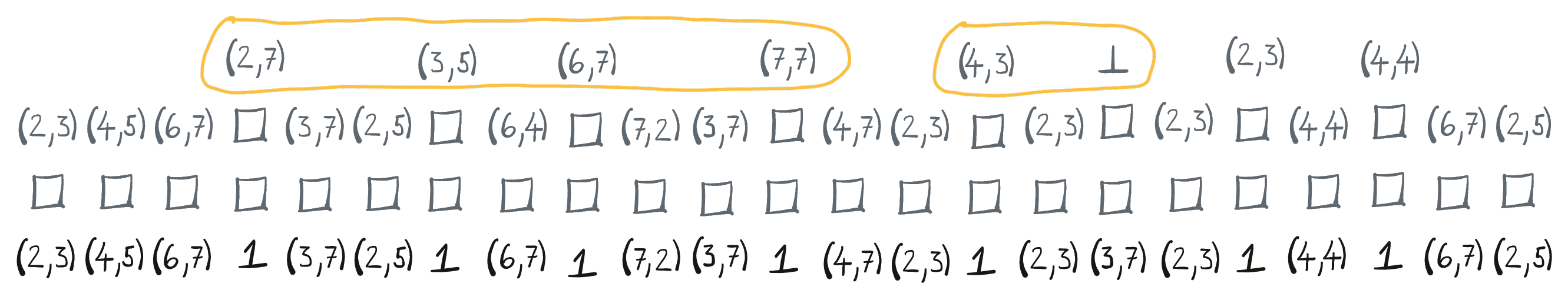}
          \medskip
    \item Then, we compute the product of each almost smooth block and show that all those products have
          $\mathcal{J}$-height of at most $(H + 1 - h) - 1$. (Section~\ref{subsec:almost-smooth-products})
          \bigpicenum{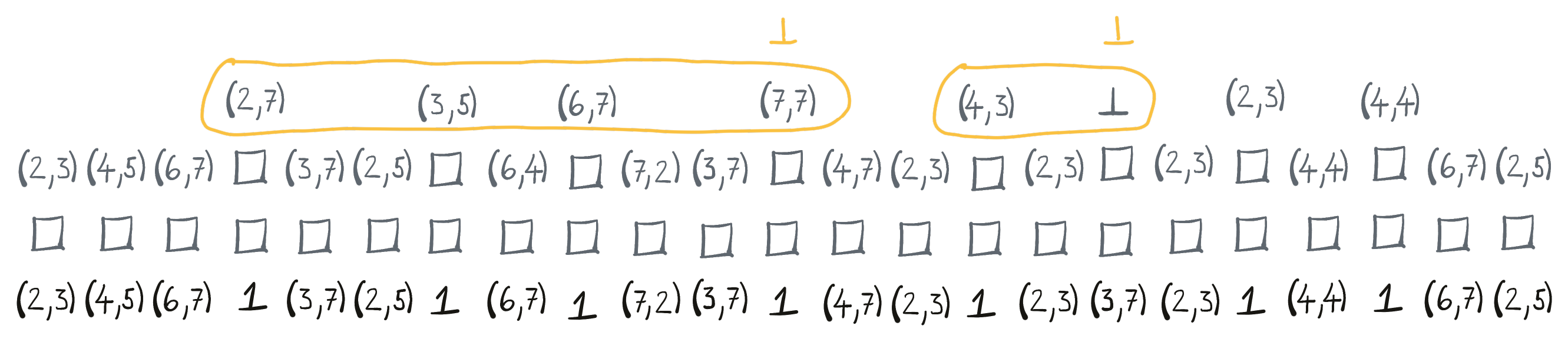}
    \item Finally, we construct a partially smooth split of height $h+1$, by 
          increasing the height of every last position in an almost smooth block by $1$, 
          and by setting the split value of each such position to the product of its almost smooth block
          (this can be done using a homomorphism):
          \bigpicenum{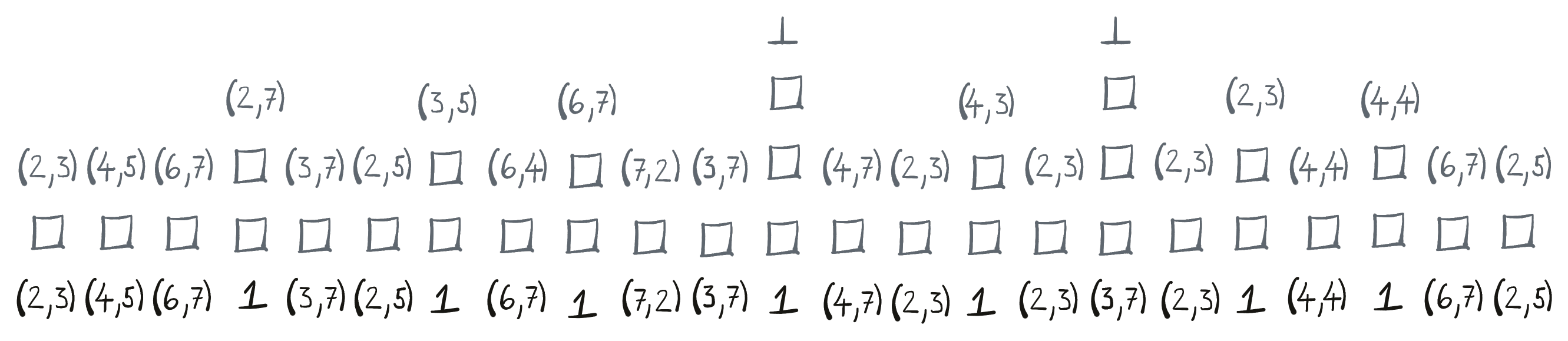}
\end{enumerate}

\subsection{Breaking up the sequence}
\label{subsec:breaking-up-the-sequence}
In this section, we construct a function that divides its input sequence into \emph{almost smooth blocks},
which are blocks that are not smooth, but would be smooth if we removed their last element:
\[ \overbrace{\underbrace{s_1, s_2, s_3, \ldots, s_n}_{\textrm{smooth}}, s_{n+1}}^{\textrm{non-smooth}} \]
We encode this as a function that underlines the last letter of each block:
\[\Sigma^* \to (\underbrace{\Sigma}_{\textrm{underlined}} + \underbrace{\Sigma}_{\textrm{not underlined}})^*\]
For example, consider the semigroup $S = {{\atoms}\choose{\leq 3}} + \bot$, with the following operation
(see Example~\ref{ex:monoid-at-most-3}):
\[ x \cdot y = \begin{cases}
    x \, \cup \, y & \textrm{if } x \neq \bot \textrm{, } y \neq \bot \textrm{, and } |x \, \cup \, y| < 3   \\\
    \bot & \textrm{otherwise }
    \end{cases}
\]
For the following input sequence:
\[
    \{1\},\ \{1\},\ \{1\},\ \{5\}, \ \{7, 8\},\ \{2\}, \ \{9\}, \ \{7\},\ \{4\},\ \{4\},\ \{4\}
\]
The function should return:
\[
    \{1\},\ \{1\},\ \{1\},\ \underline{\{5\}}, \ \{7, 8\},\ \underline{\{2\}}, \ \{9\}, \ \{9\}, \ \underline{\{7\}},\ \{4\},\ \{4\},\ \{4\}
\]
Which corresponds to the division into the following blocks:
\[
    \underbrace{\{1\},\ \{1\},\ \{1\},\ \underline{\{5\}}}_{\textrm{1st block}}, \ \underbrace{\{7, 8\},\ \underline{\{2\}}}_{\textrm{2nd block}}, \ \underbrace{\{9\}, \ \{9\}, \ \underline{\{7\}}}_{\textrm{3rd block}},\ \{4\},\ \{4\},\ \{4\}
\]
Note that the last three elements do not belong to any block. This is because their
almost smooth block is under construction --
they form a smooth sequence and they have not seen an element that would break their smoothness.
This is only allowed at the end of the input sequence:
\begin{lemma}
    \label{lem:div-smooth-blocks}
    Let $S$ be a orbit-finite semigroup and let $r: \Sigma \to S + \bot$ be its polynomial orbit-finite representation.
    The following function $f_{\textrm{divide}}$ can be constructed as a composition of primes:
    \[f_{\textrm{divide}} : \Sigma^* \to (\Sigma + \Sigma)^*\]
    The function $f_{\textrm{divide}}$ divides the input word into almost smooth blocks. Additionally, the rightmost part of the input that does not belong to any block must be smooth. (The input should only contain letters for which $r$ is defined;
    if this is not the case, the output of $f_{\textrm{divide}}$ is unspecified).
 \end{lemma}
\begin{proof}
    The main idea of the proof is noticing that being a smooth sequence is a local property:
    \begin{claim}
    \label{claim:smooth-local}
        A sequence $s_1, \ldots, s_n$ over a semigroup $S$ is smooth, if and only if all 
        of its pairs of consecutive elements are smooth sequences. 
    \end{claim}
    \begin{proof}
        ($\Rightarrow$): If $s_i,\ s_{i+1}$ is not smooth for some $i$, then $s_i \cdot s_{i+1}$
        is not an infix of either $s_i$ or $s_{i+1}$. It follows that $s_1 \cdot s_i \cdot s_{i+1}\ldots \cdot s_n$
        is not an infix of either $s_i$ or $s_{i+1}$, which means that the entire sequence is not smooth.\\

        \noindent
        ($\Leftarrow$:) We prove this by induction on $n$. For $n \leq 2$, the claim is trivially true.
        For the induction step, we assume that both $s_1, \ldots, s_n$ and $s_n, s_{n+1}$ are smooth
        and show that $s_1, \ldots, s_{n+1}$ is smooth. Every $s_i$ is clearly an infix of 
        $s_1, \ldots, s_{n+1}$. It suffices to show that $s_1 \cdot \ldots \cdot s_{n+1}$ is an infix 
        of every $s_i$. We already know that $s_1 \ldots s_n$ is an infix of every $s_i$ (for $i \leq n$), so
        we just need to show that $s_1 \cdot \ldots \cdot s_{n+1}$ is an infix of both
        of $s_{n+1}$ and of $s_1 \ldots s_n$. For this, we are going to use the following orbit-finite version of a well-known lemma about Green's relation 
        (the orbit-finite version was proved in \cite[Lemma7.1~and~Theorem5.1]{bojanczyk2013nominal}):
        \begin{lemma}
            \label{lem:green-lemma-anitchains}
            Let $S$ be an orbit-finite semigroup, and let $x, y \in S$. If $xy$ is an infix of $x$, then $xy$ is a \emph{prefix} of $x$.
            In other words, there exists an $x' \in S^1$ (as defined in Claim~\ref{claim:semigroup-to-monoid}) such that $x = xyx'$.
            Analogously, if $xy$ is an infix of $y$, then $xy$ is a \emph{suffix} of $y$.
            \end{lemma}
        \noindent
        It follows from the lemma that there exists an $s_n'$ such that $s_n'\cdot s_1 \cdot \ldots \cdot s_n = s_n$.
        Since $s_n, s_{n+1}$ is a smooth sequence, there are $a, b$ such that $a s_n s_{n+1}  b = s_{n+1}$.
        It follows that $a \cdot s_n' \cdot s_1 \cdot \ldots s_{n+1} \cdot b = s_{n+1}$, so $s_1 \cdot \ldots \cdot s_{n+1}$
        is an infix of $s_{n+1}$. To prove that 
        $s_1 \cdot \ldots \cdot s_{n+1}$ is an infix of $s_1 \cdot  \ldots  \cdot s_n$ , it suffices to see 
        that by Lemma~\ref{lem:green-lemma-anitchains}, there exists $x \in S$, such that 
        $s_{n}  s_{n+1}  x = s_n$. It follows that $s_1 \cdot \ldots \cdot s_{n+1} \cdot x = s_1 \cdot \ldots \cdot s_n$, 
        which finishes the proof. 
    \end{proof}
    
    Before we show how to construct $f_{\textrm{divide}}$, we define a couple of auxiliary functions:
    \begin{claim}
        \label{claim:gen-su-prop}
        For every polynomial orbit-finite set $\Sigma$, the following \emph{single-use letter propagation} 
        function can be constructed as a composition of primes:
        \[ f_\textrm{$\Sigma$-prop} : (\Sigma + \downarrow + \epsilon)^* \to (\Sigma + \epsilon)^* \]
        The function works in the same way as the single-use letter propagation from
        Example~\ref{ex:su-propagation}, but it propagates elements of $\Sigma$ instead of $\atoms$.
    \end{claim}
    \begin{proof}
        The claim can be shown by a straightforward induction on the construction of $\Sigma$ 
        as a polynomial orbit-finite set.
    \end{proof}
    \begin{claim}[{\cite[Lemma~36]{single-use-paper}}]
    \label{lem:delay-primes}
        For every polynomial orbit-finite alphabet $\Sigma$, the following $f_\textrm{delay}$ function can be constructed
        as a composition of primes:
        \[
            \begin{tabular}{ccccccc}
                $a_1$ & $a_2$ & $a_3$ & $a_4$ & $a_5$ & \ldots & $a_n$\\
                & & & \rotatebox[origin=c]{270}{$\mapsto$} &  &  &\\
                $\vdash$ & $a_1$ & $a_2$ & $a_3$ & $a_4$ & \ldots & $a_{n-1}$\\
            \end{tabular}
         \]
    \end{claim}
    \begin{proof} The function can be implemented in three steps:
        First, we use a classical Mealy machine to mark every position as odd or even
        (thanks to the classical Krohn-Rhodes theorem, we know that the classical Mealy machine decomposes into prime functions -- 
        see the last paragraph of the proof Lemma~\ref{lem:div-smooth-blocks} for details).
        In the next step, we use a homomorphism together with the single-use letter propagation, to 
        propagate all letters in even positions one position to the right.
        Finally, we do the same for letters in odd positions. 
    \end{proof}

    We are now ready to construct $f_\textrm{divide}$.  First, we apply the 
    \emph{delay} function, while keeping the original input.
    This is a common pattern, that is possible thanks to the $\times$ combinator:
    \[\Sigma^* \transform{\copyf^*} (\Sigma \times \Sigma)^* \transform{f_\textrm{delay} \times \idf} ((\Sigma + \vdash) \times \Sigma)\]
    This brings us to the following situation:
    \[
        \begin{tabular}{cccccccccc}
            $\vdash$ & $s_1$ & $s_2$ & $s_3$ & $s_4$ & $s_5$ & $s_6$ & $s_7$ & $s_8$ & $s_9$\\ 
            $s_1$    & $ s_2$ & $s_3$ & $s_4$ & $s_5$ & $s_6$ & $s_7$ & $s_8$ & $s_9$ & $s_{10}$\\
        \end{tabular}
     \] 
     Now, we use a homomorphism to underline every pair $s_{i-1}, s_i$ which is not smooth. We can do 
     this using a homomorphism, because the function $\Sigma^2 \to \Sigma^2 + \Sigma^2$ that underlines non-smooth pairs is equivariant.
     \[
        \begin{tabular}{cccccccccc}
            $s_1$    & $ s_2$ & $s_3$ & $\underline{s_4}$ & $\underline{s_5}$ & $\underline{s_6}$ & $s_7$ & $\underline{s_8}$ & $s_9$ & $s_{10}$\\
            $\vdash$ & $s_1$ & $s_2$ & $s_3$ & $s_4$ & $s_5$ & $s_6$ & $s_7$ & $s_8$ & $s_9$\\ 
        \end{tabular}
     \]
     We use one more homomorphism to project away the delayed letters:
     \[
        \begin{tabular}{cccccccccc}
            $s_1$ & $s_2$ & $s_3$ & $\underline{s_4}$ & $\underline{s_5}$ & $\underline{s_6}$ & $s_7$ & $\underline{s_8}$ & $s_9$ & $s_{10}$\\
        \end{tabular}
     \]
    Finally, we need to get rid of blocks of size $1$ (they are always smooth, so they cannot be almost smooth).
    We do this by removing every other underline in a contiguous block of underlined letters:
    \[
        \begin{tabular}{cccccccccc}
            $s_1$ & $s_2$ & $s_3$ & $\underline{s_4}$ & $s_5$ & $\underline{s_6}$ & $s_7$ & $\underline{s_8}$ & $s_9$ & $s_{10}$\\
        \end{tabular}
    \]
    Before we show how to implement this step, let us notice that it finishes the construction of $f_\textrm{divide}$-- 
    it follows from Claim~\ref{claim:smooth-local} that after this step the input is split into almost smooth blocks
    (and possibly one smooth block in the end). This leaves us with showing how to remove every other underline in each 
    contiguous block of underlines. First, we use a homomorphism to apply the following isomorphism to every letter:
    \[
      \begin{tabular}{ccc}
        $\underbrace{\Sigma}_{\textrm{unerline}} + \underbrace{\Sigma}_{\textrm{no underline}}$ & $\simeq$ & $\Sigma \times \{\underbrace{1}_{\textrm{underline}}, \underbrace{0}_{\textrm{no-underline}}\}$
     \end{tabular}
    \]
    This transformation extracts the finite information about underlines into a separate coordinate:
    \[
    \begin{tabular}{cccccccccc}
        0 & 0 & 0 & 1 & 1 & 1 & 0 & 1 & 0 & 0\\
        $s_1$ & $s_2$ & $s_3$ & $s_4$ & $s_5$ & $s_6$ & $s_7$ & $s_8$ & $s_9$ & $s_{10}$\\
    \end{tabular}
    \]
    \noindent
    Now, we use the parallel composition to apply the following classical (i.e. finite) Mealy machine to the $\{0, 1\}$-coordinate
    (thanks to the classical Krohn-Rhodes theorem, we know that it further decomposes into prime functions):
    \vsmallpicc{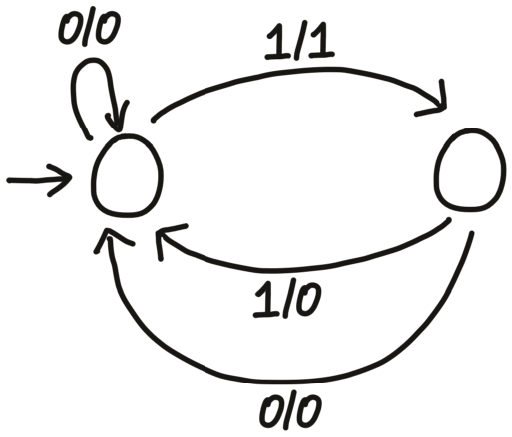}
    \noindent
    Finally, we use the same isomorphism to go back to the $\Sigma + \Sigma$ alphabet:
    \[
        \begin{tabular}{cccccccccc}
            $s_1$ & $s_2$ & $s_3$ & $\underline{s_4}$ & $s_5$ & $\underline{s_6}$ & $s_7$ & $\underline{s_8}$ & $s_9$ & $s_{10}$\\
        \end{tabular}
    \]
    \noindent
    This finishes the proof of Lemma~\ref{lem:div-smooth-blocks}. 
\end{proof}

Remember that in the proof of Lemma~\ref{lem:partial-splits-primes}, we need to apply 
$f_\textrm{divide}$ to the output of a semi-smooth split. This can be done using the following combinator: 
\begin{lemma}
    \label{lem:subsequence-combinator}
    Compositions of primes are closed under the following \emph{subsequence} combinator:
    \[
        \infer[\comma]
        { (\Sigma + \square)^* \stackrel{(f + \square)}{\longrightarrow} (\Gamma + \square)^*}
        { \Sigma^* \stackrel {f}{\longrightarrow} \Gamma^*}
    \]
    The function $(f + \square)$ applies $f$ to the word composed of $\Sigma$-letters of the input,
    and leaves the $\square$'s unchanged.
    \end{lemma}
    \begin{proof}
    We start by noticing that:
    \[
    \begin{tabular}{ccc}
    $((f \circ g) + \square) = (f + \square) \circ (g + \square)$ &  and &
    $(f \times g) + \square = (f + \square) \times (g + \square)$
    \end{tabular}
    \]
    It is not hard to see that for every prime function $f$, the function $(f+\square)$
    is a composition of primes. This finishes the proof.
    \end{proof}

\subsection{Computing almost smooth products}
\label{subsec:almost-smooth-products}
In this section, we show how to use compositions of primes to compute the products of almost smooth blocks. For example, 
we want to transform the following input:
\[
    \begin{tabular}{cccccccccc}
        $s_1$ & $s_2$ & $s_3$ & $\underline{s_4}$ & $s_5$ & $\underline{s_6}$ & $s_7$ & $\underline{s_8}$ & $s_9$ & $s_{10}$\\
    \end{tabular}
\]
\noindent
Into the following output:
\[
    \begin{tabular}{cccccccccc}
        $\square$ &  $\square$       & $\square$      & $(s_1 \cdot s_2 \cdot s_3 \cdot s_4)$ & $\square$          & $(s_5 \cdot s_6)$   &  $\square$     & $(s_7 \cdot s_8)$   &  $\square$     &  $\square$      \\
        $s_1$   &  $s_2$         & $s_3$ & $\underline{s_4}$                   & $s_5$ & $\underline{s_6}$ & $s_7$ & $\underline{s_8}$ & $s_9$ & $s_{10}$\\
    \end{tabular}
\]
This step is formalized as the following lemma:
\begin{lemma}
\label{lem:almost-smooth-blocks}
Let $S$ be an orbit-finite semigroup and let $h : \Sigma \to S + \bot$ be its
polynomial orbit-finite representation. Then the following function $f_\textrm{blocks}$ is a composition of primes:
\[ f_{\textrm{blocks}} (\Sigma + \Sigma)^* \to (\Sigma + \square)^* \]
The function $f_\textrm{blocks}$ inputs a sequence of elements divided into almost smooth blocks (for all other inputs its
behaviour is undefined) and computes
the product of each of the block:
 If the $i$-th letter of the input is underlined, then the $i$-th letter of the output is
\emph{a representation} of the product of the block that ends in that letter, i.e. a representation of the value:
\[h(s_{j+1}) \cdot h(s_{j+1}) \cdot \ldots \cdot h(s_i)\]
where $j$ is the first underlined position to the left of $i$ (or $0$ if $i$ is the first underlined position).	
If the $i$-th position is not underlined, then the $i$-th letter of the output is equal to $\square$.
\end{lemma}
\noindent
We can reduce the general construction to the case of a single block using the following \emph{map} combinator:
\begin{lemma}
    \label{lem:map-combinator-primes}
    If a function $f : \Sigma \to \Gamma$ is a composition of primes, then the function
    \[ \mathtt{map}\, f : (\Sigma + \Sigma)^* \to (\Gamma + \Gamma)^*\comma \]
    which applies $f$ independently to every block that ends with an underlined letter
    (or with the end of the word) is a composition of primes as well. 
\end{lemma}
\begin{proof}
    Since $\mathtt{map } (f \circ g) = (\mathtt{map}\, f) \circ (\mathtt{map}\, g)$ and $\mathtt{map} (f \times g) = (\mathtt{map}\, f) \times (\mathtt{map}\, g)$,
    it suffices to show that the $\mathtt{map}$ versions of all the prime functions can be expressed as compositions of primes.
    We deal only with the hardest case, i.e. group prefixes. Since 
    it works over finite alphabets, its map version can be implemented as a classical Mealy machine,
    which, by the classical Krohn-Rhodes theorem, can be decomposed into prime functions.
\end{proof}

This leaves us with showing how to use compositions of primes 
to compute a product of a single almost smooth block, i.e. a version 
of Lemma~\ref{lem:almost-smooth-blocks}, with the extra assumption that
the last letter is the only underlined letter in the input: 
\[
    \begin{tabular}{cccccccccc}
        $s_1$ & $s_2$ & $s_3$ & $s_4$ & $s_5$ & $\underline{s_6}$\\
    \end{tabular}
\]
We compute the product in two steps:
In the first step, we compute the product of the smooth sequence
of all letters except the last one (this step is formulated below as Lemma~\ref{lem:smooth-product}):
\[
    \begin{tabular}{cccccc}
        $\square$ & $\square$ & $\square$ & $\square$ & $\square$ & $(s_1 \cdot s_2 \cdot s_3 \cdot s_4 \cdot s_5)$\\
        $s_1$ & $s_2$ & $s_3$ & $s_4$ & $s_5$ & $\underline{s_6}$\\
    \end{tabular}
\]
Then we include the last element to the product by using a homomorphism to apply the binary product function
(see Claim~\ref{claim:product-poly-repr} below) to the underlined position. This finishes the construction:
\[
    \begin{tabular}{cccccc}
        $\square$ & $\square$ & $\square$ & $\square$  & $(s_1 \cdot s_2 \cdot s_3 \cdot s_4) \cdot s_5$\\
        $s_1$ & $s_2$ & $s_3$ & $s_4$ & $\underline{s_5}$\\
    \end{tabular}
\]
\begin{claim}
    \label{claim:product-poly-repr}
    Let $S$ be an orbit-finite semigroup and $r : \Sigma \to S + \bot$ its polynomial orbit-finite representation.
    There exists an equivariant function
    \[f : \Sigma \times \Sigma \eqto \Sigma + \bot\]
    such that for all
    $s_1, s_2 \in \Sigma$ that represent elements of $S$ it holds that:
    \[ f(s_1, s_2) \textrm{ is a representation of } s_1 \cdot s_2 \]
\end{claim}
\begin{proof}
    We obtain $f$ by applying Lemma~\ref{lem:straight-uniformization-eq} to the following relation $R \subseteq \Sigma^2 \times \Sigma$:
    If elements $x, y \in \Sigma$ both represent elements of $S$, then the pair $(x, y)$ is $R$-related with every
    representation of $r(x) \cdot r(y)$. If either $x$ or $y$ is not a valid representation,
    then $(x, y)$ is only $R$-related with $\bot$.\\

    In order to use Lemma~\ref{lem:straight-uniformization-eq}, we need to show 
    that for every $s_1, s_2 \in \Sigma$, there exists $s_3 \in \Sigma$ such that: 
    \[ \begin{tabular}{ccc}
        $(s_1, s_2) \ R \ s_3$ & and & $\supp(s_3) \subseteq \supp(s_1, s_2)$.
    \end{tabular}\]
    If both $s_1$ and $s_2$ represent elements from $S$, we can 
    pick any $s_3$ that represents $r(s_1) \cdot r(s_2)$ -- 
    by Definition~\ref{def:pof-representation} combined with Lemma~\ref{lem:fs-functions-preserve-supports}, 
    we know that $\supp(s_3) \subseteq \supp(s_1) \cup \supp(s_2)$.
    If either $r(s_1)$ or $r(s_2)$ is undefined, then we know that they 
    are related to $\bot$, which is equivariant. 
\end{proof}
It is worth pointing out that $\Sigma$ with $f_{\Sigma-\textrm{prop}}$ is not a semigroup,
because $f_{\Sigma-\textrm{prop}}$ does not have to be commutative:
$(a, f(b, c))$ and $f(f(a, b), c)$
might be different representations of the same element.\\

This leaves us with showing how to use compositions of primes to compute smooth products, 
we do this in the following lemma, which is the main technical result of this section:
\begin{lemma}
    \label{lem:smooth-product}
        For every orbit-finite semigroup $S$ and its polynomial orbit-finite representation $r : \Sigma \to S + \bot$, 
        the following function can be constructed as a composition of primes:
        \[f_{\textrm{smooth}} : (\Sigma + \dashv)^* \to (\Sigma + \Sigma)^*\]
        \begin{enumerate}
            \item The main case is when the input is a smooth sequence followed by the letter $\dashv$, 
                  which has to be the last letter of the input.
                  In this case the function should compute a representation of the
                  product of all its input letters:
        \[
            \begin{tabular}{ccccccc}
            $\mathtt{Input: }$ & $s_1$ & $s_2$ & \ldots & $s_n$ & $\dashv$\\
            $\mathtt{Output: }$ & $\square$ & $\square$ & \ldots & $\square$ & $s_1 \cdot \ldots \cdot s_n$\\ 
            \end{tabular}
        \]
        \item The secondary case is when the input does not contain the $\dashv$ letter.
        In this case, $f_\textrm{smooth}$ should only output $\square$'s. (This case 
        is useful for handling the last, unfinished block in Lemma~\ref{lem:almost-smooth-blocks}.)
    \end{enumerate}
    In all other cases, the output of $f_\textrm{smooth}$ is unspecified. 
\end{lemma}
\noindent
The remainder of this section is dedicated to proving Lemma~\ref{lem:smooth-product}
In the proof, we use a similar approach as in Section~\ref{subsec:eliminating-redundant-atoms}.
First, we show how to construct $f_\text{smooth}$ as a composition of \emph{finitely supported primes},
which are the single-use prime functions extended with homomorphisms based on finitely supported
(and not just equivariant) functions.
Then in Lemma~\ref{lem:primes-elim-support}, we show how to eliminate all atomic 
constants from the construction.\\

The construction of $f_\textrm{smooth}$ as a composition of finitely supported primes consists of the following six steps  
(we assume that $n \geq 2$ -- otherwise we can construct the product using the delay function):
\begin{enumerate}
\item In the first step, we fix a tuple $\bar a$ of $2 \! \cdot \! \dim(S)$ different atoms
      (remember that $\dim(S) = \max\{ |\supp(s)| : s \in S \}$)
      and we equip $s_1$ with representation of an
      idempotent $e_1$, that satisfies the following two conditions
(for the purpose of this step we say that such $e_1$ is \emph{good} for $s_1$):
\begin{enumerate}
    \item $e_1$ and $s_1$ are $\mathcal{J}$-equivalent; and
    \item $\supp(e_1) \subseteq \supp(\bar a) \cup \supp(J(s_1))$, 
          where $J(s_1)$ is the $\mathcal{J}$-class of $s_1$. 
          Notice that the set of all $\mathcal{J}$-classes is a set with atoms itself,
          which means that $\supp(J(s_1))$ is well-defined.
\end{enumerate}
The result of this step should look as follows:
\[
    \begin{tabular}{ccccccc}
        $s_1$ & $s_2$ & $s_3$ & $s_4$ & $s_5$ & $s_6$ & $\dashv$\\
        $e_1$ & $\square$ & $\square$ & $\square$ & $\square$ & $\square$ & $\square$ \\
    \end{tabular}
\]
Before we show how to construct such $e_1$, let us show that it exists:
\begin{claim}
    \label{clam:good-idempotent}
    There exists an idempotent $e$ that is good for $s_1$. 
\end{claim}
\begin{proof}
    The sequence $s_1, s_2$ is smooth.
    It follows that $s_1$, $s_2$, and $s_1 \cdot s_2$ all belong to the same
    $\mathcal{J}$-class. By \cite[Corollary~2.25]{pin2010mathematical}, 
    this means that there is some idempotent $e'$ that belongs to $J(s_1)$. 
    Let us show how to transform this $e'$ into an $e$ that is good for $s_1$.
    Define $\pi$ to be a permutation that swaps every atom from
    $(\supp(e') - \supp(J(x)))$ with a fresh atom from $(\bar a - \supp(J(x)))$.
    For such $\pi$ to exist, $\bar a$ has to be large enough. This is not 
    hard to see after noticing that function $x \mapsto J(x)$ is equivariant, 
    which means that:
    \[|\supp(J(x))| \stackrel{\textrm{Lemma~\ref{lem:fs-functions-preserve-supports}}}{\leq} |\supp(x)| 
     \leq \dim(S) \]

    Define  $e := \pi(e')$ and let us show that it an idempotent that is good for~$s_1$:
    \begin{itemize}
        \item $\pi(e')$ is idempotent, because
              $e'$ is idempotent and the product in $S$ is equivariant;
        \item the inclusion $\supp(\pi(e')) \subseteq \supp(J(x)) \cup \supp(\bar a)$
              follows from the choice of $\pi$;
        \item $\pi(e') \in J(x)$ because $\pi$ is a $\supp(J(x))$-permutation.
    \end{itemize}
\end{proof}
Now let us show how to construct $e_1$ in the first position.
First, we use a classical Mealy machine to underline the first letter. 
Then, we use a homomorphism to apply a uniformization of the 
following relation (Lemma~\ref{lem:straight-uniformisation})
to the underlined element of the input word:
\[
    R(x, y) \Leftrightarrow \begin{cases}
        y \textrm{ represents an idempotent that is good for } x&\\
        y = \bot \textrm{ and there is no idempotent that is good for } x&
    \end{cases}
\]

\item \label{step:prop-e1} In the second step, we want to propagate $e_1$ throughout the word:
      \[
          \begin{tabular}{ccccccc}
             $s_1$ & $s_2$ & $s_3$ & $s_4$ & $s_5$ & $s_6$ & $\dashv$\\
             $e_1$ & $e_1$ & $e_1$ & $e_1$ & $e_1$ & $e_1$ & $\square$\\
          \end{tabular}
       \]
       We do this in two substeps: First, we use an $\bar a$-supported 
       homomorphism based on the function $x \mapsto (x, \bar a)$ to equip 
       every input position with $\bar a$:
       \[
          \begin{tabular}{ccccccc}
             $s_1$ & $s_2$ & $s_3$ & $s_4$ & $s_5$ & $s_6$ & $\dashv$\\
             $e_1$ & $\bar a$ & $\bar a$ & $\bar a$ & $\bar a$ & $\bar a$ & $\square$\\
          \end{tabular}
       \] 
       Then, we propagate $e_1$ to every position. Notice that, 
       since $J(s_i) = J(s_1)$, and $\supp(J(s_i)) \subseteq \supp(s_i)$,
       it follows that all the atoms from $e_1$ are already present in each position:
       \[ \supp(e_1) \subseteq \supp(s_i, \bar a)\tdot\]
       The following lemma says that this is enough to perform a multiple-use propagation of $e_1$:
       \begin{lemma}
        \label{lem:multiple-use-propagation}
        For every polynomial orbit-finite $X$, the following \emph{conditional multiple-use propagation}
        function can be constructed as a composition of single-use primes:
        \[f_\textrm{prop} : X^* \to X^* \]
        Given an input word $w = x_1, \ldots, x_n$, the function $f_\textrm{prop}(w)$ replaces each $x_i$ with $x_1$,
        provided that $\supp(x_1) \subseteq \supp(x_i)$ for every input position $i$.  
        If there is at least one $i$, for which $\supp(x_1) \not\subseteq \supp(x_i)$, then 
        the output is unspecified. Here is an example (for $X = \atoms^2$):
        \[
            \begin{tabular}{ccccc}
                $(4, 7)$ & $(7, 4)$ & $(7, 4)$ & $(4, 7)$ & $(7, 4)$\\
                         &          & \rotatebox[origin=c]{-90}{$\mapsto$} &          &\\
                $(4, 7)$ & $(4, 7)$ & $(4, 7)$ & $(4, 7)$ & $(4, 7)$            
            \end{tabular}
        \]
        It is worth pointing out that the lemma crucially depends on the support-inclusion condition.
        We have already seen in Example~\ref{ex:monoid-products-not-single-use} that 
        the unrestricted version of $f_\textrm{prop}$ cannot be implemented as a single-use 
        Mealy machine. According to Section~\ref{subsec:cmp-primes-incl-su-mealy}, this 
        implies that it cannot be implemented as a composition of primes either. 
        \end{lemma}
                \begin{proof}
                    First, let us notice that by a simple induction on the 
                    structure of $X$, we can show that it is possible to extract 
                    the supports of elements of $X$ in form of a tuple
                    (the crucial assumption is that $X$ is polynomial):
                    \begin{claim}
                    \label{claim:supp-tuple}
                        For every polynomial orbit-finite $X$ ,there is an equivariant function:
                        \[ \suppt : X \eqto \atoms^{\leq \dim X}\comma \]
                        such that for every $x$,
                        the tuple $\suppt(x)$ contains the least support of $x$
                        and all atoms in $\suppt(x)$ are distinct.
                    \end{claim} 
                    Now, let us show how to equip every position with $\suppt(x_1)$.
                    We start with a homomorphism
                    that equips every position $i$ with $\suppt(x_i)$.
                    Then, we use the delay function (Lemma~\ref{lem:delay-primes})
                    and a homomorphism to compute the following relation in every position:
                    \[r_i \in P(\{1, \ldots, \dim(X)\} \times \{1, \ldots, \dim(X)\})\]
                    \[ n \; r_i \; m \
                    \Leftrightarrow \
                    \substack{
                    \textrm{$n$-th atom in $\supp(x_{i-1})$ and}\\
                    \textrm{$m$-th atom in $\supp(x_i)$ are equal}} \]
                    Notice that every $r_i$ is an element of a finite set,
                    which means that we can use a classical Mealy machine to 
                    compute the composition of $r_i$'s on each prefix.
                    This way, in every position we obtain the following $\overrightarrow r_i$:
                    \[ n \; \overrightarrow{r_i} \; m \
                    \Leftrightarrow \
                    \substack{
                     \textrm{$n$-th atom in $\supp(x_{1})$ and}\\
                     \textrm{$m$-th atom in $\supp(x_i)$ are equal}} \]
                    Thanks to the values $\overrightarrow{r_i}$, we locate 
                    each atom from $\suppt(x_1)$ in $\suppt(x_i)$, so 
                    we can use a homomorphism to compute $\suppt(x_1)$ in every position.\\

                    In the second phase of the construction, we use the values
                    $\suppt(x_1)$ to compute $x_1$ in every position. 
                    The general idea is as follows:
                    \begin{enumerate}
                        \item encode every atom in $x_1$ as its position in $\supp(x_1)$,
                              obtaining an atomless value $x_1'$;
                        \item use a classical Mealy machine to propagate $x_1'$ throughout 
                              the word;
                        \item in every position, repopulate $x_1'$ with values from 
                              $\suppt(x_1)$.
                    \end{enumerate}
                    Formally, we define $x_1'$ using
                    name abstraction and atom placeholders (from Section~\ref{subsec:eliminating-redundant-atoms}).
                    Remember that $[\atoms]X$ denotes the set of elements of $X$, where one 
                    atom might have been replaced by an atomless placeholder. It comes with two operations:
                    \[ \begin{tabular}{cc}
                        $\underbrace{\abstr{a}(x)}_{
                        \substack{\textrm{replace all } a \textrm{'s in } x\\
                                  \textrm{with the placeholder}}}$ &
                        $\underbrace{x@a}_{
                                    \substack{\textrm{replace the placeholder in } x\in [\atoms]X\\
                                              \textrm{with the atom } a}}$
                        
                    \end{tabular}
                    \]
                    Let us prove the name abstraction preserves polynomial orbit-finite sets:
                    \begin{claim}
                        If $X$ is a polynomial orbit-finite set, then so is $[\atoms]X$. 
                    \end{claim}
                    \begin{proof}
                        By Lemma~\ref{lem:abstr-products-sums-words}, we know that
                        $[\atoms](X + Y) \simeq [\atoms]X + [\atoms]Y$ and 
                        $[\atoms](X \times Y) = [\atoms]X \times \atoms[Y]$, 
                        so it suffices to notice that:
                        \[\begin{tabular}{ccc}
                            $[\atoms]1 \simeq 1$ & and &
                            $[\atoms]\atoms \simeq \underbrace{\atoms}_{\textrm{real atom}} + \underbrace{1}_\textrm{placeholder}$
                        \end{tabular} \]  
                    \end{proof}
                    Let us define $[\atoms^k]X$ and $[\atoms^{\leq k}]X$ as follows:
                    \[ \begin{tabular}{cc}
                        $[\atoms^k]X =  \underbrace{[\atoms](\ldots ([\atoms] X)\ldots)}_{
                            k \textrm{ times } [\atoms]
                        } $&
                        $[\atoms^{\leq k}] = X + [\atoms]X + \ldots + [\atoms^k]X$
                       \end{tabular}
                    \]
                    We define $x_1' \in [\atoms^{\leq \dim(X)}]X$ as $\abstr{\suppt(x_1)}(x_1)$.
                    By Lemma~\ref{lem:abstr-removes-atom}, we know that $x_1'$ 
                    is equivariant (i.e. atomless). This means that we can propagate
                    $x_1'$ throughout the word using a classical Mealy machine. 
                    Now every position is equipped with both $x_1'$ and $\suppt(x_1)$. 
                    This means that can use a homomorphism to reconstruct $x_1$ 
                    in every position as $x_1'@(\supp(x_1))$.
                \end{proof}
       
       \item\label{it:xy-values} In this step, we decompose
       every $s_i$ into $x_i \cdot y_i$ such that $e_1$ is a suffix of $x_1$ and a prefix of $y_1$. 
       We do this by using a homomorphism based on the following function:
       \begin{claim}
	     Let $S$ be an orbit-finite semigroup,
	     and $h : \Sigma \to S + \bot$ be its polynomial orbit-finite representation.
         There exists a finitely supported function that $f_{\textrm{decompose}} : \Sigma^2 \fsto \Sigma^2$, that does the following:
         \begin{itemize}
	     \item \textbf{Input:}  $s, e \in \Sigma$, such that $e$ represents an idempotent from $S$, and $s$ represents an element from the same $\mathcal{J}$-class as $e$;
         \item \textbf{Output:} $x, y \in \Sigma$, such that $h(e)$ is a suffix of $h(x)$,
         $h(e)$ is a prefix  of $h(y)$, and $h(x) \cdot h(y) = h(s)$.
         \end{itemize}
        \end{claim}
\begin{proof}
	Thanks to the Lemma~\ref{lem:straight-uniformisation}, it suffices to show that for all $s, e \in S$
    there exists at least one such decomposition.
	Since $e$ is an infix of $s$, we know that $aeb = s$ for some $a,b \in S$.
    Define $x := ae$ and $y := eb$. It follows that:
	\[x y =  aeeb = aeb = s\]
    Since $e$ is clearly a suffix of $x$ and a prefix of $y$, we can pick 
    $s = x \cdot y$ as the desired decomposition.

\end{proof}

\item Next, we apply the delay function to the $y$-coordinates and use homomorphism to compute $g_i := y_{i} \cdot x_{i + 1}$:
\[
    \begin{tabular}{ccccccc}
        $x_1$ & $x_2$ & $x_3$ & $x_4$ & $x_5$ & $x_6$ & $\dashv$\\
        $\vdash$ & $y_1$ & $y_2$ & $y_3$ & $y_4$ & $y_5$ & $y_6$ \\
        $\vdash$ & $g_1$ & $g_2$ & $g_3$ & $g_4$ & $g_5$ & $\dashv$ \\  
    \end{tabular}
\]
    We say that two elements of a semigroup
    are \emph{$\mathcal{H}$-equivalent}
    if they are both prefix equivalent and suffix equivalent. 
    The important property of the $g_i$ values is that they
    are all $\mathcal{H}$-equivalent to $e_1$:
    \begin{claim}
        \label{claim:g-h-class}
        All $g_i$ are prefix and suffix equivalent to $e_1$.
    \end{claim}
    \begin{proof}
        Let us show that $e_1$ is prefix equivalent to $g_i$
        (the proof for suffix equivalence is similar).
        Thanks to Lemma~\ref{lem:green-lemma-anitchains}, 
        it suffices to show that $e_1$ is a prefix of $g_i$
        and that $g_i$ and \emph{infix} of $e_1$.
        First, observe that $e_1$ is a prefix of $y_i$, which in turn is a prefix of $g_i$.
        Thus, $e_1$ is a prefix of $g_i$. Next, to see that $g_i$ is an infix of $e_1$, 
        observe that $x_i g_i y_{i+1} = s_i s_{i+1}$. This means that $g_i$ is an infix 
        of $s_i s_{i+1}$, which, by smoothness, is an infix of $s_i$, and $s_i$ an infix of $e_1$. 
        Thus $g_i$ is an infix of $s_i$. 
    \end{proof}

\item\label{it:g-values}  In this step, we equip the last position (i.e. the position with $\dashv$)
       with a representation of the product of all $g_i$'s:
       \[ \overrightarrow g := g_1 \cdot \ldots \cdot g_{n-1}\]
       Recall that, by Claim~\ref{claim:g-h-class}, every $g_i$ is $\mathcal{H}$-equivalent to $e_i$.
       It follows that the support of each $g_i$ is equal to the support of $e_i$:
       \begin{lemma}
        \label{lem:h-idemp-eq-supp}
        Let $e \in S$ be an idempotent and let $H(e)$ be its $\mathcal{H}$-class. 
        Then, for every $x \in H(e)$, it holds that $\supp(x) = \supp(e)$. 
       \end{lemma}
       \begin{proof}
        We start the proof by citing a few results:
        \begin{itemize}
            \item \cite[Proposition 1.13]{pin2010mathematical} states that
                   in every semigroup $S$ (possibly infinite), if
                   an $\mathcal{H}$-class contains an idempotent, 
                   then it forms a subgroup. Moreover, thanks to
                   \cite[Proposition~1.4]{pin2010mathematical}
                   we know that the idempotent is the identity element
                   of this subgroup. (In particular, this means that
                   no $\mathcal{H}$-class contains more than one idempotent.)
            \item \cite[Corollary~2.17]{ley2015logics} says that 
                   in an orbit-finite semigroup $S$ all $\mathcal{H}$-classes
                   are finite\footnote{
                    It is also worth pointing out \cite[Lemma~2.14]{ley2015logics} which says 
                    that all orbit-finite groups are finite.
                    This can be seen as the reason why Krohn-Rhodes 
                    decompositions of single-use Mealy machines use only finite (i.e. atomless)
                    groups.
                }.
        \end{itemize}
        It follows that $H(e)$ is a finite subgroup of $S$ and that $e$ is the identity element of this group.
        First, let us show that $\supp(e) \subseteq \supp(x)$:
        Since $H(e)$ is a finite group, there exists an $n$ such that $x^n = e$. 
        The function $x \mapsto x^n$ is equivariant, so by Lemma~\ref{lem:fs-functions-preserve-supports},
        $\supp(e) \subseteq \supp(x)$.
        Now, let us show that $\supp(x) \subseteq \supp(e)$: By Lemma~\ref{lem:fs-functions-preserve-supports},
        $\supp(H(e)) \subseteq \supp(e)$. Now, since $H(e)$ is finite, 
        $\supp(x) \subseteq H(e)$ for every $x \in H(x)$ 
        (or otherwise the $\supp(H(e))$-orbit of $x$ would be infinite). It follows 
        that $\supp(x) \subseteq \supp(H(e)) \subseteq \supp(e)$. 
       \end{proof}
       Now, we equip every position with $\suppt(e_1)$ -- the easiest
       way to do it is to simply not forget those values from Step~\ref{step:prop-e1}, 
       but we can also use Lemma~\ref{lem:multiple-use-propagation}. Then, 
       we use a homomorphism to compute $g'_i = \abstr{\suppt(e_1)}(g_i)$.
       By Lemmas~\ref{lem:h-idemp-eq-supp}~and~\ref{lem:abstr-removes-atom},
       we know that all $g'_i$'s are atomless.
       It follows that we can use a classical Mealy machine to compute $\overrightarrow{g'} = g'_1 \cdot \ldots g'_n$
       and save it in the last letter. The binary product used to compute $g'$ is defined
       using functoriality of $\abstr{\suppt(e_1)}$ as $\abstr{\suppt(e_1)}(\_ \cdot \_)$.
       By Lemma~\ref{lem:abstr-natural}, we know that
       $\overrightarrow{g}'@\suppt(e_1) = \overrightarrow{g}$. 
       Since, at this point, the last letter contains both $\overrightarrow{g}'$ and $\supp(e_q)$,
       it follows that we can use a homomorphism to construct $\overrightarrow{g}$ in the last letter. 

\item Finally, we notice that $s_1  \cdot \ldots \cdot s_n = x_1 \cdot \overrightarrow{g} \cdot y_n$. 
      The value $\overrightarrow{g}$ is already present in the last position.
      This means that in order to compute $s_1 \cdot \ldots \cdot s_n$, we can use 
      the generalized single-use propagation (Claim~\ref{claim:gen-su-prop}) to 
      send $x_1$ and $y_n$ to the last position (which we recognize by $\dashv$), 
      and then apply a homomorphism that multiplies values $x_1$, $\overrightarrow{g}$ and $y_i$ in the last position. 
\end{enumerate}

\noindent
This leaves us with showing how to get rid of the unnecessary atoms:
\begin{lemma}
    \label{lem:primes-elim-support}
    If $f : \Sigma^* \to \Gamma^*$ can be constructed as a composition
    of finitely supported primes, then it can also be constructed as a composition
    of $\supp(f)$-supported primes. 
    In particular, if $f$ is equivariant, then it can be constructed as a composition of (equivariant) primes.
\end{lemma}
\begin{proof}
    Given a function $f : \Sigma^* \to \Gamma^*$, which can be constructed as a composition of $\alpha$-supported primes,
    and an atom $a \not \in \supp(f)$, we show how to construct $f$ as a composition of
    $(\alpha - {a})$-supported primes. this is enough to prove the lemma, because 
    we can repeat this process for every atom in $\alpha - \supp(f)$.\\

    The construction uses name abstraction (see Section~\ref{subsec:eliminating-redundant-atoms} for details).
    Consider the function:
    \[\abstr{a}f : [\atoms](\Sigma^*) \to [\atoms](\Gamma^*)\]
    Thanks to the isomorphism  $W_X : \abstr{a}(X^*) \simeq (\abstr{a}X)^*$
    from Lemma~\ref{lem:abstr-products-sums-words}, we can treat $\abstr{a}f$ as a function on words:
    \[\abstr{a}f: ([\atoms](\Sigma))^* \to ([\atoms](\Gamma))^*\]
    Let us show that $\abstr{a}f$ can be constructed as a composition of 
    $(\alpha - \{a\})$-primes:
    \begin{lemma}
    \label{lem:unnecessary-atoms-compositions-primes}
        If $f : \Sigma^* \to \Gamma^*$ can be constructed as a composition 
        of primes supported $\alpha$,
        then $\abstr{a}f : (\abstr{a}\Sigma)^* \to (\abstr{a}\Gamma)^*$
        can be constructed as a composition of primes supported by $\alpha - \{a\}$.
    \end{lemma}
    \begin{proof}
        The proof goes by induction on the construction of $f$.
        For the induction base, we assume that $f$ is a prime function.
        The only prime function, that 
        might not be equivariant, is a homomorphism $h^*$. In this case, it is not 
        hard to see that $\abstr{a}(h^*) = (\abstr{a}h)^*$ (thanks to the isomorphism 
        $W$ from Lemma~\ref{lem:abstr-products-sums-words}). This is enough, because by 
        Lemma~\ref{lem:abstr-removes-atom} we know that if $h$ is supported by $\alpha$, 
        then $\abstr{a}h$ is supported by $\alpha - \{a\}$.\\

        \noindent
        For the induction step, we first notice that thanks to Lemma~\ref{lem:abs-functor}, 
        we have that:
        \[ \abstr{a}(f \circ g) = \abstr{a}f \circ \abstr{a}g \]
        This is enough to handle the case of $\circ$-composition.
        For the $\times$-composition, we would like to show that:
        \[ \abstr{a}(f \times g) = (\abstr{a} f) \times (\abstr{a} g) \]
        The main problem is type inconsistency:
        \[
            \begin{tabular}{cc}
                $ \abstr{a}( f \times g):$ & $([\atoms](\Sigma_1 \times \Sigma_2))^* \to ([\atoms](\Gamma_1 \times \Gamma_2))^*$\\
                &\\
                $\abstr{a}f \times \abstr{a}g:$ & $
                ([\atoms]\Sigma_1 \times [\atoms]\Sigma_2)^* \to
                ([\atoms]\Gamma_1 \times [\atoms]\Gamma_2)^*$  
            \end{tabular}
        \]
        In order to solve it, we use the isomorphism from Lemma~\ref{lem:abstr-products-sums-words}, 
        which is defined as follows:
        \[ \begin{tabular}{c}
            $P : ([\atoms](X \times Y)) \simeq ([\atoms]X \times [\atoms]Y)$\\
            \\
            $P(x) = (\abstr{a}(\proj_1 (x@a)), \abstr{a} (\proj_2 (x@a)))$\\
            \end{tabular}
        \]
        Using $P^*$, we can (implicitly) cast between 
        $([\atoms](\Sigma_1 \times \Sigma_2))^*$ and $([\atoms]\Sigma_1 \times [\atoms]\Sigma_2)^*$. 
        Then it is not hard to see that:
        \[ \abstr{a}(f \times g) = (\abstr{a} f) \times (\abstr{a} g) \]
    \end{proof}
    Now, let us use $\abstr{a}f$ to reconstruct the original $f$ without using $a$.
    Remember that in Claim~\ref{claim:abstr-embedding}, we have defined an embedding:
    $\iota_X : X \to [\atoms]X$ such that:
    \[ \begin{tabular}{cc}
    $\iota_\Sigma(x) = \abstr{a}x$ & where $a \not \in \supp(x)$  	
    \end{tabular}
    \]
    This is an embedding, so it has a partial inverse $\iota_X^{-1} : [\atoms] X \to X + \bot$.
    We use it, to construct $f'$ in the following way and show that $f = f'$:
    \[ f' : \Sigma^* \transform{\iota^*} ([\atoms]\Sigma) \transform {\abstr{a} f} ([\atoms]\Gamma)^* \transform{\iota^{-1}} (\Gamma + \bot)^* \]
    Note that there is a slight type mismatch: the type of  $f$ is $\Sigma^* \to \Gamma^*$,
    and the type of $f'$ is $\Sigma^* \to (\Gamma + \bot)^*$. We can ignore this mismatch, 
    because (as we are going to show) $f'$ never returns $\bot$. 
    Let us now proceed with the proof that $f'=f$.
    It is enough to show that the following diagram commutes:\\

    \adjustbox{width=0.4\textwidth, center}{
    \begin{tikzcd}
	{\Sigma^*} && {\Gamma^*} \\
	{([\mathbb{A}]\Sigma)^*} && {([\mathbb{A}]\Gamma)^*}
	\arrow["f", from=1-1, to=1-3]
	\arrow["{\iota_\Sigma}", from=1-1, to=2-1]
	\arrow["{\langle a \rangle f}", from=2-1, to=2-3]
	\arrow["{\iota_\Gamma}", from=1-3, to=2-3]
    \end{tikzcd}
    }
    
    \vspace{0.1cm}
    \noindent
    For that, we notice that (by assumption) $a \not \in \supp(f)$, so
    $\abstr{a}f = \iota_{\Sigma^* \to \Gamma^*}(f)$. It follows that 
    we can finish the proof by applying the following claim:
    \begin{claim}
    For every $X, Y$, and $f: X~\fsto~Y$, the following diagram commutes:\\

    \adjustbox{width=0.4\textwidth, center}{
    \begin{tikzcd}
	{\Sigma^*} && {\Gamma^*} \\
	{([\mathbb{A}]\Sigma)^*} && {([\mathbb{A}]\Gamma)^*}
	\arrow["f", from=1-1, to=1-3]
	\arrow["{\iota_\Sigma}", from=1-1, to=2-1]
	\arrow["{\iota(f)}", from=2-1, to=2-3]
	\arrow["{\iota_\Gamma}", from=1-3, to=2-3]
    \end{tikzcd}
    }
    \end{claim}
    \begin{proof}
        We take a $x \in X$, and show that $\iota(f(x)) = (\iota f)(\iota x)$. 
        Let $b$ be an atom such that $b \not \in \supp(x) \cup \supp(f)$. It follows 
        by definition of $\iota$ that $\iota(f) = \abstr{b}f$, $\iota(x) = \abstr{b}x$, 
        by and (by Lemma~\ref{lem:fs-functions-preserve-supports}) $\iota(f(x)) = \abstr{b}(f(x))$.
        Using the isomorphism from Lemma~\ref{lem:abstr-fun}, we get that:
        \[ (\iota f)(\iota x) = (\abstr{b} f)(\abstr{b} x) = \abstr{b}\left(f((\abstr{b}x)\,@\,b)\right) = \abstr{b}(f(x)) = \iota(f(x)) \]
    \end{proof}
    \noindent
    This completes the proof of Lemma~\ref{lem:smooth-product}.
\end{proof}

\section{Local semigroup transductions $\subseteq$ Compositions of primes}
\label{sec:local-monoid-transformation-composition-of-primes}
In this section, we finish the proof of Theorem~\ref{thm:kr} by proving the following lemma:
\begin{lemma}
\label{lem:local-monoid-transformation-composition-of-primes}
	If $\Sigma$ and $\Gamma$ are polynomial
	orbit-finite sets and
	$f : \Sigma^* \to \Gamma^*$ is
	a local semigroup transduction, then $f$
	can be constructed as a composition of 
	primes. 	
\end{lemma}
Take any local semigroup transduction $f : \Sigma^* \to \Gamma^*$ 
given by $(S, h, \lambda)$. All semigroup transductions are equivariant, 
so by Lemma~\ref{lem:primes-elim-support}, it suffices to show 
that $f$ can be constructed as a composition of finitely supported primes.
Moreover, we can assume that 
$h : \Sigma \eqto S$ is a polynomial orbit-finite representation\footnote{
    There is a slight type mismatch here: a polynomial orbit-finite
    representations is a partial function $\Sigma \to S + \bot$,
    and $h$ is a total function $h : \Sigma \to S$. We can deal with this 
    type mismatch by defining a semigroup $S' = S + \bot$,
    where $\bot$ is an all-absorbing error element, i.e.
    $x \cdot \bot = \bot = \bot \cdot x$. Alternatively, 
    we can assume that input does not contain any letters 
    for which $h$ is undefined. 
} of $S$. This is because if 
the original function $h : \Sigma \to S$ is not
a polynomial orbit-finite representation, we can 
pick a polynomial orbit-finite representation 
$h' : \Sigma' \eqto S$ and start the construction 
by applying to every letter the uniformization (Lemma~\ref{lem:straight-uniformisation})
of the following relation (using a homomorphism prime function):
\[
\begin{tabular}{ccc}
    $x \, R \, y$ & $ \Leftrightarrow  $ & $h(x) = h'(y)$
\end{tabular}
\]
This leaves us with implementing the transduction $(S, h', \lambda)$,
which satisfies the condition that $h'$ is a polynomial orbit-finite representation.
From now on, we assume that $h$ is a polynomial orbit-finite 
representation of $S$. Thanks to this assumption, we can use Lemma~\ref{lem:smooth-split-primes}
to construct a smooth split over the input sequence. In the remainder of this 
section, we show how to transform the split into the output of the
local semigroup transduction.\\

Before describing the construction in details, let us briefly present its plan.
First, we introduce the notion of \emph{split monotonicity} (Definition~\ref{def:monotone-split}).
Then, in Lemma~\ref{lem:monotone-smooth}, we show how to use compositions 
of primes to transform smooth splits into monotone smooth splits. 
This allows us to work only with monotone splits. 
Next, after describing some failed attempts, we introduce the central notion 
of the construction: \emph{core ancestor sequence} for a position $i$ of a split (Definition~\ref{def:core-ancestor-sequence}), 
and we prove two lemmas about it.
Lemma \ref{lem:core-defines} shows that, for a monotone split, the core ancestor sequence of a position $i$ 
uniquely defines the $i$-th letter of the output. Lemma~\ref{lem:core-construct-primes} shows how to use compositions of primes
to construct the core ancestor sequence for each position of a split.
Finally, we complete the construction by combining the two lemmas.\\

We start the detailed description of the construction, by defining some additional structure on a split:
\begin{definition}
	For every position $i$ of a split, we define its \emph{left ancestor} to be the rightmost  
    $j$ to the left of $i$, that is higher or equal than $i$. Note that every position 
    has at most one left ancestor. The \emph{ancestor sequence} of $i$ is the smallest subsequence of the split positions, 
    that contains $i$ and is closed under ancestors -- i.e. the sequence that contains
    $i$, $i$'s ancestor, the ancestor of $i$'s ancestor, and so on \ldots
\end{definition}
\begin{example}
\label{ex:split-left-ancestors}
Let us consider the semigroup $S =  {{\atoms}\choose{\leq 3}} + \bot$ with the following operation:
\[ a \cdot b = \begin{cases}
	a \cup b & \textrm{if } a \neq \bot,\ b \neq \bot, \textrm{ and } |a \cup b| \leq 3\\
	\bot & \textrm{otherwise}
 \end{cases}
 \]
 Here is a sequence over this semigroup (black) and an example split over the sequence (grey).
 The orange arrows point to the left ancestor of every position:
\bigpicc{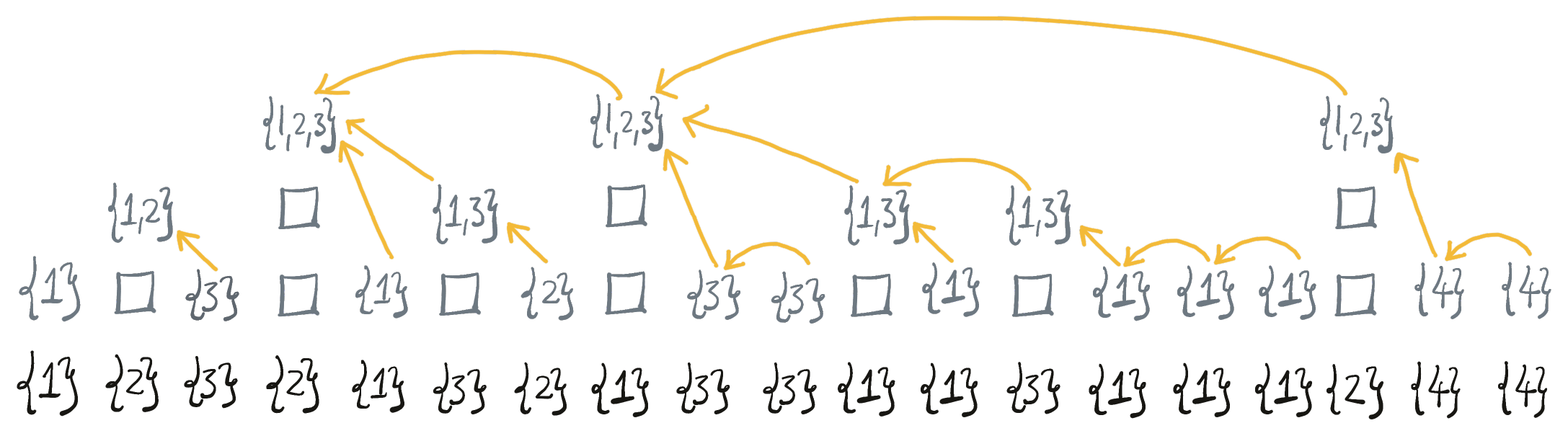}
\end{example}
Importantly, the product of the split values 
in the ancestor sequence of a position $i$ is equal to the product 
of the $i$th prefix of the input sequence. (This property follows immediately
from a definition of a split value.)
For example, consider the following split:
\bigpicc{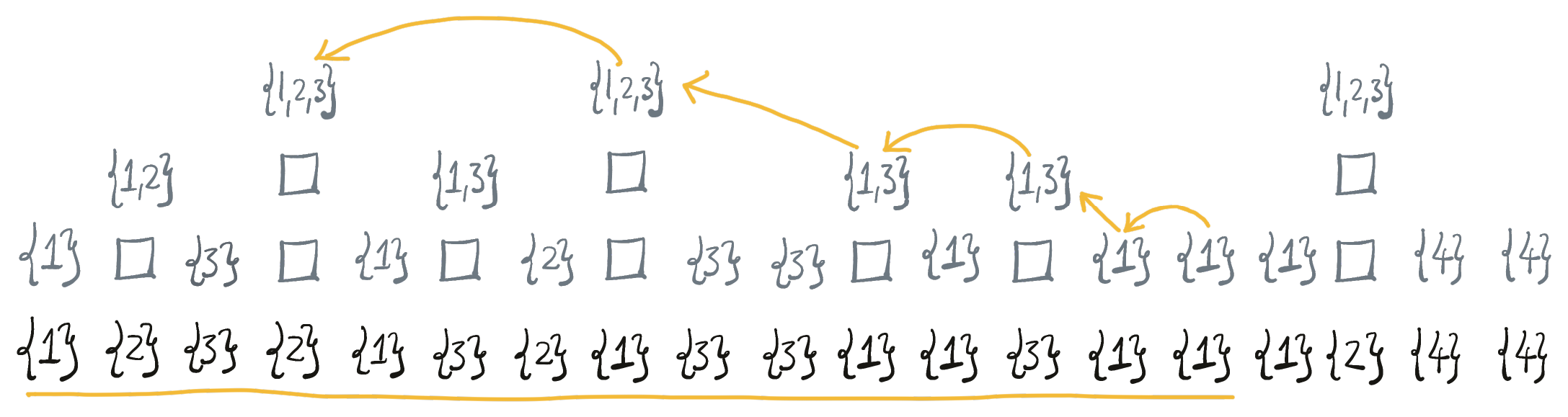}
We can calculate the product for the underlined prefix 
as the product of the highlighted ancestor sequence:
\[\{1, 2, 3\} \cdot \{1, 2, 3\} \cdot \{1, 3\} \cdot \{1, 3\} \cdot \{1\} \cdot \{1\}\]
Observe that the product of every two consecutive values in this sequence 
is $\mathcal{J}$-equivalent to the element on the left (i.e. $a_i a_{i+1}$ is $\mathcal{J}$-equivalent 
to $a_{i}$). We say that a split is $\emph{monotone}$ if all ancestor sequences
(labelled with split values) satisfy this condition:
\begin{definition}
\label{def:monotone-split}
	Let $s_1, \ldots, s_n \in S^*$ be a sequence over a semigroup  
    equipped with a smooth split and let $v_1, \ldots, v_n \in S^*$
    be the split values of this split. We say that the split is \emph{monotone}
	if for every $j, i$ such that $j$ is $i$'s left ancestor it holds that:
	\[ v_i \cdot v_j \textrm{ is infix equivalent to } v_i \]
\end{definition}
The split form Example~\ref{ex:split-left-ancestors} is not monotone, because the
following pairs do not satisfy the monotonicity condition:
\bigpicc{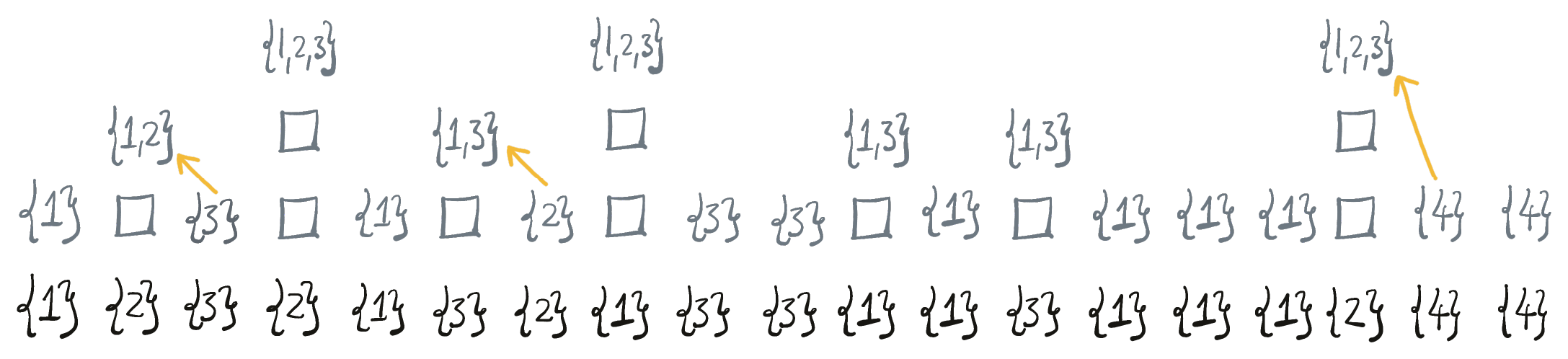}
\noindent It is, however, possible to construct another
smooth split over this sequence that is monotone (orange arrows point to ancestors):
\bigpicc{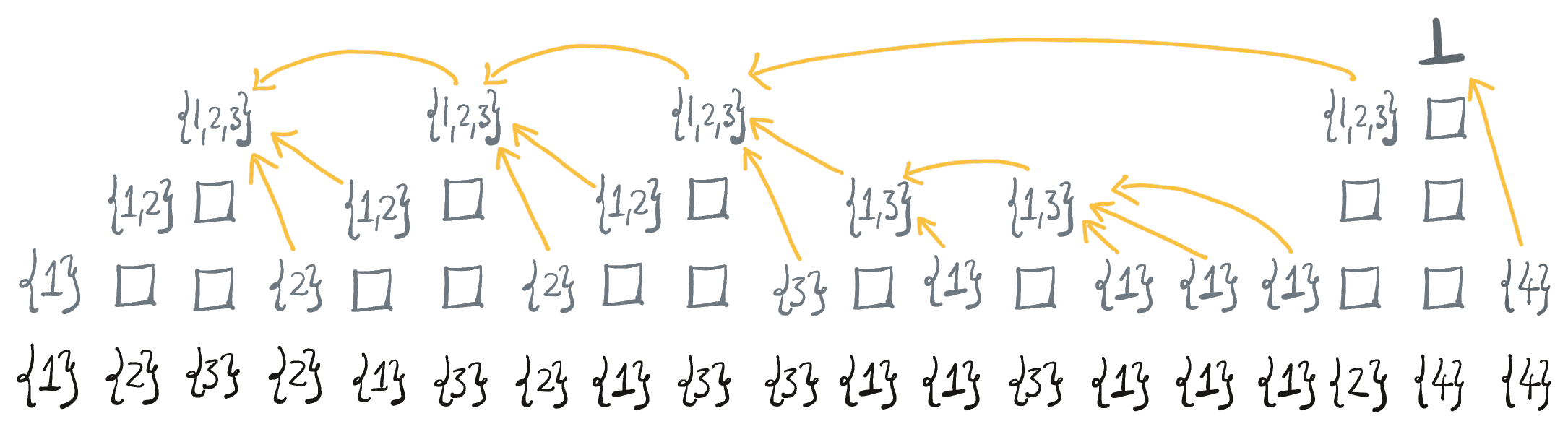}
At the end of this section, in Lemma~\ref{lem:monotone-smooth},
we will show how to transform smooth splits of bounded height into monotone smooth splits of bounded height.
For now, let us assume that the input split is monotone.\\

In order to construct the output of the local semigroup transduction, 
it would be enough to equip every position with its ancestor sequence.
Unfortunately, this cannot be achieved with compositions of primes, 
because the ancestor sequences can have unbounded lengths.
The general idea of the proof of Lemma~\ref{lem:local-monoid-transformation-composition-of-primes}
is to compress the ancestor sequences so that they
can be handled using compositions of primes, while preserving enough 
information to determine the output of the local semigroup transduction.\\

Before we continue with the proof, we need to define a few more notions:
Remember that positions $i$ and $j$ are called \emph{siblings} if
they have equal heights and there is no higher position between them.
Being siblings in an equivalence relation, and its equivalence 
classes are called \emph{sibling subsequences}. The leftmost 
position in every sibling subsequence is called the \emph{eldest sibling}. 
The \emph{sibling prefix} of a position $i$ is the prefix of the 
sibling subsequence that ends in this position. 
Finally, we say that a position of a split is
\emph{regular} if its split value is $\mathcal{J}$-equivalent 
to some idempotent in $S$. Note that if a position
of a smooth split has at least one sibling, then 
by a reasoning similar to the one in Claim~\ref{clam:good-idempotent}, 
it has to be regular.\\

\noindent
The following lemma says how to equip all sibling subsequences 
with values $x_i$, $y_i$, and $g_i$ analogous to the ones
constructed in the proof of Lemma~\ref{lem:smooth-product}:
\begin{claim}
\label{claim:xyg-values}
   Let $\Sigma$ be a polynomial orbit-finite representation of $S$. 
   There is a composition of primes that inputs a smooth split
   and equips every regular position $i$ of the split with (representations of)
   values $g_i \in S$, $y_i \in S$, and if $i$ is an eldest sibling $x_i \in S$ 
   such that for every regular $i$ and its eldest sibling $j$:
   	\begin{enumerate}
		\item $x_{j}g_iy_i$ is equal to the product of the split values of $i$'s sibling prefix;
		\item $g_{j}$ is an idempotent, and $g_i$ is $\mathcal{H}$-equivalent to $g_{j}$;
		\item $x_{j}$ is prefix equivalent to $v_j$ (i.e. the split value of $j$) and suffix equivalent to $g_j$; 
		\item $y_{i}$ is suffix equivalent to $v_i$ and prefix equivalent to $g_i$.
	\end{enumerate}
    For example, consider the following split:
    \bigpicc{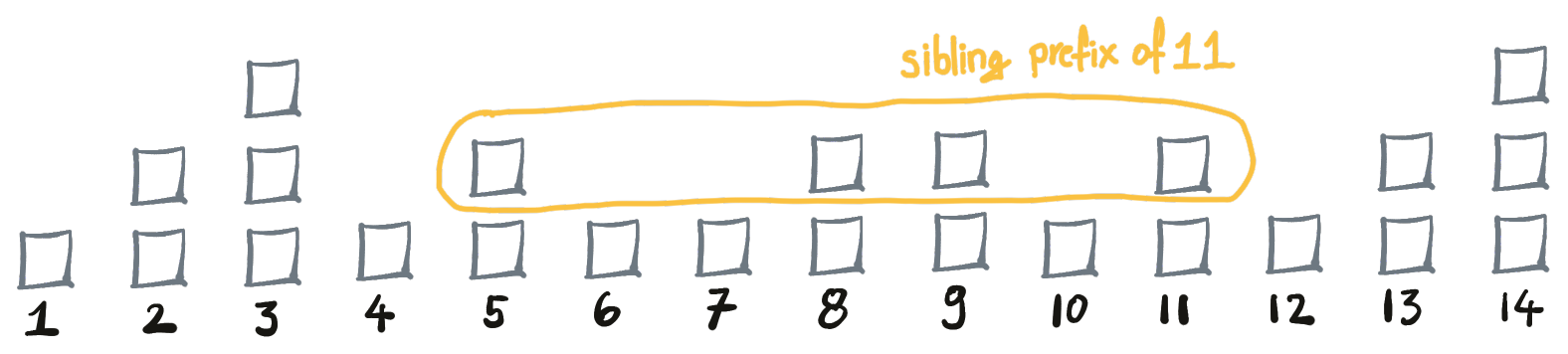}
    Highlighted is the sibling prefix of $11$. The eldest sibling of $11$ is $5$. 
    According to the lemma, this means that $x_5 \cdot g_{11} \cdot y_{11}$ is equal to the product $v_5 \cdot v_8 \cdot v_9 \cdot v_{11}$, 
    where $v_i$ denotes the split value of the position $i$. 
    Note that $v_i$ differs from $s_i$, which is the $i$-th element of the split's underlying sequence.
    Furthermore, the lemma requires that $g_5$ is idempotent and $\mathcal{H}$-equivalent to $g_{11}$.
    Also, $x_5$ is prefix equivalent to $v_5$ and suffix equivalent to $g_5$.
    Lastly, $y_{11}$ is prefix equivalent to $g_{11}$ and suffix equivalent to $v_5$.
\end{claim}
\begin{proof}
	The proof goes by induction on the split's height. 
	If it is equal to $1$, then the input is a smooth product and 
    we can use a construction similar to the one in the proof 
    of Lemma~\ref{lem:smooth-product}. The difference 
    is that in Step~\ref{it:g-values},
    we use a classical (i.e. atomless) Mealy machine to equip \emph{every}
    position $i$ with the atomless value $\overrightarrow{g_i}' = g_1' \cdot \ldots \cdots g_i'$.
    Then we can use a use homomorphism to repopulate every $\overrightarrow{g_i}'$
    with atoms from $\overline{supp}(e_1)$ (available in every position), 
    obtaining $\overrightarrow g_i  = g_1 \cdot \ldots \cdot g_i$.
    We conclude by keeping the $x_i$ and $y_i$ values constructed 
    in Step~\ref{it:xy-values} and defining $g_i$ as $\overrightarrow g_i$
    (this is not going to cause a name clash, because 
    we are not going to use the original $g_i$'s any more).\\

	For the induction step, notice that the positions of the maximal 
    height form a smooth subsequence. This means we can use 
    the construction from the induction base (combined with the subsequence 
    combinator from Lemma~\ref{lem:subsequence-combinator}), 
    to compute the values $x$, $y$ and $g$ in the positions of the maximal height. 
    Then, we observe that the positions of the maximal height divide 
    the input split into smooth splits of lower heights. For example:
	\bigpicc{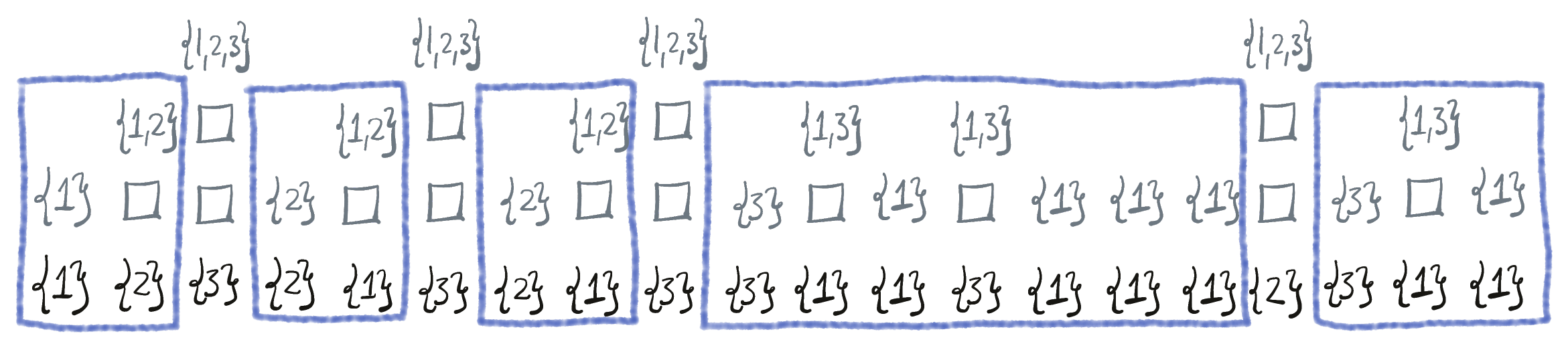}
    Thanks to this observation, we can finish the construction by combining the induction 
    assumption with the map combinator (Lemma~\ref{lem:map-combinator-primes}).
\end{proof}
Now, let us use the values $x_i$, $g_i$, $y_i$ to define the \emph{compact ancestor sequence} for every position of the split:
\begin{definition}
\label{def:compact-ancestor-sequence}
    For every position $i$ of a monotone smooth split,
    we define its \emph{compact ancestor sequence} ($\textrm{cas}(i)$) in the following way:
    \begin{enumerate}
        \item If $i$ a regular position, then $\textrm{cas}(i)$ is defined as the following tuple:
              \[ \textrm{cas}(i) =  (\textrm{cas}(p), x_j, g_i, y_i)\comma\]
              where $j$ is $i$'s eldest sibling, and $p$ is $j$'s left ancestor.
              Note that $j$ always exists, but it might happen that $i = j$. 
              It is, however, possible that $j$ might not have a left ancestor -- 
              in this case, we simply omit the $\textrm{cas}(p)$ part. 
        \item If $i$ is not regular, then $\textrm{cas}(i)$ is defined as:
              \[ \textrm{cas}(i) =  (\textrm{cas}(p), v_i)\comma\]
              where $p$ is the left ancestor of $i$. 
              Again, if $i$ does not have a left ancestor, we omit
              the $\textrm{cas}(p)$ part.
    \end{enumerate}  
\end{definition}

Notice that by Claim~\ref{claim:xyg-values} the product 
of the compact ancestor sequence is equal to the product of the
split values in the full ancestor sequence.
It follows that the product of the compact ancestor sequence of $i$
is equal to the product of the $i$th prefix of the input sequence
(i.e. $s_1\cdot \ldots \cdot s_i$). Importantly, the length 
of the compact ancestor sequence is bounded. 
\begin{claim}
\label{claim:cas-bounded-lenght}
    The length of every compact ancestor sequence 
    is bounded by $3h$, where $h$ is the height of the split. 
\end{claim}
\begin{proof}
    It suffices to see that in Definition~\ref{def:compact-ancestor-sequence}
    the position $p$ is always higher than~$i$. First, let us observe that, 
    by construction, the position $p$ is always either higher than $i$,
    or a proper sibling of $i$ (i.e. $i \neq p$). 
    This leaves us with showing that $p$ is never a sibling of $i$:
    If $i$ is regular, then we know that $p$ cannot be $i$'s sibling 
    because it is to the left if $i$'s eldest sibling $j$. 
    And if $i$ is not regular, then $p$ cannot be $i$'s sibling, 
    because the nodes that are not regular do not have 
    any siblings other than themselves 
    (see the proof of Claim~\ref{clam:good-idempotent}).
\end{proof}

Unfortunately, the compact ancestor sequence still contains too much information 
to be computed by compositions of primes --  if 
there were a composition of primes that equips every position 
$i$ with its compact ancestor sequence, then there would 
also be a composition of primes that computes the semigroup product 
of every prefix (it suffices to use a homomorphism to compute 
the product of each $\textrm{cas}(i)$). 
This is a contradiction, because by Example~\ref{ex:monoid-prefix-impossible},
we know that compositions of primes are not capable of computing products of prefixes.\\ 

In order to make the compact ancestor subsequences manageable for the 
composition of primes, we need to forget some of their atoms.
This is expressed in the language of $\alpha$-orbits. Remember that, if $\alpha$ is a finite subset of atoms, 
then the $\alpha$-orbit of an element $x$ is defined as:
\[ \orb_\alpha(x) = \{ \pi(x) \ |\ \textrm{where $\pi$ is an $\alpha$-permutation } \} \]
One can look at $\orb_\alpha(x)$ as an operation, that it forgets all the atoms from $x$
that do not belong to $\alpha$. Note that $\orb_\alpha(x)$ is usually not polynomial
orbit-finite. However, as the following lemma shows, we can use polynomial 
orbit-finite sets to represent orbits of polynomial orbit-finite sets.
\begin{lemma}
\label{lem:a-orbits-pof}
    For every $k \in \nat$ and every polynomial orbit-finite $\Sigma$, 
    there is a polynomial orbit-finite set $\orb_k(\Sigma)$ and a function:
    \[ \overline{\textrm{orb}} :  \atoms^{\leq k} \times \Sigma \to \orb_k (\Sigma)\comma\]
    such that:
    \begin{enumerate}
        \item if $\poforb(\bar a, x) = \poforb(\bar a, y)$, then $x$ and $y$ belong to the same $\bar a$-orbit; and
        \item $\supp\left(\poforb(\bar a, x)\right) \subseteq \supp(\bar a)$.
    \end{enumerate}
    The intuition behind this lemma is that $\poforb(\bar {\alpha}, x)$ is a representation 
    of $\orb_\alpha(x)$, where $\bar{\alpha}$ is any of the tuples that contains all atoms from $\alpha$. 
\end{lemma}
\begin{proof}
    We start by defining $\orb_k(\Sigma)$:
    \[ \orb_k(\Sigma) = [\atoms^{\leq k}]\Sigma = \Sigma + [\atoms]\Sigma + [\atoms^2]\Sigma + \ldots + [\atoms^k]\Sigma\]
    Now, in order to compute $\poforb(\bar a, x)$, we abstract away all the atoms that do not belong to $\bar a$:
    \[ \poforb(\bar a, x) = \abstr{\suppt(x) - \bar a}(x)\tdot\]
    (Note that $\suppt(x)$ is the support-extracting function from Claim~\ref{claim:supp-tuple}).\\

   Now let us show that this $\poforb$ satisfies the two required properties. The second one 
   (i.e. $\supp(\poforb(\bar a, x)) \subseteq \supp(\bar a)$) follows immediately from 
   Lemma~\ref{lem:abstr-removes-atom}. This leaves us with showing the first property:
   We take some $\bar a$, $x$, $y$, such that $\poforb(\bar a, x) = \poforb(\bar a, y)$
   and show that $x = \pi(y)$, for some $\bar a$-permutation $\pi$. Since
   $\poforb(\bar a, x) = \poforb(\bar a, y)$, we know that:
   \[ |\suppt(x) - \bar a| = |\suppt(y) - \bar a| \]
   Clearly, neither $(\suppt(x) - \bar a)$ nor $(\suppt(y) - \bar a)$ 
   contains repeating atoms or atoms from $\bar a$.
   This means that there is an $\bar a$-permutation $\pi$,
   that transforms $\suppt(y) - \bar a$ into $\suppt(x - \bar a)$. Let us show that this $\pi$ transforms $y$ into $x$:
   \[ \pi(y) = \pi (\poforb(a, y)@(\suppt(y) - \bar a)) = (\pi (\poforb(a, y))@(\pi(\suppt(y) - \bar a))) = \]
   \[ = (\poforb(\bar a, y))@(\suppt(x) - \bar a) = (\poforb(\bar a, x))@(\suppt(x) - \bar a) = x\] 
\end{proof}

\noindent
Let us now use $\poforb$ to define the \emph{core ancestor sequence}, which is a more compressed
version of the \emph{compact ancestor sequence}:
\begin{definition}
    \label{def:core-ancestor-sequence}
        For every position $i$, we define \emph{core ancestor sequence} ($\core(i)$) in the following way:
        \begin{enumerate}
            \item If $i$ a regular position, then $\textrm{cas}(i)$ is defined as the following tuple:
                  \[ \core(i) =  (\underbrace{\poforb\left(\suppt(g_j),\ \left(\textrm{core}(p),\ x_j\right)\right)}_
                  {
                    \substack{\textrm{A polynomial orbit-finite representation }\\
                              \textrm{of the $\supp(g_j)$-orbit of the pair $ (\core(p), x_j)$}
                             }
                  },\ g_i,\ y_i) \] 
                  where $j$ is $i$'s eldest sibling, and $p$ is $j$'s left ancestor.
                  Similarly as in Definition~\ref{def:compact-ancestor-sequence}, let us note that 
                  $j$ always exists, but it might happen that $i = j$. 
                  It is, however, possible that $j$ might not have a left ancestor -- 
                  in this case, we simply skip the $\textrm{core}(p)$ part.
            \item If $i$ is not regular, then:
                  \[ \core(i) =  (\core(p), v_i)\]
                  where $p$ is its left ancestor of $i$. 
                  Again, if $i$ does not have a left ancestor, we skip
                  the $\core(p)$ part.
        \end{enumerate}
        
\end{definition}
The intuition behind this definition is that the more distant 
an ancestor is, the more of its atoms are forgotten in $\core(i)$. 
Since $\core(i)$ does not keep all relevant atoms, we can 
no longer use it to reconstruct the product of the $i$th prefix. 
However, if the output function $\lambda : S \to \Gamma$ satisfies the locality equation,
we can use $\core(i)$ to reconstruct the $\lambda$-value of the $i$th prefix, 
i.e. $\lambda(s_1, \ldots, s_i)$:
\begin{lemma}
\label{lem:core-defines}
    For every two sequences $s_1, \ldots, s_n$, and $s_1', \ldots, s_{n'}'$ 
    over an orbit-finite semigroup $S$ that are equipped with \emph{monotone}
    smooth splits, for every $i$ and $i'$ that are positions in $s$ and $s'$, 
    and for every \emph{local} $\lambda : S \to \Gamma$, it holds that:
    \[ 
        \begin{tabular}{ccc}
            $\core(i) = \core(i')$ & $\Rightarrow$ & $\lambda(s_1 \cdot \ldots \cdot  s_i) = \lambda(s_1' \cdot \ldots \cdot s_{i'}')$
        \end{tabular}
   \]
\end{lemma}
\begin{proof}
    The proof goes by induction on the position $i$, but it requires a slight strengthening of the induction hypothesis:
    \begin{claim}
        Let $S$, $i$, $j$, $\lambda$, $s_1, \ldots, s_n$ and $s'_{1}, \ldots, s'_{n'}$
        be as in Lemma~\ref{lem:core-defines} and 
        let $z \in S$ be a value that is monotone with both $v_i$ and $v'_{i'}$, i.e.:
        \[ \begin{tabular}{ccc}
            $v_i z$ is infix equivalent to $v_i$ & and &
            $v'_{j'}  z$ is infix equivalent to $v'_{i'}$.
        \end{tabular} \]
        It follows that:
        \[ 
            \begin{tabular}{ccc}
                $\core(i) = \core(i')$ & $\Rightarrow$ & $\lambda(s_1 \cdot \ldots \cdot  s_i \cdot z) = \lambda(s_1' \cdot \ldots \cdot s_{i'}' \cdot z)$
            \end{tabular}
       \]
    \end{claim}
    \noindent
    It is easy to see that the lemma follows from the claim -- it suffices to take $S = S^1$ (from 
    Claim~\ref{claim:semigroup-to-monoid}) and $z = 1$.
    Let us now proceed with the proof of the claim. First, let us consider the case where 
    $i$ is regular. Since $\core(i) = \core(i')$ it follows that $i'$ is regular as well, and that:
    \[  (\poforb \left( \suppt(g_j), (\textrm{core}(p), x_j)\right), g_i, y_i) = (\poforb \left( \suppt(g'_{j'}), (\textrm{core}(p'), x'_{j'})\right), g'_{i'}, y'_{i'}) \]
    This means that $g_i = g'_{i'}$,  $y_i = y'_{i'}$ and
    \[\poforb \left( \suppt(g_j), (\textrm{core}(p), x_j)\right) = \poforb \left( \suppt(g'_{j'}), (\textrm{core}(p), x'_{j'})\right)\]
              By Claim~\ref{claim:xyg-values}, we know that $g'_{j'}$ is $\mathcal{H}$-equivalent to $g'_{i'}$, and that 
              $g_j$ is $\mathcal{H}$-equivalent to $g_i$. Since $g_i = g'_{i'}$, it follows that 
              $g_j$ is $\mathcal{H}$-equivalent to $g'_{j'}$. Moreover, by Claim~\ref{claim:xyg-values} 
              we know that both $g_j$ and $g'_{j'}$ are idempotent. As explained in the proof of Lemma~\ref{lem:h-idemp-eq-supp}, 
              each $\mathcal{H}$-class 
              contains at most one idempotent, which means that $g_j = g'_{j'}$.
              To simplify the notation, let us use $e$ to represent both of those values. We know that:
              \[\poforb \left( \suppt(e), (\textrm{core}(p), x_j)\right) = \poforb \left( \suppt(e), (\textrm{core}(p'), x'_{j'})\right)\]
              Thanks to Lemma~\ref{lem:a-orbits-pof}, it follows that there is a $\supp(e)$-permutation $\pi$ such that $\pi(\core(p')) = \core(p)$ and $\pi(x'_{j'}) = x_{j}$.\\

              \noindent
              We are now ready to prove that $\lambda(s_1 \cdot \ldots \cdot s_n \cdot z) = \lambda(s'_1 \cdot \ldots \cdot s'_{n'} \cdot z)$. 
              First, notice that  $s_1 \cdot \ldots \cdot s_n  = s_1 \cdot \ldots \cdot s_p \cdot x_j \cdot g_i \cdot y_i$
              (and similarly for $s'$). This means that we can proof the claim by showing that:
              \[ \lambda(\overrightarrow{s_p} \cdot x_j \cdot g_i \cdot y_i \cdot z) = \lambda(\overrightarrow{s'_{p'}} \cdot {x'}_{j'} \cdot g'_{i'} \cdot y'_{i'} \cdot z)\comma \]
              where $\overrightarrow{s_p} = s_1 \ldots s_p$ (and similarly for $s_{p'}$). We show this in two steps:
              \[ \lambda(\overrightarrow{s'_{p'}} \cdot {x'}_{j'} \cdot g'_{i'} \cdot y'_{i'} \cdot z) \stackrel{\textrm{(1)}}{=} 
                 \lambda(\pi(\overrightarrow{s'_{p'}}) \cdot \pi({x'}_{j'}) \cdot g'_{i'} \cdot y'_{i'} \cdot z) \stackrel{\textrm{(2)}}{=} 
                 \lambda(\overrightarrow{s_{p}} \cdot x_j \cdot g_{i} \cdot y_{i} \cdot z) \] 
              Let us first use the locality of $\lambda$ to show that equation (1) holds. We know that 
              $g'_{i'}$ is $\mathcal{H}$-equivalent to $e$, which means that $g'_{i'}= e\cdot g '_{i'}$. 
              This means that in order to show (1), it suffices to show that:
              \[  \lambda(\overrightarrow{s'_{p'}} \cdot {x'}_{j'} \cdot e \cdot g'_{i'} \cdot y'_{i'} \cdot z) = 
                  \lambda(\pi(\overrightarrow{s'_{p'}}) \cdot \pi({x'}_{j'}) \cdot e \cdot g'_{i'} \cdot y'_{i'} \cdot z) \]
              This is an instance of the locality equation (Definition~\ref{def:local-monoid-transduction}),
              so since $\pi$ was chosen to be a $\supp(e)$-permutation, this leaves us with showing that
              $g'_{i'} \cdot y'_{i'} \cdot z$ is an infix of $e$. This follows from
              Lemma~\ref{lem:green-lemma-anitchains} and the following facts:
              \begin{enumerate}
                \item $v'_{i'} \cdot z$ is an infix of $v'_{i'}$ (by assumption);
                \item $y'_{i'}$ is suffix equivalent to $v'_{i'}$ (by Claim~\ref{claim:xyg-values});
                \item $y'_{i'}$ is prefix equivalent to $e$ (by Claim~\ref{claim:xyg-values}, because $e = g'_{j'}$); 
                \item $g'_{i'}$ is $\mathcal{H}$-equivalent to $e$ (by Claim~\ref{claim:xyg-values}).
              \end{enumerate}
              This leaves us with showing (2): Since $\pi(x'_{j'}) = x_j$, $g'_{i'} = g_i = e$,
              and $y'_{i'} = y_i$, it suffices to proof the following:
               \[ \lambda(\pi(\overrightarrow{s'_{p'}}) \cdot x_j \cdot e \cdot y_{i} \cdot z) = 
                  \lambda(\overrightarrow{s_{p}} \cdot x_j \cdot e \cdot y_{i} \cdot z) \]
              Since $\core(i) = core(i')$ we know that
              either both $j$ and $j'$ have a left ancestor, or neither of them has one.
              If they do not have left ancestors, then the equality is trivially true 
              (as the factors $\pi(\overrightarrow{s'_{p'}})$ and $\overrightarrow{s_{p}}$ are omitted). 
              When $p$ and $p'$ exist, we prove the equality by applying the induction assumption to:
              \[
                \begin{tabular}{ccc}
                 $s = s_1, \ldots, s_p$, &
                 $s'' = \pi(s'_1), \ldots, \pi(s'_{p'})$, and
                 &$z' = x_j \cdot g_{i} \cdot y_{i} \cdot z$.
                \end{tabular}
              \]
              This leaves us with showing that:
              \begin{enumerate}
                \item $\core(p'') = \core(p)$, where $p''$ is a position in $s'' = \pi(s')$ and $p$ is a position in $s$, and
                \item $z'$ is monotone with both $v_p$ and $v''_{p''} = \pi(v'_{p'})$.
              \end{enumerate}
              First, it is not hard to see that the function $\core$ is
              equivariant with respect to the input sequence,
              so since $\pi$ was chosen to make $\pi(\core(p'))$ equal to $\core(p)$, it follows that:
              \[\core(p'') = \pi(\core(p')) = \core(p)\]
              This leaves us with showing that $v_p \cdot z'$ is infix equivalent to $v_p$ (and similarly for $v''_{p''}$).
              Since $z' = x_j \cdot g_{i} \cdot y_{i} \cdot z$, this follows from the following facts:
              \begin{enumerate}
                \item $y_i \cdot z$ is a prefix of $y_i$:
                      Thanks to Lemma~\ref{lem:green-lemma-anitchains}, 
                      this can be shown by proving that $y_i \cdot z$ is an infix of $y_i$, 
                      which follows from the assumption that $v_i \cdot z$ is an infix of $v_i$, 
                      combined with Claim~\ref{claim:xyg-values}, which states that $y_i$ is suffix equivalent to $v_i$.

                \item $x_j \cdot g_i \cdot y_i$ is a prefix of $x_i$: This 
                      is because by Claim~\ref{claim:xyg-values} combined with the smoothness 
                      of the split, we know that the product $x_j \cdot g_i \cdot y_i$ is smooth.

                \item \label{it:vpxj} $v_p \cdot x_j$ is an infix of $v_p$: Thanks to the monotonicity 
                      of the split, we know that $v_p  \cdot v_j$ is an infix of $v_p$ and 
                      by Claim~\ref{claim:xyg-values} $x_j$ is prefix equivalent to $v_j$.
                
                \item $v''_{p''} \cdot x_j$ is an infix of $v''_{p''}$:
                      This is because $v''_{p''} = \pi(v'_{p'})$ and $x_j = \pi(x'_{j'})$, 
                      which means that it suffices to show that $v'_{p'} \cdot x'_{j'}$ 
                      is an infix of $v'_{p'}$, which can be shown using an analogous 
                      argument as in Item~\ref{it:vpxj}. 
              \end{enumerate}
              This finishes the proof in the case where $i$ and $i'$ are regular. 
              In the case where $i$ and $i'$ are irregular (since $\core(i) = \core(i)$ 
              we know that $i$ and $i'$ have to be either both regular or both irregular), 
              we can directly apply the induction assumption.
\end{proof}

\noindent
It follows almost immediately from Lemma~\ref{lem:core-defines} that we can 
transform core ancestor sequences into the output letter output letters:
\begin{claim}
\label{claim:core-to-out}
    There is an equivariant function that transforms $\core(i)$ into the $i$th letter of the output
    of the local semigroup transduction.
\end{claim}
\begin{proof}
    The existence of some function that maps $\core(i)$ into the $i$th letter of the output 
    follows immediately from Lemma~\ref{lem:core-defines}. To see that this function is equivariant, 
    we can construct it as a composition of the following two equivariant relations. 
    The first relation maps $\core(i)$ to all splits with an underlined position
    such that the $\core$ of the underlined position is equal to $\core(i)$. The second relation is a function
    that  maps a split with an underlined position, into the output letter on that position. Thanks to 
    Lemma~\ref{lem:core-defines}, we know that this composition results in a function.
\end{proof}
This leaves us with showing how to equip every $i$ with its $\core(i)$.
First of all, let us notice that thanks to a reasoning similar to the one in Claim~\ref{claim:cas-bounded-lenght}
combined with Lemma~\ref{lem:a-orbits-pof}, we know that in a split of bounded height, the values of $\core(i)$
belong to a polynomial orbit-finite set. Now let us show how to use compositions of primes to compute all $\core(i)$:
\begin{lemma}
\label{lem:core-construct-primes}
For every height $h$, there is a composition
of primes that inputs a smooth
split of height at most $h$ and equips
every position $i$ with $\core(i)$.
\end{lemma}
\begin{proof}
    We start the construction by applying Claim~\ref{claim:xyg-values}
    to equip the input positions with their $x$-, $y$-, and $g$-values.
    Then, we construct the $\core$-values, level by level top-down:
    First for the positions of height $h$, then $h-1$, \ldots, and all the way down to $1$. 
    Let us assume that we have already constructed all $\core$-values for
    the positions that are strictly higher than  $k$
    (for $k=h$ this is trivially true), 
    and show how to construct $\core$-values for the positions of height $k$:\\

    First, we construct the $\core$-values for all eldest siblings of height $k$.
    For this, let us notice that if $i$ is an eldest sibling, then $\core(i)$ depends only 
    on:
    \begin{enumerate}
        \item $\core(p)$, where $p$ is the left ancestor of $i$ (if it exists); and
        \item the values $x_i, g_i, y_i$ if $i$ is regular, or on $v_i$ if $i$ is not regular.
    \end{enumerate}
    The values $v_i$, and $x_i$, $g_i$, $y_i$ (if applicable) are already
    present in all $i$'s, so it suffices to equip every eldest sibling
    with the $\core$-value of its left ancestor (if it exists).
    We do this in one round of a (generalized) single-use propagation: 
    Every node higher than $k$ transmits its $\core$-value,
    and every node of height $k$ tries to receive a value. 
    This construction works, because for every position $p$ there is at most
    one $i$ of height $k$ such that $p$ is the left ancestor of $i$.
    After this propagation, we can use a homomorphism to construct
    $\core(j)$ in every eldest sibling $j$.\\

    We are left with constructing the $\core$-values in 
    the non-eldest siblings. Notice that
    if $i$ a non-eldest sibling, then $\core(i)$ depends only on:
    \[\begin{tabular}{ccc}
        $g_i$, & $y_i$, & $\poforb\left(\suppt(g_j), (\core(p), x_j)\right)$,
    \end{tabular}\]
    where $j$ is $i$'s eldest sibling, and $p$ is the left ancestor of $j$ (if it exists). 
    Since the values $x_i$, $g_i$ are already present in $i$, 
    we only need to equip every non-eldest sibling with $\poforb\left(g_j, (\core(p), x_j)\right)$. 
    We start by using a homomorphism to construct this value in
    every eldest sibling $j$ of height $k$ (this is possible because values $x_j$
    and $\core(p)$ are already present each such eldest sibling $j$).
    Then, we notice that by Lemma~\ref{lem:h-idemp-eq-supp}, Claim~\ref{claim:xyg-values}, and Lemma~\ref{lem:a-orbits-pof},
    for every $i$ whose eldest sibling is $j$, it holds that:
    \[\begin{tabular}{ccccc}
        $\supp\left(\poforb\left(g_j, (x_j, \core(p))\right)\right)$&
        $\subseteq$&
        $\supp(g_j)$&
        $=$&
        $\supp(g_i)$
    \end{tabular}\]
    This means that we can use multiple use propagation (Lemma~\ref{lem:multiple-use-propagation})
    together with the map and subsequence combinators to equip every $i$ of height $k$
    with $\poforb\left(g_j, (x_j, \core(p))\right)$. After this step we have enough information 
    in every non-eldest node $i$ of height $k$ to use a homomorphism and compute its $\core(i)$.
\end{proof}

This almost finishes the proof of  Lemma~\ref{lem:local-monoid-transformation-composition-of-primes}.
The last thing to show is how to transform smooth splits of bounded height, into 
\emph{monotone} smooth splits of bounded height:
\begin{lemma}
\label{lem:monotone-smooth}
    For every semigroup $S$ (represented by $\Sigma$), 
    and for every height $h$,
    there is a composition of primes that inputs a smooth split of height at most $h$
    and transforms it into a monotone smooth split (over the same sequence)
    of height at most $h \cdot (h_{\mathcal{J}}(S) + 1)$,
    where $h_{\mathcal{J}}(S)$ is the $\mathcal{J}$-height of~$S$.  
\end{lemma}
\begin{proof}
The proof is an induction on $h$.
The induction base is trivial -- splits of height $1$ are necessarily smooth, 
so every smooth split of height $1$ is already monotone.
For the induction step, we assume that we have a construction for~$h$,
and we derive a construction for $h+1$. We are going to illustrate 
the construction on the following example:
\bigpicc{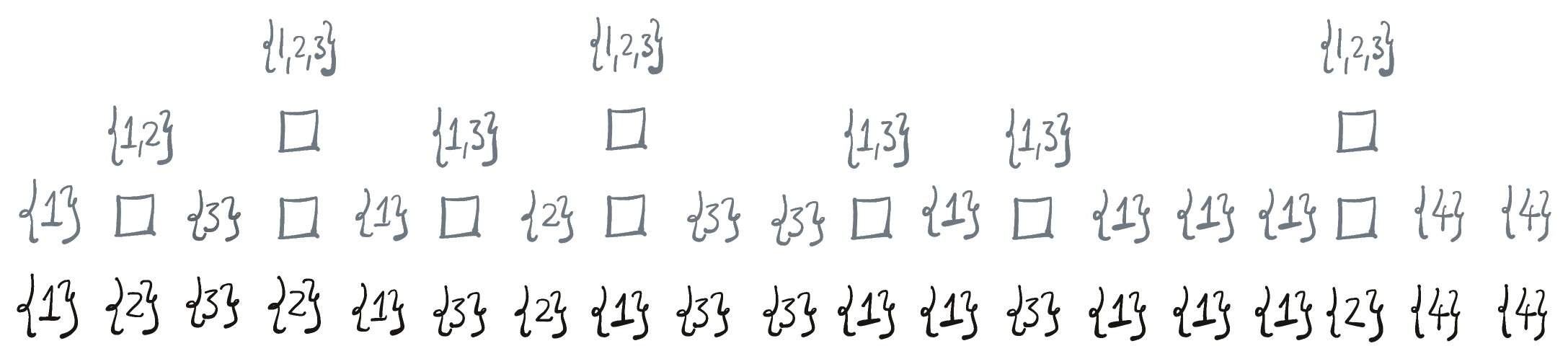}
Similarly as in the proof of Claim~\ref{claim:xyg-values}, we notice that 
the positions of the maximal height divide the input sequence into 
subsplits of heights not greater than $h$:
\bigpicc{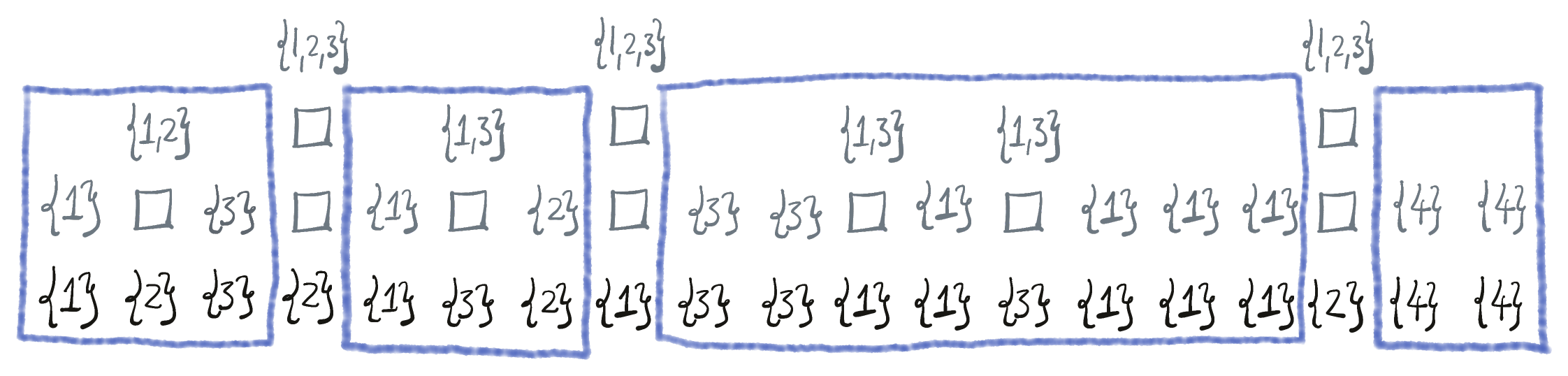}
This means that using the map combinator (Lemma~\ref{lem:map-combinator-primes}),
we can apply induction assumption to each of those subsplits.
This might increase the height of those subsplits from $h$ to 
$h' := h \cdot (h_{\mathcal{J}}(S) + 1)$.
It follows that in order to preserve the structure of the split,
we need to increase the heights of the dividing positions from $h + 1$ to $h' + 1$. 
(We can do this using a homomorphism.) In our example, this looks as follows
(for the sake of clarity, we lower $h'$ to $3$, which is still high enough to preserve the 
structure of the split):
\bigpicc{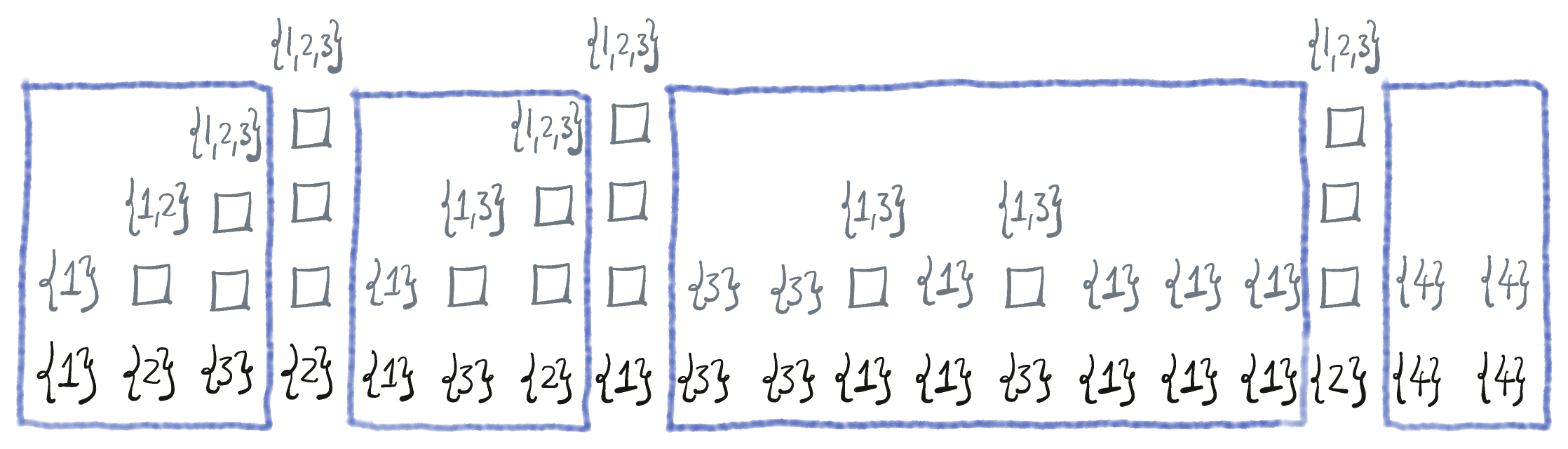}
Now let us investigate all possible remaining \emph{non-monotone positions}, 
i.e. all positions $i$,  such that $v_j \cdot v_i$ is not infix equivalent to $v_j$
(where $j$ is the left ancestor of $i$). Note that it follows from the induction assumption 
that if $(j, i)$ is such a non-monotone pair, then $j$ is of the maximal height (i.e. $h' + 1$)
and $i$ is not of the maximal height. Moreover, this can only happen in the last subsplit
(i.e. $j$ has to be the rightmost position of maximal height):
\begin{claim}
\label{claim:non-monotone-youngest-ancestor}
    If the left ancestor of $i$ is not the \emph{youngest sibling}
    (i.e. it has a sibling to the right), then $i$ is monotone (i.e. 
    $v_j v_i$ is infix equivalent to $v_j$). 
\end{claim}
\begin{proof}
    Let $j$ be $i$'s left ancestor, and let $j'$ be the first 
    sibling of $j$ (to the right):
    \smallpicc{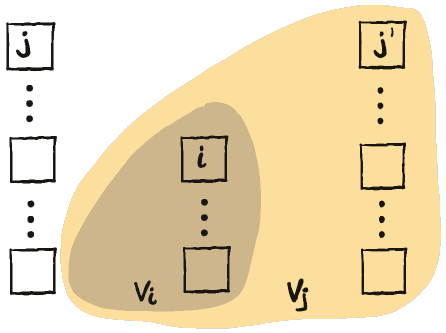} 
    Because the split is smooth, we know that $v_j v_{j'}$ 
    is an infix of $v_j$. By Lemma~\ref{lem:green-lemma-anitchains},
    it follows that $v_j v_{j'}$ is also a prefix of $v_j$. 
    By definition of the split values, we see that $v_i$ is a prefix 
    of $v_{j'}$. It follows that $v_j v_i$ is a prefix of 
    $v_j v_j'$. By transitivity, 
    this means that $v_j  v_i$ is a prefix of $v_j$.  
\end{proof}
It is easy to see that for every height $k$ and for every position $j$,
there is at most one $i$ of height $h$ such that $j$
is the left ancestor of $i$. It follows that there are at most $h'$ 
non-monotone positions. Let us show how to detect them:
We start by equipping every position $i$ whose left ancestor $j$ 
is of the maximal height (i.e. $h' + 1$) with the value $v_j$.
This can be done in $h'$ rounds of the generalized single-use propagation -- 
in the $k$th round, nodes of height $h' + 1$ transmit their split value
and nodes of height $k$ try to receive. This brings us to the following situation
(the ancestor split values are marked in orange):
\bigpicc{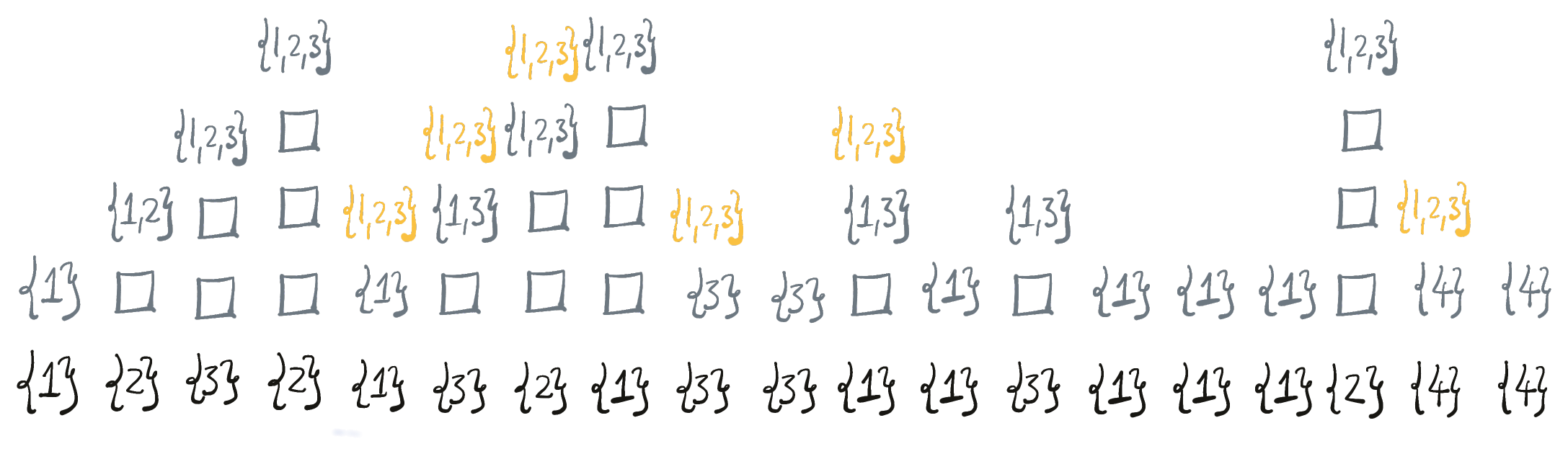}
Now we can use a homomorphism to underline all those nodes that
have received a $v_j$; and for which $v_j \cdot v_i$ is not infix equivalent to $v_j$.
In our example there is only one such position:
\bigpicc{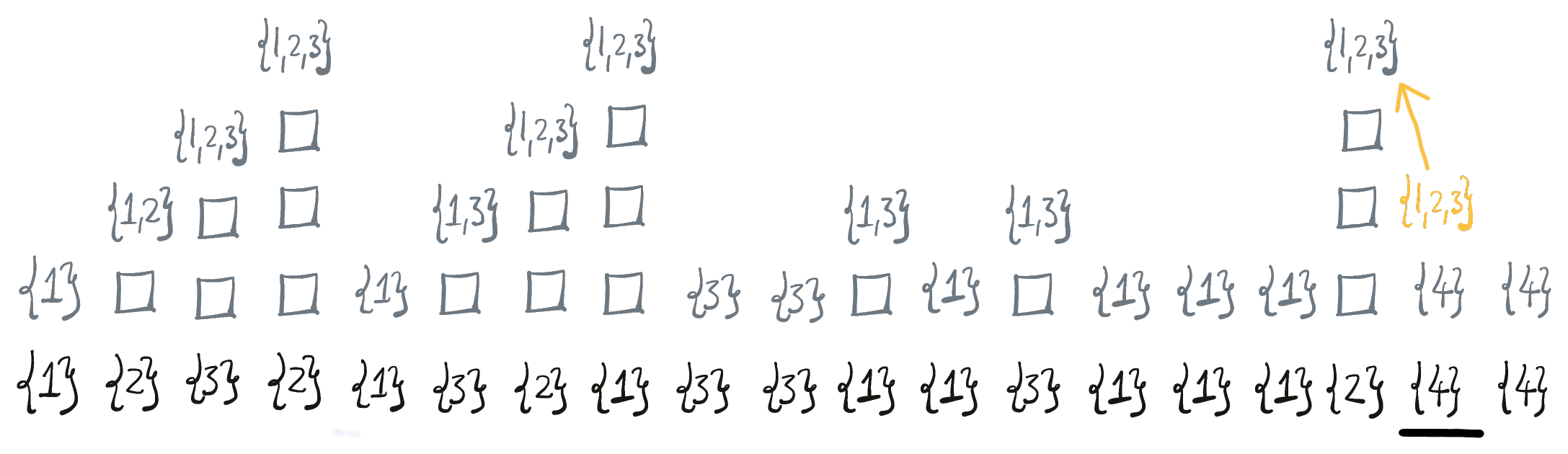}
Then, we use a classical Mealy machine (or a multiple-use bit propagation)
to detect the first underlined node, and using a homomorphism, 
we set its height to $h' + 2$ and its split value to $v_j \cdot v_i$:
\bigpicc{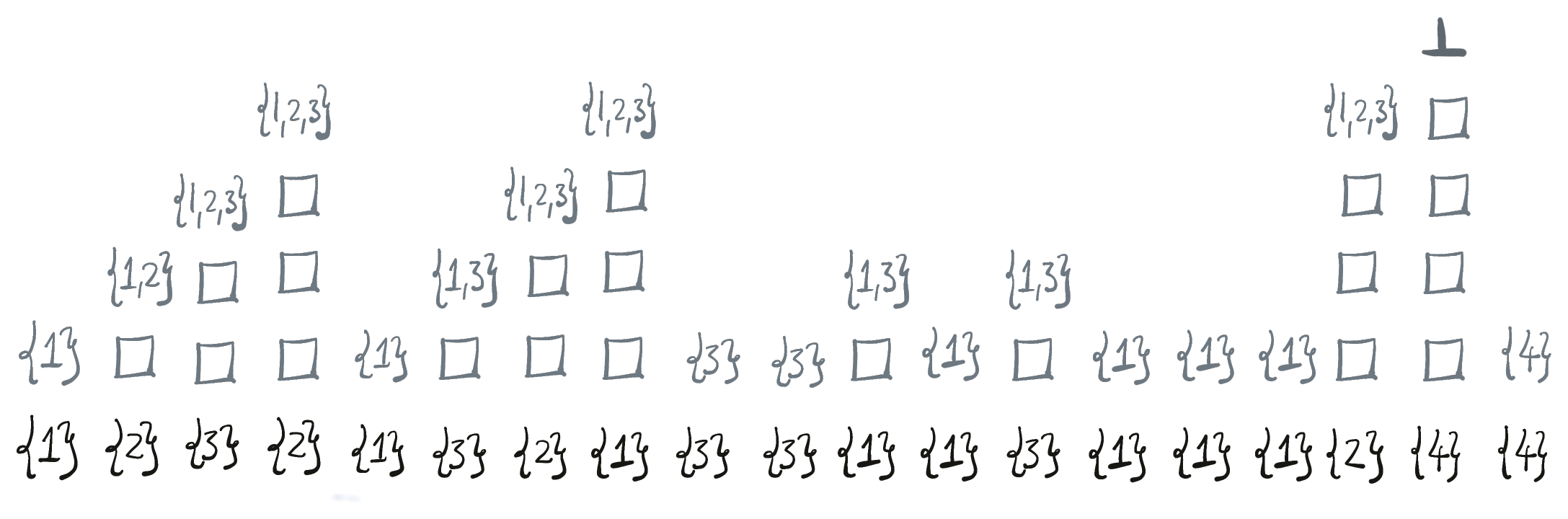}
It is not hard to see that this operation preserves the smoothness of the split. 
In our example the split is already monotone, 
but in general there might still be some non-monotone nodes left. Using a similar reasoning 
as before, we can see that if a node is non-monotone then its left ancestor has 
to be the newly elevated position of height $h' + 2$.
This means that we apply a similar construction:
\begin{enumerate}
    \item Equip every position $i$ whose left ancestor $j$ is of the maximal height with the value $v_j$;
    \item Detect the first non-monotone position (i.e. a position for which $v_j \cdot v_i$ is not infix equivalent to $v_j$);
    \item Increment the position's height, and set its split value to $v_j \cdot v_i$. 
\end{enumerate}
It is not hard to 
see that each such operation decreases the $\mathcal{J}$-height of
the split value in the position of the maximal height. This means 
that after $h_\mathcal{J}(S)$ iterations, we will obtain a monotone split 
of height at most $h' + h_{\mathcal{J}}(S)$. As $h' = h \cdot h_\mathcal{J}(S) + 1$, 
this is smaller than $(h + 1) \cdot (h_{\mathcal{J}}(S) + 1)$.
\end{proof}

\noindent
This finishes the proof of Lemma~\ref{lem:local-monoid-transformation-composition-of-primes}
and, in turn, the proof of Theorem~\ref{thm:kr}.

\section{Local semigroup transductions revisited}
\label{sec:local-semigroup-transductions-revisited}

Although local semigroup transductions were introduced as an intermediate model 
in the proof of Theorem~\ref{thm:kr}, we believe that they are of independent interest.
For this reason, we devote this (brief) section to studying them for their own sake.
The main result of this section is that the containment problem between 
general orbit-finite semigroup transductions and local semigroup transduction 
is decidable. However, this main result is mostly 
a pretext to delve deeper into the properties of local semigroup transductions.
\begin{theorem}
\label{thm:local-containment}
    The following problem is decidable:
    \[\begin{tabular}{ll}
        $\mathtt{Input: }$ & A (possibly non-local) orbit-finite semigroup transduction $(S, h, \lambda)$. \\
        $\mathtt{Output: }$ & Is there a \emph{local} semigroup transduction  $(S', h', \lambda')$\\
                            & that defines the same function?
    \end{tabular}\]
\end{theorem}
For the sake of brevity, we do not discuss finite representations of orbit-finite set 
(which is technically required to talk about the decidability of this problem). Instead, 
we will focus on the intuitive understanding of decidability for sets with atoms. 
The algorithm presented in the proof is relatively simple and we hope that it will be intuitively clear,
that it can work with any reasonable representation\footnote{
For more information on representing orbit-finite sets see \cite[Chapter~4]{bojanczyk2019slightly}.
For more information on computations on sets with atoms see \cite[Chapter~5~and~Part~III]{bojanczyk2019slightly}.}.\\

Before presenting the algorithm, let us analyse the structure of semigroup transductions.
The main result of this analysis is going to be Lemma~\ref{lem:some-local-all-local}, which says that if 
a function $f: \Sigma^* \to \Gamma^*$ is defined by \emph{some} local orbit-finite semigroup transduction, 
then \emph{all} orbit-finite semigroup transductions that define $f$ are local (as long as their semigroups 
do not contain unnecessary elements). 

\subsection{$\lambda$-morphisms}
We start by defining a suitable notion of a morphism between two semigroups with outputs:
\begin{definition}
\label{def:lambda-morphism}
    A \emph{$\Gamma$-coloured semigroup} is a pair $(S, \lambda)$, where $S$ is a semigroup, and $\lambda$ is 
    a function $S \to \Gamma$. A \emph{$\lambda$-morphism} between two $\Gamma$-coloured semigroup morphisms 
    $(S_1, \lambda_1) \to (S_2, \lambda_2)$ is a function $f : S_1 \to S_2$ such that:
    \begin{enumerate}
        \item $f$ is a semigroup morphism, i.e. $f(x\cdot y) = f(x) \cdot f(y)$ for every $x, y \in S_1$; and
        \item $f$ preserves colours, i.e. $\lambda_1(x) = \lambda_2(f(x))$ for every $x \in S_1$.
    \end{enumerate}
\end{definition}

The interesting fact about equivariant $\lambda$-morphisms is that they preserve the locality of $\lambda$
both backwards and (as long as the $\lambda$-morphism is a surjection) forwards:
\begin{lemma}
\label{lem:lambda-morphism-backwards-forwards}
    Let $(S_1, \lambda_1)$ and $(S_2, \lambda_2)$ be two orbit-finite, equivariant $\Gamma$-coloured semigroups, 
    such that there exists an equivariant $\lambda$-epimorphism (i.e. a surjective $\lambda$-morphism):
    \[f: (S_1, \lambda_1) \eqto (S_2, \lambda_2)\tdot\]
    It follows that $\lambda_1$ is local if and only if $\lambda_2$ is local. 
\end{lemma}
\begin{proof} We start with the easier $(\Leftarrow)$-implication, which holds even if $f$ not surjective.
    We assume that $\lambda_2$ is local, and show that $\lambda_1$ is local as well. For this, we take
    $x, e, y \in S_1$ and a $\supp(e)$-permutation $\pi$ such that $e$ is an idempotent and $y$ is a prefix of $e$,
    and we show that $\lambda_1(xey) = \lambda_1(\pi(x)ey)$. We start by converting the left-hand side to $\lambda_2$:
    \[\lambda_1(xey) = \lambda_2(f(xey)) = \lambda_2(f(x) \cdot f(e) \cdot f(y))\]
    At this point, we would like to apply the locality of $\lambda_2$. For this we observe the following:
    \begin{enumerate}
        \item $f(e)$ is idempotent: This is because $f(e) \cdot f(e) = f(e \cdot e) = f(e)$;
        \item $f(y)$ is a prefix of $f(e)$: Since $y$ is a prefix of $e$, there is a $z$ such that $yz = e$. 
              It follows that $f(y) \cdot f(z) = f(yz) = f(e)$. 
        \item $\pi$ is a $\supp(f(e))$-permutation: $f$ is equivariant, so by Lemma~\ref{lem:fs-functions-preserve-supports}, 
              we know that $\supp(f(e)) \subseteq \supp(e)$. It follows that every $\supp(e)$-permutation is also 
              a $\supp(f(e))$-permutation.
    \end{enumerate}
    This means that we can use the locality of $\lambda_2$:
    \[ \lambda_2(f(x) \cdot f(e) \cdot f(y)) = \lambda_2(\pi(f(x)) \cdot f(e)  \cdot f(y)) \] 
    Since $f$ is equivariant, we have that:
    \[ \lambda_2(\pi(f(x))  \cdot f(e) \cdot f(y)) = \lambda_2(f(\pi(x)) \cdot f(e) \cdot f(y))\]
    We finish the proof by going back to $\lambda_1$:
    \[\lambda_2(f(\pi(x)) \cdot f(e) \cdot f(y))  = \lambda_2(f(\pi(x)ey)) = \lambda_1(\pi(x)ey) \]
    This finishes the proof of the $(\Leftarrow)$-implication.\\

    The proof of the ($\Rightarrow$)-implication is more complicated and it requires $f$ to be surjective.
    This time we assume that $\lambda_1$ is local and show that $\lambda_2$ is local as well:
    Take $x, e, y \in S_2$ and a $\supp(e)$-permutation $\pi$ such that $e$ is an idempotent and $y$ is a prefix of $e$.
    We need to show that $\lambda_2(xey) = \lambda_2(\pi(x)ey)$. The immediate approach would be to take some $x'$, $e'$ and $y'$ that belong 
    accordingly to $f^{-1}(x)$, $f^{-1}(e)$ and $f^{-1}(y)$ and apply the locality for $\lambda_1(x'e'y')$. There are, however,
    three problems with this approach: $e'$ might not be idempotent, $y'$ might not be a prefix of $x'$, and $\pi$ might 
    not be a $\supp(e')$ permutation (as $e'$ might contain more atoms than $e$). In the next paragraphs, we 
    present a solution that deals with all three of those problems.\\

    It would be useful to assume that $\pi$ touches only finitely many atoms -- but, 
    since $\pi$ is an arbitrary $\supp(e)$-permutation, this might not be true.
    However, by Claim~\ref{claim:finite-shift}, we can assume that 
    there is another $\supp(e)$-permutation $\pi'$ that only touches 
    finitely many atoms such that $\pi'(x) = \pi(x)$. For such 
    a $\pi'$ we have that $\lambda_2(\pi(x)ey) = \lambda_2(\pi'(x)ey)$.
    This leaves us with proving that $\lambda_2(x_2ey) =  \lambda_2(\pi'(x)_2ey)$.
    We start by picking an appropriate $e'$: 
    \begin{claim}
    \label{claim:forward-locality-e'}
        There is an element $e' \in S_1$, such that: 
        \begin{enumerate}
            \item $f(e') = e$;
            \item every $s \in f^{-1}(e)$ is an infix of $e'$;
            \item $e'$ is idempotent;
            \item $\pi'$ is a $\supp(e')$-permutation.
        \end{enumerate}
    \end{claim}
    \begin{proof}
        Let us consider $E = f^{-1}(e)$.  As long as we pick $e' \in E$, it will satisfy the first condition of the claim.
        Observe that $E$ is a subsemigroup of $S_1$: This is because $\{e\}$ is a subsemigroup of $S_2$, and inverse images 
        of semigroup morphisms preserve subsemigroups. Let $E'$ be the heaviest $\mathcal{J}$-class of $E$, i.e. 
        the set of all elements whose $\mathcal{J}$-height is equal to $1$.
        (By Claim~\ref{claim:j-height-one-smooth}, we know that this is indeed a $\mathcal{J}$-class). 
        Observe that $E'$ is a subsemigroup -- this is because the $\mathcal{J}$-height of 
        $x \cdot y$ cannot be higher than the $\mathcal{J}$-height of $x$ or $y$. 
        Moreover, observe that every element of $E$ is an infix of every element of $E'$:
        if $x \in E$ and $y \in E'$, then it is not hard to see that $xy \in E'$, which
        means that $xy$ is $\mathcal{J}$-equivalent to $y$. This means that as long as we pick $e'$ from $E'$, 
        it will satisfy the second condition. Since $E'$ is a subsemigroup of $S_1$, 
        it contains at least one idempotent: Take some $x \in E'$ and let $X$ 
        be the subsemigroup generated by $x$, i.e. $X = \{x, x^2, x^3 \ldots\}$.
        By \cite[Theorem~5.1]{bojanczyk2013nominal}, we know that $X$ is finite
        and every finite (sub)semigroup contains an idempotent. Finally, 
        let us notice that $E'$ is supported by $\supp(e)$ (this follows 
        from Lemma~\ref{lem:fs-functions-preserve-supports}, as $E'$ can be 
        computed from $E$, and $E$ can be computed from $e$). Since $E'$ contains 
        an idempotent, it follows that there is at least one $\supp(e)$-orbit
        of idempotents in $E'$. Moreover, we know that $\pi'$ 
        is a $\supp(e)$-permutation that only touches finitely many atoms,
        which means that this orbit contains at least one (or, in fact, infinitely
        many) $e'$, such that $\pi'$ is a $\supp(e')$-permutation.
        Such $e'$ satisfies the conditions of the claim. 
    \end{proof}

    We are now ready to show that $\lambda_2(x_2ey) =  \lambda_2(\pi'(x)_2ey)$.
    Let us pick $e'$ as in Claim~\ref{claim:forward-locality-e'} and some 
    $x'$ and $y'$ such that $f(x') = x$ and $f(y') = y$. Observe that:
    \[ \lambda_2(xey) = \lambda_2(xe(ey)) = \lambda_2(f(x'e'(e'y'))) = \lambda_1(x' e' (ey')) \]
    In the next step, we apply the locality of $\lambda_1$ for $x'$, $e'$, $ey'$ and $\pi'$.
    We already know that $e'$ is an idempotent (in $S_1$) and that $\pi'$ is a $\supp(e')$-permutation. 
    Let us show that $e'y'$ is a prefix of $e'$: By Lemma~\ref{lem:green-lemma-anitchains},
    it suffices to show that $e'y'$ is an infix of $e'$. Since $y$ is a prefix of $e$,
    we know that there exists a $z \in S_2$ such that $yz = e$. 
    Let $z'$ be any element such that $f(z') = z$.
    It follows that $f(e'y'z') = e y z = e e = e$, so by Claim~\ref{claim:forward-locality-e'}, 
    we know that $e'y'z'$ is an infix of $e'$. It follows that $e'y'$ is an infix of $e$. 
    This means that we can apply the locality equation:
    \[ \lambda_1(x' e' (ey')) = \lambda_1(\pi'(x') \cdot  e' \cdot (ey')) \]
    We finish the proof with the following transformations:
    \[   \lambda_1(\pi'(x') \cdot  e' \cdot (ey')) = \lambda_1(\pi'(x') \cdot e' \cdot y') = \lambda_2(f(\pi(x') \cdot e' \cdot y')) =  \lambda_2(\pi'(x) e y )\]
\end{proof}

\subsection{Syntactic semigroup transduction}
Syntactic semigroups are a well-established tool for studying formal languages -- 
they were introduced\footnote{It might be worth noting that \cite{schutzenberger1955theorie} points further to \cite{dubreil1953contribution}. 
However, as I was unable to access \cite{dubreil1953contribution}, I keep \cite{schutzenberger1955theorie} as the reference.}
in \cite[Chapter~2]{schutzenberger1955theorie}, and they are discussed for example
in \cite[Section IV.4]{pin2010mathematical} or in \cite[Theorem~1.7]{bojanczyk2013nominal}. 
In this section, we discuss their generalization for word-to-word functions. We start with a definition:
\begin{definition}
\label{def:future-independence}
    We say that a length-preserving function $f : \Sigma^* \to \Gamma^*$ is \emph{future independent} 
    if for all $w, v_1, v_2 \in \Sigma^*$:
    \[ \textrm{the $|w|$-th letter of $f(wv_1)$} = \textrm{the $|w|$-th letter of $f(wv_2)$}\]
\end{definition}
It is not hard to see that the class of length-preserving, future-independent functions 
is equal to the class of functions recognized by possibly infinite semigroup transductions
(as we can always pick the free semigroup $S = \Sigma^*$ together with $\lambda$ that maps a word 
$w$ into the last letter of $f(w)$). For this reason, we are only going to define syntactic 
semigroup transduction for functions that are length-preserving and future independent.\\

When talking about syntactic semigroups, it is important to assume that the underlying semigroup 
of a semigroup transduction does not contain unreachable elements:
\begin{definition}
\label{def:full-semigroup-transduction}
    We say that a semigroup transduction $(S, h, \lambda)$ of type $\Sigma^* \to \Gamma^*$ is \emph{full}\footnote{
        An alternative for the name \emph{full} semigroup transduction might be \emph{surjective} semigroup transduction. 
        However, the name \emph{surjective semigroup transduction} might erroneously suggest that the word-to-word function $\Sigma^* \to \Gamma^*$ 
        is surjective. For this reason, we stick with \emph{full semigroup transduction}.
    }
    if $h^*$ is a surjective homomorphism: i.e. 
    if every element of $S$ corresponds to some word from $\Sigma^*$, i.e. for every $s \in S$, 
    there is a $w \in \Sigma^*$, such that: 
    \[ s = h(w_1) \cdot \ldots \cdot h(w_n) \]
\end{definition}
Note that every semigroup transduction can be transformed into an equivalent full transduction:
\begin{claim}
\label{claim:non-full-to-full}
    For every semigroup transduction $(S, h, \lambda)$, there is an equivalent 
    full-semigroup transduction $(S', h', \lambda')$. (Moreover, if $(S, h, \lambda)$ is
    equivariant, then so is $(S', h', \lambda)$.)
\end{claim}
\begin{proof}
    Define $S' \subseteq S$ to be the set of those elements that correspond to words from $\Sigma^*$. 
    It is not hard to see that $S'$ is a subsemigroup of $S$. If we define $h'$ and $\lambda'$
    to be restrictions of $h$ and $\lambda$, we obtain an equivalent semigroup transduction. 
    (It is easy to see that this construction preserves equivariance.) 
\end{proof}

\noindent
We are now ready to define the syntactic semigroup transduction for a function $f$, 
which intuitively is the minimal semigroup transduction that computes $f$:
\begin{lemma}
\label{lem:syntactic-semigroup-transduction}
    For every length-preserving and future-independent function $f : \Sigma^* \to \Gamma^*$, 
    there exists a semigroup transduction $(S_f, h_f, \lambda_f)$, callled the  \emph{syntactic semigroup transduction},  such that:
    \begin{enumerate}
        \item $(S_f, h_f, \lambda_f)$ is full;
        \item $(S_f, h_f, \lambda_f)$ computes $f$; and
        \item for every \emph{full} $(S', h', \lambda')$ that computes $f$ there exists 
              a $\lambda$-epimorphism $g : (S', \lambda') \to (S_f, \lambda_f)$.
              (Note that since $g$ is a $\lambda$-morphism, it follows that
              $(g \circ h', S_f, \lambda_f)$ is equivalent to $(S_f, h_f, \lambda_f)$,
              which means that $(g \circ h', S_f, \lambda_f)$ is an implementation of $f$.)
    \end{enumerate}
\end{lemma}
\noindent
(It is worth pointing out that the syntactic semigroup transduction does not have to be finite or orbit-finite.)
\begin{proof}
    The proof is analogous to the proof of \cite[Theorem~1.7]{bojanczyk2020languages}:
    For every $f$, we define the two-sided Myhill–Nerode equivalence relation
    $\sim_f$, which identifies two words $w_1, w_2 \in \Sigma^*$ if:
    \[ 
        \begin{tabular}{ccc}
            $\textrm{the last letter of } f(uw_1v)  =  \textrm{the last letter of } f(uw_2v)$, & 
            for all $u,v \in \Sigma^*$
        \end{tabular}
    \]
    The syntactic semigroup of $f$ is constructed as follows: $S_f = (\Sigma^+)_{/\sim_f}$
    (i.e. the set of non-empty words over $\Sigma$ divided by $\sim_f$), with the following operation:
    \[ [u]_f \cdot [v]_f  = [uv]_f \comma\]
    where $[u]_f$ denotes the abstraction class of $u$. Functions $h_f$ and $\lambda_f$ are defined 
    as follows:
    \[ \begin{tabular}{ccc}
        $h_f(w) = [w]_{f}$ & $\lambda_f([w]_f) = \textrm{the last letter of}\ f(w)$
    \end{tabular}  \]
    It is not hard to see that definitions of $\lambda$ and of the semigroup operation do 
    depend on the choice of the representatives. It is also not hard to see that
    $(S_f, h_f, \lambda_f)$ implements $f$.
    For every full $(S', h', \lambda')$ that implements $f$, we define the following $g$:
    \[ \begin{tabular}{cc}
        $g(x) = [w_x]_f$ & where $w_x$ is any word whose $h'$-image's product is equal to $x$
       \end{tabular}
    \]
    To see that $g$ is a well-defined, surjective, and a monoid morphism, we 
    can use the same argument that is used in the classical construction of a semantic 
    monoid of a language (see, for example, \cite[Theorem~1.7]{bojanczyk2020languages}).
    This leaves us with showing that $g$ preserves $\lambda$-values. Since both 
    $(S', h', \lambda')$ and $(S_f, h_f, \lambda_f)$ are implementations of $f$, it follows that:
    \[ \lambda'(s) = \textrm{the last letter of } f(w_s) = \lambda_f([w_s]_f) = \lambda_f(g(s))\]
\end{proof}

It can be shown that the syntactic semigroup transduction of $f$ is unique (with respect to $\lambda$-isomorphisms), 
which means that we can define it abstractly using only the statement of Lemma~\ref{lem:syntactic-semigroup-transduction}.
However, we stick with the concrete definition using the construction from the proof of Lemma~\ref{lem:syntactic-semigroup-transduction}
(i.e. $S_f = \Sigma^+_{/\sim_f}$) as it is sometimes easier to reason about it. For example, 
by analysing the construction, it is not hard to see that it preserves equivariance:
\begin{claim}
\label{claim:syntatic-semigroup-equivariant}
    If $f : \Sigma^* \to \Gamma^*$ is equivariant, then so is its syntactic 
    semigroup transduction $(S, h, \lambda)$. Moreover for every equivariant 
    full $(S', h', \lambda')$ that corresponds to $f$, the $\lambda$-epimorphism 
    $g: (S', \lambda') \to (S, \lambda)$ is equivariant. 
\end{claim}

In the context of formal languages, the syntactic monoid often serves as a useful tool for
checking if a language possesses certain properties (such as first-order definability).
A similar approach can be used for syntactic semigroup transductions
and the locality restriction:
\begin{lemma}
\label{lem:local-syntactic-iff}
    An equivariant function $f : \Sigma^* \eqto \Gamma^*$ is recognized by 
    some local semigroup transduction if and only if 
    its syntactic semigroup transduction is local. 
    (Remember that by Definition~\ref{def:local-monoid-transduction} every 
    local semigroup transduction is in particular orbit-finite.) 
\end{lemma}
\begin{proof}
    The ``only if'' part is immediate. Let us focus on the ``if'' part:
    Let $(S, h, \lambda)$ be the local semigroup transduction that recognizes $f$. 
    As described in Claim~\ref{claim:non-full-to-full}, we can transform
    it into an equivalent full $(S', h', \lambda')$. It is not hard to see 
    that $(S', h', \lambda')$ is local as well. 
    Let $(S_f, h_f, \lambda_f)$ be the syntactic semigroup transduction of $f$.
    It follows by Lemma~\ref{lem:syntactic-semigroup-transduction}, 
    and Claim~\ref{claim:syntatic-semigroup-equivariant} that there is
    an equivariant, surjective $\lambda$-morphism
    $k : (S_1, \lambda_1) \eqto (S_f, \lambda_f)$. Observe 
    that $S'$ is orbit-finite (because it is local).
    Since $k$ is surjective, 
    it follows that $S_f$ is orbit-finite as well.  
    By Lemma~\ref{lem:lambda-morphism-backwards-forwards}, $\lambda_f$ is local. 
\end{proof}

Interestingly, instead of checking whether the syntactic semigroup transduction of a function $f$
is local, it suffices to check whether \emph{some} full semigroup transduction that recognizes 
$f$ is local. This is formalized by the following lemma which follows directly from 
combining Lemma~\ref{lem:local-syntactic-iff} with
Lemmas~\ref{lem:lambda-morphism-backwards-forwards}~and~\ref{lem:syntactic-semigroup-transduction}:
\begin{lemma}
\label{lem:some-local-all-local}
    If $f : \Sigma^* \to \Gamma^*$ is recognized by some local semigroup transduction, 
    then all full semigroup transductions that recognize $f$ are local. 
\end{lemma}

\noindent
We are now ready to prove Theorem~\ref{thm:local-containment}.
As noted in the introduction to this section, in the proof 
we are going to use an informal intuition of what it means
to be a computable function over sets with atoms.
For a formal definition, see \cite[Chapter~8]{bojanczyk2013nominal}.
\begin{proof}[Proof of Theorem~\ref{thm:local-containment}] We are given 
an orbit-finite semigroup transduction $(S, h, \lambda)$ and we want to check 
if it can be implemented as a local semigroup transduction. Thanks to
Lemma~\ref{lem:some-local-all-local} we can do this in the following two steps:
\begin{enumerate}
    \item In the first step, we use the construction from Claim~\ref{claim:non-full-to-full} 
          to compute an equivalent full-semigroup transduction $(S', h', \lambda')$.
          This can be done using the following fixpoint algorithm.
          Initiate $S_0' = h(\Sigma)$, and keep computing $S_{i+1}'$ as follows:
          \[ S_{i+1}' = S_i' \cup \{x \cdot y \ |\ x, y \in S_i'\}\comma\]
          until $S_{i}' = S_{i+1}'$. When this process stabilizes, set $S' := S_{i}'$.\\

          To see that this procedure terminates, 
          observe that if $S_{i}$ and $S_{i+1}$ have an equal number of orbits,
          then $S_{i} = S_{i+1}$. Consequently, the running time of this procedure 
          is limited by the number of orbits in $S$. 
    \item In the second step, we check if $(S', h', \lambda')$ is local: For every tuple 
          $(x, x', y, z) \in S^4$ such that $yz$ is an idempotent and
          $x$ and $x'$ belong to the same $\supp(yz)$-orbit, we verify that:
          \[
            \lambda(x y z ) = \lambda(x'yz)
          \]
          Since $\lambda$ and $S$ are equivariant, 
          it suffices to check this condition for only one representative of every orbit in $S^4$.
          By Lemma~\ref{lem:of-preserved}, $S^4$ is orbit-finite, which means that we can do this
          in finite time.
\end{enumerate}
\end{proof}

We conclude our discussion of local semigroup transition with the following lemma, which
underlines the connection between local semigroup transduction and single-use Mealy machines:
\begin{lemma}
\label{lem:full-single-use-Mealy-local-iff}
    A full semigroup transduction over polynomial orbit-finite alphabets is 
    equivalent to some single-use Mealy machine, if and only if it is local. 
\end{lemma}
\begin{proof}
    The ``if'' part follows directly from Lemma~\ref{lem:single-use-mealy-monoid-transduction}.
    Let us focus on the ``only if'' part:
    Take a full $(S, h, \lambda)$ that is equivalent to some single-use Mealy machine. 
    It follows by Lemma~\ref{lem:single-use-mealy-monoid-transduction} that $(S, h, \lambda)$ 
    it is equivalent to some local semigroup transduction.
    By Lemma~\ref{lem:some-local-all-local} it follows that $(S, h, \lambda)$ is local itself.
\end{proof}

A corollary of Theorem~\ref{thm:local-containment} and Lemma~\ref{lem:full-single-use-Mealy-local-iff}
is that it is decidable whether an orbit-finite semigroup transduction can be translated into a single-use Mealy machine.

\section{Rational transductions with atoms and their Krohn-Rhodes decompositions}
\label{sec:rational-transductions-with-atoms}
In this section, we discuss rational transduction (i.e. the class of transductions 
defined by unambiguous Mealy machines -- see the introduction to this chapter for details)
and their possible extension to polynomial orbit-finite alphabets. The main result of this section 
is a Krohn-Rhodes-like decomposition theorem for this extension. 
Apart from being of its own significance, this result plays an important role in the next chapter.\\

One possible approach to define rational transductions over polynomial orbit-finite alphabets 
would be to use \emph{unambiguous single-use Mealy machines}.
However, unambiguity is a form of nondeterminism, 
and so far we do not have a good notion of nondeterminism compatible with
the single-use restriction (see Section~\ref{sec:non-determinsitic-single-use} for 
details\footnote{\label{ftn:unambigous}
It is worth noting that both examples
    from Section~\ref{sec:non-determinsitic-single-use},
    which demonstrate that nondeterministic single-use automata are stronger than deterministic ones,
    use automata that are ambiguous (which means that some accepted words will always have more than one accepting run).
    It follows that those examples cannot be used to show that unambiguous automata are stronger than deterministic ones.
    In fact, the question of whether unambiguous nondeterministic automata are equivalent
    to deterministic single-use automata remains open. If the two models 
    turned out to be equivalent, it would open a path to a machine-based definition of 
    single-use rational transductions.
}).
For this reason, we define rational transductions with atoms using an  algebraic approach.
Before discussing the definition for infinite alphabets, we start by recalling the classic algebraic definition\footnote{I was unable to find this definition in the literature.
However, it is consistent with the field's folklore, as it can be viewed as a semigroup version of Eilenberg's bimachine \cite[Section~XI.7]{eilenberg1974automata}.}
for rational transductions over \emph{finite} alphabets:
\begin{definition}
\label{def:rational-monoid-transduction}
    A \emph{rational semigroup transduction} of type $\Sigma^* \to \Gamma^*$
    (for finite $\Sigma$ and $\Gamma$) consists of a finite semigroup $S$,
    an input function $h : \Sigma \to S$, and an output function $\lambda$: 
    \[ \lambda : \underbrace{(S + \vdash)}_{\substack{\textrm{prefix product}\\ \textrm{$\vdash$ represents empty prefix}}} \times
                 \underbrace{\Sigma}_{\textrm{current letter}} \times
                 \underbrace{(S + \dashv)}_{\substack{\textrm{suffix product}\\ \textrm{$\dashv$ represents empty suffix}}} \to
                 \underbrace{\Gamma}_{\textrm{output letter}} \]
    The rational semigroup transduction defines the function $f : \Sigma^* \to \Gamma^*$,
    where the $i$th letter of $f(w)$ is equal to:
    \[ \lambda(h(w_1) \cdot \ldots \cdot h(w_{i-1}),\ w_i,\ h(w_{i+1}) \cdot \ldots \cdot h(w_{n})) \]
\end{definition}
\begin{example}
\label{ex:swap-first-last-rational}
    For example, let us construct a rational semigroup transduction of type $\Sigma^* \to \Sigma^*$ that
    defines the function:
    \[ \textrm{``Swap the first and the last letter''}\]
    The transduction is based on the semigroup $S = \Sigma^2$, with 
    the following operation:
    \[ (x_1, y_1) \cdot (x_2, y_2) = (x_1, y_2) \]
    The input and output functions of the transduction are as follows:
    \[
        \begin{tabular}{cc}
            $h(a) = (a, a)$ &
            $\lambda(p, a, s) = \begin{cases}
                y_s & \textrm{if } p =\, \dashv \textrm { and } s = (x_s, y_s)\\
                x_p & \textrm{if } s =\, \vdash \textrm { and } p = (x_p, y_p)\\
                a   & \textrm{otherwise}
            \end{cases}$
        \end{tabular}
    \]
\end{example}
For finite alphabets, rational semigroup transductions define the class of rational transductions:
\begin{lemma}
    \label{lem:rational-equiv-unambigous}
    Rational semigroup transductions define the same class of functions as unambiguous Mealy machines. 
\end{lemma}
\begin{proof}
    $\subseteq$: Observe that a rational semigroup transduction can be computed by a composition of 
    three deterministic right-to-left and left-to-right Mealy machines --
    the first one (right-to-left) computes the semigroup product of each suffix, 
    the second one computes the product of each prefix (left-to-right), and 
    the third one (right-to-left or left-to-right) computes the output letters.
    This finishes the proof, because 
    by reasoning similar as in Lemma~\ref{lem:su-mealy-compose} unambiguous Mealy 
    machines are closed under compositions with both left-to-right and right-to-left 
    Mealy machines.\\

    $\supseteq$: In order to translate an unambiguous Mealy $\mathcal{A}$, 
    we can use the following semigroup of behaviours. The behaviour of 
    a word $\Sigma^*$ is the following relation $b_w \subseteq Q \times Q$: 
    \[ \begin{tabular}{ccc}
        $(q_1, q_2) \in b_w$ & $\Leftrightarrow$ & $\substack{\mathcal{A} \textrm{ has a run over } w\\
                                                \textrm{that enters in } q_1 \textrm{ on the left,}\\
                                                \textrm{and exits in } q_2 \textrm{ on the right}}$
    \end{tabular} \]
    Similarly, as it was for the behaviour functions, the behaviour of concatenation is a composition of behaviours:
    $b_{uv} =  b_v \circ b_u$. This means that the set of all behaviours forms a semigroup. It is not hard to see that 
    the $i$-th output letter can be computed based on the behaviour of the $(i-1)$-th prefix, 
    the behaviour of the $(i + 1)$-th suffix, and the $i$-th input letter (the unambiguity restriction guarantees 
    that there is only one possible output letter). 
\end{proof}

Now let us define \emph{local rational semigroup transductions}, which is an extension
of rational semigroup transductions to orbit-finite alphabets. To the best of my knowledge, 
this definition is an original contribution of this thesis.
Similarly as in Definition~\ref{def:local-monoid-transduction}, the key idea is to restrict 
the power of $\lambda$ with a locality equation:

\begin{definition}
\label{def:local-rational-monoid-transduction}
    A $\emph{local rational semigroup transduction}$ is a version of rational semigroup transduction, 
    where $\Sigma$, $\Gamma$, and $S$ are orbit-finite, $h$ and $\lambda$ are equivariant, 
    and $\lambda$ satisfies the following locality equation
    for all $x_1, x_2, y_1, y_2, e \in S$ such that $e$ is an idempotent,
    for all $a \in \Sigma$, and for all $\supp(e)$-permutations $\pi$:
    \[ e = y_1 h(a) x_2 \ \ \Rightarrow \ \ \lambda(x_1 \cdot e \cdot y_1,\ a,\ x_2 \cdot e \cdot y_2) = \lambda(\pi(x_1) \cdot e \cdot y_1,\ a,\ x_2 \cdot e \cdot \pi(y_2)) \]
    A local rational monoid transduction defines a function $\Sigma^* \eqto \Gamma^*$
    (defined in the same way as in Definition~\ref{def:rational-monoid-transduction}).
    It is worth pointing out that the locality restriction does not restrict the values 
    $\lambda(\vdash, \cdot, \cdot)$ or $\lambda(\cdot, \cdot, \dashv)$. The intuitive reason for this is that 
    those values are computed only once,
    and only repetitive behaviours create obstacles 
    for the single-use restriction.
\end{definition}
\begin{example}
\label{ex:local-rational-monoid-transduction}
The transduction ``swap the first and the last letter'' from Example~\ref{ex:swap-first-last-rational}
is a local rational semigroup transduction for every orbit-finite alphabet $\Sigma$.
Since the locality restriction only talks about situations where both prefix and suffix 
are real semigroup elements (and not $\dashv$ or $\vdash$), the locality restriction is trivially satisfied:
\[ \lambda(x_1 \cdot e \cdot y_1, a, x_2 \cdot e \cdot y_2) = a = \lambda(\pi(x_1) \cdot e \cdot y_1, a, x_2 \cdot e \cdot \pi(y_2))\tdot\]
\end{example}

\noindent
Finally, let us explore the relationship between semigroup transductions and  local rational semigroup transductions:
\begin{claim}
\label{claim:rational-future-independent-one-way}
    The class of local rational semigroup transductions that are future independent
    (see Definition~\ref{def:future-independence})
    is equal to the class of local semigroup transactions.\\
\end{claim}
\begin{proof}
    We start the proof by observing that a rational transduction is future independent
    if and only if for every $x$, $a$, $y$, $y'$, it holds that:
    \[ \lambda(x, a, y) = \lambda(x, a, y')\]
    With this observation, we are ready to prove the claim.
    The $\subseteq$ inclusion is easy: In order to transform a local
    semigroup transduction into a local rational semigroup transduction,
    we can use the following output function:
    \[ \lambda'(x, a, y) = \lambda(x \cdot h(a))\]
    It is easy to see that $\lambda'$ is future independent. Moreover, locality of 
    $\lambda$ implies the locality of $\lambda'$. This is because
    if $y_1 \cdot h(a) \cdot y_2 = e$, then $y_1 \cdot h(a)$ is a prefix of $e$.\\

    The proof of $\supseteq$ is similar. The main difference is that we need to define 
    another semigroup $S'$ that keeps track of the last letter of a word. We define it as
    $S' = (S + 1) \times \Sigma$ with the following operation:
    \[ (x_1,\ a_1) \cdot (x_2,\ a_2) = (x_1 \cdot h(a_1) \cdot x_2,\ a_2) \]
    Now, we define $h'(a) = (1, a)$ and $\lambda'((x, a)) = \lambda(x, a, \dashv)$
    (or $\lambda(\vdash, a, \dashv)$ if $x = 1$).
    Thanks to the future independence of $\lambda$, we know that $(S', h', \lambda')$ defines 
    the same function as $(S, h, \lambda)$. Moreover, using a similar idea as in the 
    proof of $\subseteq$, we can show that $\lambda'$ satisfies the locality equation.
\end{proof}

\subsection{Rational Krohn-Rhodes decompositions}
In this section, we formulate and prove a version of the Krohn-Rhodes theorem for 
local semigroup transductions. Observe that all classical prime functions 
(see Theorem~\ref{thm:classical-kr}) and single-use prime functions (see Theorem~\ref{thm:kr}) 
except for the homomorphisms are left-to-right oriented. 
The rational version of Krohn-Rhodes theorem extends the set of prime functions
with their right-to-left counterparts. Let us start with the classical version of the 
theorem for finite alphabets\footnote{
    The theorem belongs to folklore. It follows immediately from the Elgot-Mezei theorem
    (\cite[Theorem~7.8]{elgot1963two}, see introduction to this chapter for details) 
    combined with the Krohn-Rhodes theorem (Theorem~\ref{thm:classical-kr}). 
}:
\begin{theorem}
\label{thm:classical-rational-kr}
    The class of rational transductions (over finite alphabets) is equal to the smallest 
    class closed under $\circ$ and $\times$, that contains the following \emph{rational prime functions}:
    \begin{enumerate}
        \item Length-preserving homomorphism $h^* : \Sigma^* \to \Gamma^*$, for every $h : \Sigma \to \Gamma$,
              where $\Sigma$ and $\Gamma$ are finite.
        \item Left-to-right multiple-use bit propagation (from Example~\ref{ex:flip-flop}) and 
              right-to-left multiple-use bit propagation (analogous);
        \item The $G$-prefix function (from Example~\ref{ex:monoid-prefix}) and the $G$-suffix\footnote{Actually, 
        it can be shown that we do not need the $G$-suffix function, as we can derive it from the other rational primes.
        However, the $G$-suffixes function does not cause any problems later on,
        and keeping it makes the formulation more symmetrical.} function (analogous)
              for every finite group $G$. 
    \end{enumerate}
\end{theorem}


The following Krohn-Rhodes theorem for local rational semigroup transductions is the 
main result of this section. To the best of my knowledge, it is an original 
contribution of this thesis.
\begin{theorem}
    \label{thm:rational-kr}
        The class of local rational semigroup transductions over polynomial orbit-finite alphabets\footnote{
            \label{ftn:kr-semigroups-general-rational}
            Note that local rational semigroup transductions are defined for all orbit-finite alphabets,
            but the theorem only holds for polynomial orbit-finite alphabets.
            A counterexample is the single-use propagation 
            of $\binom{\mathbb{A}}{2}$,
            which can be constructed as a local rational semigroup transduction but not as a composition of single-use rational primes.
            One way to address this issue would be to extend the set of single-use rational primes with the generalized
            single-use propagation for every orbit-finite $\Sigma$
            (i.e. an extended version of the function from Claim~\ref{claim:gen-su-prop}). However, 
            the current proof of the theorem only works for polynomial orbit-finite alphabets,
            leaving the question of whether compositions of these generalized primes are equivalent to local rational semigroup transductions over orbit-finite alphabets open.
            A similar (but simpler) open question can also be asked about Krohn-Rhodes decomposition of local semigroup transductions.
        }
        is equal to the smallest class closed under $\circ$ and $\times$,
        that contains the following \emph{single-use rational prime functions}:
        \begin{enumerate}
            \item All rational prime functions over finite alphabets (from Theorem~\ref{thm:classical-rational-kr});
            \item Length-preserving equivariant homomorphism $h^* : \Sigma^* \to \Gamma^*$, for every equivariant
                  $h : \Sigma \eqto \Gamma$, where $\Sigma$ and $\Gamma$ are polynomial orbit-finite;
            \item Left-to-right single-use atom propagation (Example~\ref{ex:su-propagation})
                  and right-to-left single-use atom propagation (analogous).
        \end{enumerate}
\end{theorem}

\noindent
The remainder of this section is dedicated to the proof of Theorem~\ref{thm:rational-kr}. 

\subsubsection{Local rational transductions $\subseteq$ Compositions of rational primes}
We begin the proof of Theorem~\ref{thm:rational-kr} by showing that all local rational semigroup transductions
can be constructed as compositions of (single-use rational) primes.
Thanks to Theorem~\ref{thm:kr},
it is enough to show that we can construct every local rational monoid transduction as 
a composition of left-to-right and right-to-left single-use Mealy machines\footnote{
    \label{ftn:su-rational-two-mealy}
    For finite alphabets, it was enough to use one left-to-right 
    and one right-to-left Mealy machine. In contrast, in the proof of Theorem~\ref{thm:rational-kr}, 
    we are going to use multiple left-to-right and multiple right-to-left single-use Mealy machines.
    It remains an open question whether every local rational semigroup transduction can be constructed using
    one left-to-right and one right-to-left single-use Mealy machine.
}. The proof follows the approach of Lemma~\ref{lem:local-monoid-transformation-composition-of-primes}:
First, we construct a smooth factorization of the input sequence, and then we show how to transform 
it into the output of the local rational semigroup transduction. This time, leveraging the more powerful computation model,
we do not construct smooth splits. Instead, we directly
construct smooth factorization trees using the following
\emph{right-aligned encoding} of trees. We illustrate the encoding using the following example:
\picc{Simons-tree}
The key idea of the encoding is to write every node in the position of its rightmost descendant.
To make things cleaner, we also assume that all leaves are on the same depth -- for this reason,
we introduce unary nodes. Here is a right-aligned version of the example above (with the unary nodes marked as triangles):
\picc{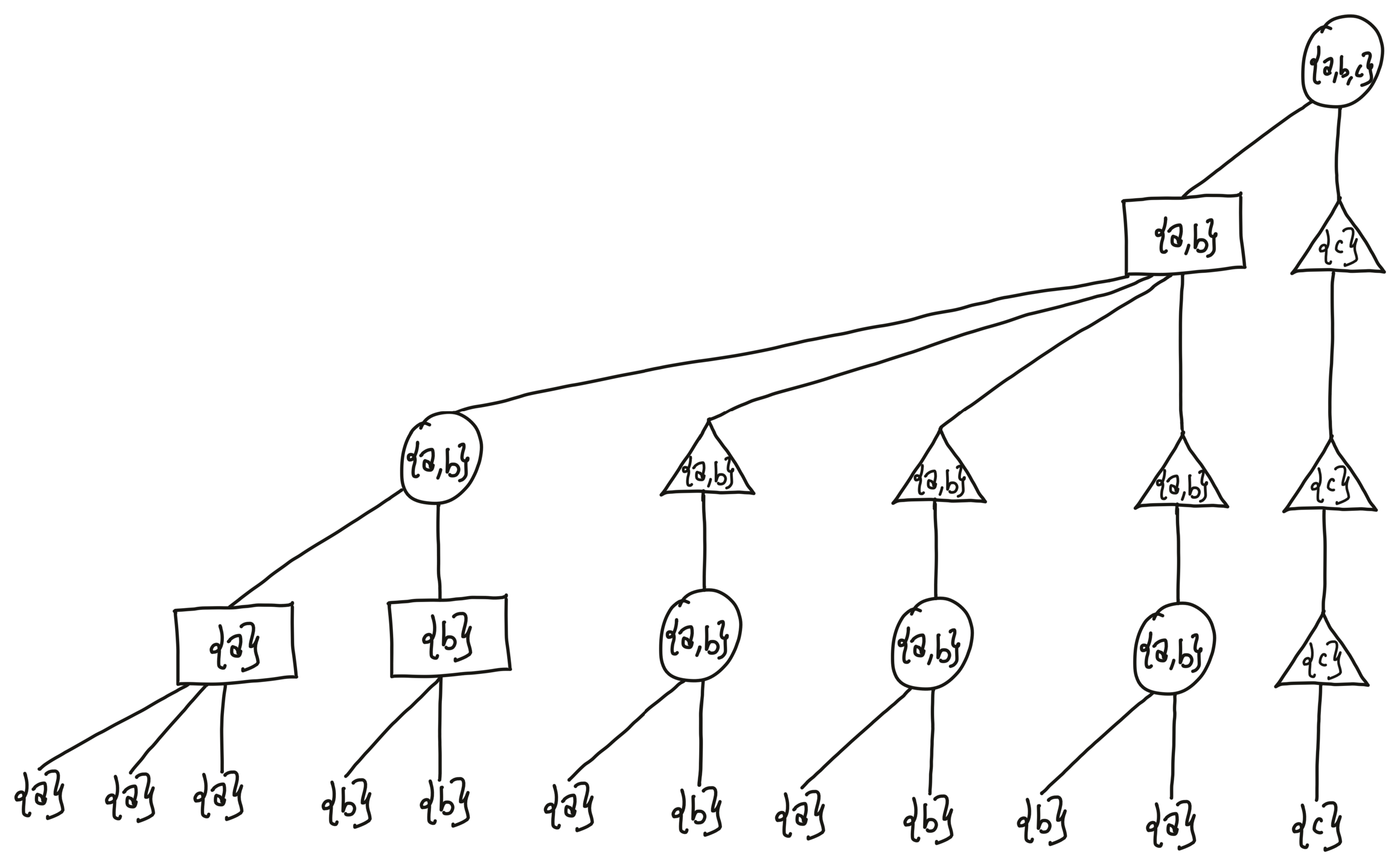}
Observe that in a right-aligned tree of height $h$, each position contains at most $h$ nodes.
It follows that we can encode right-aligned trees as words over the following alphabet
(where $\Sigma$ is a polynomial orbit-finite representation of the semigroup):
\[ {\left( \underbrace{\Sigma}_{\textrm{leaf}} + \underbrace{\Sigma}_{\textrm{binary node}} + \underbrace{\Sigma}_{\textrm{smooth node}} + \underbrace{\Sigma}_{\textrm{unary node}}\right)}^{\leq k} \]
Our example tree corresponds to the following word:
\bigpicc{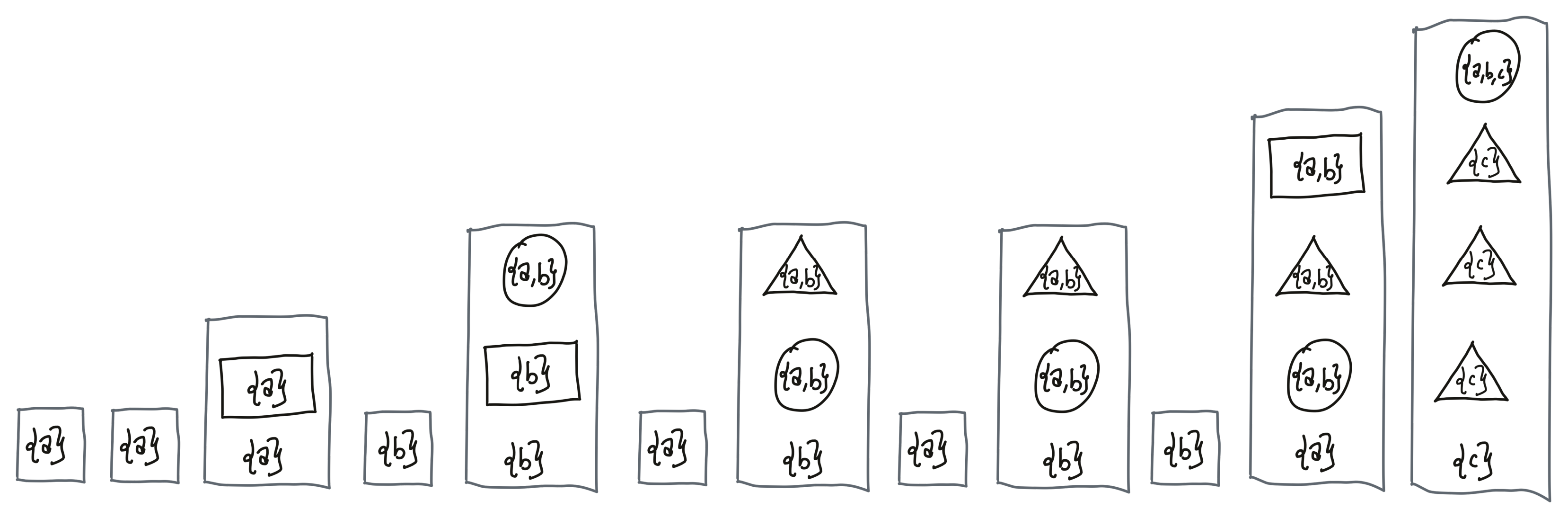}
Such a word contains enough information to recreate the original tree:
the parent of a node $q$ is the first node to the right of $q$ whose height 
is $h_q + 1$ (where $h_q$ is $q$'s height).\\

Let us now prove that we can use compositions of rational primes to construct smooth factorization trees:
\begin{lemma}
\label{lem:smooth-factorisation-tree-rational}
    Let $S$ be an orbit-finite semigroup, and let $\Sigma$ be its polynomial orbit-finite 
    representation. There exists $h$ and a function $f_{\textrm{tree}}$ of the following type
    that can be constructed as a composition of rational single-use primes:
    \[ f_\textrm{tree} : \Sigma^* \to (\underbrace{(\Sigma + \Sigma + \Sigma + \Sigma)^{\leq k}}_{
        \substack{\textrm{The alphabet for representing}\\\textrm{right-aligned smooth factorications}\\\textrm{(described earlier)}}
    } )^*\comma\]
    such that $f_\textrm{tree}$ outputs a right-aligned smooth factorization tree for the input sequence.
\end{lemma}
\begin{proof}
    The proof is analogous to the proof of Lemma~\ref{lem:smooth-split-primes}.
    It is an induction on the maximal $\mathcal{J}$-height of the input elements.
    If all elements have $\mathcal{J}$ heights equal to $1$
    then, by Claim~\ref{claim:j-height-one-smooth}, we know that the input sequence
    is smooth. This means that we can use Lemma~\ref{lem:smooth-product} 
    to compute its product, and construct a smooth factorization tree
    by inserting a smooth root in the last position. Note that we can detect 
    the last node using a right-to-left Mealy machine -- this is not possible using only 
    left-to-right Mealy machines.\\

    This leaves us with the induction step: We start by partitioning the input sequence into maximal smooth 
    blocks, i.e. smooth blocks that would 
    lose their smoothness if extended by one element left or right.
    The construction is similar to Lemma~\ref{lem:div-smooth-blocks}, 
    but without the ``keep every other underline'' phase. Instead, there is an additional step,
    where a right-to-left Mealy machine shifts all underlines one position to the left.
    Next, we use a construction similar to Lemma~\ref{lem:smooth-product} to compute 
    the product of each of those blocks. Then, for each block, we insert a smooth node that groups its elements together --
    with right-aligned encoding, this means inserting a smooth node in the last position of each smooth block
    (for blocks of length one we use a unary node instead of a smooth one). Next,
    we insert binary nodes to group the new nodes in pairs (if their number is odd, 
    we use two binary nodes to group the last three nodes together). We can do this 
    with a single-use Mealy machine that keeps track of the parity
    and remembers enough copies of the previous value to compute binary products (this is possible thanks to Lemma~\ref{lem:fs-k-fold}).
    Now, observe that those newly inserted binary nodes contain values whose 
    $\mathcal{J}$-heights are strictly lower than the maximal $\mathcal{J}$-height of the input sequence. 
    This means that we can finish the construction by combining the induction assumption with the subsequence combinator
    from Lemma~\ref{lem:subsequence-combinator}.
\end{proof}

We are now ready to prove the $\subseteq$-inclusion of Theorem~\ref{thm:rational-kr}.
For this we fix a local rational semigroup transduction $f = (S, h, \lambda)$ of type $\Sigma^* \to \Gamma^*$ 
and we show how to construct it as a composition of left-to-right and right-to-left single-use Mealy machines. 
Observe first, that by a similar reasoning as in the proof of Lemma~\ref{lem:local-monoid-transformation-composition-of-primes}
combined with the reasoning from the proof of the $\supseteq$-inclusion in Lemma~\ref{claim:rational-future-independent-one-way}, 
we can assume that $h : \Sigma \to S$ is a polynomial orbit-finite representation of $S$. This means that
we can use Lemma~\ref{lem:smooth-factorisation-tree-rational} to construct a smooth factorization tree over the input sequence.
In the remainder of this section, we show how to transform the smooth factorization tree into the output of the transduction.\\

Before we continue with the proof, let us present a few definitions:
First, we define a \emph{pointed word} (denoted as $\underline{w}$)
to be a word with one underlined letter.
We define the output of a pointed word $\underline{w}$ (with respect to the fixed local rational 
semigroup transduction $f$) to be the $i$-th letter of $f(w)$, 
where $i$ is the index of the underlined letter in $\underline{w}$,
and $w$ is the word $\underline{w}$, but without the underline.
Finally, we define the \emph{profile}
of $\underline{w}$ to be the following element of $S^1 \times \Sigma \times S^1$
(remember that $S^1$ is defined in Claim~\ref{claim:semigroup-to-monoid} -- we use it 
to handle the case where the prefix or the suffix is an empty word):
\[ (h(w_1) \cdot \ldots \cdot h(w_{i-1}),\ w_i,\ h(w_{i+1}) \cdot \ldots \cdot h(w_n)) \]
It is easy to see that the output of a pointed word depends entirely on its profile.\\

Let us now consider an input sequence $s_1, \ldots, s_n \in \Sigma^*$
 equipped with a smooth factorization tree. 
Notice that every node of this tree splits the input sequence into a suffix, an infix, and a prefix:
\smallpicc{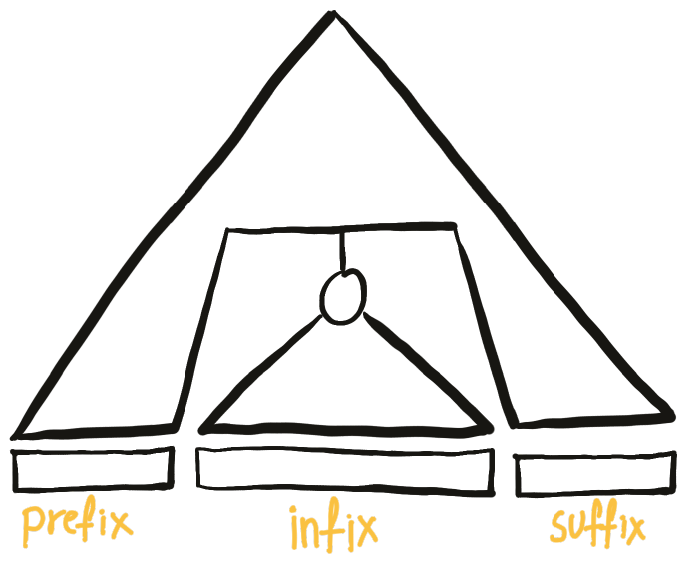}
\noindent
We say the \emph{context} of a node is the following pair from $S^1 \times S^1$:
\[ (\textrm{product of the prefix},\ \textrm{product of the suffix}) \]
Observe that if $q$ is a node that contains the $i$-th position of the input, 
then the $i$-th position of the output depends entirely on 
the context of $q$, and on the profile of $q$'s infix with $i$ as the underlined position. 
Moreover, observe that if $q$ is a leaf, then its infix consists of a single letter $s_i$,
whose profile is equal to $(1,\ s_i,\ 1)$. This means that the $i$-th output letter 
depends only on $s_i$ and on the context of the $i$-th leaf. 
This means that if we were able to use compositions of primes to compute the context of every node (or even every leaf), 
then we would also be able to compute the output word. Unfortunately, compositions 
of rational primes are unable to compute all contexts. 
The reason for this is analogous to the reason why compositions of left-to-right primes
cannot compute the products of all prefixes (see Example~\ref{ex:monoid-prefix-impossible}).
Instead, we define the \emph{reduced context} of a node which is an analogue of the core ancestor subsequences
from Definition~\ref{def:core-ancestor-sequence}. We start the definition by formulating an analogue of Claim~\ref{claim:xyg-values}. 
(We skip the proof, as it is analogous to Claim~\ref{claim:xyg-values}.)
\begin{claim}
    \label{claim:xyg-values-rational}
    There is a composition of rational primes that equips every smooth sequence $s_1, \ldots, s_n \in \Sigma^*$ 
    with the (representations) of $x_i, y_u, \overrightarrow{g_i}, \overleftarrow{g_i}$, such that:
    \begin{enumerate}
        \item the values $x_i$, $y_i$ are a decomposition of $s_i$, i.e. $x_i \cdot y_i = s_i$;
        \item the values $x_1, \overrightarrow{g_i}, y_i$ can be used to compute the $i$-th prefix, 
              i.e. \[x_1 \cdot \overrightarrow{g_i} \cdot y_i = s_1 \cdot \ldots s_i\]
        \item the values $y_i, \overleftarrow{g_i}, x_n$ can be used to compute the $i$-th suffix, 
              i.e. \[x_i \cdot \overleftarrow{g_i} \cdot y_n = s_i \cdot \ldots s_n\]
        \item all $\overrightarrow{g_i}$s and $\overleftarrow{g_i}$s are pairwise $\mathcal{H}$-equivalent;
        \item $\overrightarrow{g_1} = \overleftarrow{g_n}$ and they are both idempotent;
        \item all $x_i$'s are suffix-equivalent to $\overrightarrow{g_1} = \overleftarrow{g_n}$,
              and all $y_i$'s are prefix-equivalent to $\overrightarrow{g_1} = \overleftarrow{g_n}$.
        \item for all $i > 1$, it holds that  $\overrightarrow{g_i} = \overrightarrow{g_{i-1}} \cdot (y_{i-1} \cdot x_i)$, 
              and for all $i < n$, it holds that $\overleftarrow{g_i} = (y_i \cdot  x_{i+1}) \cdot  \overleftarrow{g_{i + 1}}$. 
    \end{enumerate}
\end{claim}

Note the claim talks about finitely supported primes and not about equivariant primes. 
However, we do not have to worry about that, because in the end we will be able to remove the 
unnecessary atoms from our construction, using Lemma~\ref{lem:unnecessary-atoms-compositions-primes}
(which can be easily extended to compositions of rational primes).
We are now ready to define \emph{reduced contexts}:
\begin{definition}
    The \emph{reduced context} (denoted as $\ctx(q)$) of a node $q$ is defined as follows:
    \begin{enumerate}
        \item If $q$ is the root of the tree, then its reduced context is empty.
        \item If $q$ is a child of a unary node $r$, then its reduced context is equal 
              to the reduced context of $r$.
        \item If $q$ is a child of a binary node $r$, and $q'$ is its sibling then 
              the reduced context of $q$ is equal to $(\ctx(r), v_{q'})$, 
              where $v_{q'}$ is the value of the node $q'$.
        \item If $q$ is the first child of a smooth node $r$, whose children are 
              $q_1, \ldots, q_n$ (i.e. $q = q_1$), then its reduced context  
             t contains the reduced context of $r$, and the (smooth) product
              of the rest of its sibling values, i.e.:
              \[ (\ctx(r), v_{q_2} \cdot \ldots \cdot v_{q_n})  \]
              If $q$ is the last child of a smooth node $r$, then the context of $q$ is defined analogously.
        \item If $q$ is an inner child of a smooth node $r$, whose 
              children are $q_1, \ldots, q_n$, i.e. $q = q_i$ for some 
              $1 < i < n$, then its reduced context is equal to the following tuple:
              \[ \poforb_e(x_1, y_n, \ctx(r)),\ \overrightarrow{g_{i-1}}, y_{i-1}, \overleftarrow{g_{i+1}}, x_{i+1} \comma\] 
              where $x, y, \overrightarrow{g}, \overleftarrow{g}$ are the values defined by Claim~\ref{claim:xyg-values-rational}
              for the smooth sequence $v_{q_1}, \ldots, v_{q_n}$, and $e$ is the only idempotent that is $\mathcal{H}$-equivalent 
              to $\overrightarrow{g_i}$. (As explained in the proof of Lemma~\ref{lem:h-idemp-eq-supp},
              each $\mathcal{H}$-class has at most one idempotent, which means that $e$ is well-defined.)
    \end{enumerate}
    Moreover, the reduced context remembers to which one of those cases $q$ belongs. This means that 
    we can represent the reduced context of a node at depth $k$, using the polynomial orbit-finite set $C_k$, 
    defined recursively: $C_0$ is the singleton set and $C_{k+1}$ is the following disjoint sum:
    \[
        \underbrace{C_k}_{\substack{\textrm{child of a}\\
        \textrm{unary node}}} +
        \underbrace{C_k \times \Sigma}_{\substack{\textrm{left child of a}\\
        \textrm{binary node}}} +
        \underbrace{C_k \times \Sigma}_{\substack{\textrm{right child of a}\\
        \textrm{binary node}}} +
        \underbrace{C_k \times \Sigma}_{\substack{\textrm{first child of a}\\
        \textrm{smooth node}}} +
        \underbrace{C_k \times \Sigma}_{\substack{\textrm{last child of a}\\
        \textrm{smooth node}}} +
        \underbrace{\orb_{d}(\Sigma^2 \times C_k) \times \Sigma^4}_{\substack{\textrm{inner child of a}\\
        \textrm{smooth node}\\
        (\textrm{where } d=\dim(\Sigma))}}
       \]
\end{definition}

Let us now show that the reduced context of a node and the profile of the infix 
is enough information to calculate the output letter (this is an analogue of Lemma~\ref{lem:core-defines}):
\begin{lemma}
\label{lem:rational-core-defines}
    Let $\underline{w}$ and $\underline{w'}$ be two pointed words, 
    let $T$ and $T'$ be smooth factorization trees for $w$ and $w'$
    (without the underlines), and let $q$ and $q'$ be nodes of $T$ and $T'$, 
    such that:
    \begin{enumerate}
        \item the underlined positions in $\underline{w}$ and $\underline{w'}$ belong
              respectively to the infix of $q$ and to the infix of $q'$;
        \item the infix profile of $q$ is equal to the infix profile of $q'$;
        \item the reduced context of $q$ is equal to the reduced context of $q'$. 
    \end{enumerate}
    Then the output letter of $\underline{w}$ is equal to the output letter of $\underline{w'}$. 
\end{lemma}
\begin{proof}
    The proof is very similar to the proof of Lemma~\ref{lem:core-defines}.
    It is an induction on the depth of $q$ and $q'$ (note that if $\ctx(q) = \ctx(q')$, then 
    $q$ and $q'$ have equal depths). The induction base is trivial: We know that 
    $q$ and $q$ are roots of $T$ and $T'$. It follows that 
    the infix profile of $q$ is equal to the profile of $\underline{w}$ and the 
    infix profile of $q'$ is equal to the profile of $\underline{w}'$. By assumption, 
    the infix profiles of $q$ and $q'$ are equal, which means that the profiles 
    $\underline{w}$ and $\underline{w'}$ are equal as well. This finishes the proof, 
    because the output letter of a pointed word depends entirely on its profile.\\ 
    
    For the induction step, we only deal with the most interesting case, which is when $q$ 
    is an inner node of a smooth node (other cases are immediate). We start by introducing 
    some notation: Let $r$, $q_1, \ldots, q_i, \ldots q_n$, and $e$ be as in the definition 
    of $\ctx(q)$ (in particular this means that $q = q_i$ and that $r$ is the parent of $q$).
    Let $(p_1,\ a,\ s_1)$ be the infix profile of $q$,
    let $(s_3, p_3)$ be the context of $r$, and let $p_2$ and $s_2$ be defined as follows:
    \[ 
        \begin{tabular}{cc}
        $p_2 = v_1 \cdot \ldots \cdot v_{i-1}$ & $s_2 = v_{i+1} \cdot \ldots \cdot v_n$,
        \end{tabular}
    \]
    where $v_i$ denotes the tree value of the node $q_i$.
    Finally, let $Z$ be an incomplete tree obtained by cutting out $r$'s subtree from $T$.
    Here is a sketch:
    \custompicc{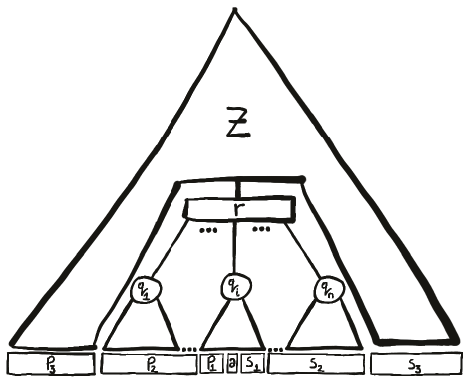}{0.65}
    Observe that $q'$ is an inner node of a smooth node as well -- this is because $\ctx(q)= \ctx(q')$.
    It follows that we can introduce the same notation for $T'$, i.e. we define 
    $r'$, $q_1', \ldots, q'_{n'}$, $p_1', p_2', p_3', a', s_3', s_2', s_1'$ and $Z'$ analogously as for $T$.
    By assumption, we know that infix profiles of $q$ and $q'$ are equal, which means that
    $p_1 = p_1'$, $a = a'$, and $s_1 = s_1'$. Moreover, since $\ctx(q) = \ctx(q')$, we also know that:
    \begin{enumerate}
        \item  $\overrightarrow{g_{i-1}} = \overrightarrow{g_{i'-1}}'$ and $y_{i-1} = y'_{i'-1}$;
        \item  $\overleftarrow{g_{i+1}} = \overleftarrow{g_{i'+1}}' $ and $x_{i+1} = x'_{i'+1}$;
        \item \label{it:orbits-eq}  $\poforb_e(x_1, y_n, \ctx(r)) = \poforb_{e'}(x'_1, y'_{n'}, \ctx(r'))$.
    \end{enumerate} 
    Since $\overrightarrow{g_{i-1}} = \overrightarrow{g_{i'-1}}'$, we know that $e = e'$,
    which (by Item~\ref{it:orbits-eq}) means that there exist a $\supp(e)$-permutation $\pi$
    such that $\pi(x_1) = x'_1$, $\pi(y_{n}) = y'_{n'}$ and $\pi(\ctx(r)) = \ctx(r')$. Let us use 
    this $\pi$ to construct an intermediate tree $T''$ with an underlined 
    element, and us prove that both the output of $T$ and the output of $T'$ are equal to the output of $T''$.
    We obtain $T''$ by taking $T$ and replacing $Z$ with $\pi(Z)$, 
    replacing the node $q_1$ with a node whose value is equal to $\pi(x_1) \cdot y_1$
    (it can be a leaf combined with unary nodes to keep the heights aligned), 
    and replacing $q_n$ with a node whose value is equal to $x_n \cdot \pi(y_n)$. 
    Here is a sketch of $T''$:
    \custompicc{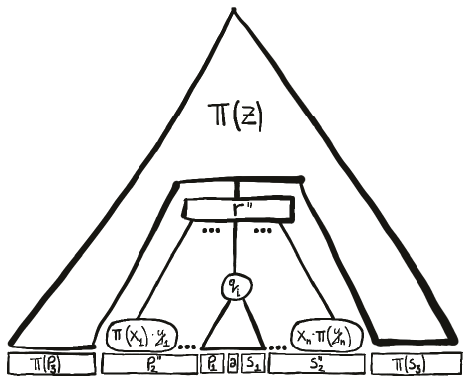}{0.7}
    As noted in the sketch, it holds by construction of $T''$ that $p_1'' = p_1$,
    $s_1'' = s_1$, $a'' = a$,  $p_3'' = \pi(p_3)$, and $s_3'' = \pi(s_3)$.
    We are now left with showing that (a) $T''$ is a valid
    smooth decomposition tree, (b) the output of $T'$ is equal to the output of $T''$ and 
    that (c) the output of $T$ is equal to the output of $T''$.\\
    
    Let us first prove that $T''$ is a valid smooth factorization tree.
    It suffices to show that $v_{r''} = \pi(v_{r})$,
    which means that it fits to $\pi(Z)$, and that $r''$ is a smooth node.
    Let us with $v_{r''} =\pi(v_{r})$. Observe that, by construction of $T''$:
    \[ v_{r''} = \pi(x_1) \cdot y_1 \cdot v_2 \cdot \ldots \cdot v_{n-1} \cdot x_n \cdot \pi(y_n)\tdot \]
    By Claim~\ref{claim:xyg-values-rational}, it follows that:
    \[ v_{r''} = \pi(x_1) \cdot \overrightarrow{g_n} \cdot \pi(y_n)\]
    By Claim~\ref{lem:h-idemp-eq-supp}, we know that $\supp(\overrightarrow{g_n}) = \supp(e)$,
    which means that $\pi(\overrightarrow{g_n}) = \overrightarrow{g_n}$
    (as $\pi$ is a $\supp(e)$-permutation).
    It follows that:
    \[ v_{r''} = \pi(x_1) \cdot \pi(\overrightarrow{g_n}) \cdot \pi(y_n) = \pi(x_1 \cdot \overrightarrow{g_n} \cdot y_n) = \pi(v_r)\]
    Now, to show that $r''$ is a smooth node, it suffices to show that the following product is smooth:
    \[ \pi(x_1) \cdot y_1 \cdot v_2 \cdot \ldots \cdot v_{n-1} \cdot x_n \cdot \pi(y_n)\tdot \]
    We have already shown that the value of this product is equal to $\pi(v_r)$,
    which, thanks to the smoothness of $r$, belongs to the $\mathcal{J}$-class of $\pi(e)=e$.
    According to Claim~\ref{claim:xyg-values-rational}, $\pi(x_1)$, $y_1$, $x_n$ and $\pi(y_n)$
    also belong to the $\mathcal{J}$-class of $\pi(e) = e$. Thanks to the smoothness of $r$,
    we know that each $v_j$ also belongs to this $\mathcal{J}$-class. It follows that the product is smooth.\\

    Now, let show that the output $T'$ is equal to the output of $T''$. For this we
    use the induction assumption for $r'$ and $r''$. This means that we need to show that
    the infix profile of $r'$ is equal to the infix profile of $r''$, 
    and that $\ctx(r') = \ctx(r'')$. 
    First, observe that $\ctx(r'') = \pi(\ctx(r))$. This is because the reduced context of 
    both $r''$ and $r$ depends entirely on the $Z$-part of the tree,
    and this dependence is easily seen to be equivariant. We have chosen 
    $\pi$, so that $\pi(\ctx(r)) = \ctx(r')$, which means that $\ctx(r'') = \ctx(r')$.
    Now, let us show that the infix profiles of $r'$ and $r''$ are equal, 
    i.e. that:
    \[(p''_2 \cdot p_1, a, s_1 \cdot s''_2 ) = (p'_2 \cdot p_1, a, s_1 \cdot s'_2 )\]
    It suffices to show that $p''_2 = p'_2$ (the proof for $s''_2 = s'_2$ is analogous).
    Notice that $p_2'' = \pi(x_1) \cdot y_1 \cdot v_{q_2} \cdot \ldots \cdot v_{q_{i-1}}$.
    By Claim~\ref{claim:xyg-values-rational} it follows that
    $p''_2 = \pi(x_1) \cdot \overrightarrow{g_{i-1}} \cdot y_{i-1}$. 
    By assumption, we know that $\pi(x_1) = x_1'$, $\overrightarrow{g_{i-1}}  = \overrightarrow{g_{i'-1}}'$, 
    and $y_{i-1} = y_{i'-1}'$. It follows that $p''_2 = x'_1 \cdot \overrightarrow{g_{i'-1}}' \cdot y'_{i'-1}$, 
    which by Claim~\ref{claim:xyg-values-rational} means that $p''_2 = p_2'$.\\

    Finally, let us show that the output of $T$ is equal to the output of $T''$. We do this by directly 
    showing that:
    \[ \lambda(p_3 \cdot p_2 \cdot p_1,\ a,\ s_1 \cdot s_2 \cdot s_3) =  \lambda(\pi(p_3) \cdot p''_2 \cdot  p_1,\ a,\ s_1 \cdot  s''_2 \cdot \pi(s_3)) \]
    We define $c_1 = p_3 \cdot x_1$ and $b_1 = y_1 \cdot v_{q_{2}} \cdot \ldots v_{q_{i-1}} \cdot p_1$ 
    (and analogously for $c_2$ and $b_2$). This means that $p_3 \cdot p_2 \cdot p_1 = c_1 \cdot b_1$ and 
    $\pi(p_3) \cdot p''_2 \cdot p_1 = \pi(c_1) \cdot b_1$ (and analogously for $b_2$, $c_2$ and the $s$-values). 
    This leaves us with showing that:
    \[\lambda(c_1 \cdot b_1,\ a,\ b_2 \cdot c_2) = \lambda(\pi(c_1) \cdot b_1,\ a,\ b_2 \cdot \pi(c_2))\]
    Observe that that $c_1 \cdot e = c_1$ -- this is because $e$ is an idempotent that is a suffix of $x_1$,
    which is a suffix of $c_1$ -- and that
    $\pi(c_1) \cdot e = \pi(c_1)$, as $\pi$ is a $\supp(e)$-permutation.
    (And analogously for $c_2$.) This means that it is enough to show that:
    \[\lambda(c_1 \cdot e \cdot b_1,\ a,\ b_2 \cdot e \cdot c_2) = \lambda(\pi(c_1) \cdot e \cdot b_1,\ a,\ b_2 \cdot e \cdot  \pi(c_2))\]
    This resembles the locality equation for $\lambda$, but we do not know if $c_1 a c_2 = e$
    (in fact, this might not be true). However, it is not hard to see that $c_1 a c_2$ is 
    $\mathcal{H}$-equivalent to $e$. We finish the proof, by showing that this is enough to apply the locality equation: 
    \begin{claim}
    \label{claim:rational-locality-h-equiv}
    If $\lambda$ satisfies the locality equation, then for every idempotent $e$, 
    for every $\supp(e)$-permutation $\pi$, and for all $a, b_1, b_2, c_1, c_2, a$, such that 
    $b_1 a b_2$ is $\mathcal{H}$-equivalent to $e$, it holds that:
    \[ \lambda(c_1 \cdot e \cdot b_1,\ a,\ b_2 \cdot e \cdot c_2) = \lambda(\pi(c_1) \cdot e \cdot b_1,\ a,\ b_2 \cdot e \cdot \pi(c_2)) \]
    \end{claim}
    \begin{proof}
        Let $g := b_1 a b_2$. By assumption, we know that $g$ is $\mathcal{H}$-equivalent to $e$.
        By \cite[Proposition 1.13]{pin2010mathematical}, we know that the $\mathcal{H}$-class 
        of $e$ is a subgroup of $S$, and that $e$ is the identity element of this subgroup.
        This means that there exists a $g^{-1}$ (also $\mathcal{H}$-equivalent to $e$)
        such that $g^{-1} \cdot g = g \cdot g^{-1} = e$. Define $c_1' := c_1 g$, 
        and $b_1' := g^{-1} b_1$. Observe that $c_1' \cdot e \cdot b_1' = c_1 \cdot e \cdot b_1$, 
        which means that:
        \[\lambda(c_1 \cdot e \cdot b_1,\ a,\ b_2 \cdot e \cdot c_2) = \lambda(c_1' \cdot e \cdot b_1',\ a,\ b_2 \cdot e \cdot c_2) \]
        We know that $b_1' \cdot a \cdot b_2 = g^{-1} \cdot g = e$, which means that we can use the locality equation:
        \[ \lambda(c_1' \cdot e \cdot b_1',\ a,\ b_2 \cdot e \cdot c_2) = \lambda(\pi(c_1') \cdot e \cdot b_1',\ a,\ b_2 \cdot e \cdot \pi(c_2))\]
        By Lemma~\ref{lem:h-idemp-eq-supp}, we know that $\supp(g) = \supp(e)$. It follows that $\pi(c_1') = \pi(c_1) \cdot g$, 
        which leads to $\pi(c_1') \cdot e \cdot b_1' = \pi(c_1) \cdot e \cdot b_1$. Thus, we have:
        \[ \lambda(\pi(c_1') \cdot e \cdot b_1',\ a,\ b_2 \cdot e \cdot \pi(c_2)) = \lambda(\pi(c_1) \cdot e \cdot b_1,\ a,\ b_2 \cdot e \cdot \pi(c_2))\]
        This completes the proof of the claim.
    \end{proof}   
\end{proof}

At this point, it is not hard to show how to construct the output 
of the local rational semigroup transduction. First, we notice that, 
thanks to Lemma~\ref{lem:rational-core-defines}, it is not hard 
to show the following analogue of Claim~\ref{claim:core-to-out}:
\begin{claim}
    Let $w_i$ be the $i$th input letter, and let $l_i$ be the $i$-th 
    leaf. There exists an equivariant function, that transforms 
    every pair $(w_i,\ctx(l_i))$ into the $i$-th output letter of the local rational semigroup transduction.
\end{claim}
\begin{proof}
    Once, we notice that the infix profile of $l_i$ is equal to $(1, w_i, 1)$, 
    the claim follows from Lemma~\ref{lem:rational-core-defines}, in the 
    same way as Claim~\ref{claim:core-to-out} follows from Lemma~\ref{lem:core-defines}.
\end{proof}

This finishes the construction, because by reasoning very similar to
the proof of Lemma~\ref{lem:core-construct-primes}, we can show that
we can use compositions of rational prime functions to construct $\ctx(q)$
in each node of a smooth factorization tree.

\subsubsection{Compositions of rational primes $\subseteq$ Local rational transductions}
In this section,  we show that every function $f : \Sigma^* \to \Gamma^*$ 
that can be constructed as a composition of single-use rational primes,
can be implemented as a rational semigroup transduction. The proof goes by 
induction on the construction of $f$ as a composition of primes.
However, in order to simplify the proof, we introduce an alternative way of constructing 
compositions of primes, that only uses sequential composition:
\begin{claim}
\label{claim:seq-comppositions-of-primes}
    Let $P$ be a set of prime functions that contains all length-preserving homomorphisms -- 
    for example, $P$ could be the set of single-use rational primes. 
    It follows that every word-to-word function can be constructed  
    as a $(\times, \circ)$-composition of primes from $P$, if and only if it can be constructed as a $(\circ)$-composition 
    of letter-to-letter homomorphisms and functions of the form $p \times \idf$, where $p \in P$.
\end{claim}
\begin{proof}
    Let $P'$ be the extended set of primes. The ($\Leftarrow$)-implication is easy:
    it suffices to notice that every function in $P'$ is a $(\times, \circ)$-composition 
    of primes. The proof of the ($\Rightarrow$)-implication uses the following equality
    to push down the ($\times$)-compositions:
    \[ (f \times g) = (f \times \idf) \circ (\idf \times f) \]
    Formally, the proof goes by induction on the construction of $f$ as a $(\circ, \times)$-composition 
    of primes from $P$:
    \begin{enumerate}
        \item The induction base is trivial: if $f$ is equal to a $p \in P$, it can be constructed 
              as follows:
              \[ {{\leftI}^{-1}}^* \circ (p \times \idf) \circ \leftI^*\comma\]
              where $\leftI$ is the function $\Sigma \to \Sigma \times 1$. 
        \item The case where $f$ is a composition $f = g \circ h$ is also simple.
              By induction, we know that:
              \[\begin{tabular}{cc}
                $g =  g_n \circ \ldots \circ g_1$ & $h = h_m \circ \ldots \circ h_1$,
              \end{tabular}\]
              where all $g_i$'s and $h_i$'s belong to $P'$. This means
              that we can construct $f$ in the following way:
              \[ f =  g_n \circ \ldots \circ g_1 \circ h_m \circ \ldots \circ h_1 \]
        \item Finally, the most interesting case is where $f = g \times h$. 
              Again, by induction assumption, we know that: 
              \[\begin{tabular}{cc}
                $g =  g_n \circ \ldots \circ g_1$ & $h = h_m \circ \ldots \circ h_1$,
              \end{tabular}\]
              It follows that: 
              \[ f = (g_n \circ \ldots \circ g_1) \times (h_m \circ \ldots h_1) = (g_n \times \idf) \circ \ldots \circ (g_1 \times \idf) \circ (\idf \times h_m) \circ \ldots \circ (\idf \times h_1)\]
              This leaves us with showing that for every $p \in P'$ functions $p \times \idf$ and $\idf \times p$ are $(\circ)$-compositions of functions from $P'$. We only show this 
              for $p \times \idf$, as the proof for $\idf \times p$ is analogous. There are two cases to consider: If $p$ is a letter-to-letter homomorphism then, 
              so is $p \times \idf$. If, on the other hand, $p$ is of the form $p' \times \idf_{\Sigma^*}$, where $p' \in P$, then $p \times \idf = (p' \times \idf_\Sigma) \times \idf_\Gamma$, 
              which is almost the same as  $p' \times \idf_{\Sigma \times \Gamma} \in P'$. This finishes the proof, as we can fix this slight type mismatch 
              using $(\circ)$-compositions with $\assoc^*$ and ${\assoc^{-1}}^*$. 
    \end{enumerate}
\end{proof}

This leaves us with translating $(\circ)$-compositions of the extended primes into local rational semigroup transductions:
\begin{lemma}
\label{lem:extended-primes-to-rational-transductions}
    Let $p_1, \ldots, p_n$ be a sequence of extended single-use rational primes (i.e. every $p_i$ is either a letter-to-letter homomorphism 
    or a $p' \times \idf$, where $p'$ is a single-use rational prime). It follows that the composition 
    $p_n \circ \ldots \circ p_1$ can be expressed as a local rational semigroup transduction.
\end{lemma}
The remaining part of this section is dedicated to proving Lemma~\ref{lem:extended-primes-to-rational-transductions}: The proof goes 
by induction on $n$. The induction base is trivial -- if $n=0$ then the composition of primes is equal to  $\idf^* : \Sigma^* \to \Sigma^*$.
For the induction step, it suffices to show that local rational semigroup transductions are closed under \emph{pre-compositions}\footnote{
    Showing this for post-compositions would be an equally valid 
    proof strategy, but pre-compositions seem to be more compatible with local rational semigroup transductions.
    We are going to see a similar reasoning in Chapter~4 where, depending on the model, 
    we are going to use either pre- or post-compositions. 
}
with extended single-use rational prime functions. We start with the length-preserving homomorphisms:
\begin{claim}
\label{claim:rational-precomp-homomorphism}  
    Let $f: \Gamma^* \to \Delta^*$ be defined by some local rational semigroup transduction,
    and let $g : \Sigma \eqto \Gamma$ be an equivariant function over orbit-finite sets. It follows 
    that the following composition can also be defined by a local rational semigroup transduction:
    \[ \Sigma^* \transform{g^*} \Gamma^* \transform{f} \Delta^*\]
\end{claim}
\begin{proof}
    If $f$ is recognized by $(S, h, \lambda)$, then $f \circ g^*$ is recognized by $(S, h \circ g, \lambda)$. 
\end{proof}

Next, we show that local rational semigroup transductions are closed under pre-composition with single-use left-to-right 
atom propagation -- the right-to-left variant of the proof is analogous.
\begin{claim}
\label{claim:rational-precomp-su-prop}
    Let $f: ((\atoms + \epsilon) \times \Sigma)^* \to \Gamma^*$ be defined by some local rational semigroup transduction. 
    It follows that the following composition can also be defined by a local rational semigroup transduction:
    \[ ((\atoms + \{\downarrow, \epsilon\}) \times \Sigma)^* \longtransform{\overrightarrow{{\textrm{su-prop}}}\ \times\  \idf} ((\atoms + \epsilon) \times \Sigma)^* \transform{f} \Gamma^*\]
\end{claim}
\begin{proof}
    Let $(S, h, \lambda)$ be a local rational semigroup transduction that defines $f$.
    We show how to construct $(S', h', \lambda')$ that defines $f \circ (f_{\overrightarrow{\textrm{su-prop}}} \times \idf)$. 
    First, let us analyse how $f_{\overrightarrow{\textrm{su-prop}}} \times \idf$ modifies an input infix. Here is an example:
    \bigpicc{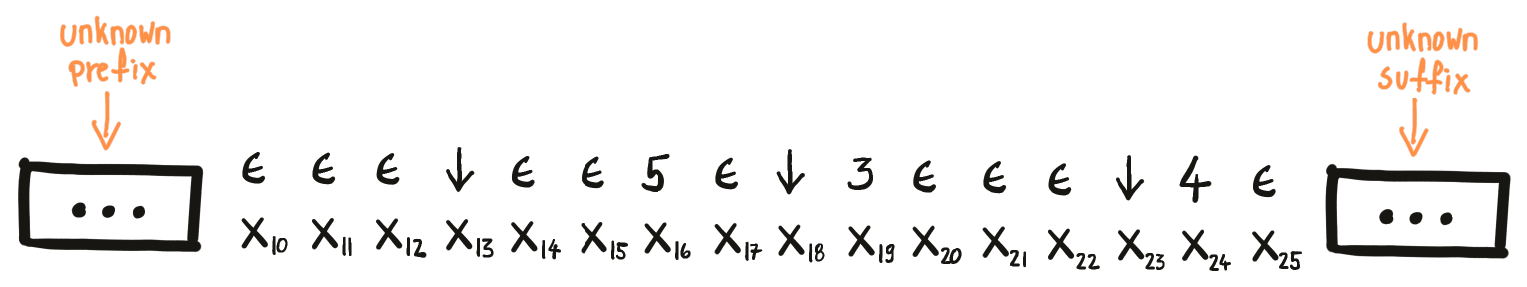}
    Since we do not know the prefix, we cannot exactly predict the output of $f_{\overrightarrow{\textrm{su-prop}}} \times \idf$
    on the infix. (Observe that the unknown suffix does not influence the output on the infix.) We can, however, predict almost 
    all of the output letters:
    \bigpicc{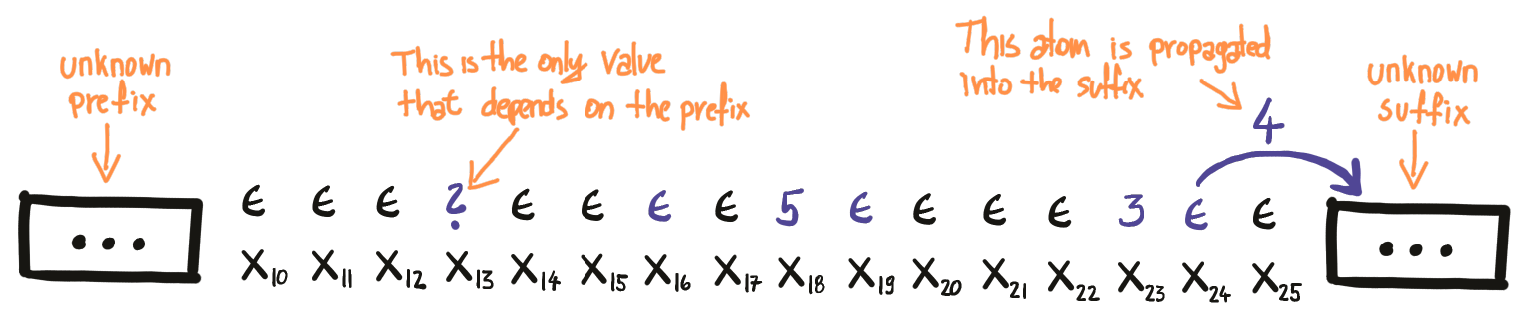}
    Since we do not know all output letters, we cannot exactly compute the product of their $h$-values. 
    We can, however, use the product of the $h$-values to compress the information about the 
    infix:
    \bigpicc{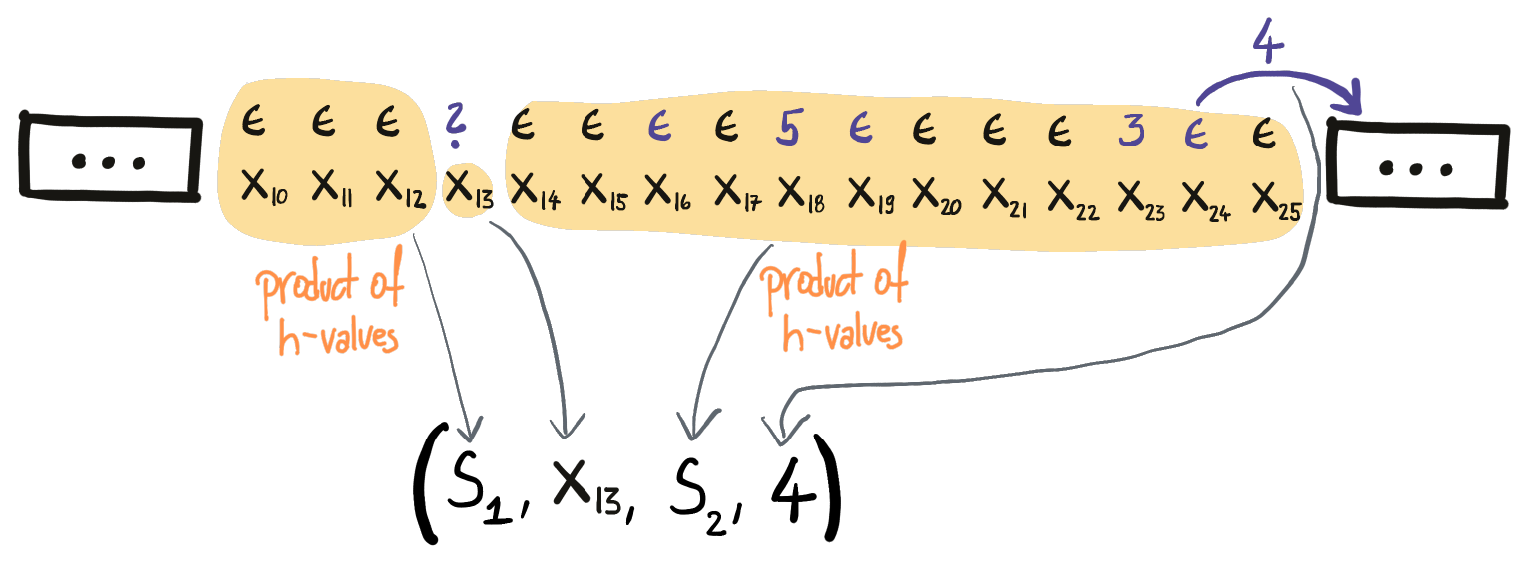}
    Observe that $s_1$ and $s_2$ come from $S^1$ (defined in Claim~\ref{claim:semigroup-to-monoid}) --
    this is because both the prefix before the first $\downarrow$ and the suffix after the first $\downarrow$ might be empty.
    Moreover, in the general case, the last non-$\epsilon$ letter might be $\downarrow$ -- in this case the 
    value propagated into the suffix is going to be $\bot$. It is also possible that 
    the first non-$\epsilon$ value is not a $\downarrow$, or that the input word 
    contains only $\epsilon$'s. 
    In order to account for all those cases, we represent the compressed 
    information about an infix as an element of the following set:
    \[ S' = \underbrace{S^1 \times \Sigma \times S^1 \times (\atoms + \bot)}_{
        \substack{\textrm{An infix whose first non-$\epsilon$}\\ \textrm{is equal to $\downarrow$}}
    }\ \ \  + \underbrace{S \times (\atoms + \bot)}_{
        \substack{\textrm{An infix whose first non-$\epsilon$}\\ \textrm{is an element of $\atoms$}}
    } + \underbrace{S}_{
        \substack{\textrm{An infix that}\\ \textrm{only contains $\epsilon$'s}}\\
    } \]
    It is not hard to see that this compression is compositional -- i.e. the compressed information about $w_1w_2$ 
    depends only on the compressed information about $w_1$ and the compressed information about $w_2$. 
    It follows by (the semigroup version of) the proof of Lemma~\ref{lem:compositional-to-monoids} that this compression operation 
    induces a semigroup structure on $S'$. Moreover, since both $S$ and $\Sigma$ are orbit-finite, 
    it follows that $S'$ is orbit-finite as well. It follows that we can use $S'$ as the underlying 
    semigroup for the rational semigroup transduction for $f_{\overrightarrow{\textrm{su-prop}}} \times \idf$.
    This leaves with defining $h'$ and $\lambda'$. We define $h'$ to be the compression function. 
    For $\lambda'$ we are going to need a couple of auxiliary functions. We start with $\mathtt{feed}$:
        \[ \mathtt{feed} : \underbrace{(\atoms + \bot)}_{
            \substack{ \textrm{Value propagated}\\ \textrm{from the prefix} }
        } \times \underbrace{S'}_{
            \substack{\textrm{The compressed information}\\ \textrm{about the infix}}
        } \to \underbrace{S}_{
            \substack{\textrm{The product of the $h$-values}\\ \textrm{on the infix after applying $f_{\overrightarrow{\textrm{su-prop}}}$}}
        } \]
        It is not hard to see that the function $\mathtt{feed}$ can be computed using the following formula:
        \[
            \begin{tabular}{ccc}
                $\mathtt{feed}(v, (s_1, x, s_2, a) ) = s_1 h\binom{v}{x} s_2$ & $\mathtt{feed}(v, (s, a)) = s$ & $\mathtt{feed}(v, s) = s$
            \end{tabular}
        \]
        Now, we define the function $\mathtt{get} : S' \to (\atoms + \bot)$ that computes the atom propagated by 
        the infix:
        \[ 
            \begin{tabular}{ccc}
                $\mathtt{get}(s_1, x, s_2, a) = a$ & $\mathtt{get}(s, a) = a$ & $\mathtt{get}(s) = \bot$
            \end{tabular}
        \]
        Now, we can use $\mathtt{feed}$, $\mathtt{get}$, and the original output function $\lambda$, 
        to define $\lambda'$ that computes the output letter for $w_1 \underline{\binom{y}{x}} w_2$
        with respect to $f \circ (f_{\overrightarrow{\textrm{su-prop}}} \times \idf)$:
        \[ 
            \lambda'\left(p',\ \binom{\downarrow}{x},\ s'\right) = \lambda\left(\mathtt{feed}(\bot,\, p'),\ \binom{\mathtt{get}(p')}{x},\ \mathtt{feed}(\bot,\, s') \right)
        \]
        \[ 
            \lambda'\left(p',\ \binom{\epsilon}{x},\ s'\right) = \lambda\left(\mathtt{feed}(\bot,\, p'),\ \binom{\epsilon}{x},\ \mathtt{feed}(\mathtt{get}(p'),\, s') \right)
        \]
        \[ 
            \lambda'\left(p',\ \binom{a \in \atoms}{x},\ s'\right) = \lambda\left(\mathtt{feed}(\bot,\, p'),\ \binom{\epsilon}{x},\ \mathtt{feed}(a,\, s') \right)
        \]
        It is not hard to see that this $\lambda'$ is equivariant, and that $(S', h', \lambda')$ defines $f \circ (f_{\overrightarrow{\textrm{su-prop}}} \times \idf)$.
        This leaves us showing that $\lambda'$ is local: Let us take $x_1, x_2, y_1, y_2, e \in S'$, $a \in (\atoms + \{\downarrow, \epsilon\}) \times \Sigma$,
        and a $\supp(e)$-permutation $\pi$, such that $e$ is an idempotent, and $y_1 h'(a) y_2 = e$, and let us show that: 
        \[ \lambda'(x_1 e y_1,\; a, \; y_2 e x_2) = \lambda'(\pi(x_1)\, e y_1,\; a,\; y_2 e \pi(x_2)) \]
        For this, we consider two cases:\\
        
        First, we consider the case where $e$ contains only $\epsilon$'s,  i.e. $e \in S$ (see the definition of $S'$).
                  Since $y_1 h'(a) y_2 = e$, it follows that also  $y_1, a$, and $y_2$ contain only $\epsilon$'s --
                  i.e. $y_1, y_2 \in S$ and $a$ is of the form $\binom{\epsilon}{a}$ (for some $a \in \Sigma$). By definition of $S'$ and $\lambda'$,
                  it follows that:
                  \[ \lambda'(x_1 e y_1, \binom{\epsilon}{a}, y_2 e x_2) = \lambda(\mathtt{feed}(\bot, x_1) \cdot e \cdot y_1, \binom{\epsilon}{a}, y_2 \cdot e \cdot \mathtt{feed}(\mathtt{get}(x_1), x_2 ))\]
                  Observe that (by definition of $S'$), $e$ is an idempotent as an element of $S$, and also in $S$ it holds that $y_1 h\binom{\epsilon}{a} y_2 = e$.
                  This means that we can use the locality of $\lambda$ to transform the right-hand side of the equality further into:
                  \[\lambda(\pi(\mathtt{feed}(\bot, x_1)) \cdot e \cdot y_1, \binom{\epsilon}{a}, y_2 \cdot e \cdot \pi(\mathtt{feed}(\mathtt{get}(x_1), x_2 )))\]
                  By equivariance of $\mathtt{feed}$ and $\mathtt{set}$, we know that this is equal to:
                  \[ \lambda(\mathtt{feed}(\bot, \pi(x_1)) \cdot e \cdot y_1,  \binom{\epsilon}{a}, y_2 \cdot e \cdot \mathtt{feed}(\mathtt{get}(\pi(x_1)), \pi(x_2)))\comma\]
                  which by definition of $S'$ and $\lambda'$ is equal to:
                  \[ \lambda'(\pi(x_1) e y_1, \binom{\epsilon}{a}, y_2 e \pi(x_2))\]
                  This finishes the proof for the first case.\\

                  In the second case, we assume the complement of the first case, i.e. that $e$ contains at least one $\downarrow$ or an element of $\atoms$.
                  In this case $a$ could be of any form: $\binom{\epsilon}{a}$, $\binom{\downarrow}{a}$, or $\binom{b \in \atoms}{a}$.
                  All of those cases are analogous, so we only show how to deal with the most interesting case, which is $\binom{\downarrow}{a}$. 
                  First, we notice that since $e$ is an idempotent, we have that $\lambda'(x_1 e y_1, \binom{\downarrow}{a}, y_2 e x_2) = \lambda'(x_1 e e y_1, \binom{\downarrow}{a}, y_2 e x_2)$. 
                  Then, we use the definition of $S'$ and $\lambda'$ to unfold the right-hand side in the following way:
                  \[ \lambda\left(\mathtt{f}(\bot, x_1) \cdot \mathtt{f}(x_1, e) \cdot \mathtt{f}(e, e) \cdot \mathtt{f}(e, y_1), \binom{\mathtt{get}(ey_1)}{a}, \mathtt{f}(\bot, y_2) \cdot \mathtt{f}(y_2, e) \cdot \mathtt{f}(e, x_2) \right) \comma \]
                  where $\mathtt{f}(x_1, e)$ is a notational shortcut for $\mathtt{feed}(\mathtt{get}(x_1), e)$.
                  Now, let us observe 
                  that (by definition of $S'$), $\mathtt{feed}(e, e)$ is an idempotent in $S$: 
                  \[ \mathtt{feed}(e, e) \cdot \mathtt{feed}(e, e) = \mathtt{feed}(e, e \cdot e) = \mathtt{feed}(e, e) \]
                  Moreover, since $y_1 h'\binom{\downarrow}{a} y_2 = e$, we know that $\mathtt{feed}(y_2, e) = \mathtt{feed}(e, e)$, 
                  and that
                  \[ \mathtt{feed}(e, y_1) \cdot h\binom{\mathtt{get}(ey_1)}{a} \cdot \mathtt{feed}(\bot, y_2) = \mathtt{feed}(e, e) \]
                  Finally, we observe that by Lemma~\ref{lem:fs-functions-preserve-supports}, it holds that $\supp(\mathtt{feed}(e, e)) \subseteq(\supp(e))$, 
                  which means that $\pi$ is a $\supp(\mathtt{feed}(e, e))$-permutation. It follows that we can apply the locality equation obtaining:
                  \[
                  \lambda\left( \pi\left(\mathtt{f}(\bot, x_1) \cdot \mathtt{f}(x_1, e)\right) \cdot \mathtt{f}(e, e) \cdot \mathtt{f}(e, y_1),
                   \binom{\mathtt{get}(ey_1)}{a}, \mathtt{f}(\bot, y_2) \cdot \mathtt{f}(y_2, e) \cdot \pi(\mathtt{f}(e, x_2)) \right)
                  \]
                  By definition of $S'$ and $\lambda'$ this is equal to $\lambda(\pi(x_1) e e y_1, \binom{\downarrow}{a}, y_2 e \pi(x_2))$,
                  which is in turn equal to $\lambda'(\pi(x_1) e y_1, \binom{\downarrow}{a}, y_2 e \pi(x_2))$.
                  It follows that $\lambda'$ is local.
\end{proof}

Finally, we show that local rational prime functions are closed under pre-compositions with left-to-right multiple-use bit propagation
and with group-prefix functions (left-to-right propagation and group-suffix functions
can be handled analogously). We can deal with both of those cases by proving the following claim:
\begin{claim}
    \label{claim:rational-precomp-fintie-Mealy}
    Let $\mathcal{A} : \Sigma_1^* \to \Sigma_2^*$ be a Mealy machine over \emph{finite} (and not only orbit-finite) alphabets, 
    let $\Gamma, \Delta$ be orbit-finite alphabets, and let $f : (\Sigma_2 \times \Gamma)^* \to \Delta^*$
    be a local rational semigroup transduction. It follows that the following composition is a local rational semigroup transduction as well:
    \[ f \circ (\mathcal{A} \times \idf) : (\Sigma_1 \times \Gamma)^* \to \Delta^* \]
\end{claim}
\begin{proof}
    By Lemma~\ref{lem:mealy-monoids-classical}, we know that $\mathcal{A}$ is equivalent to some finite semigroup transduction 
    (see Definition~\ref{def:finite-semigroup-transductions}), which means that it can be expressed 
    as  $\lambda_1^* \circ f_{S_1\textrm{-pref}}\circ h_1^*$, where $f_{S_1\textrm{-pref}}$ is the semigroup prefix function 
    for some \emph{finite} (and not only orbit-finite) semigroup $S_1$. Since, by Claim~\ref{claim:rational-precomp-homomorphism}, we already know 
    that local rational semigroup transductions are closed under pre-compositions with homomorphisms, it suffices 
    to show how to construct $f \circ (f_{S_1\textrm{-pref}} \times \idf)$ as a local rational semigroup transduction:\\

    Let $(S_2, h_2, \lambda_2)$ be a local rational semigroup transduction for $f$. 
    We show how construct a local rational semigroup construction for $f \circ (f_{S_1\textrm{-pref}} \times \idf)$, 
    using the classical wreath product construction for finite semigroups\footnote{See \cite[Definition~1.7]{krohn1965prime}.}.
    First, let us analyse how $f_{S_1\textrm{-pref}} \times \idf$ modifies an input infix. Here is an example:
    \smallpicc{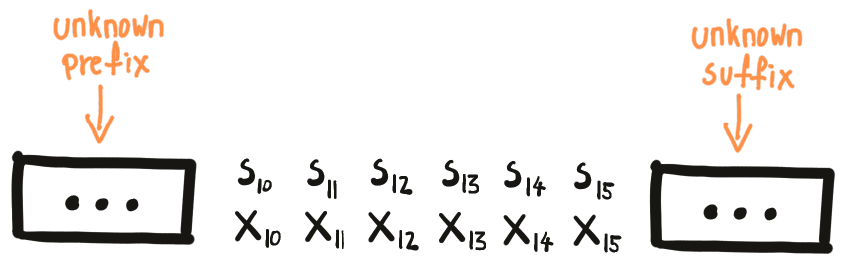}
    \noindent
    In order to compute the output of $f_{S_1\textrm{-pref}} \times \idf$, 
    it suffices to know the $S_1$-product of the prefix. For example:
    \bigpicc{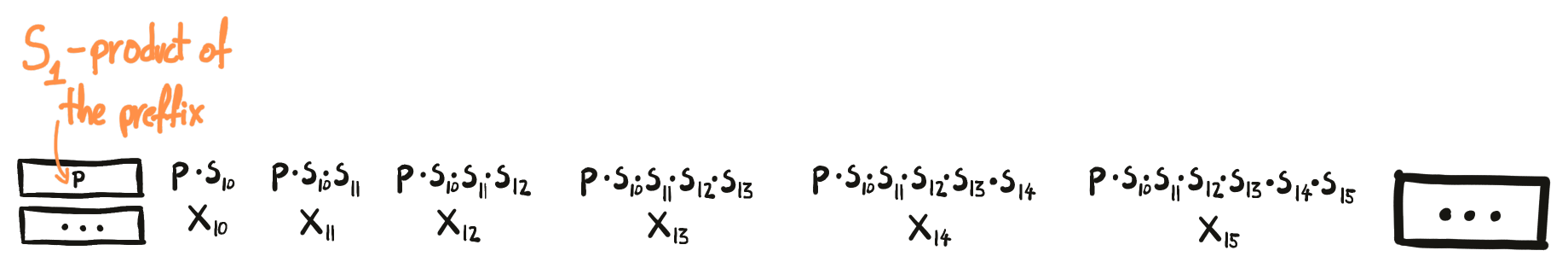}
    \noindent
    At this point, we can apply $h_2$ to every letter, and compute the $S_2$ product of the infix. 
    In general, this can be represented as a function of the following type:
    \[ \underbrace{S_1^1}_{\substack{
        \textrm{Given the $S_1$-value of a prefix}\\
        \textrm{(which might be empty) \ldots}
    }} \to \underbrace{S_2}_{\substack{
        \textrm{\ldots what is the}\\
        \textrm{$S_2$-value of the infix}\\
        \textrm{after applying $f_{S_1-\textrm{pref}}$ and $h_2^*$}
    }} \]
    In order to obtain compositionality, we also need to remember the $S_1$-value of the infix.
    In total, we compress the information about each infix into an element of the following set:
    \[S_3 = \underbrace{S_1}_{\substack{
        \textrm{$S_1$-product}\\
        \textrm{of the prefix}
    }} \times
    \underbrace{(S_1^1 \to S_2)}_{\substack{
        \textrm{The function}\\
        \textrm{explained above}
    }}
    \]
    The key observation is that since $S_1$ is finite (and not only orbit-finite), 
    we know that $S_1 \to S_2$ is orbit-finite, as it is isomorphic to a finite power $S_2^{|S_1|}$. 
    This procedure of compressing is compositional, which (as explained in Lemma~\ref{lem:compositional-to-monoids})
    imposes a semigroup product on $S_3$. The explicit formula for this product looks as follows:
    \[ (x, X) \cdot (y, Y) = (xy,\ p \mapsto X(p) \cdot Y(px)) \]
    The intuition behind this formula is that $p$ is the $S_1$-product of the $X$'s prefix, 
    and $px$ is the $S_1$-product of the $Y$'s prefix (as $Y$'s prefix also includes $x$).
    This is a well-studied construction called the wreath product of $S_1$ and $S_2$
    sometimes denoted\footnote{In \cite{krohn1965prime} this is denoted as $S_3 = S_2 w S_1$.} as $S_3 = S_2 \wr S_1$.
    We use $S_3$ as the underlying semigroup 
    for the local rational semigroup transduction for $f \circ (f_{S_1\textrm{-pref}} \times \idf)$. 
    The input function $h_3 : S_1 \times \Gamma \to S_3$ is given by the following formula:
    \[ h_3(s, a)  = (s,\ p \mapsto h_2(ps, a))\comma \]
    and the output function $\lambda_3 : S_3^1 \times (S_1 \times \Gamma) \times S_3^1 \to \Delta$, is given by the 
    following formula:
    \[ \lambda_3((x, X),\ (s, a),\ (y, Y)) = \lambda_2(X(1), (xs, a), Y(xs))\]
    The intuition behind this formula for $\lambda_3$ is that after applying $f_{S_1-\textrm{pref}}$ to the input word, 
    the $S_2$-product of the $h_2$-values of the prefix is equal to $X(1)$, 
    the current letter is equal to $(xs, a)$, and  the $S_2$-product of the $h_2$-values 
    of the suffix is equal to $Y(xs)$. Thanks to this intuition, it is not hard to see that $(S_3, h_3, \lambda_3)$
    implements $f \circ (f_{S_1\textrm{-pref}} \times \idf)$.\\
    
    This leaves us with showing that $\lambda_3$ is local.
    For this we take $(x_1, X_1)$, $(x_2, X_2)$, $(y_1, Y_1)$, $(y_2, Y_2)$, $(e, E) \in S_3$, 
    $(s, a) \in S_1 \times \Gamma$ and a $\supp((e, E))$-permutation $\pi$
    such that $(e, E)$ is an idempotent, and $(y_1, Y_1) \cdot  h_3(s, a) \cdot (y_2, Y_2) = (e, E)$, 
    and we show  that:
    \begin{gather*}
        \lambda_3((x_1, X_1) (e, E) (y_1, Y_1), \ (s, a), \ (y_1, Y_1) (e, E) (x_1, X_1) \\
        =\\
        \lambda_3(\pi((x_1, X_1)) (e, E) (y_1, Y_1), \ (s, a), \ (y_2, Y_2) (e, E)\; \pi((x_2, X_2)
    \end{gather*}
    First, since $(e, E)$ is an idempotent, we can transform the initial expression into:
    \[ \lambda_3((x_1, X_1) (e, E) (e, E) (y_1, Y_1), \ (s, a), \ (y_1, Y_1) (e, E) (x_1, X_1) )\]
    Then, using the definition of $\lambda_3$ and of the product in $S_3$, we transform it into:
    \[ \lambda_2(X_1(1) \cdot E(x_1) \cdot E(x_1 e) \cdot Y_1(x_1 e e),\ (x_1 e e y_1 s, a),\ Y_2(x_1 e e y_1 s) \cdot E(x_1 e e y_1 s y_2) \cdot X_2(x_1 e e y_1 s y_2 e))\]
    Now, we observe that since $(e, E)$ is an idempotent in $S_3$, we know that $e$ is an idempotent in $s_1$, and 
    that since $(y_1, Y_1) \cdot h_3(s, a) \cdot (y_2, Y_2) = (e, E)$, we know that $y_1 s y_2 = e$.
    Thanks to this observation, we can transform our expression into:
    \[ \lambda_2(X_1(1) \cdot E(x_1) \cdot E(x_1 e) \cdot Y_1(x_1 e),\ (x_1 e y_1 s, a),\ Y_2(x_1 e y_1 s) \cdot E(x_1 e) \cdot X_2(x_1 e))\] 
    At this point, we would like to apply the locality equation for $\lambda_2$. For this, we need to
    show proof all of its assumptions. 
    First, we observe that since $(e, E)$ is an idempotent, we know that:
    \[ (e, E) = (e, E) (e, E) = (ee, p \mapsto E(p) \cdot E(pe) ) \] 
    In particular, this means that $E(p) = E(p) \cdot E(pe)$ for every $p \in S_1^1$. If we take $p = x_1 e$, we obtain that 
    $E(x_1 e) = E(x_1 e) E(x_1 e e) = E(x_1e)E(x_1e)$, which means that $E(x_1e)$ is an idempotent. Next, 
    by a similar reasoning, we observe that since $(y_1, Y_1) \cdot h_3(s, a) \cdot (y_2, Y_2) = (e, E)$,
    we know that for all $p$, it holds that:
    \[ Y_1(p) \cdot h_2(py_1s, a) \cdot Y_2(py_1s) = E(p) \]
    If we take $p = x_1e$, we get that:
    \[ Y_1(x_1 e) \cdot h_2(x_1ey_1s, a) \cdot Y_2(x_1ey_1s) = E(x_1 e)\]
    Finally, we notice that $S_1$ is atomless (this is because every equivariant set that is both finite and orbit-finite has
    to be atomless). This means that both $x_1$ and $e$ have empty supports, which means that $\pi$ is a $\supp(E(x_1e))$-permutation, 
    as by Lemma~\ref{lem:fs-functions-preserve-supports} it holds that $\supp(E(x_1e)) \subseteq \supp(E)$. It follows that we can apply the locality equation, obtaining:
    \[ \lambda_2( \pi(X_1(1) \cdot E(x_1)) \cdot E(x_1 e) \cdot Y_1(x_1 e),\ (x_1 e y_1 s, a),\ Y_2(x_1 e y_1 s) \cdot E(x_1 e) \pi( X_2(x_1 e)))\] 
    Observe now that since $x_1$ is atomless (i.e. equivariant), we know that $\pi(x_1) = x_1$, and since $\pi$ is a $\supp(e, E)$-permutation, 
    we know that $\pi(E) = E$. This means that (after unfolding some of the $e$'s back to $y_1 s y_2$ or to $ee$)
    we can fold the definitions of $\lambda_3$ and of the product in $S_3$, obtaining:
    \[ \lambda_3( \pi(x_1, X_1) (e, E) (e, E) (y_1, Y_1),\ (s, a),\ (y_2, Y_2) (e, E) \pi(x_2, X_2))\]
    After folding $(e, E) (e, E)$ back to $(e, E)$, we conclude that $\lambda_3$ is local. 
\end{proof}

\chapter{Two-way transductions with atoms}
\label{ch:su-regular-functions}

In this chapter, we define and study the class of word-to-word functions recognized by \emph{single-use two-way transducers}.
Our main result is Theorem~\ref{thm:regular-equivalence-with-atoms}, which states that this class of transductions 
admits three more equivalent definitions: \emph{copyless streaming string transducers with atoms}, 
\emph{regular list functions with atoms}, and \emph{compositions of single-use two-way primes}.
Furthermore, we show that single-use two-way transducers (and their equivalent models)
are closed under compositions and have decidable equivalence.\\

In my opinion, these results demonstrate that single-use models are better behaved than their multiple-use counterparts:
Two-way multiple-use register automata lack decidable equivalence (see Theorem~\ref{thm:2DOFA-logspace}), and copyless 
streaming string transducers with multiple-use registers are not closed under composition
(see \cite[Proposition~4]{alur2011streaming}).
I believe that for this reason,
the class of function defined by single-use two-way transducers deserves the name of \emph{regular transductions with atoms}.\\

This chapter, along with the four models it introduces, is based on \cite{single-use-paper}  (specifically, on \cite[Theorems~13~and~14]{single-use-paper}).
However, the presentation of the models and the results is new and, hopefully, improved. The new approach builds upon the idea 
of single-use functions from Chapter~2, and uses Theorem~\ref{thm:rational-kr} from Chapter~3.

\section{Definitions}
We use this section to formulate and briefly discuss the four equivalent definitions of regular transductions with atoms:
\subsection{Single-use two-way transducers}
\label{subsec:su-2-transducer}
A \emph{single-use two-way transducer} is a single-use two-way automaton with output.
We have already seen all the building blocks of the definition -- we simply need to combine them:
\begin{definition}
    A \emph{single-use two-way transducer} consists of:
    \begin{enumerate}
        \item A polynomial orbit-finite input alphabet $\Sigma$ and a polynomial orbit-finite output alphabet $\Gamma$;
        \item A polynomial orbit-finite set of states $Q$;
        \item An equivariant initial state $q_0 \in Q$;
        \item A single-use transition function:
        \[
            \delta \ : \ \underbrace{\Sigma + \{\vdash, \dashv\}}_{\substack{\textrm{current letter, or an} \\ \textrm{end-of-word marker}}} \ \eqto\  \left( \underbrace{Q}_{\substack{\textrm{current} \\ \textrm{state}}}\  \suto \underbrace{Q}_{\substack{\textrm{new} \\ \textrm{state}}} \times \underbrace{(\Gamma + \epsilon)}_{\substack{
                \textrm{output}\\
                \textrm{letter}
            }} \times \underbrace{\{\leftarrow, \rightarrow\}}_{\substack{\textrm{which way}\\\textrm{to go}}}\ +\, \underbrace{\mathtt{finish}}_{\substack{\textrm{or finish}\\\textrm{the run}}}\right)
        \]
\end{enumerate}

A single-use two-way transducer defines the following function $\Sigma^* \to \Gamma^*$:
Given an input word $w \in \Sigma^*$, we equip it with end-of-word markers obtaining $\vdash w \dashv$,
and we place the transducer's head at the first letter of $w$ in the initial state $q_0$. Then, we 
start applying the transition function, feeding it with the transducer's state and with the input letter 
seen by the transducer's head, and using its output to update the transducer's state 
and move the transducer's head.
We continue this process until\footnote{There is the usual problem with looping --
a looping transducer will never finish the run, which could mean that it defines a partial function.
The simplest way to fix this
is to require that the transducers do not loop.
However, it would not make any difference if we assumed instead 
that the output of a looping run is the empty word, or even if we agreed 
that the two-way transductions are partial functions.} the transition function returns $\mathtt{finish}$,
at which point we construct the output word as the concatenation of all output letters produced by the transition function
(excluding $\epsilon$'s). 
\end{definition}

\begin{example}
\label{ex:map-duplicate}
    Consider the following $\emph{map duplicate}$ function, which duplicates every 
    $\#$-separated block:
    \[ f_\textrm{map-dup} : (\atoms + \#)^* \to (\atoms + \#)^* \]
    For example:
    \[ f_\textrm{map-dup}(1\ 2\ 3\ \#\ 5\ 7\ \#\ 1\ 2) = 1\ 2\ 3\ 1\ 2\ 3\ \#\ 5\ 7\ 5\ 7\  \#\ 1\ 2\ 1 \ 2\ \] 
    Here is an example\footnote{Observe that this transducer only uses finitely many states. Since this is the only example of a single-use two-way transducer presented in this thesis,
    it is worth pointing out that single-use two-way are allowed to use polynomial orbit finite set of states.} of a single-use two-way transducer that implements $f_\textrm{map-dup}$:
    \custompicc{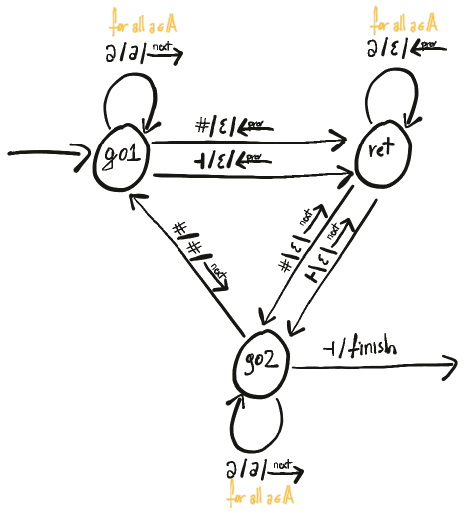}{0.6}
    \noindent
    And here is an example run of this transducer:
    \bigpicc{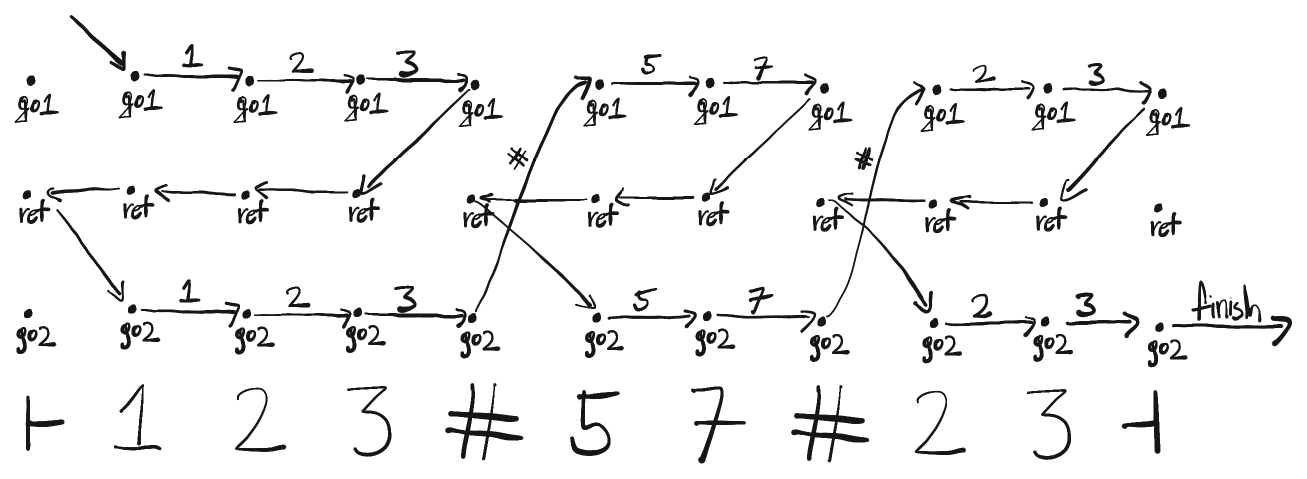}
\end{example}

\subsection{Single-use streaming string transducers with atoms}
\label{subsec:sst-def}
\emph{Single-use streaming string transducers with atoms} are a variant of \emph{copyless streaming string transducers}.
The original model was simultaneously defined in \cite[Section 2.2]{alur2011streaming}
and in \cite[Section~3]{alur2010expressiveness} (the papers cite one another).
Interestingly, the two papers provide different definitions:
\cite{alur2011streaming} defines streaming string transducers as a model over
infinite alphabets (using multiple-use atom registers, in style of \cite{kaminski_finite_memory_paper}),
while \cite{alur2010expressiveness} limits the definition to finite alphabets.
One of the main results of \cite{alur2010expressiveness} is \cite[Theorem~3]{alur2010expressiveness}, which states that 
the finite-alphabet version of copyless streaming string transducers is equivalent to \textsc{MSO}-transductions.
As \textsc{MSO}-transductions are equivalent to two-way transducers, it follows that
copyless streaming string transducers benefit from the robustness of 
regular transductions over finite alphabets. In particular, they are closed under compositions and have decidable equivalence.
On the other hand, the infinite-alphabet version from \cite{alur2011streaming} retains only some of those properties -- 
it has decidable equivalence (\cite[Theorem~12]{alur2011streaming}), 
but it is not closed under compositions (\cite[Proposition~4]{alur2011streaming}). Possibly for this reason 
the finite-alphabet version from \cite{alur2010expressiveness} is much more prevalent in the literature.\\

Before defining single-use string streaming transducers with atoms, 
let us discuss its finite-alphabet version\footnote{As mentioned before, the model was originally defined in \cite[Section~3]{alur2010expressiveness}. 
However,  our presentation differs slightly from the original one, as we want to maintain consistency with the previous chapters of this thesis.}.
We start by an informal description (to be followed by a formal definition):
The \emph{copyless streaming string transducer} is a variant of a one-way automaton, which constructs its 
output using a finite number of string registers over the output alphabet.
It can concatenate its registers  (e.g. $\mathtt{r}_1 \, := \mathtt{r}_2 \cdot \mathtt{r}_3$),
but the \emph{copyless} restriction requires that this operation destroys the contents of 
the concatenated registers (i.e. $\mathtt{r}_2$ and $\mathtt{r}_3$ in the example), by overriding them with $\epsilon$. 
Observe the analogy between the copyless and the single-use restrictions.
A streaming string transducer is not allowed to query its registers (unlike the register automaton 
from Section~\ref{sec:dra}).\\ 

Here is an example of a copyless streaming string transducer that implements the finite-alphabet version
of $f_\textrm{map-dup}$ from Example~\ref{ex:map-duplicate} over the alphabet $\{a, b, \#\}$.
It has only one state, and three registers: $\mathtt{current}_1$, $\mathtt{current}_2$, and $\mathtt{output}$
(due to the copyless restriction the transducer has to maintain two copies of the current block: $\mathtt{current}_1$ and $\mathtt{current}_2$):
\custompicc{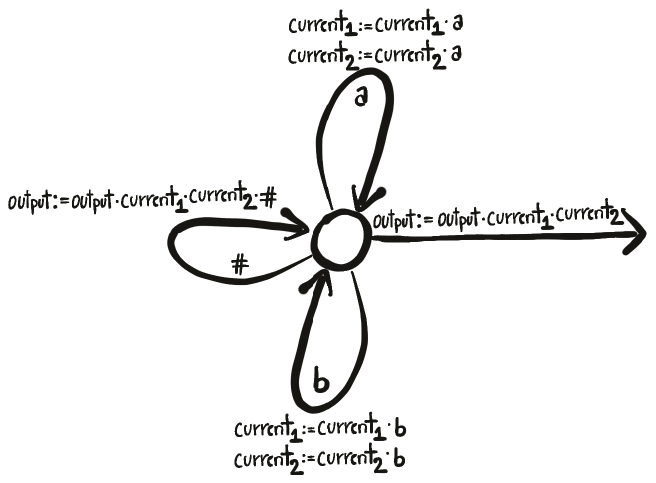}{0.7}

For the formal definition, we are going to use  a variant of the single-use function 
(from Definition~\ref{def:single-use-functions}). First, let us
define the class of \emph{polynomial $\Gamma^*$-register sets} (over a finite output alphabet $\Gamma$)
to be the smallest class closed under $+$ and $\times$ that contains the following sets:
\[
    \begin{tabular}{cc}
        $\underbrace{1}_{\textrm{Singleton set}}$ &
        $\underbrace{\Gamma^*}_{\substack{\textrm{Contents of a}\\ \textrm{string register}}}$
    \end{tabular}
\]

Next, let us define the class of \emph{single-use functions over polynomial $\Gamma^*$-register sets} to be the smallest class 
of functions closed under the combinators from Definition~\ref{def:single-use-functions} (i.e. $\circ$, $\times$, and $+$),
that contains all the basic functions from Definition~\ref{def:single-use-functions} except the functions about $\atoms$,
and all the following basic functions about $\Gamma^*$:
\[
    \begin{tabular}{|ll|}
        \hline
        \multicolumn{2}{|c|}{Functions about $\Gamma^*$}\\
        \hline
        $\concat:$ & $ \Gamma^* \times \Gamma^* \to \Gamma^*$ \\
        $\singleton:$ & $\Gamma \to \Gamma^*$ \\
        $\const_\epsilon :$ & $1 \to \Gamma^*$\\
        \hline
    \end{tabular}    
\]
(In the type of  $\singleton$, we use the fact that since $\Gamma$ is finite, it can be represented
as the following polynomial $\Gamma^*$-register set:  $1 +\ldots + 1$).\\

\begin{example}
    Consider the following function $f : \{a, b\} \times {(\Gamma^*)}^2 \to {\Gamma^*}$:
    \[ \begin{tabular}{cc}
        $f(a, \mathtt{r}_1, \mathtt{r}_2) = \mathtt{r_1} \cdot \mathtt{r_2}$ &  $f(b, \mathtt{r}_1, \mathtt{r}_2) = \mathtt{r}_2 \cdot \mathtt{r}_1$
    \end{tabular}\]
    This is a single-use function over polynomial $\Gamma^*$-register sets, as it can be implemented in the following way (assuming that $\{a, b\}$ is represented as $1 + 1$):
    \[ (1 + 1) \times {(\Gamma^*)}^2 \transform{\distr} {(\Gamma^*)}^2 + {(\Gamma^*)}^2 \longtransform{\idf + \sym} {(\Gamma^*)}^2 + {(\Gamma^*)}^2 \transform{\merge} {(\Gamma^*)}^2 \transform{\mathtt{concat}} \Gamma^* \] 
\end{example}
We denote the set of all 
single-use functions between two polynomial $\Gamma^*$-register sets as $X \suto Y$. 
Overloading the notation should not cause any confusion, as it is usually clear 
from the context whether $X$ and $Y$ are polynomial $\Gamma^*$-register sets or polynomial orbit-finite sets.\\ 

We are now ready to give a formal definition of copyless string streaming 
transducers for finite alphabets:
\begin{definition}
\label{def:sst}
    A \emph{copyless streaming string transducer} of type $\Sigma^* \to \Gamma^*$ (where $\Sigma$ and $\Gamma$ are finite sets)
    consists of:
    \begin{enumerate}
        \item a polynomial $\Gamma^*$-register set $Q$ of states;
        \item an initial state $q_0 \in Q$;
        \item a single-use transition function $\delta : \Sigma \to (Q \suto Q)$; and
        \item an output function $\lambda : Q \suto \Gamma^*$.
    \end{enumerate}
    Every copyless string streaming transducer defines a function $\Sigma^* \to \Gamma^*$: In order to compute the output 
    for a $w \in \Sigma^*$, the transducer processes $w$ letter by letter, updating its state according 
    to the transition function. After processing the entire $w$,
    it computes the output word by applying $\lambda$ to its final state. 
\end{definition}

It is not hard to extend this definition to polynomial orbit-finite alphabets. First, we define the class 
of \emph{polynomial orbit-finite $\Gamma^*$-register sets} (where $\Gamma$ is a polynomial orbit-finite output alphabet)
to be the smallest class of sets closed under $\times$ and $+$, that contains the sets $1$, $\atoms$, and $\Gamma^*$.
Then, we define the class of \emph{single-use functions over polynomial orbit-finite $\Gamma^*$-register sets} to be the smallest class of functions 
that is closed under the combinators from Definition~\ref{def:single-use-functions} (i.e. $\circ$, $\times$ and $+$), 
contains all basic functions from Definition~\ref{def:single-use-functions}, and contains the three basic functions 
about $\Gamma^*$ (i.e. $\concat$, $\const_\epsilon$, and $\singleton$). Now, we can define 
\emph{single-use string streaming transducers} (for infinite alphabets) in the same way as in Definition~\ref{def:sst}:
\begin{definition}
\label{def:sst-with-atoms}
    A \emph{single-use streaming string transducer} of type $\Sigma^* \to \Gamma^*$ (where $\Sigma$ and $\Gamma$ are polynomial orbit-finite sets)
    consists of:
    \begin{enumerate}
        \item a polynomial orbit-finite $\Gamma^*$-register set $Q$ of states;
        \item an initial state $q_0 \in Q$;
        \item a single-use transition function $\delta : \Sigma \to (Q \suto Q)$; and
        \item an output function $\lambda : Q \suto \Gamma^*$.
    \end{enumerate}
    A single-use streaming string transducer defines a function $\Sigma^* \eqto \Gamma^*$ in the 
    same way as a finite streaming string transducer. 
\end{definition}

\noindent
Finally, let us point out that if we extend single-use functions with 
\[\copyf_\atoms : \atoms \to \atoms \times \atoms \comma\] 
we obtain a definition equivalent to the one from \cite[Section 2.2]{alur2011streaming}.
If we, instead, include the function:
\[\copyf_{\Gamma^*} : \Gamma^* \times \Gamma^* \to \Gamma^*\comma\]
we obtain \emph{copyful streaming string transducers}, a model whose finite-alphabet version is studied in 
\cite{filiot2017copyful}. (The main result of the paper is that copyful streaming string transducers
over finite alphabets have decidable equivalence, see \cite[Section~3]{filiot2017copyful}.)
Of course, it is also possible to include both $\copyf_\atoms$ and $\copyf_{\Gamma^*}$ and obtain
a version of the transducer that is both multiple use and copyful, but I was unable to find 
this variant in the literature.

\subsection{Regular list functions with atoms}
\emph{Regular list functions with atoms} are based on \emph{regular list functions} -- 
a model for finite alphabets that was introduced in \cite{bojanczyk2018regular}.
One of the main results of the paper is \cite[Theorem 6.1]{bojanczyk2018regular}, which
proves that regular list functions are equivalent to \textsc{MSO}-transductions. This means that,
similarly to two-way transducers or copyless string streaming transducers,
regular list functions (over finite alphabets) exhibit the robustness of
regular transductions.\\

Extending regular list functions to infinite alphabets is very natural. (In fact, 
the extension was already suggested in \cite[Section~7]{bojanczyk2018regular}.)
For this reason, we proceed directly to the definition of  list functions \emph{with atoms}. 
It is structurally similar\footnote{In fact, the current shape of the definition 
of single-use functions was inspired by \cite[Definition~2.1]{bojanczyk2018regular}.}
to the definition of single-use functions (i.e. Definition~\ref{def:single-use-functions}):\\

First, we define \emph{polynomial sets with atoms} to be the smallest class of sets that 
contains $1$ and $\atoms$, and that is closed under $\times$, $+$, and $X^*$ (i.e. lists of finite length).
Here is an example:
\[ ((\atoms \times \atoms)^* + \atoms)^*\]
It is worth pointing out that every polynomial orbit-finite $\Gamma^*$-register set (as defined in Section~\ref{subsec:sst-def})
is a polynomial set with atoms, but not vice versa. This is because
polynomial sets with atoms treat $X^*$ as an independent set constructor, whereas 
polynomial orbit-finite $\Gamma^*$-register sets only allow lists over one fixed polynomial orbit-finite $\Gamma$.\\

\noindent
We are now ready to define regular list functions with atoms:
\begin{definition}
\label{def:regular-list-functions-atoms}
    The class of \emph{regular list functions with atoms} is the smallest class of functions that
    is closed under the combinators from Definition~\ref{def:single-use-functions} (i.e. $+$, $\times$, and $\circ$), 
    contains all basic functions from Definition~\ref{def:single-use-functions}, and additionally is closed 
    under the following combinator $\map$ and contains the following basic functions:
    \[
            \begin{tabular}{c}
                    \begin{tabular}{|c|}
                    \hline
                    The combinator $\map$\\
                    \hline
                    $\infer
                    { X^* \longtransform{\map f} Y^*}
                    { X \longtransform{f} Y }$\\
                    \hline
                    \end{tabular}\\
                \\
                \begin{tabular}{|ll| l|}
                    \hline
                    \multicolumn{3}{|c|}{Functions about lists}\\
                    \hline
                    $\const_\epsilon :$ & $ 1 \to X^*$ & Returns the empty lists\\
                    $\cons :$ & $X \times X^* \to X^*$ & Adds an element to the front of a list.\\
                    $\destruct : $ & $X^* \to X \times X^* + 1$ & Extracts the head and tail of a list (if possible).\\
                    $\concatf : $ & ${(X^*)}^* \to X^*$ & Flattens nested lists.\\
                    $\reverse : $ & $X^* \to X^*$ & Reverses the list.\\
                    $\blocks :$ & $(X + 1)^* \to {(X^*)}^*$ & Groups elements of X into maximal blocks.\\
                    $\group_G : $ & $(X \times G)^* \to (X \times G)^*$ & Computes group prefixes (see below).\\ 
                    \hline
                \end{tabular}\\
                \\
                \begin{tabular}{|ll|}
                    \hline
                    \multicolumn{2}{|c|}{Copying}\\
                    \hline
                    $\copyf_\atoms$ & $\atoms \to \atoms \times \atoms$ \\
                    $\copyf_{X^*}$ & $X^* \to X^* \times X^*$ \\
                    \hline
            \end{tabular}\\
            \end{tabular}
            \]
            The function $\group_G : (G \times X)^* \to (G \times X)^*$ is defined for every finite group $G$. 
            It computes group prefixes on the first coordinate and leaves the second coordinate unchanged: 
            \[ \group_G([(g_1, x_1), (g_2, x_2), \ldots, (g_n, x_n)]) = [(g_1, x_1), (g_1 \cdot g_2, x_2), \ldots, (g_1 \cdot \ldots \cdot g_n, x_n)] \]
   
            Finally, we define \emph{regular list transductions} to be all regular list functions 
            of the type $\Sigma^* \to \Gamma^*$, where $\Sigma$ and $\Gamma$ are polynomial orbit-finite.\\
\end{definition}
It is worth pointing out that if we remove $\atoms$ and related basic functions from Definition~\ref{def:regular-list-functions-atoms},
we obtain the original regular list functions from \cite{bojanczyk2018regular}.\\

Observe that regular list functions are a copyful model: They include $\copyf_\atoms$ and $\copyf_{X^*}$
which allow for copying both atoms and lists. Moreover, using a construction similar to the one from
Example~\ref{ex:copy-general-definable}, one can derive a general $\copyf_X$ function which 
works for every polynomial set with atoms $X$. Despite that, regular list functions are equivalent to the single-use versions 
of two-way transducers and string streaming transducers. Moreover, without the $\copyf$ functions,
regular list functions would become too weak: Without $\copyf_\atoms$,
they would not be able to simulate multiple-use access to input letters,
and without $\copyf_{X^*}$, they would not be able to implement the duplicate function (i.e. $w \mapsto ww$). 
For this reason, regular list functions are an interesting link between single-use and multiple-use 
models  that could help us understand the connections between the two approaches. For finite 
alphabets, a related line of research has been recently explored in \cite{bojanczyk2023folding}.

\subsection{Compositions of single-use two-way primes}
In this section we define \emph{compositions of single-use two-way primes}, which
is a model based on the Krohn-Rhodes decomposition theorems.\\

Observe that the functions computed by two-way transducers might not preserve length, 
so if we want to define an equivalent class of compositions of primes,
we need to include some primes that are not length-preserving. Those are going to be the 
$f_{\textrm{map-dup}}$ function (from Example~\ref{ex:map-duplicate}) and \emph{letter-to-word} homomorphisms, 
defined as follows:
\begin{example}
    Let $\Sigma$, $\Gamma$ be polynomial orbit-finite sets, and let $f: \Sigma \eqto \Gamma^*$ be an equivariant function.
    (Observe that since $\Sigma$ is orbit-finite and $f$ is equivariant, it follows the length of words in $f(\Sigma)$ is bounded.) Define the
    \emph{letter-to-word homomorphism based on $f$} to be the function $f^* : \Sigma^* \to \Gamma^*$, that applies 
    $f$ to every input letter and concatenates the results. For example, if we take $\Sigma = \atoms + \bot$, $\Gamma= \atoms$, 
    and $f$ defined as $f(a \in \atoms) = aa$ and $f(\bot) = \epsilon$, then:
    \[ f^*(123\bot45\bot\bot) = 1122334455 \]
\end{example}

The parallel composition (i.e. $\times$) only makes sense for length-preserving functions, which
means that we can no longer use it in the definition of composition of single-use two-way primes. Instead, 
we use a similar approach as in Claim~\ref{claim:seq-comppositions-of-primes},
and define compositions of primes only in terms of $(\circ)$-compositions:
\begin{definition}
    We define the \emph{compositions of single-use two-way primes}, to be the smallest class of 
    word-to-word functions that is closed under the sequential composition (i.e. $\circ$) and contains all 
    the following prime functions: 
    \begin{enumerate}
        \item \emph{letter-to-word}\footnote{In other words, a letter-to-words homomorphism $\Sigma^* \to \Gamma^*$ is simply a monoid morphism 
                                             between the free monoids $\Sigma^*$ and $\Gamma^*$. The phrase \emph{letter-to-word} is used 
                                             to distinguish this general class of homomorphism from the letter-to-letter homomorphisms used in the previous chapter.} \emph{homomorphisms}, i.e. functions of the form
                                             $f^* : \Sigma^* \to \Gamma^*$,
               for every $f: \Sigma \eqto \Gamma^*$, where $\Sigma$ and $\Gamma$ are polynomial orbit-finite sets (the function $f^*$ is defined as a function that applies $f$ to every letter and concatenates the results);
        \item functions of the form $p \times \idf_{\Delta^*}  : (\Sigma \times \Delta)^* \to (\Gamma \times \Delta)^*$,
              where $p : \Sigma^* \to \Gamma^*$ is one of the single-use prime functions
              from Theorem~\ref{thm:kr} (i.e. the Krohn-Rhodes decomposition theorem for single-use Mealy machines), 
              and $\Delta$ is a polynomial-orbit-finite set;
        \item \emph{map duplicate}, i.e. function $f_\textrm{map-dup}: (\Sigma + \#)^* \to (\Sigma + \#)^*$, which
               is a generalization of the map duplicate function from Example~\ref{ex:map-duplicate} for an arbitrary polynomial 
               orbit-finite alphabet $\Sigma$;
        \item \emph{map reverse}, i.e. the function $f_\textrm{map-rev}: (\Sigma + \#)^* \to (\Sigma + \#)^*$ 
              defined below in Example~\ref{ex:map-reverse}, for every polynomial orbit-finite $\Sigma$;
              \begin{example}
                \label{ex:map-reverse}
                    For every polynomial orbit-finite $\Sigma$, we define the 
                    \emph{map reverse} function, which independently reverses 
                    every $\#$-separated block:
                    \[f_{\textrm{map-rev}} : (\Sigma + \#)^* \to (\Sigma + \#)^* \]
                    For example:
                    \[ f_\textrm{map-rev}(1\ 2\ 3\ \#\ 5\ 7\ \#\ 1\ 2) = 3\ 2\ 1\ \#\ 7\ 5\  \#\ 2\ 1\] 
            \end{example}
        \item \emph{end of word marker}, i.e. the function $w \mapsto w \! \dashv$. This function is only\footnote{
            If we only want to consider non-empty words, we can skip this function. Or if we want to consider 
            the empty word, we can replace it with the 
            \emph{start of word marker} (i.e. $w \mapsto\; \vdash\!w$). If we want a symmetric function,  
            we can also use the \emph{both ends marker} (i.e. $w \mapsto\; \vdash\!w\!\dashv$), 
            or the \emph{empty word selector}, 
            i.e. the function which maps every word to itself, with the exception of the empty word which 
            is mapped to a single letter $\fullmoon$. All of those functions result in an equivalent class of compositions of two-way single-use primes.}
               required to deal with the empty word,
                as for all other prime functions it holds that $p(\epsilon) = \epsilon$.
    \end{enumerate}
\end{definition}

\section{Equivalence of the models}
\label{sec:regular-equiv}
As mentioned in the introduction, all models defined in the previous section are equivalent:
\begin{theorem}
\label{thm:regular-equivalence-with-atoms}
    All the following models recognize the same class of transductions over polynomial orbit-finite alphabets:
    \begin{enumerate}
    \item Single-use two-way transducers;
    \item Single-use streaming string transducers;
    \item Regular list transductions with atoms;
    \item Compositions of single-use two-way primes;
    \end{enumerate}
\end{theorem}

\noindent
Before proving Theorem~\ref{thm:regular-equivalence-with-atoms}, let us point out its two important corollaries:
\begin{claim}
\label{claim:two-way-automata-closed-under-compositions}
    Single-use two-way transducers and single-use streaming string transducers are closed under composition. 
\end{claim}
\begin{proof}
    The claim is an immediate consequence of Theorem~\ref{thm:regular-equivalence-with-atoms}, 
    as compositions of single-use two-way primes are trivially seen to be closed under compositions. 
\end{proof}
\begin{claim}
    The equivalence problem is decidable for the four models that appear in Theorem~\ref{thm:regular-equivalence-with-atoms}.
\end{claim}
\begin{proof}
   By analysing the proof of Theorem~\ref{thm:regular-equivalence-with-atoms} (presented later in this section),
   one can show that it is effective, i.e. translating between any two of the models is a computable function.
   This leaves us with showing that just one of the models has decidable equivalence. 
   Let us focus on single-use string streaming string transducers:
   As mentioned in the last paragraph of Section~\ref{subsec:sst-def},
   single-use streaming string transducers are a special case of streaming string transducers 
   from \cite{alur2011streaming}, which by \cite[Theorem~12]{alur2011streaming} have decidable equivalence.
   It follows that single-use streaming string transducers have decidable equivalence as well. 
   (Alternatively, we can use \cite[Theorem~14]{single-use-paper} which
   is a direct proof of decidable equivalence for single-use streaming string transducers.)
\end{proof}

\noindent
The rest of Section~\ref{sec:regular-equiv} is dedicated to proving Theorem~\ref{thm:regular-equivalence-with-atoms},
according to the following plan:\\

\vspace{0.3cm}
    \adjustbox{width=1\textwidth, center}{
        \begin{tikzcd}
            & {\textrm{Regular list transductions}} \\
            \\
            \\
            {\textrm{Two-way transducers}} && {\textrm{Compositions of primes}} \\
            \\
            \\
            \\
            & {\textrm{Streaming string transducers}}
            \arrow["{\rotatebox{0}{\textrm{Section~\ref{subsec:su-2way-to-primes}}}}"'{text={rgb,255:red,255;green,58;blue,51}}, curve={height=24pt}, from=4-1, to=4-3]
            \arrow["{\rotatebox{0}{\textrm{Section~\ref{subsec:su-primes-to-2way}}}}"'{text={rgb,255:red,255;green,58;blue,51}}, curve={height=12pt}, from=4-3, to=4-1]
            \arrow["{\rotatebox{-33}{\textrm{Section~\ref{subsec:su-sst-to-2way}}}}"{description, text={rgb,255:red,255;green,58;blue,51}}, from=8-2, to=4-1]
            \arrow["{\rotatebox{33}{\textrm{Section~\ref{subsec:su-primes-to-sst}}}}"{description, text={rgb,255:red,255;green,58;blue,51}}, from=4-3, to=8-2]
            \arrow["{\rotatebox{27}{\textrm{Section~\ref{subsec:su-list-fun-2way}}}}"{description, text={rgb,255:red,255;green,58;blue,51}}, from=1-2, to=4-1]
            \arrow["{\rotatebox{-27}{\textrm{Section~\ref{subsec:su-primes-to-list-fun}}}}"{description, text={rgb,255:red,255;green,58;blue,51}}, from=4-3, to=1-2]
        \end{tikzcd}
      }
\hspace{0.3cm}\\

\noindent
It is worth pointing out that Section~\ref{subsec:su-primes-to-2way} seems to be redundant. 
However, together with Section~\ref{subsec:su-2way-to-primes} it completes the proof of Claim~\ref{claim:two-way-automata-closed-under-compositions},
which is later used in Sections~\ref{subsec:su-sst-to-2way}~and~\ref{subsec:su-list-fun-2way}.

\subsection{Two-way transducers $\subseteq$ Compositions of primes}
\label{subsec:su-2way-to-primes}
In this section, we show how to translate two-way transducers into compositions of single-use two-way primes.
This is the most complicated part of the proof of Theorem~\ref{thm:regular-equivalence-with-atoms}.
We start by showing that every non-looping single-use two-way transducer
can visit every position only a bounded number of times. This is not obvious, 
because  the number of states of a single-use two-way transducer 
is usually infinite -- in particular, an analogous lemma does not hold for the multiple-use 
variant of two-way transducers. (For example, the two-way automaton from Claim~\ref{claim:2dofa-all-distinct}
visits each position $i$ a number of times that grows linearly with $i$).
\begin{lemma}
\label{lem:2way-bounded}
    For every non-looping single-use two-way transducer $\mathcal{A}$, there is a bound $k$,
    such that $\mathcal{A}$ visits every position in every input word at most $k$ times.
\end{lemma}
\begin{proof}
    Let $\mathcal{A}$ be a transducer of type $\Sigma^* \to \Gamma^*$ and
    let $Q$ be $\mathcal{A}$'s set of states. Consider the following semigroup $S_{\mathcal{A}}$
    of $\mathcal{A}$'s behaviours (see the proof of Lemma~\ref{lem:2sua-incl-ofm} for details):
    \[ Q \times \{\leftarrow, \rightarrow\} \ \longsuto \ Q \times \{\leftarrow, \rightarrow\}  + \textrm{finish}\]
    Thanks to Theorem~\ref{thm:su-orbit-finite}, we know that $S_{\mathcal{A}}$ is orbit-finite.
    It follows that there is a bound $p$ such that every element of $S_{\mathcal{A}}$ is supported by
    at most $p$ atoms. Similarly there is a $q$ such that every letter in $\Sigma$ is supported
    by at most $q$ atoms. Moreover, since $Q$ is orbit-finite, it follows that every
    finite set of atoms supports only a finite number of elements in $Q$ (this follows from
    \cite[Lemma~5.2]{bojanczyk2013nominal} or alternatively, since 
    $Q$ is \emph{polynomial} orbit-finite, it can also be shown 
    by structural induction on $Q$). Additionally, it is not hard to see 
    that the number of elements in $Q$ that are supported by a subset of atoms
    depends only on the subset's size.
    It follows that there is a function $f_Q$, 
    such that $f_Q(n)$ is the number of elements in $Q$ that are supported 
    by a subset of $n$ atoms. 
    We are going to show that $\mathcal{A}$ 
    visits each position no more than $k := f_Q(2 p + q)$ times.\\ 
    
    Let us take a word $w \in \Sigma^*$, a position $i$ in $w$, 
    and let us show that $\mathcal{A}$ visits $i$ at most $k$ times. First, 
    we split $w$ into the prefix $w_{<i} \in \Sigma^*$, the letter $w_i \in \Sigma$,
    and the suffix $w_{>i} \in \Sigma^*$.  Define $b_{<i}, b_{>i} \in S_\mathcal{A}$ to be
    the behaviour of $\mathcal{A}$ on $w_{<i}$ and $w_{>i}$.
    By Lemma~\ref{lem:fs-functions-preserve-supports} (applied multiple times),
    we know that every state in which $\mathcal{A}$ visits $i$, 
    can contain only the atoms that appear in $b_{<i}$, $b_{i}$, or $w_i$.
    It follows that there are at most $f_Q(2 \cdot p + q)$ different states in which
    $\mathcal{A}$ can visit $i$. Since $\mathcal{A}$ is non-looping, it follows 
    that each of those states is visited at most once. This means that the number of visits in $i$ is 
    bounded by $k = f_Q(2 \cdot p + q)$.
\end{proof}

A useful abstraction for the construction of translating two-way automata into compositions 
of primes is the $\emph{shape of a run}$, which is a record of all visits of $\mathcal{A}$
in each of the input position.  
A visit is represented as an element of the following set:
\[ \underbrace{\{\leftarrow, \rightarrow\}}_{\substack{
    \textrm{wheather $\mathcal{A}$ entered}\\
    \textrm{from left or right}}
    }\times \underbrace{\{\leftarrow, \rightarrow\}}_{\substack{
    \textrm{wheather $\mathcal{A}$ left}\\
    \textrm{towards left or right}}}
    \times \underbrace{(\Gamma + \epsilon)}_{\substack{
        \textrm{the letter that $\mathcal{A}$ outputs}\\
        \textrm{when leaving the position}
    }}
\]
For the sake of simplicity, let us assume that $\mathcal{A}$ finishes all of its runs $\dashv$,
and that it never outputs any letters in $\dashv$ or in $\vdash$. 
(It is not hard to see that every $\mathcal{A}$ can be transformed into an equivalent
$\mathcal{A'}$ that satisfies this restriction\footnote{The only problem arises when the input word is empty -- 
in this case, the restriction prohibits the automaton from outputting any letters at all. 
One way to deal with this problem is to start the translation, by applying 
the end-of-word marker prime function. This, way we equip the input word with 
a copy of $\dashv$ (followed by the actual $\dashv$). Now $\mathcal{A}'$ can use this 
extra letter to construct its output for $\epsilon$.}.)
By Lemma~\ref{lem:2way-bounded}, we know that each position admits at most $k$ visits. 
It follows that the shape of the run can be represented as a word over the following alphabet:
\[ (\{\leftarrow, \rightarrow \} \times \{\leftarrow, \rightarrow\} \times (\Gamma + \epsilon))^{\leq k}\comma\]
where each position stores its visits in chronological order. 
For example, here is the shape of the run from Example~\ref{ex:map-duplicate}:
\bigpicc{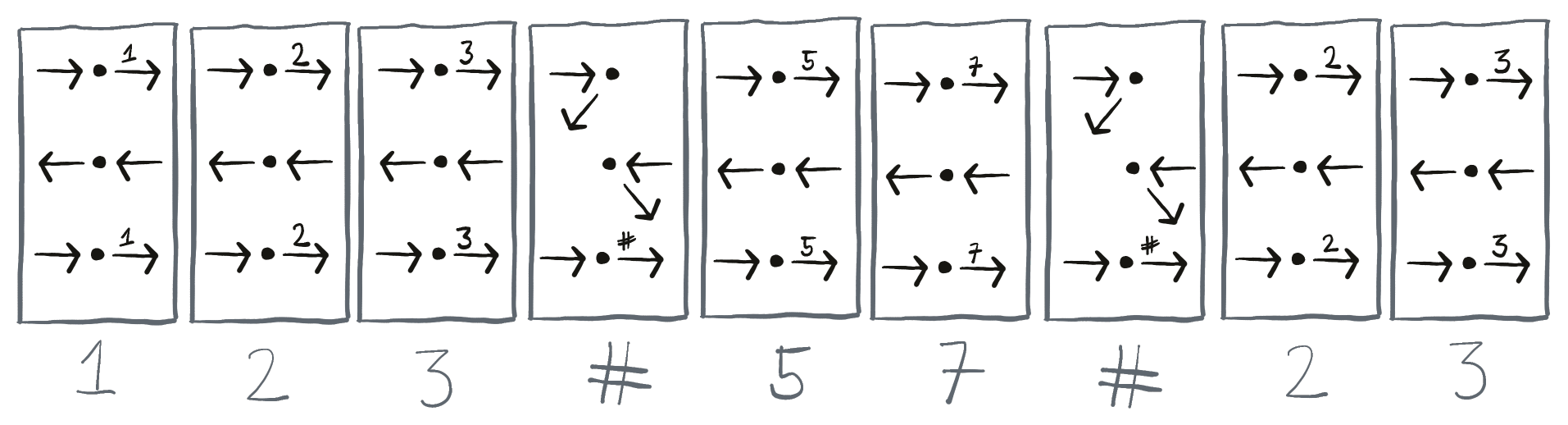}
Observe that the shape of $\mathcal{A}$'s run uniquely determines its output -- thanks
to the chronological order of the events we can retrace the steps of the automaton:
\bigpicc{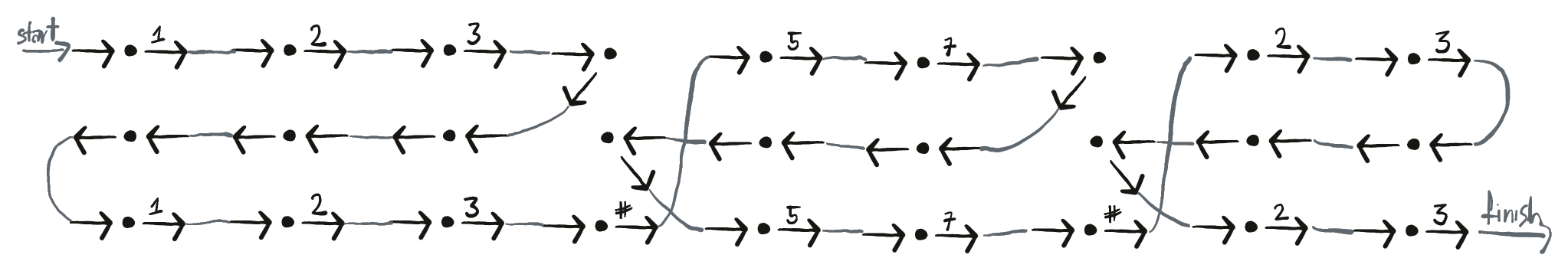}
We split the construction of $\mathcal{A}$ as a composition of primes into two steps. 
First, we show how to use compositions of primes to construct the shape 
of $\mathcal{A}$'s run over the input word. Then, we show 
how to use compositions of primes to transform the shape of a run 
into the output.

\subsubsection{Constructing the shape of the run}
In this section, we show that the following function, that outputs the shape of $\mathcal{A}$'s run over its input word,
can be constructed as a composition of primes:
\[ f_\textrm{$\mathcal{A}$-shape} : \Sigma^* \ \to \ {\left({(\{\leftarrow, \rightarrow\} \times \{\leftarrow, \rightarrow\}  \times (\Gamma + \epsilon))}^{\leq k}\right)}^*\]
It is enough to show that $f_{\mathcal{A}\textrm{-shape}}$ is a local rational semigroup transduction.
This is because, by Theorem~\ref{thm:rational-kr}, every local rational semigroup transduction 
can be decomposed into single-use rational primes, which can be further decomposed into single-use two-way primes:
\begin{claim}
\label{claim:rational-primes-two-way-primes}
    Every composition of single-use rational primes is also a composition of single-use two-way primes.
\end{claim} 
\begin{proof}
    Thanks to Claim~\ref{claim:seq-comppositions-of-primes},
    it is enough to show 
    that for every rational single-use prime function $p$, the function
    $p \times \idf_\Sigma$ is a composition of single-use two-way primes.
    For homomorphisms and left-to-right functions, this is trivial. 
    For the right-to-left multiple-use bit propagation, we can use the following decomposition:
    \[  f_{\overleftarrow{\textrm{prop}}} \times \textrm{id}_\Sigma  = \textrm{reverse} \circ (f_{\overrightarrow{\textrm{prop}}} \times \textrm{id}_\Sigma  ) \circ \textrm{reverse}\comma\]
    where the reverse function is easily seen to be a special case of $f_{\textrm{map-reverse}}$ with no $\#$'s.
    We finish the proof, by observing that we can use the same approach for all other right-to-left rational prime functions.
\end{proof}

This leaves us with showing that $f_\textrm{$\mathcal{A}$-shape}$ is a local rational semigroup transduction. 
For that we take $S$ to be the semigroup of $\mathcal{A}$'s behaviours as described in Section~\ref{sec:su-two-way-to-monoids}:
\[ \begin{tabular}{ccc}
    $ \underbrace{Q}_{\substack{\textrm{In what state}\\
    \textrm{does } \mathcal{A} \textrm{ enter }\\
    \textrm{the word}}} \times
    \underbrace{\{\leftarrow, \rightarrow\}}_{\substack{
        \textrm{Does } \mathcal{A} \textrm{ enter}\\
        \textrm{from the right}\\
        \textrm{from the left}
    }}$ & $\suto$ & $\underbrace{Q}_{\substack{\textrm{In what state}\\
    \textrm{does } \mathcal{A} \textrm{ exit }\\
    \textrm{the word}}} \times  \underbrace{\{\leftarrow, \rightarrow\}}_{\substack{
        \textrm{Does } \mathcal{A} \textrm{ exit }\\
        \textrm{from the right}\\
        \textrm{from the left}}} \times
        \underbrace{(\Gamma + \epsilon)}_{\substack{\textrm{What letter}\\
        \textrm{does } \mathcal{A} \textrm{ outputs}\\
        \textrm{when it leaves}\\
        \textrm{the word } }} + \textrm{ finish}$
\end{tabular} \]

It is not hard to see that for every $u, w, v \in \Sigma^*$, the $w$-part of $\mathcal{A}$'s 
run on $uwv$ depends only on $w$ and on $\mathcal{A}$'s behaviours of $u$ and $v$. It follows 
that there is a function:
\[ \lambda : S \times \Sigma \times S \eqto (\{\leftarrow, \rightarrow \} \times \{\leftarrow, \rightarrow\} \times \Gamma)^{\leq k}\comma \]
that computes the shape of the run on the single-letter infix. (It is not hard to see that this function is equivariant.)
This means that $f_\textrm{$\mathcal{A}$-shape}$ can be implemented as $(S, h, \lambda)$, where $S$ and $\lambda$ are as described above, 
and $h$ is a function that maps single-letter words to their behaviours. The hard part of the proof is showing that $\lambda$
satisfies the locality equation, i.e. that: 
\[ \lambda(x_1 e y_1, a, y_2 e x_2) =  \lambda(\pi(x_1)\, e \, y_1, a, y_2\, e\, \pi(x_2))\comma \]
provided that $e$ is idempotent, $y_1 \cdot h(a) \cdot y_2 = e$, and $\pi$ is a $\supp(e)$-permutation.
We show this by using a slightly stronger result: Remember that the $w$-part of the shape 
of $\mathcal{A}$'s run on $uwv$ depends only on $w$ and on $\mathcal{A}$'s behaviours on $u$ and $v$. 
It follows that we can extend $\lambda$ to:
\[ \lambda' : S \times \Sigma^* \times S \eqto {\left(\{\leftarrow, \rightarrow \} \times  (\{\leftarrow, \rightarrow\} \times \Gamma)^{\leq k}\right) }^* \comma \]
Now, in order to prove the locality of $\lambda$, we take some words $\bar{y}_2, \bar{y}_2 \in \Sigma^*$ whose 
$\mathcal{A}$-behaviours are equal to $y_1$ and $y_2$ (thanks to a reasoning 
similar to  Claim~\ref{claim:non-full-to-full}, we can assume that $S$ only contains behaviours 
correspond to actual words) and apply the following lemma for $x_1, x_2, e$ and $w = \bar{y}_2 a \bar{y}_2$:
\begin{lemma}
\label{lem:lambda-prim-local}
    Let $e \in S$ be an idempotent behaviour, and let  
    $w$ be a word whose behaviour is equal to $e$.
    For all behaviours $x_1, x_2$ and for every $\supp(e)$-permutation $\pi$, it holds that:
    \[ \lambda'(x_1e, w, ex_2) = \lambda'(\pi(x_1)\,e, w, e\,\pi(x_2)) \]
\end{lemma}
\begin{proof}
    In order to prove the lemma, we pick some $\bar{x}_1, \bar{x}_2, \bar{e} \in \Sigma^*$
    (again, thanks to a similar reasoning as in Claim~\ref{claim:non-full-to-full}, we can assume that 
    those words exist), and we trace and compare the runs of $\mathcal{A}$ on 
    $\vdash \! \bar{x}_1\, \bar{e}\, w\, \bar{e}\, \bar{x}_2 \! \dashv$ and on
    $\vdash \! \pi(\bar{x}_1)\, \bar{e}\, w\, \bar{e}\, \pi(\bar{x}_2) \! \dashv$.\\

    Consider the first part of the run, that starts in the initial state $q_0$ 
    and ends when $\mathcal{A}$ enters the $w$-part of its input:
    \picc{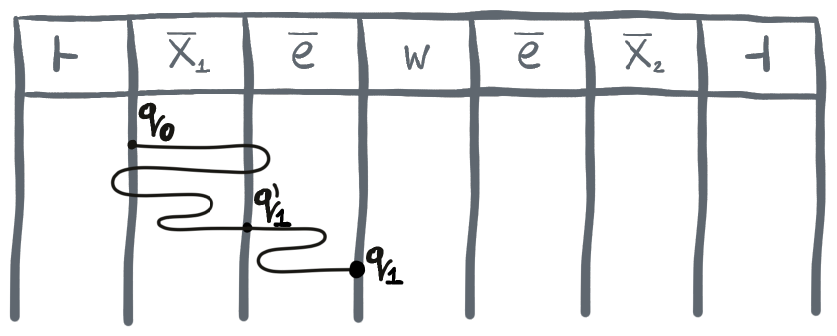}
    First, let us consider the run on $\vdash \! \bar{x}_1\, \bar{e}\, w\, \bar{e}\, \bar{x}_2 \! \dashv$:
    We define $q_1$ be the state in which $\mathcal{A}$ first enters $w$ and $q_1'$ to be the last state in which $\mathcal{A}$ enters
    $\bar{e}$, before entering $w$ (see the picture). Let us now notice that 
    analogous states for  $\vdash \! \pi(\bar{x}_1)\, \bar{e}\, w\, \bar{e}\, \pi(\bar{x}_2) \! \dashv$
    are equal to $\pi(q_1)$ and $\pi(q_1')$.  This is because 
    those states depend equivariantly on the behaviour of the prefix up until $w$, 
    which in this case is equal to $\pi(x_1) \cdot e = \pi(x_1) \cdot \pi(e) = \pi(x_1 \cdot e)$.\\
    
    Consider now the second part of the run, which lasts until $\mathcal{A}$ leaves the $\bar{e}w\bar{e}$-part of the input. 
    Let us show that this part of the run exits the $\bar{e}w\bar{e}$-part of the input on the right in $q_1$
    (or $\pi(q_1)$):
    \picc{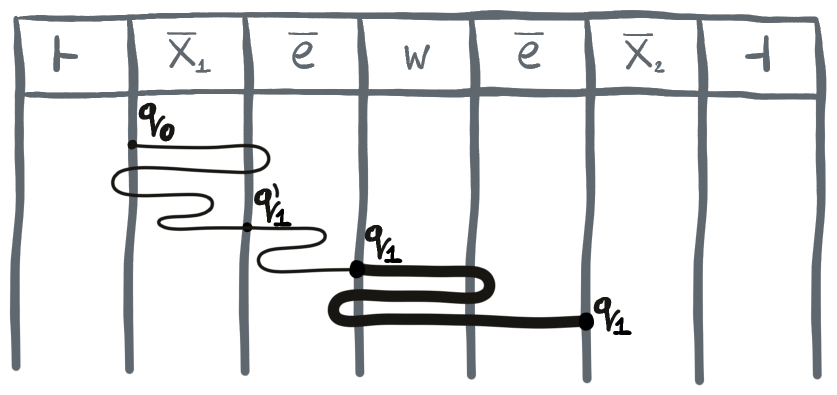}
    Observe that the behaviours of both $\bar{e}$ and $w$ are equal to $e$. 
    Since $e$ is idempotent, it follows that the behaviour of the word $\bar{e}w\bar{e}$ is also equal to $e$. 
    By definition of $q_1'$, we know that $e(q_1', \rightarrow) = (q_1, \rightarrow)$.
    It follows that the second part of the run exits $\bar{e}w\bar{e}$ from the right in state $q_1$:
    \[(\bar{e}w\bar{e})(q_1', \rightarrow) = e(q_1', \rightarrow) = (q_1, \rightarrow)\]

    Now, let us show that during this second part of the run (marked as a bold line), $\mathcal{A}$ 
    has to preserve all atoms from $q_1$ that do not appear in $\supp(e)$ -- in particular, 
    this means that $\mathcal{A}$ cannot query or output those atoms. The proof is analogous 
    to the one in Section~\ref{sec:su-mealy-to-local-monoid-transduction}. 
    First, let us consider the following function, which describes the behaviour of $\mathcal{A}$
    on $\bar{e}w\bar{e}$, when it enters $w$ from the left:
    \[ (\bar{e} \toprightarrow w \bar{e}) : \underbrace{Q}_{ \substack{
        \textrm{The state in which}\\
        \textrm{$\mathcal{A}$ is placed in}\\
        \textrm{the first letter of $w$.}
    }
    } \suto \underbrace{Q \times \{\rightarrow, \leftarrow\}}_{ \substack{
        \textrm{The state and the direction}\\
        \textrm{in which $\mathcal{A}$ exits $\bar{e}w\bar{e}$}
    }
    }\]
    By an argument similar to the one in 
    Claim~\ref{claim:two-way-behaviour-single-use}, we know that $(\bar{e} \toprightarrow w \bar{e})$ is, indeed, a single-use function.
    Moreover, it is not hard to see 
    that $(\bar{e} \toprightarrow w \bar{e})$ depends (in an equivariant way) only on the behaviours on $\bar{e}$ and $w$, 
    which are both equal to $e$.
    By Lemma~\ref{lem:fs-functions-preserve-supports} it follows that:
    \[\supp(\bar{e} \toprightarrow w \bar{e} ) \subseteq \supp(e)\]
    Moreover since $e(\rightarrow, q_1') = q_1$, we know that:
    \[ (\bar{e} \toprightarrow w\bar{e})(q_1) = (ewe)(\rightarrow, q_1') = e(\rightarrow, q_1') = (\rightarrow, q_1) \]
    As the second part of the run corresponds to $(\bar{e} \toprightarrow w\bar{e})(q_1)$, 
    it follows that it can only destroy only those atoms from $q_1$ that appear in $\supp(e)$.
    This is because each atom from $q_1$ that is destroyed during the second part of the run 
    has to be restored before $\mathcal{A}$ exits $\bar{e}w\bar{e}$,
    as $(\bar{e}\toprightarrow w \bar{e})(q_1) = q_1$. 
    By a reasoning similar to the one from Section~\ref{sec:su-mealy-to-local-monoid-transduction},
    we know that each such restored atom has to appear in $\supp(\bar{e} \toprightarrow w \bar{e})$,
    and we know that $\supp(\bar{e} \toprightarrow w\bar{e}) \subseteq \supp(e)$.\\

    Now, let us consider the second part of $\mathcal{A}$'s run on $\vdash\pi(\bar{x}_1)\, \bar{e} w \bar{e} \, \pi(\bar{x}_2) \dashv$,
    i.e. the part that starts in $\pi(q_1)$ (in the first letter of $w$), and ends when $\mathcal{A}$ leaves $\bar{e}w\bar{e}$.
    Thanks to the same arguments as for $\vdash \! \bar{x}_1 \bar{e} w \bar{e} \bar{x}_2 \! \dashv$, 
    we know that the second part of the run also leaves $\bar{e}w\bar{e}$ from the right in $\pi(q_1)$,
    and during the second part of the run, $\mathcal{A}$ has to preserve all atoms from $\pi(q_1)$ that 
    do not appear in $\supp(e)$. In particular, this means that $\mathcal{A}$ does 
    not query or output any atoms from $\pi(q_1)$ that do not appear in $\supp(e)$.
    Since $\pi$ is a $\supp(e)$-permutation, it follows 
    that, in the second part of the run, $\mathcal{A}$ cannot distinguish between $q_1$ and $\pi(q_1)$, 
    as the only difference between the two states are the atoms outside of $\supp(e)$. 
    It follows that the second part of the run has the same shape when starting in $q_1$ and in $\pi(q_1)$.
    In particular, this means that the shape of the $w$-part in the second part of the run is the same on both
    $\bar{x}_1\bar{e}w\bar{e}\bar{x}_2$ and $\pi(\bar{x}_1)\bar{e}w\bar{e}\pi(\bar{x}_2)$.\\

    Now, let us consider the third part of the run -- from exiting $\bar{e}w\bar{e}$ in $q_1$
    (or $\pi(q_1)$) until reentering $w$:
    \picc{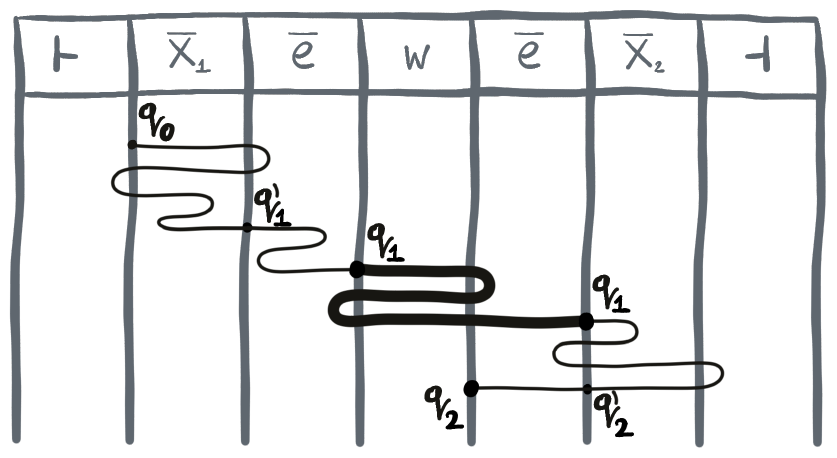}
    We define $q_2$ and $q_2'$ analogously to $q_1$ and $q_1'$, i.e. $q_2$ is the state in which $\mathcal{A}$ reenters $w$, 
    and $q'_2$ is the last state in which $\mathcal{A}$ enters $\bar{e}$, 
    before reentering~$w$. Similarly as before, we observe that the analogous states for
    $\vdash \! \pi(\bar{x}_1)\, \bar{e}\, w\, \bar{e}\, \pi(\bar{x}_2)\! \dashv$ are equal to $\pi(q_2)$ and $\pi(q_2')$.\\
    
    Let us now consider the fourth part of the run, which starts in $q_2$ (or $\pi(q_2)$)
    and continues until $\mathcal{A}$ leaves $\bar{e}w\bar{e}$:
    \picc{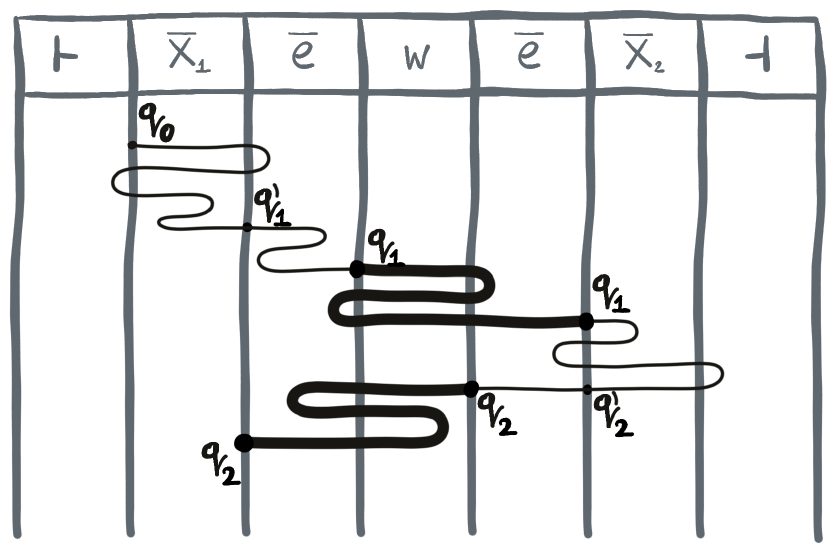}
    An analysis, similar to the one for the second part of the run, shows that $\mathcal{A}$
    leaves $\bar{e}w\bar{e}$ from the left in $q_2$ (or in $\pi(q_2)$),
    and that both in the run on $\bar{x}_1\bar{e}w\bar{e}\bar{x}_2$ and on $\pi(\bar{x}_1)\bar{e}w\bar{e}\pi(\bar{x}_2)$,
    $\mathcal{A}$ cannot query atoms from $q_2$ (or $\pi(q_2)$) that do not appear in $\supp(e)$. 
    It follows, by the same argument as before, that the fourth part of the run has the same shape
    for both $\vdash\! \bar{x}_1 \bar{e} w \bar{e}  \bar{x}_2\! \dashv$ and $\vdash\! \pi(\bar{x}_1)\, \bar{e}\, w\, \bar{e}\, \pi(\bar{x}_2)\! \dashv$.
    In particular, this means that during the fourth part of the run, the $w$-parts of the two runs have equal shapes.\\

    We finish the proof by observing that a similar reasoning can be continued 
    until the end of the two runs.
\end{proof}

\subsubsection{Untangling the run graphs}
In this section, we show how to use compositions of single-use two-way primes to untangle
the run graphs:
\begin{lemma}
\label{lem:untangle-compositions-of-primes}
    For every polynomial orbit-finite $\Gamma$, and for every $k \in \nat$, the following function can be constructed as a composition of single-use two-way primes:
    \[ f_{\textrm{untangle}}: \ \underbrace{ {\left(\left(\{\leftarrow, \rightarrow\}^2 \times (\Gamma + \epsilon)\right)^{\leq k}\right)}^*}_{\textrm{encoding of a shape of the run}} \to
    \underbrace{\Gamma^*}_{\textrm{\textrm{the untangled string}}} \]
    We assume that the input shape is a single path, i.e. 
    all nodes except the initial one have exactly one successor,
    and all nodes except the final one have exactly one predecessor.
\end{lemma}

\noindent
For example, consider the following input:
\bigpicc{two-way-example-shape}
\noindent
It corresponds to the following \emph{run graph}:
\bigpicc{two-way-example-shape-connected}
\noindent
Which means that the function $f_\textrm{untangle}$ should return the following word:
\[ 123123\#5757\#2323\]

This section is entirely dedicated to proving Lemma~\ref{lem:untangle-compositions-of-primes}. 
We present the same proof as in \cite[Section E.4.3]{single-use-paper}.
Interestingly, the proof looks almost the same as it would 
for a finite\footnote{To the best of my knowledge, the proof for a finite
$\Gamma$ was first presented in \cite[Section~6.1]{bojanczyk2018polyregular}.
We present the proof following the lines of \cite[Section~E.4]{single-use-paper}.} $\Gamma$.\\

The proof goes by induction on $k$, 
which  represents the \emph{width} of the run graph,
i.e. the maximal number of times a position is visited.
For the induction base, we notice that there are only two possible types of run graphs for $k=1$:
\picc{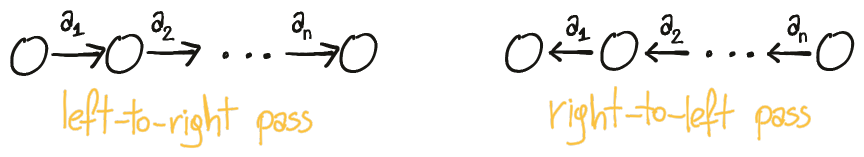}
Both of those cases are easy to untangle: The left-to-right pass is almost already 
untangled -- it suffices to apply a homomorphism that extracts the letters.
For the right-to-left pass, we can use the same homomorphism followed by the reverse function.
Unexpectedly, the hardest part of the induction base
is combining the two procedures into a single function.
We do this using the following lemma:
\begin{lemma}
\label{lem:2-way-primes-cases}
    Let $L \subseteq \Sigma^*$ be a language recognized by a single-use two-way automaton.
    If $f_1 : \Sigma^* \to \Gamma^*$ and $f_2 : \Sigma^* \to \Gamma^*$
    are both compositions of two-way primes, then so is the following function:
    \[ (\mathtt{if}\ L\ \mathtt{then}\ f_1\ \mathtt{else}\ f_2)(w) = \begin{cases}
            f_1(w) & \textrm{if } w \in L\\
            f_2(w) & \textrm{if } w \not \in L
        \end{cases}\] 
\end{lemma}
\begin{proof}
    The construction consists of the following six steps:
    \begin{gather*}
        \Sigma^*  \longtransform{w \mapsto w \! \dashv}
        (\Sigma + \dashv)^*  \transform{f_L} 
        (\Sigma + \{\yes, \no\})^*  \longtransform{f_\textrm{colour}} (\Sigma + \Sigma)^*\\
        (\Sigma + \Sigma)^* \longtransform{f_1 + \idf} 
        (\Gamma + \Sigma)^*  \longtransform{\idf + f_2} 
        (\Gamma + \Gamma)^*  \longtransform{\merge^*} 
        \Gamma^*,
    \end{gather*}
    \begin{enumerate}
        \item  First, we equip the input word with the end-of-word marker.
               This function (i.e. $w \mapsto w\! \dashv$) is easily seen to be a composition 
               of a rational transduction, that underlines the last letter of the input,
               with a letter-to-word homomorphism that inserts $\dashv$ after the underlined letters. 
               (Thanks to Claim~\ref{claim:rational-primes-two-way-primes}, the rational 
               transduction can be further decomposed into two-way primes.)\\
        \item  Next, we transform the end of word marker $\dashv$ into either 
               $\yes$ or $\no$ depending on whether the input word belongs 
               to $L$. This step is implemented as the following function $f_L : (\Sigma + \dashv)^* \to (\Sigma + \{\yes, \no\})^*$:
               \[
                   f_L(w \dashv) = \begin{cases}
                       w\,\yes & \textrm{if } w \in L\\
                       w\,\no & \textrm{if } w \not \in L\\
                   \end{cases}
               \]
               To see that $f_L$ is a composition of primes, observe that 
               thanks to Theorem~\ref{thm:dsua-2dsua-ofm}, $L$ can be recognized by a one-way single-use automaton. 
               This automaton can be easily modified into a single-use Mealy machine 
               that recognizes $f_L$. It follows by Theorem~\ref{thm:kr} and Claim~\ref{claim:rational-primes-two-way-primes}
               that $f_L$ is a composition of single-use two-way primes. 
        \item Next, we propagate the output of $f_L$, by colouring the entire word into 
              blue or yellow, depending on whether its last letter is equal to $\yes$ or $\no$.
              For this, we use the following function:
              \[ f_{\textrm{colour}} : (\Sigma + \{\yes, \no\})^* \to (\underbrace{\Sigma}_{\textrm{blue copy}} + \underbrace{\Sigma}_{\textrm{yellow copy}})^*\comma\]
              It is not hard to see that $f_\textrm{colour}$ can be implemented by a right-to-left single-use Mealy machine
              composed with a homomorphism that forgets about the $\{\yes, \no\}$-letter.
              It follows, by (an analogue) of Theorem~\ref{thm:kr} and Claim~\ref{claim:rational-primes-two-way-primes} that
              $f_\textrm{colour}$ is a composition of single-use two-way primes. 
        \item Next, we apply the function $(f_1 + \idf)$, defined in the following claim:
        \begin{claim}
            If $f : \Sigma^* \to \Gamma^*$ is a composition of single-use two-way primes then so is 
            the following function $(f + \idf ) : (\Sigma + \Delta)^* \to (\Gamma + \Delta)^*$:
            \[ \begin{tabular}{c}
                $(f + \idf)(w) = \begin{cases}
                    f(w) & \textrm{if all letters in } w \textrm{ belong to } \Sigma\\
                    w & \textrm{if all letters in } w \textrm{ belong to } \Delta\\
                    \emph{\textrm{(unspecified)}} & \textrm{if } w \textrm{ contains letters from both } \Sigma \textrm{ and } \Delta\\
                \end{cases}$
            \end{tabular}\]
        \end{claim}
        \begin{proof}
            The proof goes by induction on the construction of $f$ as a composition of single-use two-way primes.
            The induction step follows from the following observation:
            \[(g \circ h) + \idf  = (g + \idf ) \circ (h + \idf)\comma\]
            This leaves us with the induction base,  which states that  for every single-use two-way 
            prime $p$, the function $(p + \idf)$ is a composition of single-use two-way primes. We only 
            show it for $p = f_\textrm{su-prop}\times\idf_X$, and $p = f_\textrm{map-rev}$, as 
            other cases are either trivial or analogous.\\
    
            First, we show how to construct $f_\textrm{su-prop}\times\idf_X \, + \, \idf$. Here is its type:
            \[ \left( (\atoms + \{\downarrow, \epsilon\}) \times X \, + \, \Delta  \right)^*\ \ \to \ \  \left((\atoms + \epsilon) \times X  \, + \, \Delta\right)^*  \]
            We start the construction with a homomorphism that equips every letter from $\Delta$ with $\epsilon$ -- i.e. the neutral letter of $f_\textrm{su-prop}$. 
            This gives us a word over $(\atoms + \{\downarrow, \epsilon\}) \times (X + \Delta)$. 
            Then, we apply $f_\textrm{su-prop} \times \idf_{X + \Delta}$, obtaining a word over $(\atoms + \epsilon) \times (X + \Delta)$. 
            Finally, we use homomorphism to remove the $\epsilon$'s from $\Delta$.\\
    
            \noindent
            This leaves us with constructing $(f_{\textrm{map-rev}} + \idf)$, whose type is:
            \[  ((\Sigma + \#) + \Delta)^* \to ((\Sigma + \#) + \Delta)^*\]
            We start the construction with a letter-to-word homomorphism that expands every $a$ from $\Delta$ into $\#a\#$,
            and keeps elements of $(\Sigma + \#)$ unchanged. 
            Then, we apply $f_\textrm{map-rev}$. Finally, we remove all $\#$'s that are adjacent to at least one letter from $\Delta$.
            (This last step is a composition of a local rational semigroup transduction
            that underlines neighbours of $\Delta$, and a letter-to-word homomorphism that removes underlined $\#$'s.)
        \end{proof}
    \item Next, we apply $\idf + f_2$, defined analogously to $f_1 + \idf$. 
    \item Finally, we forget about the colours by applying homomorphism:
    \[\merge^* : (\Gamma + \Gamma)^* \to \Gamma^*\]
    \end{enumerate}
\end{proof}

This finishes the proof of the induction base for Lemma~\ref{lem:untangle-compositions-of-primes}.
We start the proof of the induction step, with the special case of \emph{right loops}, which 
are those run graphs whose both initial and final 
nodes belong to the first position. In order to untangle 
a right loop, let us consider the following way
of dividing it into two parts: the first one 
contains all the nodes up to (and including)
the first visit in the last position, 
and the second one contains all other nodes. 
Here is an example:
\picc{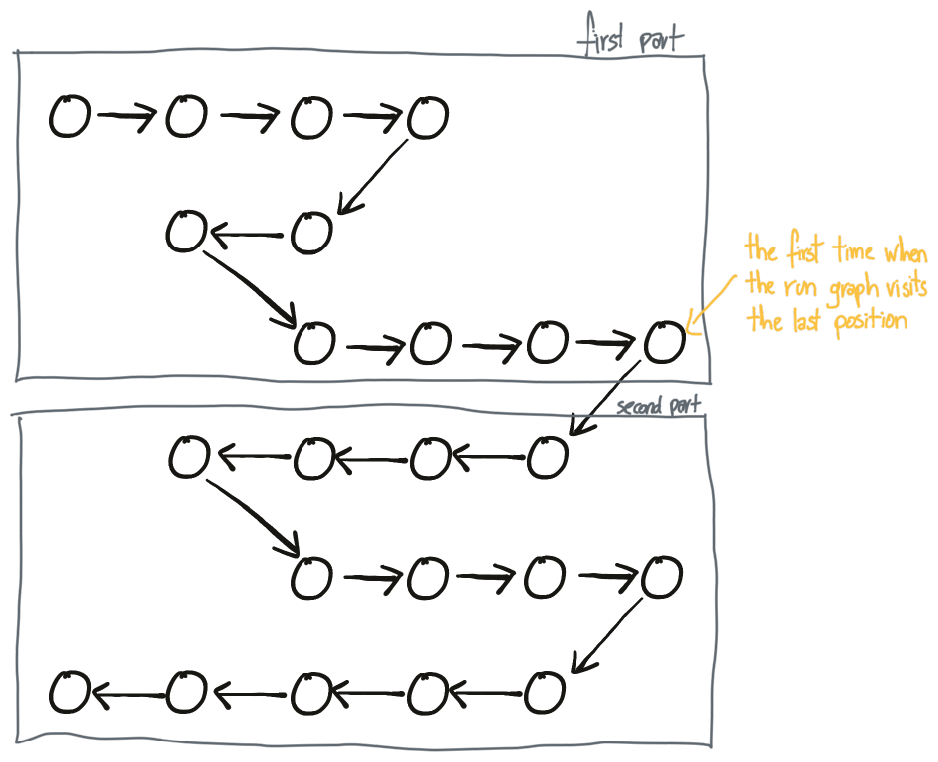}
The idea behind this division is that the width of either part
is smaller than the width of the whole run graph -- 
this will enable us to apply the induction assumption.
Moreover, the division can be constructed as a composition of primes:
\begin{lemma}
\label{lem:loop-decomposition}
    The following function $f_\textrm{loop-div}$, which inputs a right loop and
    splits it into two $\#$-separated parts (as described earlier), can be constructed as a composition of single-use two-way primes.
    \[ f_\textrm{loop-div} : {\left((\{\leftarrow, \rightarrow\}^2 \times (\Gamma + \epsilon))^{\leq k}\right)}^* \to \left((\{\leftarrow, \rightarrow\}^2 \times  (\Gamma + \epsilon) )^{\leq k-1} + \# \right)^*\comma  \]
\end{lemma}
\noindent
Here is an example:
\bigpicc{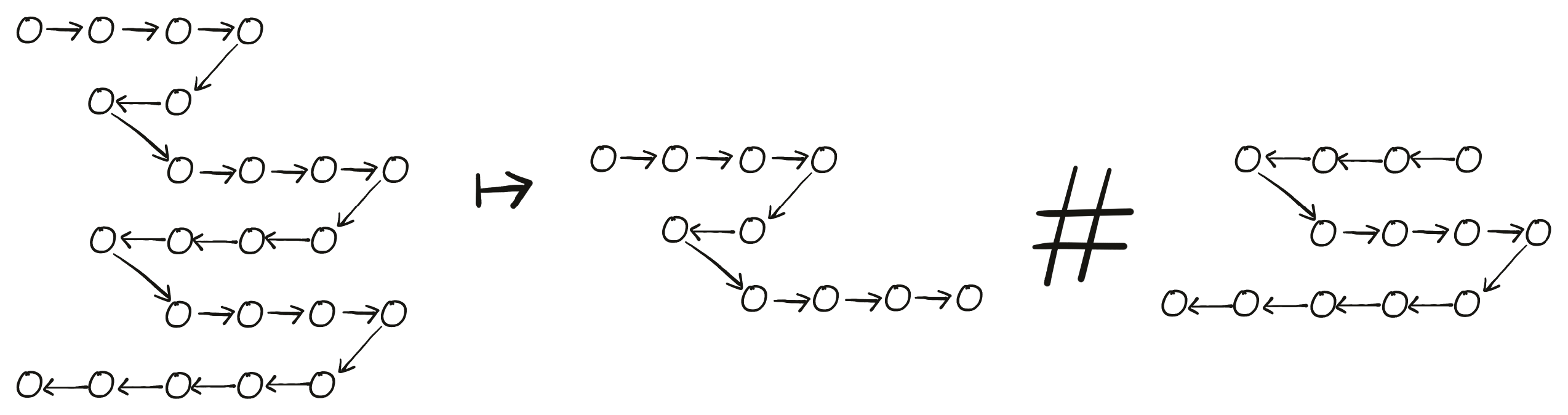}
\begin{proof}
    First, we show how to colour each input node into yellow or blue, depending 
    on whether it belongs to the first or to the second part of the decomposition. We start 
    with the special case where $\Gamma$ is the singleton set, i.e. $\Gamma = 1$.
    In this case, both the input and the output alphabets are finite, so 
    it is enough to construct the colouring as an unambiguous Mealy machine -- 
    thanks to Lemma~\ref{lem:rational-equiv-unambigous}, Theorem~\ref{thm:classical-rational-kr}, 
    and Lemma~\ref{claim:rational-primes-two-way-primes}, we know that 
    the unambiguous Mealy machine can be decomposed into single-use two-way primes.\\

    The unambiguous Mealy machine uses nondeterminism to guess the colour of each node,
    and verifies that the colouring is correct, by checking a few local conditions:
    \begin{enumerate}
        \item the initial node is yellow;
        \item the first node in the last position is yellow, and its successor is blue;
        \item no yellow node is followed by a blue node.
    \end{enumerate}
    The machine is unambiguous because there is only one correct colouring.\\

    Now, let us go back to the general case: Since the output alphabet is now infinite, we cannot use an unambiguous Mealy machine.
    However, it is not hard to see the colouring does not depend on the $\Gamma$-values, 
    which means that we can reduce the polynomial orbit-finite case to the finite case: 
    First, we apply a homomorphism $f^*_{\Gamma-\textrm{extr}}$, where $f_{\Gamma-\textrm{extr}}$ is 
    a function that splits each input letter into its $\Gamma$-free shape and its $\Gamma$-labels:
    \[ f_{\Gamma-\textrm{extr}}\ \ : \ \  \underbrace{(\{\leftarrow, \rightarrow\}^2 \times (\Gamma + \epsilon))^{\leq k}}_{
        \textrm{The input letter}
    } \ \ \longrightarrow \ \  \underbrace{(\{\leftarrow, \rightarrow\}^2)^{\leq k}}_{
        \substack{\textrm{The $\Gamma$-free shape part of}\\ \textrm{the input letter} }
    }\times \underbrace{(\Gamma + \epsilon)^{\leq k}}_{
        \substack{\textrm{The labels of}\\\textrm{the input letter}}
    }\]
    Next, we use an unambiguous Mealy machine to construct the colouring for the $\Gamma$-free version
    (using Claim~\ref{claim:unabigous-finite-parallel} defined below) and finally, we use a homomorphism 
    to transfer the $\Gamma$-labels back to the run graph. 
    \begin{claim}
        \label{claim:unabigous-finite-parallel}
            Let $A$ and $B$ be finite alphabets, and let $f: A^* \to B^*$ be a function recognized by an unambiguous Mealy machine. 
            It follows that for every polynomial orbit-finite $C$, the following function is a composition 
            of rational single-use primes (which by Lemma~\ref{claim:rational-primes-two-way-primes} means 
            that it is also a composition of single-use two-way primes):
            \[ f \times \idf_C : (A \times C)^* \to (B \times C)^* \]
    \end{claim}
    \begin{proof}
        The claim is a direct consequence of Lemma~\ref{lem:rational-equiv-unambigous} and Theorem~\ref{thm:rational-kr}.
    \end{proof}

    Once we have coloured the nodes, we can easily produce the output of $f_\textrm{loop-div}$:
    First, we duplicate the coloured run (using $f_\text{map-dup}$ with no separators),
    and then we apply a single-use  Mealy machine that filters out
    all blue nodes from the first copy and all yellow nodes from the second copy.
\end{proof}

After dividing the right loop into the two parts, we can untangle it  by 
independently untangling each of the parts (using $f_\textrm{untangle}$ from the induction assumption),
and concatenating the results:
\[ \left((\{\leftarrow, \rightarrow\}^2 \times \Gamma)^{\leq k-1} + \# \right)^* \longtransform{\mathtt{map}\; f_\textrm{untangle}} (\Gamma + \#)^* \transform{\textrm{concat}} \Gamma^* \]
The list-flattening function $\textrm{concat} : (\Gamma + \#)^* \to \Gamma$ is a simple letter-to-word homomorphism that filters out all $\#$'s,
and the $\mathtt{map}$ combinator is defined by the following lemma:
\begin{lemma}
\label{lem:2-way-primes-map}
    If $f : A^* \to B^*$ is a composition of single-use two-way primes,
    then so is the following function $\mathtt{map} \;f$, which applies $f$ independently 
    to each $\#$-separated block:
    \[ \mathtt{map} \;f : (A + \#)^* \to (B + \#)^* \]
\end{lemma}
\begin{proof}
We start the proof by noticing that:
\[ \mathtt{map}(f \circ g) = (\mathtt{map\;  } f) \circ (\mathtt{map\; } g)\]
This leaves us with showing that for every prime function $p$, the function $\mathtt{map}\; p$ is a composition of primes. 
Most of the cases are either easy or handled analogously as in the proof of Lemma~\ref{lem:map-combinator-primes}.
The only interesting cases are $p = f_\textrm{map-rev}$ and $p = f_\textrm{map-dup}$. Moreover, the two 
cases are analogous, so we only show how to construct $\mathtt{map}\; f_\textrm{map-rev}$. Observe 
that it uses two types of separators:
\[ (\Sigma + \underbrace{\#_1}_{
    \substack{\textrm{separator}\\
               \textrm{for } \mathtt{map}
    }
} + \underbrace{\#_2}_{\substack{
    \textrm{separator}\\
    \textrm{for } f_\textrm{map-rev}
}
})^* \to (\Sigma + \underbrace{\#_1}_{
    \substack{\textrm{separator}\\
               \textrm{for } \mathtt{map}
    }
} + \underbrace{\#_2}_{\substack{
    \textrm{separator}\\
    \textrm{for } f_\textrm{map-rev}
}
})^*\]
However, as one can easily verify, both of those separators are treated in the same way:
$\mathtt{map }\; f_\textrm{map-rev}$ is a version of $f_\textrm{map-rev}$ that
treats both $\#_1$ and $\#_2$ as its separator. It follows that we can implement $\mathtt{map }\; f_\textrm{map-rev}$
by mapping both $\#_1$ and $\#_2$ to $\#$, and applying $f_\textrm{map-rev}$.
Finally, we have to restore the $\#$'s back to $\#_1$ or $\#_2$.
In order for this to be possible, we need to modify the first step:
instead mapping both $\#_1$ and $\#_2$ to $\#$, we map them 
respectively to $a_1\#$ and $a_2\#$ (where $a_1$ and $a_2$ are letters that do not appear in $\Sigma$).
After this modification, we can map the $\#$'s back to $\#_1$ and $\#_2$ using a single-use Mealy machine. 
\end{proof}

Let us now deal with arbitrary \emph{loops}, i.e.
run graphs that start and finish in the same position. It is not hard to see 
that they can be split into at most $k$ left loops and right loops, as illustrated by the following example:
\vsmallpicc{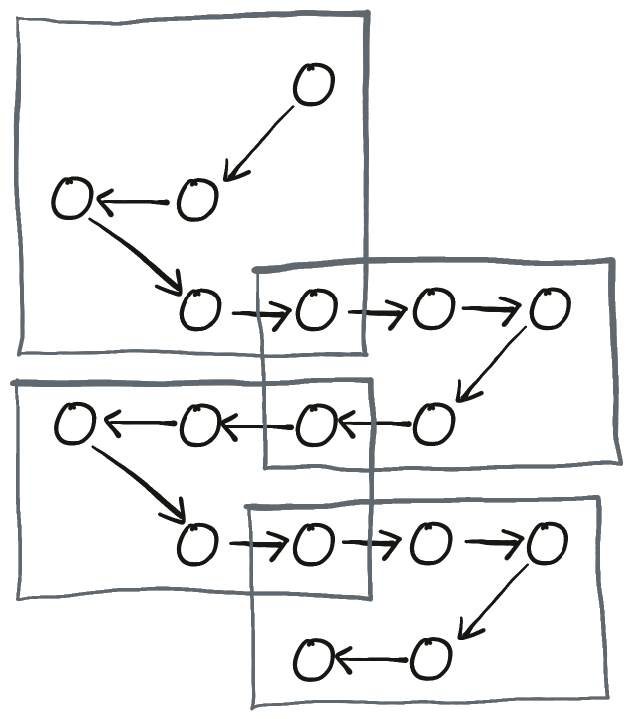}
Observe that, thanks to a reasoning similar to the one from Lemma~\ref{lem:loop-decomposition}, 
we can construct this decomposition as a composition of single-use two-way primes. 
This way, we reduce untangling an arbitrary loop into 
untangling right loops and left loops. This finishes the construction, 
as left loops can be untangled analogously to right loops.\\

Finally, let us show how to untangle arbitrary run graphs. Without loss of generality, 
we assume that the initial node is to the left of the final node -- 
the other case can be handled analogously.  First, let us inductively define
\emph{stations}, \emph{sweeps} and \emph{loops} of a run graph:
The \emph{first station} is the position that contains the initial node.
The \emph{$i$-th loop} is the part of the run between the first and the last 
visit in the $i$-th station. The \emph{$(i + 1)$-st station} is the first position
to the right of the $i$-th station, that was not visited by the $i$-th loop.
Finally, the \emph{$i$-th sweep} is the part of the run graph between the 
last visit in the $i$-th station and the first visit in the $(i + 1)$-th station.
Here is an example\footnote{
    Note, that the sweeps might contain U-turns (even though the sweeps in the example do not contain them).
}:
\bigpicc{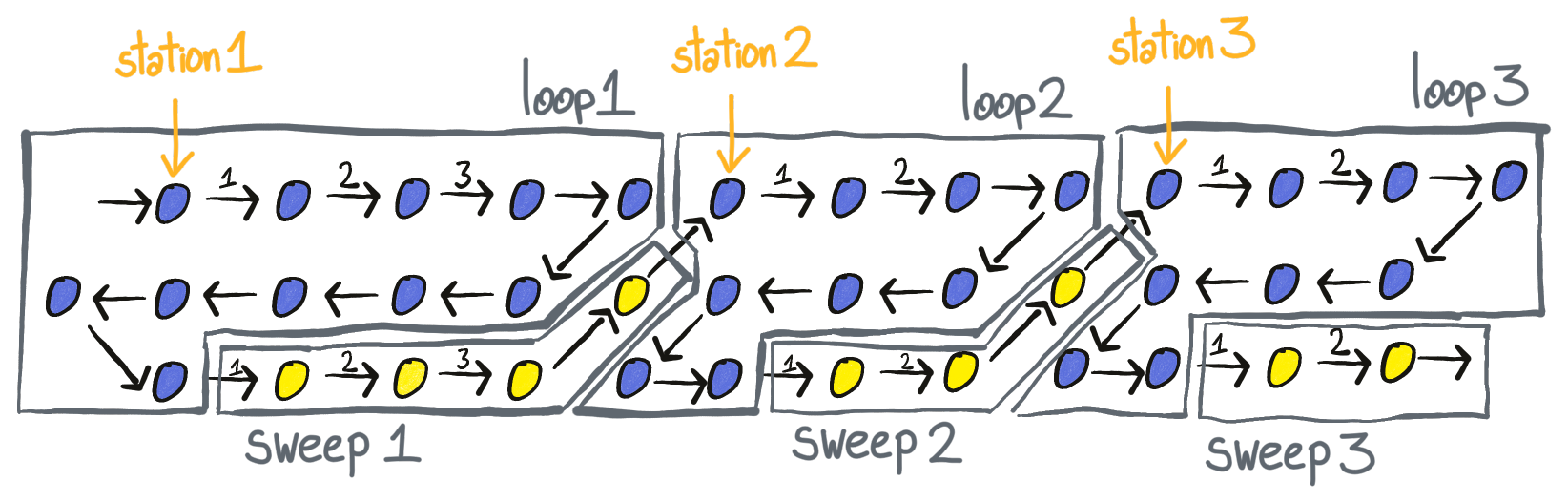}
Notice that every position that is visited by the $i$-th sweep is also visited by the
$i$-th loop, so the width of each sweep is smaller than the width of the original graph -- 
this will allow us to untangle the sweeps using the induction assumption.\\

\noindent
Here is the procedure for untangling a run graph:
\begin{enumerate}
    \item First, we underline all stations, and colour each node into either yellow or
          blue, depending on whether it belongs to a loop or to a sweep -- 
          the construction is analogous to the one for the right-loop decomposition 
          from Lemma~\ref{lem:loop-decomposition}.
          
    \item Next, we transform a run graph into a $\#$-separated list of its \emph{windows},
          where the $i$-th window is defined to be the maximal interval that contains the 
          $i$-th station and no other station. (Note that the windows are usually overlapping.) 
          Here is an example:\\
          \bigpicenum{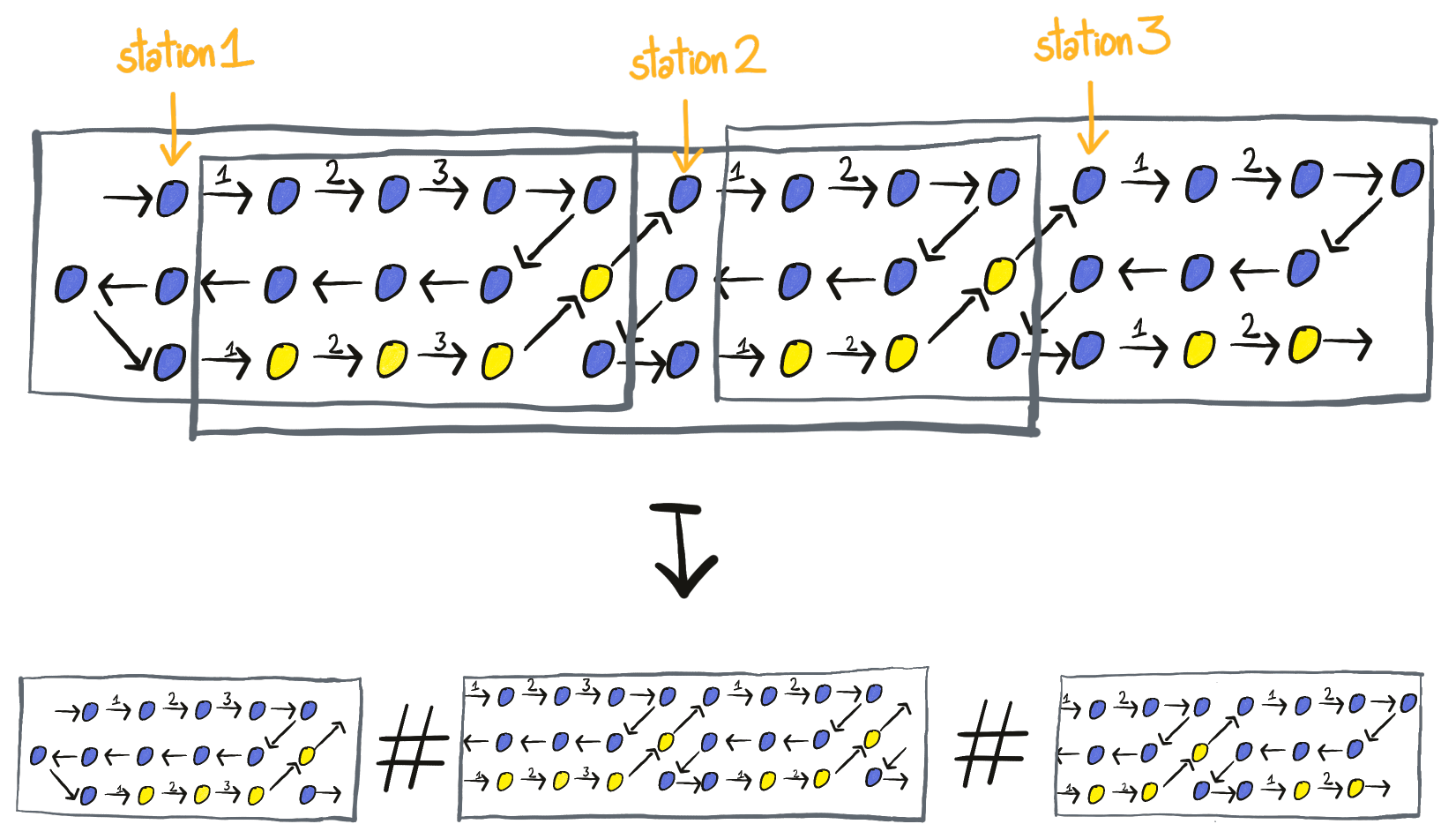}\\

          We do this by surrounding every station with $\#$'s,
          applying the map duplicate function, and cleaning
          up the output with a single-use Mealy machine.
    \item Observe that the $i$-th window contains the entire $i$-th 
          loop and $i$-th sweep. In this step, we extract the loop 
          and the sweep from each of the windows. 
          For example, for the first window, this looks as follows:\\
          \bigpicenum{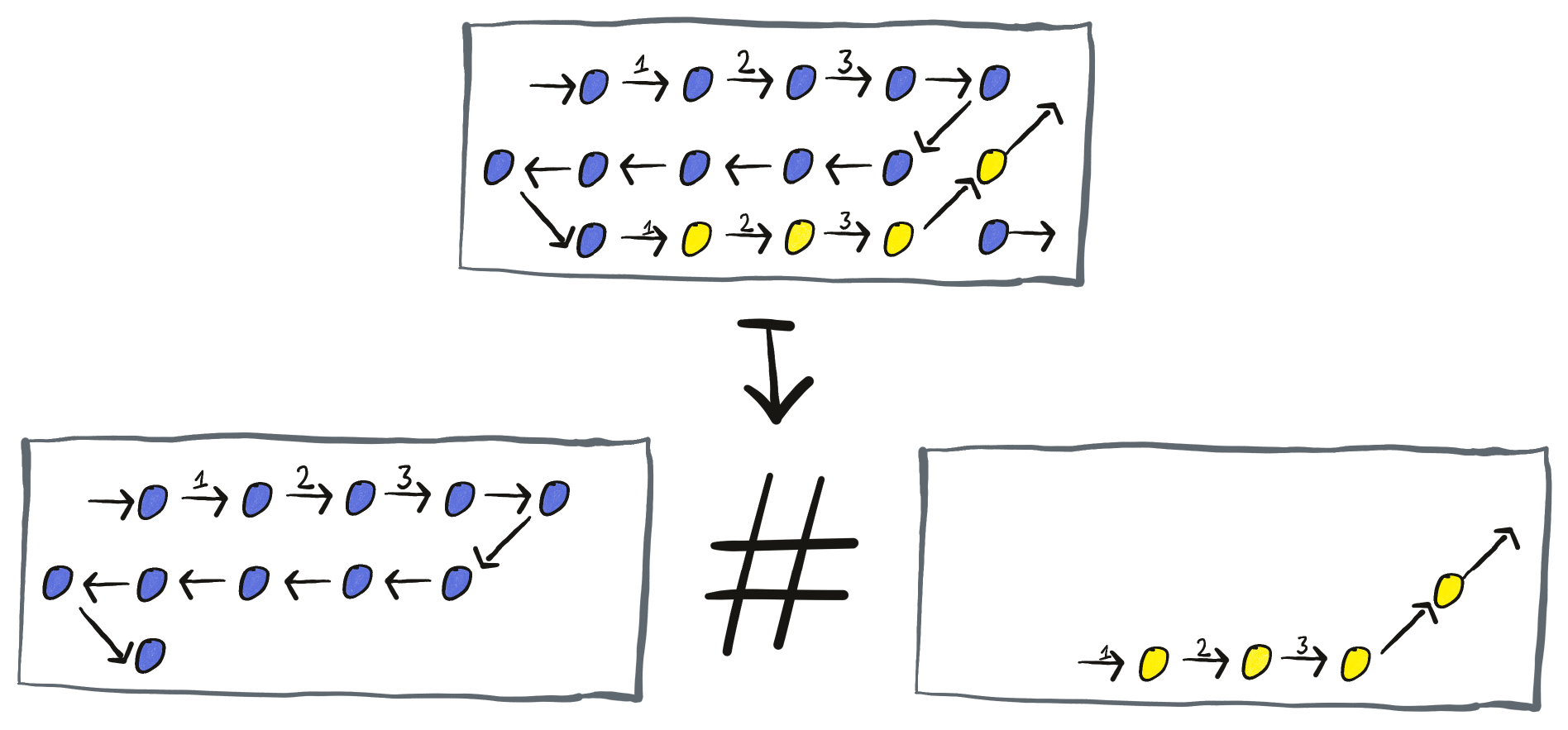}\\

          We do this by duplicating every window and filtering out the spurious nodes.
          The filtering phase can be handled by an 
          unambiguous Mealy machine combined with Claim~\ref{claim:unabigous-finite-parallel}.
          (To apply this construction to all windows,  we use the $\mathtt{map}$ combinator 
          from Lemma~\ref{lem:map-combinator-primes}.)
    \item Now, we untangle each loop and each sweep. For the loops, we use the construction described 
          earlier in this section, and for the sweeps, we use the induction assumption (as mentioned before, the sweeps 
          are always thinner than the whole run graph). We combine those constructions,
          using the $\mathtt{if-then-else}$ and $\mathtt{map}$ combinators (see Lemmas~\ref{lem:2-way-primes-cases}~and~\ref{lem:2-way-primes-map}). 
    \item Finally, we obtain the untangled result by concatenating all outputs produced by the previous step. (As mentioned before,
          this simply means filtering out the $\#$'s.)
\end{enumerate}

\subsection{Compositions of primes $\subseteq$ Two-way automata}
\label{subsec:su-primes-to-2way}
In this section, we show how to translate compositions of two-way primes into single-use two-way transducers
(this construction is also described in \cite[Section~E.1]{single-use-paper}). 
\begin{lemma}
\label{lem:su-primes-to-2way}
    For every composition of primes $f = p_1 \circ  \ldots \circ p_n$, there is an equivalent
    single-use two-way transducer.
\end{lemma}
The proof goes by induction on $n$. For $n=0$, the function $f$ is the identity,
which makes the lemma trivial.
For the induction step, it suffices to show that single-use two-way automata are closed under
pre-compositions with prime functions:
\[ 
        \left(\substack{\textrm{2-way single-use}\\
                 \textrm{transducers}}\right)
        \circ
        \left(\substack{\textrm{2-way single-use}\\
        \textrm{primes}}\right) 
        \subseteq
        \left(\substack{\textrm{2-way single-use}\\
                 \textrm{transducers}}\right) 
\]
This is shown in the following claim:
\begin{claim}
    For every single-use two-way prime function $p$,
    and for every single-use two-way transducer $\mathcal{A}$, 
    there is a two-way transducer $\mathcal{A}'$
    that computes $\mathcal{A} \circ p$.
\end{claim}
\begin{proof}
    We prove the claim by the case analysis of $p$. For the sake of conciseness, we only 
    show the proof for multiple-use bit propagation, single-use atom propagation, and map-reverse
    (other cases are either simple or analogous). 
    \begin{enumerate}
        \item  \textbf{Multiple-use bit propagation}
                To recognize $\mathcal{A} \circ (f_\textrm{prop} \times \idf)$, we use $\mathcal{A}'$ which simulates $\mathcal{A}$
                and additionally keeps track of the current register value for $f_\textrm{prop}$.
                Every time it makes a transition, $\mathcal{A}'$ combines the input letter with
                the current register value and feeds this pair to $\mathcal{A}$. It outputs 
                the same letter and moves in the same direction as $\mathcal{A}$.
                When $\mathcal{A}$ goes forward, then $\mathcal{A'}$ can easily update the register value.
                When $\mathcal{A}$ goes backwards, then $\mathcal{A}'$ checks the current register operation
                (in its input): if it is $\epsilon$, then the register value stays the same;
                otherwise $\mathcal{A'}$ goes left to the first non-$\epsilon$ operation,
                updates its register value, and finds its way back by going forward to the first non-$\epsilon$ operation.
        \item   \textbf{Single-use atom propagation} 
                The transducer $\mathcal{A'}$ for $\mathcal{A} \circ (f_{\textrm{su-prop}} \times \idf)$ resembles the 
                transducer for $f_\textrm{prop}$ from the previous item, but it does not keep track of the register value. 
                Instead, it computes it on demand: every time $\mathcal{A}'$ enters a position
                with a $\downarrow$ (i.e. the \emph{read} operation), 
                it goes left to the first non-$\epsilon$ operation,
                saves enough copies (see Lemma~\ref{lem:su-transition-functions-equivalent}) of the register value
                (or if the operation is $\downarrow$, remembers that the register is empty)
                and finds its way back by going forward to the first $\downarrow$.
                Then it is ready to simulate the transition of $\mathcal{A}$. 
        \item   \textbf{Map-reverse}    
                To recognize $\mathcal{A} \circ f_{\textrm{map-rev}}$ we use $\mathcal{A}'$ that simulates $\mathcal{A}$,
                but modifies the order of the input letters -- every time $\mathcal{A}$ wants to go right,
                $\mathcal{A}'$ goes left, and every time $\mathcal{A}$ wants to go right, $\mathcal{A'}$ goes left. Moreover:
                  \begin{itemize}
                    \item whenever $\mathcal{A}'$ enters $\#$ (or $\vdash$) from the right,
                          it goes to the next $\#$ (or $\dashv$);
                    \item whenever $\mathcal{A}'$ enters $\#$ (or $\dashv$) from the left,
                          it goes to the previous $\#$ (or $\vdash$);
                    \item whenever $\mathcal{A}'$ exits $\#$ (or $\vdash$) to the right,
                          it goes to the rightmost element of the block to the right;
                    \item whenever $\mathcal{A}'$ exits $\#$ or ($\dashv$) to the left,
                          it goes to the leftmost element of the block to the left.
                  \end{itemize}
                  Here is an example equipped with a schematic illustration of this order:
                  \picc{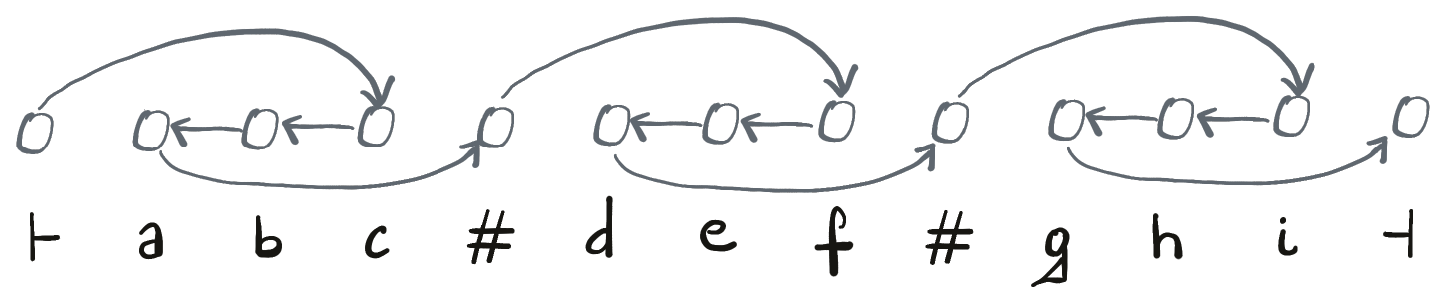}
        \end{enumerate}
    \end{proof}
This finishes the proof that compositions of single-use two-way primes are included in single-use two-way automata. 
Together with Section~\ref{subsec:su-2way-to-primes}, it follows that the two classes are equal. 
In particular, this means that we already have the proof that single-use two-way automata are closed 
under compositions (see Claim~\ref{claim:two-way-automata-closed-under-compositions}). We are going 
to use this fact in the remaining parts of this section.

\subsection{Streaming string transducers $\subseteq$ Two-way automata}
\label{subsec:su-sst-to-2way}
In this section, we show how to translate single-use streaming string transducers into single-use two-way transducers.
The translation has two parts: First, we translate the streaming string transducer into a single-use two-way automaton
that transforms its input into the sequence of operations on $\Gamma^*$-registers
performed by the streaming string transducer while processing this input.
Here is an example of such a sequence of operations:
\bigpicc{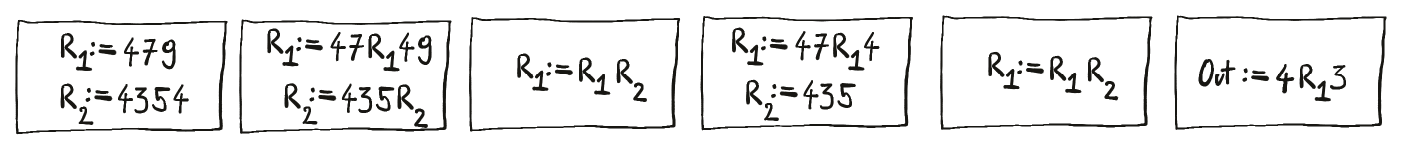}
In the second step, we construct a single-use two-way transducer that interprets the register operations and produces their output.
This finishes the translation because single-use two-way transducers are closed under compositions
(see Claim~\ref{claim:two-way-automata-closed-under-compositions}, and the last paragraph of Section~\ref{subsec:su-primes-to-2way}).\\

Let us start by showing how to construct the sequence of register operations. Notice that
the $\Gamma^*$-registers are write-only, 
i.e. they can be used to construct the output, but the transducer is not allowed to query their contents. 
It follows, by a reasoning similar to the one in Section~\ref{subsec:single-use-decision-trees},
that all single-use functions over polynomial orbit-finite $\Gamma^*$-register sets, can be represented
by single-use decision trees such as the following one (for $\Gamma = \atoms$):
\[ \atoms^3 \times (\atoms^*)^2 \suto \atoms^3 \times (\atoms^*)^2 + \atoms^2 \times \atoms^*\]
\picc{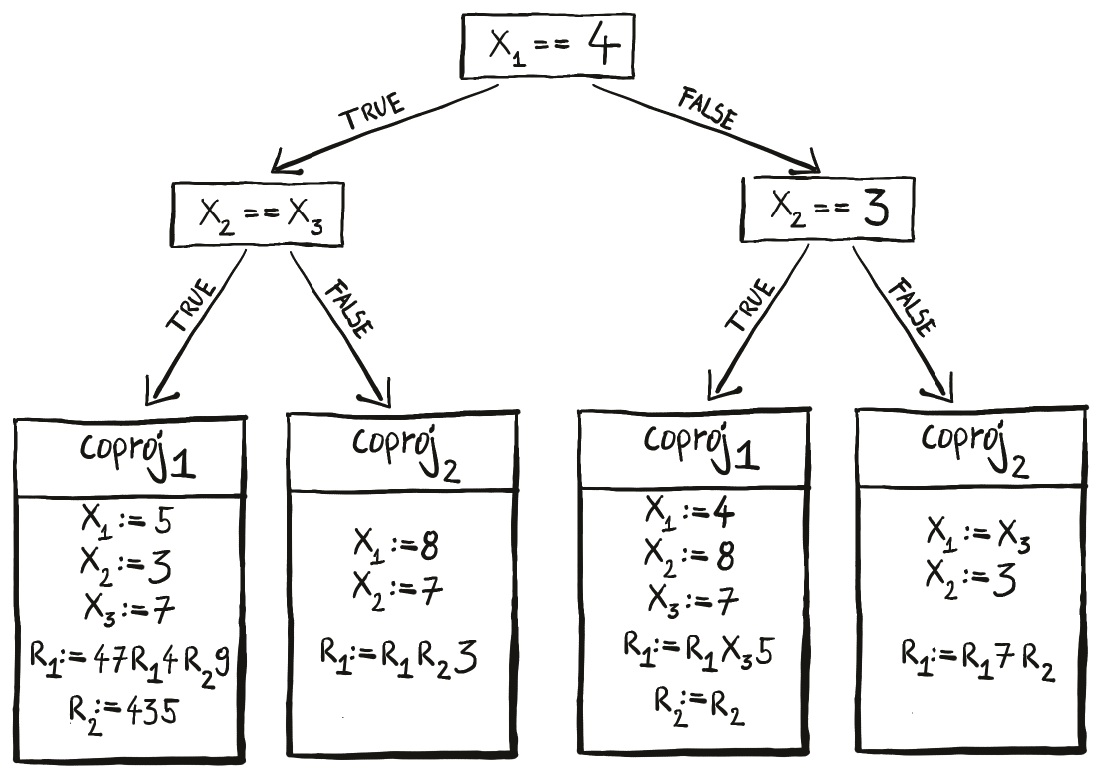}
In this type of tree, the inner nodes look the same as in the usual single-use trees (from Section~\ref{subsec:single-use-decision-trees}) -- 
each inner node contains a query about the $\atoms$-variables. The difference is in the leaves,
which can now include constructors for the $\Gamma^*$-variables.
Each such constructor is a finite word over $\Gamma^*$-values, where 
each $\Gamma^*$-value is either a $\Gamma^*$-variable ($R_i$) or a $\Gamma$-literal. Each 
$\Gamma$-literal is, in turn, an expression of the following form:
\[\coproj_i(v_1, \ldots, v_{k_i})\comma\]
where each $v_j$ is either an $\atoms$-variable ($x_i$) or an atomic constant.
The single-use restriction says that every $x_i$ and every $R_i$ can appear 
at most once on each path from the root to a leaf.\\

Thanks to this tree representation, we can look on a transition function of a
single-use streaming string transducer, as a function  $\Sigma \eqto (\Trees(Q, Q))$.
Notice that for every $Q$ there is a limit $n$ of how many $\Gamma^*$-values 
can be stored in the elements of $Q$. Moreover, since $\Sigma$ is orbit-finite and the transition function is equivariant,
we know that the length of $\Gamma^*$ constructors (which are words over $R_i$'s and $\Gamma$-constructors)
is bounded by some $m \in \nat$. This means that the set of all possible $\Gamma^*$-constructors that 
appear in the transition functions can be represented as the following polynomial orbit-finite set:
\[\overline{\Gamma^*} = (\underbrace{(1 + \ldots + 1)}_{\textrm{representing } R_i\textrm{'s}} + \underbrace{\Gamma}_{\textrm{single letters}})^{\leq m} \]
This means that we can think about $\Trees(Q, Q)$ as a function that ignores the $\Gamma^*$-registers from the input,
and instead of performing the $\Gamma^*$-operation, it outputs them:
\picc{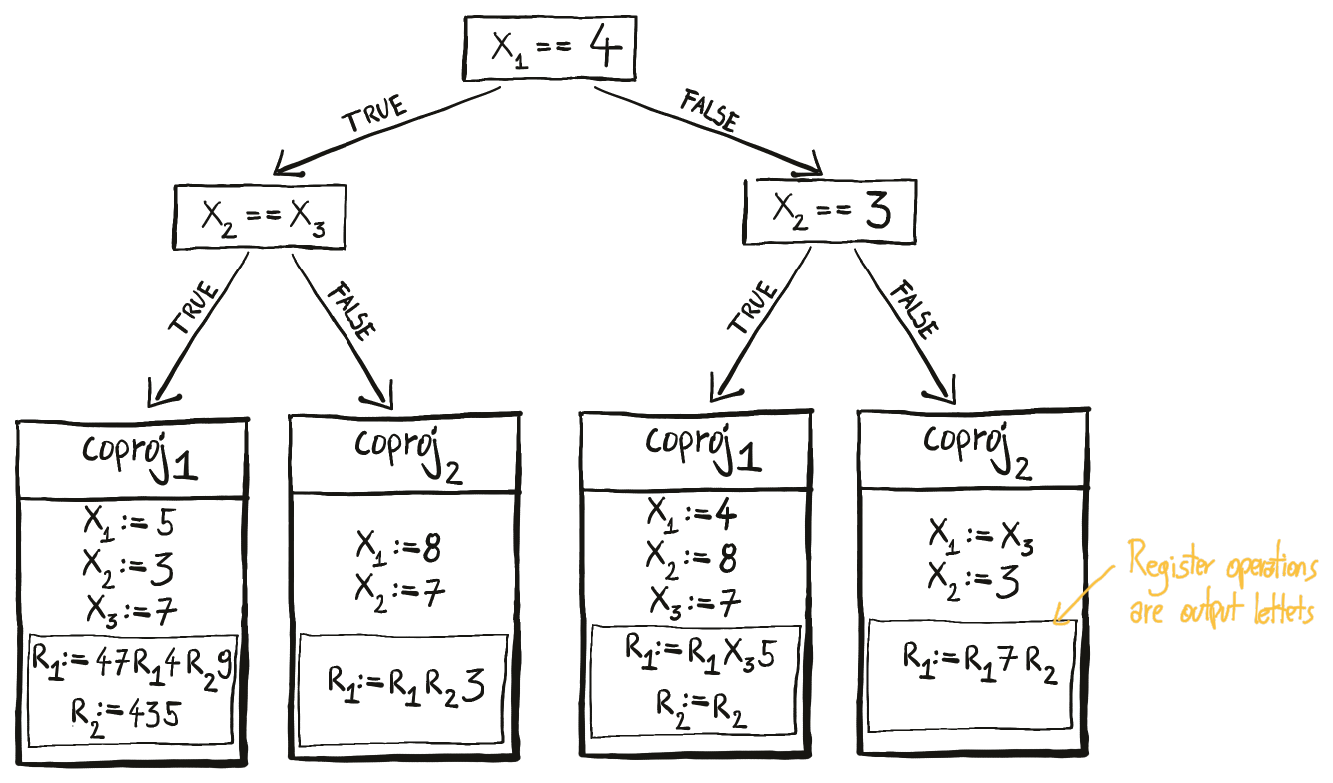}
\noindent
This means that we can translate $\Trees(Q, Q)$ into  $\Trees(Q', Q' \times \overline{\Gamma^*}^{\leq n})$,
where $Q'$ is the version of $Q$ where every $\Gamma^*$ has been replaced by $1$. 
Observe that both $Q'$ and  $\overline{\Gamma^*}^{\leq n}$ are polynomial orbit-finite sets.
It follows that we can translate the transition function of the streaming string transducer into 
a function of the following type:
\[\Sigma \eqto (Q' \suto Q' \times \overline{\Gamma^*}^{\leq n})\]
(Observe that this is a single-use function for polynomial orbit-finite sets,
as defined in Definition~\ref{def:single-use-functions}).
It follows that we can treat each single-use string streaming transducer as a single-use 
Mealy machine that outputs its register operations (instead of performing them).
This almost finishes the first part of the construction.
The last (technical) part is to produce 
the final register operations generated by the output function $\lambda$
(see Definition~\ref{def:sst-with-atoms}). Since a Mealy machine has 
to finish its run as soon as it reaches $\dashv$, it it is too weak for this purpose. 
Instead, we use a single-use two-way automaton\footnote{A more adequate model would be a single-use variant of a one-way transducer
(see last paragraph of \cite[Section~2]{single-use-paper}).
However, since we have not defined this variant in this thesis, we need to use the stronger model of a two-way transducer.}
that simulates the Mealy machine 
and outputs the $\lambda$-operations as soon as it reaches the $\dashv$ marker.
(We use the same construction as for the transition function, to translate $\lambda$ into 
 $Q' \suto \overline{\Gamma^*}^{\leq n}$.)\\

This leaves us with showing how to use a single-use two-way transducer to interpret the sequence of operations on $\Gamma^*$.
We present a construction (also described in \cite[Section~E.6]{single-use-paper}) that is almost identical
to the construction for finite alphabets (see \cite[Lemma~13.4]{bojanczyk2018automata} or \cite[Theorem~6]{dartois2018aperiodic}).
First, we notice that thanks to the single-use restrictions for $\Gamma^*$-registers, each sequence of 
register operations can be described as a forest. For example, consider the following sequence of operations:
\bigpicc{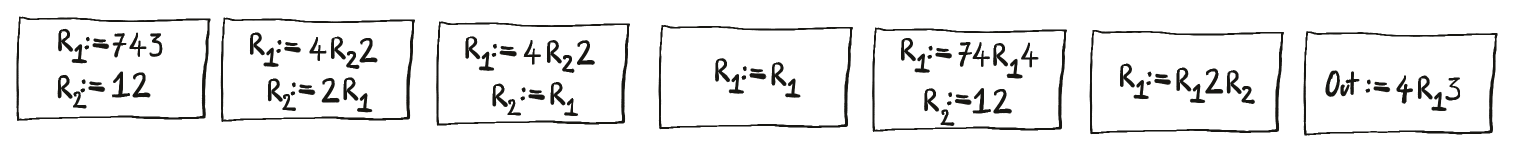}
\noindent
It corresponds to the following forest:
\bigpicc{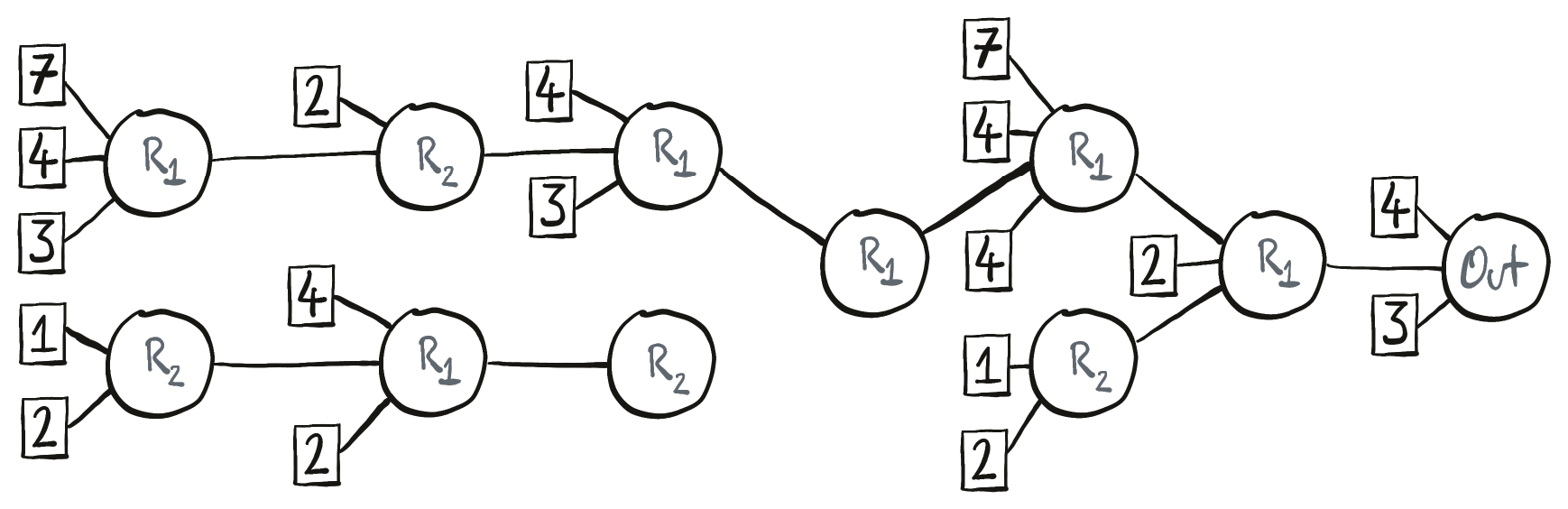}
In order to execute the register operations, it suffices to perform a DFS traversal starting in the final
node (i.e. \texttt{Out}), outputting the encountered letters as we visit them:
\bigpicc{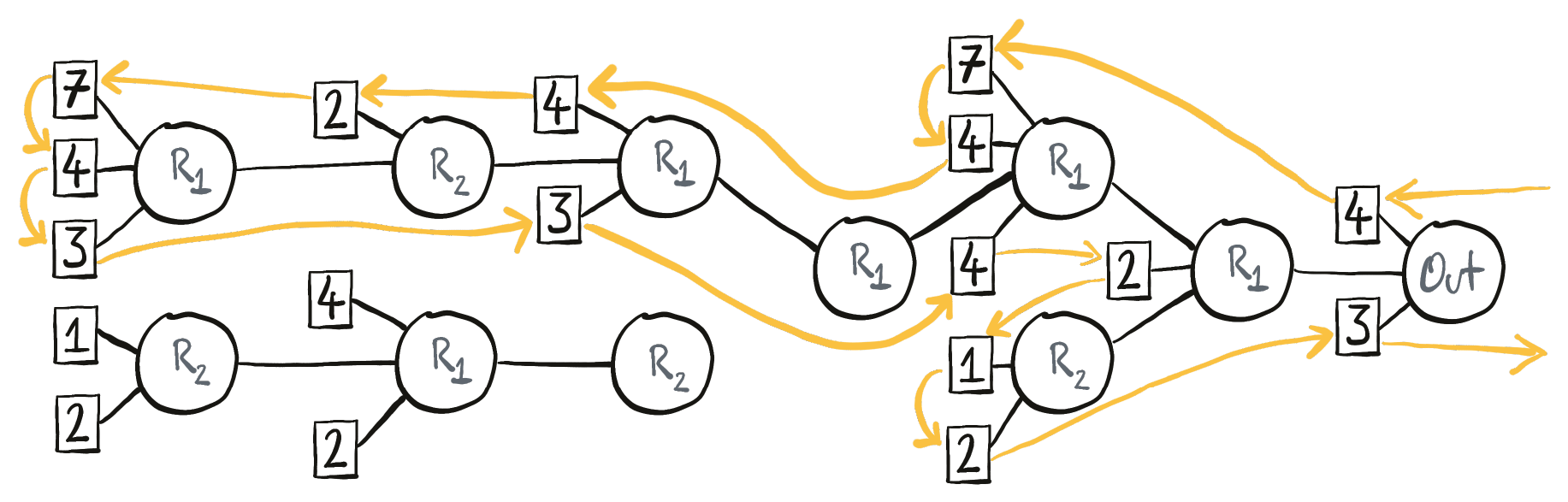}
It is not hard to see that a single-use two-way transducer is capable of performing such a traversal.

\subsection{Compositions of primes $\subseteq$ Streaming string transducers}
\label{subsec:su-primes-to-sst}
In this section, we show how to translate compositions of two-way primes into single-use streaming string transducers:
\begin{lemma}
\label{lem:su-primes-to-sst}
    For every composition of primes $p_1 \circ  \ldots \circ p_n$, there is an equivalent
    single-use streaming string transducer.
\end{lemma}
The proof of the lemma uses a similar induction as the proof for two-way transducers (Lemma~\ref{lem:su-primes-to-2way}), but this 
time we show that single-use streaming string transducers are closed under post-compositions with
single-use two-way primes:
\[ 
        \left(\substack{\textrm{single-use}\\
              \textrm{2-way primes}}\right)
        \circ
        \left(\substack{\textrm{single-use string streaming}\\
              \textrm{transducers}}\right)
        \subseteq
        \left(\substack{\textrm{single-use string streaming}\\
        \textrm{transducers}}\right)
\]
(Note the difference with Lemma~\ref{lem:su-primes-to-2way}, where we have used pre-compositions).
\begin{claim}
\label{claim:sst-postcomp-primes}
    For every single-use two-way prime function $p$,
    and for every streaming string transducer $\mathcal{A}$, 
    there is a two-way transducer $\mathcal{A}'$, 
    that recognizes $\mathcal{A} \circ p$.
\end{claim}
\begin{proof}
    We prove the claim, by case analysis of $p$. This time, we present the proof for all possible $p$'s:
    \begin{enumerate}
        \item  \textbf{Map reverse} Observe that since $f_\textrm{map-rev}$ is of the type $(\Gamma + \#)^* \to (\Gamma + \#)^*$, 
        then both $\mathcal{A}$ and $\mathcal{A'}$ are of the type $\Sigma^* \to (\Gamma + \#)^*$. 
        We construct $\mathcal{A'}$ in the following way: The set of states of $\mathcal{A}'$
        is equal to the set of states of $\mathcal{A}$, where every $(\Gamma + \#)^*$ is replaced by:
        \[ \hspace{-1cm}\underbrace{\Gamma^*}_{
            \substack{\textrm{The initial part of}\\ \textrm{the register value in } \mathcal{A} \\ \textrm{up to the first } \#\comma\\ \textrm{reversed.}}
         }
        \times \underbrace{(\Gamma + \#)^*}_{
            \substack{\textrm{The middle part of}\\ \textrm{the register value in } \mathcal{A} \\ \textrm{from the first to the last } \#\comma\\ \textrm{map-reversed.}} 
        } \times
        \underbrace{\Gamma^*}_{
            \substack{\textrm{The final part of}\\ \textrm{the register value in } \mathcal{A} \\ \textrm{after the last } \#\comma\\ \textrm{reversed.}}
         }\times
        \underbrace{\{\Yes, \No\}}_{
            \substack{ \textrm{ Does the register } \\ \textrm{value in } \mathcal{A} \textrm{ contain} \\ \textrm{at least one} \\ \textrm{separator?}}\\
        }\]
        If the register in $\mathcal{A}$ contains no separator, then its entire content is stored in the first $\Gamma^*$. 
        It is also worth pointing out that the $\Gamma^*$-registers are implemented as $(\Gamma + \#)^*$-registers
        (which happen to have the semantic property of never containing any $\#$'s). 
        Here is an example of a register value in $\mathcal{A}$ and the corresponding value in $\mathcal{A'}$ (for $\Gamma = \atoms$):
        \[ \begin{tabular}{ccc}
            $\boxed{12\#345\#67\#89}$ & $\leftrightsquigarrow$ &
            $\left(\boxed{21},\ \boxed{\#543\#76\#}, \boxed{98}, \Yes\right)$
        \end{tabular} \]
        The transition function of $\mathcal{A'}$ is a version of the transition function of $\mathcal{A}$, 
        where $\concatf$, $\singleton$, and $\const_\epsilon$ are interpreted as follows:
        \[ \hspace{-0.5cm}(A_1, A_2, A_3, a_4)\, \cdot \, (B_1, B_2, B_3, b_4) = \begin{cases}
            (A_1,\ A_2 B_1  A_3  B_2,\ B_3,\ \Yes), &\textrm{if } b_4 = \Yes\\
            (A_1,\ A_2,\ B_1  A_3,\ a_4), & \textrm{otherwise }
        \end{cases} \]
        \[ \begin{tabular}{cc}
            $\singleton'(a \in \Gamma) = (a, \epsilon, \epsilon, \No)$ & $\singleton'(\#) = (\epsilon, \#, \epsilon, \Yes)$
        \end{tabular}
        \]
        \[
            \const_\epsilon = (\epsilon, \epsilon, \epsilon, \No)
        \]
        (It is not hard to see that all of those functions are single-use.)
        Similarly, we define $\lambda'$ (i.e. the output function of $\mathcal{A'}$)
        to be $\lambda$, where $\concatf$, $\singleton$ and $\const_{\epsilon}$ are interpreted in the same way 
        as for the transaction function. Since such $\lambda'$ produces an element of
        \[\Gamma^* \times (\Gamma + \#)^* \times \Gamma^* \times \{\Yes, \No\}\comma\]
        we need to compose it with the following \emph{exit function}, 
        which collapses this compound register type back to $(\Gamma +\#)^*$:
        \[ (R_1, R_2, R_3, r_4) \mapsto  R_1  R_2 R_3\]
    \item  \textbf{Map duplicate} In this case, we use the same idea as for $f_\textrm{map-rev}$, 
           but we keep two copies of the initial and final blocks. For example:
           \[ \begin{tabular}{ccc}
            $\boxed{12\#345\#67\#89}$ & $\leftrightsquigarrow$ &
            $\left(\boxed{12},\ \boxed{12}, \boxed{\#345345\#6767\#}, \boxed{89}, \boxed{89}, \Yes\right)$
        \end{tabular} \]
    \item \textbf{Letter-to-word homomorphism} 
          We show how to construct $\mathcal{A'}$ for $h^* \circ \mathcal{A}$, where $h$ is 
          a function of type $\Gamma \eqto \Delta^*$. In this case,
          $\mathcal{A}$ is of type $\Sigma^* \to \Gamma^*$ and $\mathcal{A}'$ is of type
          $\Sigma^* \to \Delta^*$. We define $\mathcal{A'}$ to be a version 
          of $\mathcal{A}$ where every $\Gamma^*$-register is replaced by a $\Delta^*$-register,
          which keeps the $h^*$-image of the original $\Gamma$-register from $\mathcal{A}$. 
          Simulating register concatenation is trivial -- 
          whenever $\mathcal{A}$ concatenates two $\Gamma^*$-registers, 
          $\mathcal{A}'$ can simply concatenate the corresponding $\Delta^*$-registers.
          Simulating $\singleton$ requires some explanation. The principle 
          is easy --  $\mathcal{A}'$ needs to interpret $\singleton$ as follows:
          \[ \singleton'(a) = h(a) \]
          The harder part is implementing $h$ as a single-use function:
          Observe, first, that since $\Gamma$ is orbit-finite, 
          and $h$ is equivariant, it follows that there exists a limit $l$ on the length
          of the output of $h$. This means that we can translate 
          $h$ into an equivalent $h_1 : \Gamma \eqto \Delta^{\leq l}$. 
          Now, we can apply Lemma~\ref{lem:fs-k-fold} to obtain
          an equivalent $h_2 : \Gamma \suto_k \Delta^{\leq l}$, which 
          (by composing it with at most $l$ $\concatf$'s) can be easily transformed into
          an equivalent $h_3 : \Gamma \suto_k \Delta^*$. Finally, by Definition~\ref{def:k-fold-use-functions},
          we can transform $h_3$ into an equivalent $h_4 : \Gamma^k \suto \Delta^*$.
          This finishes the construction, as $\mathcal{A'}$ can be easily modified to maintain 
          $k$ copies of every $\atoms$-register from $\mathcal{A}$. 
    \item \textbf{Single-use atom propagation} For the sake of simplicity, we show 
           how to construct $\mathcal{A}'$ for $f_\textrm{su-prop} \circ \mathcal{A}$.
           (The construction can be easily modified to construct the actual
           $(f_\textrm{su-prop} \times \idf) \circ \mathcal{A}$.)
           As $f_\textrm{su-prop}$ has the type $(\atoms + \downarrow + \epsilon)^* \to (\atoms + \epsilon)^*$, 
           we know that $\mathcal{A}$ and $\mathcal{A}'$ have the following types:
           \[ \begin{tabular}{cc}
                $\mathcal{A} : \Sigma^* \to (\atoms + \downarrow + \epsilon)^*$ & 
                $\mathcal{A}' : \Sigma^* \to (\atoms + \epsilon)^*$
           \end{tabular} \]
           We construct $\mathcal{A'}$ as a version of $\mathcal{A}$, where every $(\atoms + \downarrow + \epsilon)^*$-register 
          is replaced by the following set (compare with the proof of Claim~\ref{claim:rational-precomp-su-prop}):
           \[ \hspace{-1cm}\underbrace{\epsilon^*}_{
            \substack{\textrm{Maximal}\\ 
            \textrm{$\epsilon$-prefix.}}
           }
           \times 
           \underbrace{\{\downarrow, \square, \epsilon \}}_{
            \substack{\textrm{The first non-$\epsilon$}\\ \textrm{operation, or $\epsilon$}\\ \textrm{if there is none.} \\ \textrm{($\square$ represents elements of $\atoms$)} }
           }
           \times
           \underbrace{(\atoms + \epsilon)^*}_{
            \substack{\textrm{The output for}\\ \textrm{the suffix after}\\ 
            \textrm{the first non-$\epsilon$.}\\
            } 
           } \times
           \underbrace{(\atoms + \downarrow + \epsilon)}_{
            \substack{\textrm{The final non-$\epsilon$}\\ \textrm{operation, or $\epsilon$}\\  \textrm{if there is none.} }
           } \]
           (Again, the $\epsilon^*$-registers are actually implemented as $(\epsilon + \atoms)^*$-registers.)
           Observe that the output of the suffix after the first non-$\epsilon$ does not depend on the
           initial register value for $f_\textrm{su-prop}$.\\ 

           Here is an example of a register in $\mathcal{A}$ and the corresponding value in~$\mathcal{A}'$:
           \[ \begin{tabular}{ccc}
            $\boxed{\epsilon \epsilon \downarrow 1 2 \epsilon \downarrow 3 \epsilon \epsilon \epsilon \downarrow 4 \epsilon \epsilon}$ & $\leftrightsquigarrow$ &
            $\left(\boxed{\epsilon \epsilon }, \downarrow, \ \boxed{\epsilon \epsilon \epsilon 2 \epsilon \epsilon \epsilon \epsilon 3 \epsilon \epsilon \epsilon }, 4\right)$
        \end{tabular} \]
        This representation allows $\mathcal{A}'$ to simulate register concatenation. For example:
        \[ (A_1, a_2, A_3, a_4 \in \atoms) \cdot (B_1, \downarrow, B_3, b_4) = (A_1, a_2, A_3 B_1 a_4 B_3, b_4) \]
        Other cases are handled analogously.
        Operations $\singleton$ and $\const_\epsilon$ are trivial, and the exit function looks as follows:
        \[ \begin{tabular}{ccc}
            $(A_1, a_2, A_3, a_4)$ & $\mapsto$ & $\begin{cases}
                A_1 \epsilon A_3 & \textrm{if } a_2 \in \{\square, \downarrow\}\\
                A_1 & \textrm{if } a_2 = \epsilon
            \end{cases}$
        \end{tabular} \]

    \item \textbf{Multiple-use bit propagation} Again, for simplicity, we present the construction for
          $f_\textrm{prop} \circ \mathcal{A}$. The construction is similar to the one 
          for $f_\textrm{su-prop}$. This time both $\mathcal{A}$ and $\mathcal{A}'$ have 
          the type $\Sigma^* \to \{\fullmoon, \newmoon, \epsilon\}^*$. 
          We construct $\mathcal{A'}$ as a version of  $\mathcal{A}$, where 
          every $\{\fullmoon, \newmoon, \epsilon\}^*$-register is replaced by:
          \[ \hspace{-1cm}\underbrace{(\{\fullmoon,\newmoon, \epsilon \}^*)^{\{\fullmoon, \newmoon, \epsilon\}}}_{
            \substack{\textrm{$f_\textrm{prop}$-output for the}\\ \textrm{ maximal $\epsilon$-prefix,}\\
            \textrm{depending on the}\\ \textrm{initial register value.}}}
           \times \underbrace{\{\fullmoon, \newmoon, \epsilon\}^*}_{
            \substack{\textrm{$f_\textrm{prop}$-output for the suffix}\\ \textrm{that starts in the first non-$\epsilon$.}\\
            \textrm{(Note, that this does not depend}\\
            \textrm{on the initial register value of $f_\textrm{prop}$.)}
            }
          } \times \underbrace{\{\fullmoon, \newmoon, \epsilon\}}_{
            \substack{\textrm{The final non-$\epsilon$ value,}\\ \textrm{or $\epsilon$ if there is none.}}
          } \]
          We represent $(\{\fullmoon, \newmoon, \epsilon\}^*)^{\{\fullmoon, \newmoon, \epsilon\}}$ as $({\{\fullmoon, \newmoon, \epsilon\}^*})^3$.\\
            
          For example:
          \[ \begin{tabular}{ccc}
            $\boxed{\epsilon \epsilon \newmoon \fullmoon \epsilon \epsilon \newmoon \epsilon  \epsilon }$ & $\leftrightsquigarrow$ &
            $\left(\left(\substack{{\epsilon} \;\mapsto\; \boxed{\epsilon \epsilon }\\
            \fullmoon \;\mapsto\; \boxed{\fullmoon \fullmoon}\\
            \newmoon \;\mapsto\; \boxed{\newmoon \newmoon}
           }\right),\ \boxed{ \newmoon \fullmoon \fullmoon \fullmoon \newmoon \newmoon \newmoon }, \newmoon\right)$
        \end{tabular} \]

          The concatenation of registers is interpreted as follows:
          \[ \hspace{-1cm}(A_1, A_2, a_3) \cdot (B_1, B_2, b_3) = \begin{cases}
            (A_1,\ A_2 \cdot B_1(a_3) \cdot B_2,\ b_3) & \textrm{if } a_3 \neq \epsilon\\
            (x \mapsto A_1(x) \cdot B_1(x),\ B_2,\ b_3) & \textrm{if } a_3 = \epsilon\\
          \end{cases} \]
          Observe that this is a single-use function -- in particular, each $A_1(x)$ and $B_1(x)$ is used at most once. 
          Functions $\singleton$ and $\const_\epsilon$ are trivial, and the exit function looks as follows:
          \[\begin{tabular}{ccc}
            $(A_1, A_2, a_3)$ & $\mapsto$ & $A_1(\epsilon) \cdot  A_2$
          \end{tabular}
          \]
    \item \textbf{Group prefixes} Again, for simplicity, we present the construction for
    $f_\textrm{G-pref} \circ \mathcal{A}$. This time both $\mathcal{A}$ and $\mathcal{A'}$ 
    are of the type $\Sigma \to G^*$  (where $G$ is a finite group). 
    We construct $\mathcal{A}'$ as a version of $\mathcal{A}$ where every $G^*$ is replaced by:
    \[ \underbrace{(G^*)^G}_{
        \substack{\textrm{The output for the register,}\\
                  \textrm{depending on the initial $G$-value }}
    } \times \underbrace{G}_{
        \substack{\textrm{The $G$-product of}\\
                  \textrm{the entire register}}
    } \]
    (Again, we represent $(G^*)^G$ as $(G^*)^{|G|}$).\\

    For example, if we take $G = \mathbb{Z}_3$:
    \[ \begin{tabular}{ccc}
        $\boxed{001211}$ & $\leftrightsquigarrow$ &
        $\left(\substack{0 \;\mapsto\; \boxed{001012}\\
                         1 \;\mapsto\; \boxed{112120}\\
                         2 \;\mapsto\; \boxed{220201}},\ 2\right)$
    \end{tabular} \] 
    The register concatenation can be interpreted as follows (note the similarity to the wreath product):
    \[ (A_1, a_2) \cdot (B_1, b_2) = (g \mapsto A_1(g) \cdot B_1(g \cdot a_2),\ a_2 \cdot b_2 ) \]
    Observe, that since $G$ is a group, then $g \mapsto g \cdot a_2$ is a bijection on $G$.
    It follows that each $A_1(x)$ and each $B_1(x)$ is used at exactly once,
    which makes the function single-use. Operations $\singleton$ and $\const_\epsilon$ are trivial,
    and the exit function looks as follows:
    \[\begin{tabular}{ccc}
      $(A_1, a_2)$ & $\mapsto$ & $A_1(1)$
    \end{tabular}
    \]
    \item \textbf{End of word marker} This is the simplest case. It can be simulated by $\mathcal{A}'$ that looks exactly of 
          like $\mathcal{A}$ except of the output function, which for $\mathcal{A}'$ is defined in the following way:
          \[ \lambda'(q) = \lambda(w) \dashv \]
    \end{enumerate}
\end{proof}

It might be worth mentioning that, as observed in the last paragraph of \cite[Section~E3]{single-use-paper},
this proof of Claim~\ref{claim:sst-postcomp-primes} also works in the presence 
of both $\copyf_{\Gamma^*}$ and $\copyf_{\atoms}$. It follows that:
\[ 
        \left(
            \substack{
              \textrm{ single-use}\\
              \textrm{two-way primes}
            }\right)
        \circ
        \left(\substack{\textrm{multiple-use string streaming}\\
              \textrm{transducers}}\right)
        \subseteq
        \left(\substack{\textrm{multiple-use string streaming}\\
        \textrm{transducers}}\right)
\]
The same proof also works in the finite case, which might be of independent interest.
In particular, it follows that:
\[ 
        \left(
              \textrm{two-way transducers}
        \right)
        \circ
        \left(\substack{\textrm{copyful string streaming}\\
              \textrm{transducers}}\right)
        \subseteq
        \left(\substack{\textrm{copyful string streaming}\\
        \textrm{transducers}}\right)
\]
\subsection{Regular list transductions $\subseteq$ Two-way transducers}
\label{subsec:su-list-fun-2way}
In this section, we show how to translate regular list transductions with atoms
into single-use two-way transducers (the construction, which is also presented in \cite[Section~E.5]{single-use-paper}, 
is an adaptation of the left-to-right implication from \cite[Theorem~4.3]{bojanczyk2018regular}).\\

Remember that regular list functions with atoms work over the class of 
polynomial sets with atoms -- i.e. the smallest class that contains $1$ and $\atoms$ 
and is closed under $\times$, $+$ and $X^*$. We start the construction by observing 
that every element of every polynomial set with atoms can be encoded as 
a word over the following alphabet: 
\[  \Sigma_{[\atoms]} = \{ \underbrace{\circ}_{\substack{
    \textrm{element}\\
    \textrm{of $1$}
}}, \underbrace{\mathtt{[}\ ,\  \mathtt{]}}_{\substack{
    \textrm{used for}\\ X^*}\\
}, \underbrace{\mathtt{(}\ ,\  {)}}_{\substack{
    \textrm{used for}\\  X \times Y}
}, \underbrace{\mathtt{coproj}_1, \mathtt{coproj}_2}_{\substack{
    \textrm{used for }\\ X + Y
}}, \underbrace{\raisebox{0.9mm}{$\mathtt{,}$}}_{\substack{
    \textrm{separator for}\\
    X \times Y \textrm{ and } X^*
}}\}+ \underbrace{\atoms}_{\substack{
    \textrm{elements}\\
    \textrm{of } \atoms
}}  \]
For example, here is an encoding of an element from $(\atoms^2 + 1)^*$:
\[ [\coproj_1\; (5, 8),\ \coproj_2 \;\circ,\ \coproj_1\; (1, 2),\ \coproj_1\; (7, 8),\ \coproj_2\; \circ]  \]
Thanks to this encoding, we can translate every regular list function with atoms into a two-way transducer:
\begin{lemma}
    For every regular list function with atoms $f : X \to Y$, there is a two-way transducer:
    \[\mathcal{F} : \Sigma_{[\atoms]}^* \to \Sigma_{[\atoms]}^*\comma\]
    such that $\mathcal{F}(w_x)$ outputs the encoding of $f(x)$ (where $w_x$ denotes the 
    $\Sigma_{[\atoms]}$-encoding of $x$). 
\end{lemma}
\begin{proof}
    The proof is a standard induction on the construction of $f$ as a regular list function with atoms.
    The most interesting case is function composition but, as explained in the
    last paragraph of Section~\ref{subsec:su-primes-to-2way}, we already 
    know that single-use automata are closed under compositions.
\end{proof}

We finish the construction, by observing, that for every polynomial \emph{orbit-finite} $\Gamma$, the following two 
functions can be implemented as single-use two-way transducers:
\[
    \begin{tabular}{cc}
    $\mathcal{T}_\Gamma: \Gamma^* \to \Sigma_{[\atoms]}^*$ & $\mathcal{T}_\Gamma^{-1}: \Sigma_{[\atoms]}^* \to \Gamma^*$,
    \end{tabular}
\]
where $\mathcal{T}_\Gamma$ is a function that translates the input word 
(which is a polynomial set with atoms) into its  $\Sigma_{[\atoms]}$-encoding,
and $\mathcal{T}_\Gamma^{-1}$ is a one-way inverse of $\mathcal{T}_\Gamma$.

\subsection{Compositions of primes $\subseteq$ Regular list transductions}
\label{subsec:su-primes-to-list-fun}
Finally, let us show how to translate compositions of single-use two-way primes into regular list transductions
with atoms (the construction is also presented in \cite[Section~E.2]{single-use-paper}).
Since regular list functions are (by definition) closed under compositions, it suffices to 
show how to translate every prime function:
\begin{enumerate}
    \item \textbf{Letter-to-word homomorphism} In this step, 
    we construct $h^*$, for any $h:\Sigma \eqto \Gamma^*$ where $\Sigma$ and $\Gamma$
    are polynomial orbit-finite. Similar as in Section~\ref{subsec:su-primes-to-sst}, 
    we observe that there is a limit $k$ on the length of the outputs of $h$, 
    which means that we can translate $h$ into an equivalent $h': \Sigma \eqto \Gamma^{\leq k}$. 
    As $\Gamma^{\leq k}$ is polynomial orbit-finite, it follows by Lemma~\ref{lem:su-plus-copy-fs}, 
    that $h'$ is a regular list function with atoms. This means that we can construct $h^*$ 
    in the following way:
    \[ \Sigma^* \longtransform{\mathtt{map}\;f'} (\Gamma^{\leq k})^* \longtransform{\mathtt{map\; toList_{\leq k}}} 
    (\Gamma^*)^* \longtransform{\mathtt{concat}} \Gamma^*\comma\]
    where $\mathtt{toList_{\leq k}}$ is the following tuple-to-list transformation:
    \[ \mathtt{toList_{\leq k}} :  X^{\leq k} \to X^*\tdot  \]
    We finish the construction by showing that $\mathtt{toList_{\leq k}}$ is a regular list function.
    We start with $\mathtt{toList}_1$:
    \[ \mathtt{toList_1} :  X \transform{\rightI} X \times I \transform{\const_\epsilon} X \times X^* \transform{\cons} X^* \]
    Now, let us continue with $\mathtt{toList_{\leq 2}}$:
    \[ \mathtt{toList_2} :  X^2 \longtransform{\idf \times \mathtt{toList_1}} X \times X^* \longtransform{\cons} X^* \] 
    By continuing in this manner, we can construct $\mathtt{toList_i}$ for every $i \leq k$.
    Finally, we combine all those functions into 
    $\mathtt{toList_{\leq k}}$ using the $[f_1, \ldots, f_n]$ combinator from Lemma~\ref{ex:maybe-comb-su}.
\item \textbf{Single-use propagation} For the sake of clarity, we show how to construct $f_{\textrm{su-prop}}$ -- 
      the construction for $f_{\textrm{su-prop}} \times \idf$ is analogous. Consider the following example input:
      \[ [1, \ \epsilon,\ \epsilon,\ \downarrow,\ \epsilon,\ \epsilon,\ 3,\ \epsilon,\ \downarrow,\ 3,\ 2,\ \epsilon,\ \downarrow,\ \downarrow,\ 3,\ \epsilon] \]
      We start by applying $h^*$ for the following $h$ (we have already shown in the previous item 
      that $h^*$ is a regular list function):
      \[ \begin{tabular}{ccc}
        $h(\epsilon) = \epsilon$ & $h(\downarrow) = \downarrow \#$ &  $h(a \in \atoms) = \# a$ 
      \end{tabular}
      \]
      This transforms the input word into:
      \[ [\#,\ 1, \ \epsilon,\ \epsilon,\ \downarrow,\ \#,\ \epsilon,\ \epsilon,\ \#,\  3,\ \epsilon,\ \downarrow,\ \#,\ \#,\  3,\ \#,\ 2,\ \epsilon,\ \downarrow,\ \#,\  \downarrow,\ \#,\ \#,\ 3,\ \epsilon] \]
      Next, we apply the $\mathtt{block}$ function:
      \[ [\;[\; ],[1, \ \epsilon,\ \epsilon,\ \downarrow],\ [\epsilon,\ \epsilon],\ [3,\ \epsilon,\  \downarrow],\ [\;],\ [3],\ [2,\ \epsilon,\ \downarrow],\ [\downarrow],\ [\; ], \ [3,\ \epsilon]\;] \]
      Observe that every $\downarrow$ is the last element of some block,
      and that the output produced by each $\downarrow$ 
      is equal to the first letter of the $\downarrow$'s block (as long as the block contains at most two letters -- 
      the output of a singleton block $[\downarrow]$ is equal to $[\epsilon]$).  
      This means that we can produce the output for each block using the following function:
      \[
        f_{\textrm{block-out}}(b_1\,  b' \,  b_n) = \begin{cases}
            \epsilon \, b' \, b_1 & \textrm{if } b_n = \downarrow\\
            \epsilon \, b' \, \epsilon & \textrm{otherwise}\\
        \end{cases} \comma
      \]
      where $b_1$ denotes the block's first letter, $b_n$ denotes the last letter, and $b'$ denotes 
      inner letters of the block. Observe, that due to the way the blocks are constructed, 
      we know that:
      \[
      \begin{tabular}{ccc}
        $b_1 \in \atoms + \epsilon$, & $b_n \in \{\epsilon, \downarrow\}$,  & and $b' \in \epsilon^*$. 
      \end{tabular}
      \]
      This definition of $f_{\textrm{block-out}}$ only works blocks with at least two letters,
      but the other cases are very simple: the empty block produces empty output and all one-letter blocks 
      produce $[\epsilon]$. Observe that we can implement the $f_{\textrm{block-out}}$
      function as a regular list function with atoms: First, we use $\mathtt{destr}$ and $\mathtt{reverse}$
      (together with some structural transformations, such as $\mathtt{distr}$)
      to split $b$ into $(b_1, b', b_n)$. Then, we can use the if-then-else combinator
      from Example~\ref{ex:if-else-su} together with $\mathtt{cons}$,
      $\mathtt{reverse}$ (and some structural functions) to construct the output.
      This finishes the construction, as we can now use $\mathtt{map}$ to apply $f_\textrm{block-out}$ to 
      every block, and combine the results using $\mathtt{concat}$. 
\item \textbf{Multiple-use propagation} Again, we only show how to construct $f_{\textrm{prop}}$ -- 
       the construction for $f_{\textrm{prop}} \times \idf$ is analogous. The idea is very similar to 
       the one for single-use propagation (from the previous item). For this reason, 
       we only show how to implement the key component of the construction, which is the following function:
       \[ \mathtt{replace} : \{\fullmoon, \newmoon\} \times \epsilon^* \to \{\fullmoon, \newmoon\}^*\comma \]
       which replaces every $\epsilon$ in the input word with the letter from the first coordinate. For example:
       \[ (\newmoon, [\epsilon, \epsilon, \epsilon, \epsilon]) \mapsto  [\newmoon, \newmoon, \newmoon, \newmoon]\]
       Since $\{\fullmoon, \newmoon\}$ is encoded as $1 + 1$,
       we can use the if-then-else combinator (from Example~\ref{ex:if-else-su}),
       and implement $\mathtt{replace}$ as follows\footnote{It is worth mentioning that this idea does not work for
       $\mathtt{replace}_\atoms : \atoms \times \epsilon^* \to \atoms^*$.
       In fact, it is not hard to see that $\mathtt{replace}_\atoms$
       cannot be implemented as a single-use two-way automaton, 
       which by Section~\ref{subsec:su-list-fun-2way}, means that it is not a regular list function with atoms.}:
       \[ (\map\, \const_{\fullmoon}) \; \orf \; (\map\, \const_{\newmoon})\]
       \item \textbf{Map reverse} This function can be implemented in the following way:
       \[ f_\textrm{map-rev} (\Sigma + \#)^* \longtransform{\mathtt{block}} {(\Sigma^*)}^* \longtransform{\map\; \reverse} {(\Sigma^*)}^*  \longtransform{\concat} \Sigma^* \]
       The only problem with this construction is that it erases the $\#$-separators. However, it is not hard to see that we can reintroduce them 
       before applying the $\mathtt{concat}$ function.
       \item \textbf{Map duplicate} We use a similar idea as for $f_\textrm{map-reverse}$. The only difference is that we 
       have to implement the function $\mathtt{duplicate} : X^* \to X^*$. We start by showing how to concatenate two lists:
       \[ \mathtt{append} : X^* \times X^* \longtransform{\idf \times \mathtt{toList_1}} X^* \times {(X^*)}^* \transform{\mathtt{append}} {(X^*)}^* \transform{\mathtt{\concat}} X^*\]
       Now we can implement $\mathtt{duplicate}$ in the following way:
       \[ X^* \longtransform{\copyf_{X^*}} X^* \times X^* \longtransform{\mathtt{append}} X^*\]
    \end{enumerate}
\chapter*{Further work}
Throughout the thesis, I have pointed out several open problems. This section 
gathers them all in one place. For more information, about each 
of those problems follow the references to the sections where they were originally discussed,
or contact me directly.

\begin{description}

    \item [1. Semantic definition of single-use functions.] The definition of single-use functions
    is quite syntactic in its nature, which limits their domain to polynomial orbit-finite sets.
    This open question asks for a semantic definition of single-use functions.
    (Compare with syntactic and semantic definitions of equivariance in Section~\ref{sec:dra}.)
    A possible approach might be to consider a version of sets with atoms equipped with an action of all functions $\atoms \to \atoms$ 
    or even all relations $\atoms \times \atoms$ (rather than only atom bijections). See
    the introduction of Section~\ref{sec:su-functions} for context. 

    \item [2. Nondeterministic single-use automata.]
    A straightforward way of introducing nondeterminism to single-use automata results in a model
    that is too strong:
    It is not hard to construct a \emph{nondeterministic single-use automaton}
    that recognizes the language ``The first letter appears again'' (which cannot be recognized 
    by a deterministic single-use automaton). This does not fit well in the picture
    of definitional robustness presented in this thesis. This open question asks if there 
    is a notion of nondeterminism compatible with the single-use restriction.
    This seems to be connected with developing a good notion of \emph{single-use relations}.
    For context, see Section~\ref{sec:non-determinsitic-single-use}. 

    \item [3. Unambiguous single-use automata.] It is worth noting that both examples
    from Section~\ref{sec:non-determinsitic-single-use},
    which demonstrate that nondeterministic single-use automata are stronger than deterministic ones,
    use automata that are ambiguous (which means that some accepted words will always have more than one accepting run).
    It follows that those examples cannot be used to show that unambiguous automata are stronger than deterministic ones.
    In fact, the question of whether unambiguous nondeterministic automata are equivalent
    to deterministic single-use automata remains open. If the two models 
    turned out to be equivalent, it would open a path to a machine-based definition of 
    single-use rational transductions. For context, see Footnote~\ref{ftn:unambigous} on Page~\pageref{ftn:unambigous}. 

    \item [4. Local semigroup transductions over arbitrary orbit-finite sets.]
    Note that local semigroup transductions (Definition~\ref{def:local-monoid-transduction})
    are defined for all orbit-finite alphabets, but their Krohn-Rhodes theorem only works 
    for polynomial orbit-finite alphabets.
    A counterexample is the single-use propagation 
    of $\binom{\mathbb{A}}{2}$,
    which can be constructed as a local semigroup transduction but not as a composition of single-use primes.
    One way to address this issue would be to extend the set of single-use primes with the generalized
    single-use propagation for every orbit-finite $X$
    (i.e. an extended version of the function from Claim~\ref{claim:gen-su-prop}). However, 
    the current proof of the theorem only works for polynomial orbit-finite alphabets. 
    This leaves the question of whether compositions of these generalized primes are equivalent to local rational semigroup transductions over orbit-finite alphabets open.
    An analogous open question can be asked about local rational semigroup transductions.
    For context, see the footnotes in Theorems~\ref{lem:single-use-mealy-monoid-transduction}~and~\ref{thm:rational-kr}.

    \item [5. Local rational semigroup transductions as two Mealy machines.]~\\
    Elgot-Mezei Theorem (\cite[Theorem~7.8]{elgot1963two}) shows that every 
    rational function can be expressed as a composition of one 
    left-left-to-right and one right-to-left Mealy machine. In contrast, in the proof of Theorem~\ref{thm:rational-kr}, 
    we have used multiple left-to-right and multiple right-to-left single-use Mealy machines.
    It remains an open problem whether the Elgot-Mezei Theorem can be generalized for single-use Mealy machines.
    For context see Footnote~\ref{ftn:su-rational-two-mealy}
    on Page~\pageref{ftn:su-rational-two-mealy}.

    \item [6. Single-use restriction for total-order atoms.] 
    In this thesis, we have considered a set of atoms, whose elements can only be compared 
    with respect to equality. However, there are also other types of atoms
    that are studied in the literature (see \cite[Chapter~3]{bojanczyk2019slightly}). 
    One example is the total-order atoms, i.e. the set $\mathbb{Q}$ equipped with the relation $\leq$.
    Interestingly, Claim~\ref{claim:finite-shift} fails for some 
    of those other atoms (including the total-order atoms)
    which breaks most of the proofs presented in this thesis. 
    For this reason, developing a single-use theory for other kinds 
    of atoms remains an open problem. We are currently working on it
    together with Nathan Lhote. For context, see Footnote~\ref{ftn:total-order} 
    on Page~\pageref{ftn:total-order}. 

    \item [7. Polyregular functions over infinite alphabets.]
    \emph{Polyregular functions} is a
    class over finite alphabets that extends \emph{regular functions} while keeping many of their 
    desirable properties (see \cite{bojanczyk2022transducers}). 
    Our final open question concerns finding a well-behaved class of \emph{polyregular functions over infinite alphabets}.
    It is harder than one might expect, as polyregular functions seem to be very good at bypassing the single-use restriction.
    For context, see the introduction to Chapter~\ref{ch:mealey}.

\end{description}

\bibliography{doktorat} 

\newcommand{\etalchar}[1]{$^{#1}$}
\begin{thebibliography}{BDM{\etalchar{+}}11}

\bibitem[A{\v{C}}11]{alur2011streaming}
Rajeev Alur and Pavol {\v{C}}ern{\'y}.
\newblock Streaming transducers for algorithmic verification of single-pass
  list-processing programs.
\newblock In {\em Proceedings of the 38th annual ACM SIGPLAN-SIGACT symposium
  on Principles of programming languages}, pages 599--610, 2011.

\bibitem[A{\v{C}e}10]{alur2010expressiveness}
Rajeev Alur and Pavol {\v{C}e}rn{\'y}.
\newblock Expressiveness of streaming string transducers.
\newblock In {\em Foundations of Software Technology and Theoretical Computer
  Science}, 2010.

\bibitem[AFR14]{alur2014combinators}
Rajeev Alur, Adam Freilich, and Mukund Raghothaman.
\newblock Regular combinators for string transformations.
\newblock In {\em Proceedings of the Joint Meeting of the Twenty-Third EACSL
  Annual Conference on Computer Science Logic (CSL) and the Twenty-Ninth Annual
  ACM/IEEE Symposium on Logic in Computer Science (LICS)}, CSL-LICS '14, New
  York, NY, USA, 2014. Association for Computing Machinery.

\bibitem[Asp98]{asperti1998light}
Andrea Asperti.
\newblock Light affine logic.
\newblock In {\em Proceedings. Thirteenth Annual IEEE Symposium on Logic in
  Computer Science (Cat. No. 98CB36226)}, pages 300--308. IEEE, 1998.

\bibitem[BC18]{bojanczyk2018automata}
Miko{\l}aj Boja{\'n}czyk and Wojciech Czerwi{\'n}ski.
\newblock Automata toolbox.
\newblock {\em URL:
  \url{https://www.mimuw.edu.pl/\~bojan/upload/reduced-may-25.pdf}}, 2018.

\bibitem[BDK18]{bojanczyk2018regular}
Miko{\l}aj Boja{\'n}czyk, Laure Daviaud, and Shankara~Narayanan Krishna.
\newblock Regular and first-order list functions.
\newblock In {\em Proceedings of the 33rd Annual ACM/IEEE Symposium on Logic in
  Computer Science}, pages 125--134, 2018.

\bibitem[BDM{\etalchar{+}}11]{bojanczyk2011two}
Miko\l{}aj Boja\'{n}czyk, Claire David, Anca Muscholl, Thomas Schwentick, and
  Luc Segoufin.
\newblock Two-variable logic on data words.
\newblock {\em ACM Trans. Comput. Logic}, 12(4), jul 2011.

\bibitem[BKM21]{bojanczyk2021orbit}
Miko{\l}aj Boja{\'n}czyk, Bartek Klin, and Joshua Moerman.
\newblock Orbit-finite-dimensional vector spaces and weighted register
  automata.
\newblock In {\em 2021 36th Annual ACM/IEEE Symposium on Logic in Computer
  Science (LICS)}, pages 1--13. IEEE, 2021.

\bibitem[BN{\^e}23]{bojaczyk2023algebraic}
Mikolaj Boja{\'n}czyk and L{\^e}~Th{\`a}nh~Dũng Nguy\~{\^e}n.
\newblock {Algebraic Recognition of Regular Functions}.
\newblock working paper or preprint, February 2023.

\bibitem[Boj13]{bojanczyk2013nominal}
Miko{\l}aj Boja{\'n}czyk.
\newblock Nominal monoids.
\newblock {\em Theory of Computing Systems}, 53(2):194--222, 2013.

\bibitem[Boj18]{bojanczyk2018polyregular}
Miko{\l}aj Boja{\'n}czyk.
\newblock Polyregular functions.
\newblock {\em arXiv preprint arXiv:1810.08760}, 2018.

\bibitem[Boj19]{bojanczyk2019slightly}
Miko{\l}aj Bojanczyk.
\newblock Slightly infinite sets.
\newblock {\em A draft of a book available at
  \url{https://www.mimuw.edu.pl/~bojan/paper/atom-book}}, 2019.

\bibitem[Boj20]{bojanczyk2020languages}
Miko{\l}aj Boja{\'n}czyk.
\newblock Languages recognised by finite semigroups, and their generalisations
  to objects such as trees and graphs, with an emphasis on definability in
  monadic second-order logic.
\newblock {\em arXiv e-prints}, pages arXiv--2008, 2020.

\bibitem[Boj22]{bojanczyk2022transducers}
Mikolaj Bojanczyk.
\newblock Transducers of polynomial growth.
\newblock In {\em Proceedings of the 37th Annual ACM/IEEE Symposium on Logic in
  Computer Science}, pages 1--27, 2022.

\bibitem[Boj23]{bojanczyk2023folding}
Miko{\l}aj Boja{\'n}czyk.
\newblock Folding interpretations.
\newblock {\em arXiv preprint arXiv:2301.05101}, 2023.

\bibitem[BS20]{single-use-paper}
Mikołaj Bojańczyk and Rafał Stefański.
\newblock Single use register automata for data words, 2020.

\bibitem[CE12]{courcelle2012graph}
Bruno Courcelle and Joost Engelfriet.
\newblock {\em Graph structure and monadic second-order logic: a
  language-theoretic approach}, volume 138.
\newblock Cambridge University Press, 2012.

\bibitem[CHL{\etalchar{+}}19]{chen2019decision}
Taolue Chen, Matthew Hague, Anthony~W. Lin, Philipp R\"{u}mmer, and Zhilin Wu.
\newblock Decision procedures for path feasibility of string-manipulating
  programs with complex operations.
\newblock {\em Proc. ACM Program. Lang.}, 3(POPL), jan 2019.

\bibitem[CLP15]{ley2015logics}
Thomas Colcombet, Clemens Ley, and Gabriele Puppis.
\newblock Logics with rigidly guarded data tests.
\newblock {\em Log. Methods Comput. Sci.}, 11(3), 2015.

\bibitem[Col07]{colcombet2007combinatorial}
Thomas Colcombet.
\newblock A combinatorial theorem for trees.
\newblock In {\em International Colloquium on Automata, Languages, and
  Programming}, pages 901--912. Springer, 2007.

\bibitem[Cou94]{courcelle1994monadic}
Bruno Courcelle.
\newblock Monadic second-order definable graph transductions: a survey.
\newblock {\em Theoretical Computer Science}, 126(1):53--75, 1994.

\bibitem[DFJL17]{dartois2017reversible}
Luc Dartois, Paulin Fournier, Isma{\"e}l Jecker, and Nathan Lhote.
\newblock On reversible transducers.
\newblock {\em arXiv preprint arXiv:1702.07157}, 2017.

\bibitem[DH19]{dolatian2019learning}
Hossep Dolatian and Jeffrey Heinz.
\newblock Learning reduplication with 2-way finite-state transducers.
\newblock In {\em International Conference on Grammatical Inference}, pages
  67--80. PMLR, 2019.

\bibitem[DJR18]{dartois2018aperiodic}
Luc Dartois, Isma{\"e}l Jecker, and Pierre-Alain Reynier.
\newblock Aperiodic string transducers.
\newblock {\em International Journal of Foundations of Computer Science},
  29(05):801--824, 2018.

\bibitem[DL09]{demri2009ltl}
St{\'e}phane Demri and Ranko Lazi{\'c}.
\newblock Ltl with the freeze quantifier and register automata.
\newblock {\em ACM Transactions on Computational Logic (TOCL)}, 10(3):1--30,
  2009.

\bibitem[Dub41]{dubreil1953contribution}
Paul Dubreil.
\newblock Contribution {\`a} la th{\'e}orie des demi-groupes. i.
\newblock {\em Mem. Acad. Sci. Paris}, 63:1--52, 1941.

\bibitem[EH01]{engelfriet2001mso}
Joost Engelfriet and Hendrik~Jan Hoogeboom.
\newblock Mso definable string transductions and two-way finite-state
  transducers.
\newblock {\em ACM Transactions on Computational Logic (TOCL)}, 2(2):216--254,
  2001.

\bibitem[Eil74]{eilenberg1974automata}
Samuel Eilenberg.
\newblock {\em Automata, languages, and machines}, volume~B.
\newblock Academic press, 1974.

\bibitem[EM63]{elgot1963two}
Calvin~C Elgot and Jorge~E Mezei.
\newblock Two-sided finite-state transductions.
\newblock In {\em Proceedings of the Fourth Annual Symposium on Switching
  Circuit Theory and Logical Design (swct 1963)}, pages 17--22. IEEE, 1963.

\bibitem[FR17]{filiot2017copyful}
Emmanuel Filiot and Pierre-Alain Reynier.
\newblock Copyful streaming string transducers.
\newblock In {\em Reachability Problems: 11th International Workshop, RP 2017,
  London, UK, September 7-9, 2017, Proceedings 11}, pages 75--86. Springer,
  2017.

\bibitem[Gir87]{girard1987linear}
Jean-Yves Girard.
\newblock Linear logic.
\newblock {\em Theoretical computer science}, 50(1):1--101, 1987.

\bibitem[GP02]{gabbay2002new}
Murdoch~J Gabbay and Andrew~M Pitts.
\newblock A new approach to abstract syntax with variable binding.
\newblock {\em Formal aspects of computing}, 13(3):341--363, 2002.

\bibitem[Gre51]{green1951structure}
James~A Green.
\newblock On the structure of semigroups.
\newblock {\em Annals of Mathematics}, pages 163--172, 1951.

\bibitem[Har72]{hartmanis1972non}
Juris Hartmanis.
\newblock On non-determinancy in simple computing devices.
\newblock {\em Acta Informatica}, 1(4):336--344, 1972.

\bibitem[Iba71]{ibarra1971characterizations}
Oscar~H Ibarra.
\newblock Characterizations of some tape and time complexity classes of turing
  machines in terms of multihead and auxiliary stack automata.
\newblock {\em Journal of Computer and System Sciences}, 5(2):88--117, 1971.

\bibitem[JL11]{jurdzinski2011alternating}
Marcin Jurdzi{\'n}ski and Ranko Lazi{\'c}.
\newblock Alternating automata on data trees and xpath satisfiability.
\newblock {\em ACM Transactions on Computational Logic (TOCL)}, 12(3):1--21,
  2011.

\bibitem[KF94]{kaminski_finite_memory_paper}
Michael Kaminski and Nissim Francez.
\newblock Finite-memory automata.
\newblock {\em Theoretical Computer Science}, 134(2):329--363, 1994.

\bibitem[KR65]{krohn1965prime}
Kenneth Krohn and John Rhodes.
\newblock Algebraic theory of machines. i. prime decomposition theorem for
  finite semigroups and machines.
\newblock {\em Transactions of the American Mathematical Society},
  116:450--464, 1965.

\bibitem[LKB14]{lasota2014automata}
S{\l}awomir Lasota, Bartek Klin, and Miko{\l}aj Boja{\'n}czyk.
\newblock Automata theory in nominal sets.
\newblock {\em Logical Methods in Computer Science}, 10, 2014.

\bibitem[Mea55]{mealy1955method}
George~H Mealy.
\newblock A method for synthesizing sequential circuits.
\newblock {\em The Bell System Technical Journal}, 34(5):1045--1079, 1955.

\bibitem[MSS{\etalchar{+}}17]{moerman2017learning}
Joshua Moerman, Matteo Sammartino, Alexandra Silva, Bartek Klin, and Micha{\l}
  Szynwelski.
\newblock Learning nominal automata.
\newblock In {\em Proceedings of the 44th ACM SIGPLAN Symposium on Principles
  of Programming Languages}, pages 613--625, 2017.

\bibitem[NDRP20]{nelson2020probing}
Max Nelson, Hossep Dolatian, Jonathan Rawski, and Brandon Prickett.
\newblock Probing rnn encoder-decoder generalization of subregular functions
  using reduplication.
\newblock In {\em Proceedings of the Society for Computation in Linguistics
  2020}, pages 167--178, 2020.

\bibitem[N{\^e}21]{nguyen2021automates}
L{\^e}~Th{\`a}nh~Dũng Nguy\~{\^e}n.
\newblock {\em {Automates implicites en logique lin{\'e}aire et th{\'e}orie
  cat{\'e}gorique des transducteurs}}.
\newblock Theses, {Universit{\'e} Paris-Nord - Paris XIII}, December 2021.

\bibitem[NSV04]{neven2004finite}
Frank Neven, Thomas Schwentick, and Victor Vianu.
\newblock Finite state machines for strings over infinite alphabets.
\newblock {\em ACM Transactions on Computational Logic (TOCL)}, 5(3):403--435,
  2004.

\bibitem[Pin10]{pin2010mathematical}
Jean-{\'E}ric Pin.
\newblock Mathematical foundations of automata theory.
\newblock {\em Lecture notes LIAFA, Universit{\'e} Paris}, 7, 2010.

\bibitem[Pit13]{pitts2013nominal}
Andrew~M. Pitts.
\newblock {\em Nominal Sets: Names and Symmetry in Computer Science}.
\newblock Cambridge Tracts in Theoretical Computer Science. Cambridge
  University Press, 2013.

\bibitem[Sch55]{schutzenberger1955theorie}
Marcel-Paul Sch{\"u}tzenberger.
\newblock Une th{\'e}orie alg{\'e}brique du codage.
\newblock {\em S{\'e}minaire Dubreil. Algebre et th{\'e}orie des nombres},
  9:1--24, 1955.

\bibitem[She59]{shepherdson1959reduction}
John~C Shepherdson.
\newblock The reduction of two-way automata to one-way automata.
\newblock {\em IBM Journal of Research and Development}, 3(2):198--200, 1959.

\bibitem[Sim90]{simon1990factorization}
Imre Simon.
\newblock Factorization forests of finite height.
\newblock {\em Theoretical Computer Science}, 72(1):65--94, 1990.

\bibitem[Ste18]{stefanski2018automaton}
Rafa{\l} Stefa{\'n}ski.
\newblock An automaton model for orbit-finite monoids.
\newblock Master's thesis, University of Warsaw - Faculty of Mathematics,
  Informatics and Mechanics, 2018.

\bibitem[UMB23]{urbat2023nominal}
Henning Urbat, Stefan Milius, and Fabian Birkmann.
\newblock Nominal topology for data languages.
\newblock {\em arXiv preprint arXiv:2304.13337}, 2023.

\bibitem[Wad89]{wadler1989theorems}
Philip Wadler.
\newblock Theorems for free!
\newblock In {\em Proceedings of the fourth international conference on
  Functional programming languages and computer architecture}, pages 347--359,
  1989.

\bibitem[Wad90]{wadler1990linear}
Philip Wadler.
\newblock Linear types can change the world!
\newblock In {\em Programming concepts and methods}, volume~3, page~5.
  Citeseer, 1990.

\end{thebibliography}
\bibliographystyle{alpha}

\end{document}